\documentclass[12pt,a4paper]{amsart}
\makeatletter
\renewcommand\normalsize{%
    \@setfontsize\normalsize{11.7}{14pt plus .3pt minus .3pt}%
    \abovedisplayskip 10\p@ \@plus4\p@ \@minus4\p@
    \abovedisplayshortskip 6\p@ \@plus2\p@
    \belowdisplayshortskip 6\p@ \@plus2\p@
    \belowdisplayskip \abovedisplayskip}
\renewcommand\small{%
    \@setfontsize\small{9.5}{12\p@ plus .2\p@ minus .2\p@}%
    \abovedisplayskip 8.5\p@ \@plus4\p@ \@minus1\p@
    \belowdisplayskip \abovedisplayskip
    \abovedisplayshortskip \abovedisplayskip
    \belowdisplayshortskip \abovedisplayskip}
\renewcommand\footnotesize{%
    \@setfontsize\footnotesize{8.5}{9.25\p@ plus .1pt minus .1pt}
    \abovedisplayskip 6\p@ \@plus4\p@ \@minus1\p@
    \belowdisplayskip \abovedisplayskip
    \abovedisplayshortskip \abovedisplayskip
    \belowdisplayshortskip \abovedisplayskip}
\ifdefined\pdfpagewidth
\else
\fi
\calclayout
\makeatother

\usepackage{xcolor}
\usepackage{amsfonts,amsmath,amssymb,mathtools,amsthm}
\usepackage{graphicx}
\usepackage{stmaryrd}
\usepackage{thmtools}
\DeclareGraphicsExtensions{.eps}

\numberwithin{equation}{section} 
\renewcommand{\theequation}{\thesection.\arabic{equation}}

\renewcommand{\thesection}{\arabic{section}}
\makeatletter
\@addtoreset{equation}{section}
\makeatother
\newtheorem{theorem}{Theorem}[section]

\newtheorem{proposition}{Proposition}[section]
\newtheorem{lemma}{Lemma}[section]
\newtheorem{corollary}{Corollary}[section]

\theoremstyle{definition}
\newtheorem{remark}{Remark}[section]
\newtheorem{definition}{Definition}[section]

\def\br{\begin{remark}\rm\small}
\def\er{\end{remark}}
\def\bt{\begin{theorem}}
\def\et{\end{theorem}}
\def\bd{\begin{definition}}
\def\ed{\end{definition}}
\def\bp{\begin{proposition}}
\def\ep{\end{proposition}}
\def\bl{\begin{lemma}}
\def\el{\end{lemma}}
\def\bc{\begin{corollary}}
\def\ec{\end{corollary}}
\def\beaq{\begin{eqnarray}}
\def\eeaq{\end{eqnarray}}

\theoremstyle{definition}

\newcommand{\be}{\begin{equation}}
\newcommand{\ee}{\end{equation}}
\newcommand{\beq}{\begin{equation}}
\newcommand{\eeq}{\end{equation}}
\newcommand{\bea}{\begin{eqnarray}}
\newcommand{\eea}{\end{eqnarray}}
\newcommand{\beqq}{\begin{equation*}}
\newcommand{\eeqq}{\end{equation*}}
\newcommand{\beaa}{\begin{eqnarray*}}
\newcommand{\eeaa}{\end{eqnarray*}}

\newcommand{\Tr}{{\operatorname {Tr}}}

\newcommand{\diag}{{\operatorname{diag}}}

\newcommand{\om}{\omega}

\newcommand{\td}{\tilde}

\newcommand\blfootnote[1]{%
  \begingroup
  \renewcommand\thefootnote{}\footnote{#1}%
  \addtocounter{footnote}{-1}%
  \endgroup
}

\newcommand{\Res}{\mathop{\,\rm Res\,}}
\usepackage[pdftex,hypertexnames=false]{hyperref}
\hypersetup{colorlinks,urlcolor=magenta,citecolor=red,linkcolor=blue,filecolor=black}

\title{\bf{Hamiltonian representation of isomonodromic deformations of general rational connections on $\mathfrak{gl}_2(\mathbb{C})$}}

\date{\vspace{-5ex}}
 
\author{$_{1}$Olivier Marchal\footnote{Universit\'{e} Jean Monnet Saint-\'{E}tienne, CNRS, Institut Camille Jordan UMR 5208, Institut Universitaire de France, F-42023, Saint-\'{E}tienne, France}\,\,,
$_{2}$Nicolas Orantin\footnote{Universit\'e de Gen\`eve, Section de math\'ematiques, 24 rue du G\'en\'eral Dufour, 1211 Gen\`eve 4, Suisse} \,\,,
$_{3}$Mohamad Alameddine\footnote{Universit\'{e} Jean Monnet Saint-\'{E}tienne, CNRS, Institut Camille Jordan UMR 5208, F-42023, Saint-\'{E}tienne, France.}
}

\begin{document}

\begin{center}
\huge{Hamiltonian representation of isomonodromic deformations of general rational connections on $\mathfrak{gl}_2(\mathbb{C})$}
\end{center}
\vspace{0.5cm}
\begin{center}
$_{1}$Olivier Marchal\footnote{Universit\'{e} Jean Monnet Saint-\'{E}tienne, CNRS, Institut Camille Jordan UMR 5208, Institut Universitaire de France, F-42023, Saint-\'{E}tienne, France}\,\,,
$_{2}$Nicolas Orantin\footnote{Universit\'e de Gen\`eve, Section de math\'ematiques, 24 rue du G\'en\'eral Dufour, 1211 Gen\`eve 4, Suisse} \,\,,
$_{3}$Mohamad Alameddine\footnote{Universit\'{e} Jean Monnet Saint-\'{E}tienne, CNRS, Institut Camille Jordan UMR 5208, F-42023, Saint-\'{E}tienne, France.}
\end{center}

\vspace{1.0cm}

\textbf{Abstract}: In this paper, we construct the Hamiltonian systems attached to generic $\mathfrak{gl}_2(\mathbb{C})$ meromorphic connections with an arbitrary number of unramified poles of arbitrary orders, under the assumption that every leading polar coefficient is regular semisimple. In particular, we propose the Lax pairs and Hamiltonian evolutions expressed in terms of irregular times and monodromies associated to the poles as well as $g$ pairs of Darboux coordinates defined as the apparent singularities arising in the oper gauge. Moreover, we also provide a reduction of the isomonodromic deformations to a subset of $g$ non-trivial isomonodromic deformations. This reduction is equivalent to a map reducing the set of irregular times to only $g$ non-trivial isomonodromic times. We apply our construction to all genus-one cases covered by these hypotheses and recover the standard Painlev\'{e} equations $2$--$6$. We finally make the connection with the topological recursion and the quantization of classical spectral curves from this perspective.

\blfootnote{\textit{Email Addresses:}$_{1}$\textsf{olivier.marchal@univ-st-etienne.fr} (Corresponding author), $_{2}$\textsf{nicolas.orantin@gmail.com}, $_{3}$\textsf{mohamad.alameddine@univ-st-etienne.fr}}

\section*{Statements and Declarations}
The authors declare that they have no conflict of interest.\\
Data sharing is not applicable to this article as no datasets were generated or analyzed during the current study.




\newpage

\section{Introduction}
The isomonodromic deformations of meromorphic connections on vector bundles over Riemann surfaces remain a broad subject of modern mathematics and have undergone many developments since the pioneering works of E. Picard, L. Fuchs, R. Fuchs, P. Painlev\'{e}, R. Garnier and L. Schlesinger \cite{Picard,Fuchs,Garnier,Painleve,Gambier,schlesinger1912klasse}. Fuchsian singularities in $\mathfrak{gl}_2(\mathbb{C})$ were first studied by R. Fuchs and led to the analysis of the Painlev\'{e} $6$ equation. These works, their links with the Painlev\'{e} property and with isomonodromic deformations of rational connections on $\mathfrak{gl}_2(\mathbb{C})$ with only simple poles were later pursued by R. Garnier and K. Okamoto \cite{Okamoto1986Iso,Okamoto1986} for arbitrary Fuchsian systems in $\mathfrak{gl}_2(\mathbb{C})$, i.e. to Garnier systems in their scalar version or to Schlesinger systems \cite{schlesinger1912klasse} in their matrix form. Nowadays, connections with Fuchsian singularities and their underlying Hamiltonian systems are well-understood and we refer to  Chap. $3$ of \cite{FromGaussToPainleve} for a review on the subject. Another significant result obtained by R. Garnier is the generalization to the other five Painlev\'{e} equations proving that these equations may be obtained by complete integrability conditions. But it was J. Malmquist \cite{Malmquist1922} who pointed out for the first time that all Painlev\'{e} equations can be written as Hamiltonian systems, while the relations with isomonodromic deformations of linear ordinary differential equations with irregular singularities were given in \cite{Okamoto1980}. Other interesting works related to irregular Garnier systems, obtained by confluence of simple poles, can also be found in the literature \cite{Kimura1989,Kawamuko1,Kawamuko2,Diarra_Loray_2020,Komyo2022ANA,Mazzocco2004}. However, the case of general isomonodromic deformations of linear ordinary differential equations with irregular singularities is much more complicated than the case of Fuchsian singularities and is still actively studied. For example, seminal contributions to the theory were made by the Japanese school of M. Jimbo, T. Miwa, and K. Ueno \cite{JimboMiwaUeno,JimboMiwa} in the case of generic (in the sense that the leading term of the connections at each pole has distinct eigenvalues) singularity structures with poles of arbitrary order on bundles of arbitrary rank. The vast family of nonlinear differential equations resulting from the isomonodromic deformations is the largest known to have the  ``Painlev\'{e} property''. Of special interest are the Painlev\'{e} equations that have been studied in various ways: these equations are required to preserve a set of generalized monodromy data, which is obtained by considering Stokes data around the irregular singularities in addition to the usual monodromies. The space of deformations itself needs to be extended beyond the space of complex structures of the Riemann surface by including the irregular type of the monodromy considered \cite{BoalchThesis,Boalch2001,Boalch2014,BoalchYamakawa,BertolaMo2005}. There are currently many ways to view and to approach the problem. Indeed, it has been known since the work of K. Okamoto \cite{okamoto1979deformation} in 1979 that isomonodromy equations may be rewritten using Hamiltonian formulations. They may also be regarded as a natural extension of the symplectic structures investigated by Narasimhan and Seshadri \cite{NarasimhanSeshadri}, Atiyah and Bott \cite{AtiyahBott}, Goldman \cite{Goldman84} and Fock and Rosly \cite{FockRosly99} on spaces of flat connections on Riemann surfaces to the world of isomonodromic systems - a perspective pursued by P. Boalch \cite{Boalch2001,Boalch2012,Boalch2022}. 
There exists another perspective, using moment map embeddings and (central extensions of) loop algebras, a viewpoint introduced by J. Harnad \cite{Harnad:1993hw}, J. Hurtubise \cite{HarnadHurtubise}. By considering data associated with families of Riemann surfaces branched over the singularities, one may also consider the isomonodromy equations as non-Abelian Gauss-Manin connections.
This observation was made by P. Boalch \cite{Boalch2001} leading to the notion of wild non-Abelian Hodge theory \cite{biquard_boalch_2004,maninfrobenius} of central importance in Seiberg-Witten theory and theoretical physics and providing a modern viewpoint on isomonodromic equations as non-Abelian Gauss-Manin connections for families of wild Riemann surfaces \cite{ManinBook,ManinArticle,Yamakawa3,Boalch2014872,Boalch2014,BoalchYamakawa,Boalch2022}.

\medskip

Due to the many possible perspectives and the relations with the Painlev\'{e} property and transcendents, isomonodromic deformations have an extremely wide range of applications in mathematical physics. Without being exhaustive, let us mention that they play an important role in the study of random matrix theory and provide generating functions for moduli spaces of two-dimensional topological quantum field theories and the study of quantum cohomology and Gromov-Witten invariants. 
In addition, a renewed interest in the study of these isomonodromic systems in the mathematical physics community was motivated by the discovery of a correspondence between conformal blocks of conformal field theories and isomonodromic tau functions \cite{GIL-CFT,CGL-CFT} and irregular Higgs bundles.  Moreover, the recent development of exact WKB theory \cite{IwakiNakanishi,Iwaki-P1} in relation with the Painlev\'{e} equations and topological recursion \cite{EO07} also offers new insights on the subject.

\medskip

The non-autonomous Hamiltonian nature of isomonodromic deformations has been uncovered using the spectral invariants of a corresponding Hitchin system \cite{Harnad:1993hw,Hitchin-Ham,Boalch2001,Krichever2002}.
Although the geometry of the Hamiltonian representation of the isomonodromic equations for a generic meromorphic connection is now well understood \cite{HURTUBISE20081394,BertolaHarnadHurtubise2022}, an explicit expression for the Hamiltonians has generally been derived on a case-by-case basis. It is believed that any autonomous Hamiltonian encountered in the Hitchin system, hence corresponding to an isospectral deformation,  can be turned into a Hamiltonian corresponding to an isomonodromic deformation by making it explicitly dependent on the Casimirs of the system. However, this procedure is easy to implement in simple examples such as Painlev\'e equations or Fuchsian systems but finding how to de-autonomize any Hamiltonian for generic irregular connections is a more complicated task. Recently Gaiur, Mazzocco and Rubtsov \cite{MartaPaper2022} solved this problem for isomonodromic deformations arising from confluences of isomonodromic deformations of Fuchsian systems and we discuss the comparison of residue-trace formulas and the role of gauge choices in Remark \ref{RemarkHInTermsOfTraceSquared}.

\medskip

In this article, we contribute to the subject of isomonodromic deformations by proposing an explicit expression for the Hamiltonian systems in terms of Darboux coordinates corresponding to apparent singularities. Our construction applies to regular or irregular unramified poles of arbitrary orders in $\mathfrak{gl}_2(\mathbb{C})$, with regular semisimple leading polar coefficients, on the generic chart specified in Section \ref{Section2Mero}. We also give an explicit expression of the corresponding Lax pairs whose compatibility equations provide the Hamiltonian systems. In other words, this article extends the unramified cases of the Lax pairs proposed by Jimbo-Miwa (Appendix C of \cite{JimboMiwa}) to $\mathfrak{gl}_2(\mathbb{C})$ meromorphic connections satisfying these hypotheses, and generalizes Schlesinger Hamiltonians to non-Fuchsian singularities in $\mathfrak{gl}_2(\mathbb{C})$. Our approach consists in using the geometric knowledge at each pole, described by the irregular times (sometimes also referred to as ``KP times'' or ``spectral times'') to provide a first natural space of isomonodromic deformations. Building the corresponding Lax pairs and writing the compatibility equations in terms of the spectral Darboux coordinates given by the apparent singularities $\left(q_i\right)_{1\leq i\leq g}$ and their dual coordinates $\left(p_i\right)_{1 \leq i\leq g}$ allows us to derive the evolutions with respect to this set of coordinates and prove that these evolutions are indeed Hamiltonian (Theorem \ref{HamTheorem}). The present work complements \cite{MartaPaper2022} by expressing the
full unreduced $\mathfrak{gl}_2$ deformation equations in spectral
Darboux coordinates, including the scalar trace directions.
For the trace-free rank-two problem, the confluence parameters of
\cite{MartaPaper2022} already give the full number of irregular
deformation parameters; comparing the two constructions requires a
local change of time coordinates and compatible gauge choices. 
In the second part of the paper, we show that the initial space of isomonodromic deformations can be reduced to only $g$ non-trivial deformations thus giving a minimal set of non-trivial isomonodromic times and isomonodromic deformations providing a compatible nonautonomous Hamiltonian system. This reduction of the tangent space provides an explicit map between the irregular times and the non-trivial isomonodromic times as well as simplifications for the Hamiltonians (Theorem \ref{HamTheoremReduced}) for our choice of non-trivial isomonodromic times. Finally, in order to illustrate our method, we present several examples: the Painlev\'{e} equations and Chiba's genus-two system $(P_{\mathrm{II}-2})_2$ recovering the standard results of Jimbo-Miwa \cite{JimboMiwa} and H. Chiba \cite{Chiba}.

\medskip

Recent works of P. Boalch and D. Yamakawa \cite{Boalch2001,BoalchYamakawa,Yamakawa2017TauFA,Yamakawa2019FundamentalTwoForms} are also closely related to the present article. Indeed, in \cite{Boalch2001} P. Boalch proved that the isomonodromy equations of Jimbo-Miwa-Ueno constitute a flat Ehresmann connection on a non-linear fibre bundle and then constructed a symplectic structure on each fibre to prove afterwards that the parallel transport preserves the symplectic structure of the fibres providing a symplectic connection. Moreover, in \cite{Yamakawa2017TauFA}, the author gives a complete flat symplectic Ehresmann connection on the total space of deformation parameters and provides a completely integrable non-autonomous Hamiltonian system associated to isomonodromic deformations. This result is complemented (in the much more general context of non-resonant isomonodromic deformations of meromorphic connections on the trivial principal $G$-bundle over $\mathbb{P}^1$, where $G$ is any complex reductive group) in \cite{Yamakawa2019FundamentalTwoForms} with an explicit description of the fundamental two-forms for isomonodromy equations. As we will see below, we shall recover many features of these results for $G=GL_2(\mathbb{C})$. In particular, we shall provide explicit formulas for the Hamiltonians associated to general isomonodromic deformations and provide a decomposition of the deformation space in which the fundamental two-form may be reduced (Theorem \ref{MainTheotau0}). In other
words, we propose (see Theorem \ref{TheoremCorrespondence}) an explicit birational map between the Jimbo-Miwa-Ueno/Boalch symplectic isomonodromy connection and the symplectic Ehresmann connection built from our Hamiltonians by providing explicit formulas for the time-dependent Hamiltonian systems characterizing the symplectic Ehresmann connection and explicit formulas for the Lax pairs characterizing the Jimbo-Miwa-Ueno/Boalch connection. Our construction recovers the Jimbo-Miwa-Ueno/Boalch connection since our Lax matrix $\td{L}\in \hat{F}_{\mathcal{R},\mathbf{r}}$ is built from Birkhoff local diagonalizations at each pole (characterizing the base $B$ of times $\mathbf{t}$) and solving the isomonodromy compatibility equations $\partial_{t_k} \td{L}= \partial_\lambda \td{A}_{t_k}+\left[\td{A}_{t_k},\td{L}\right]$ using a specific set of Darboux coordinates $(\mathbf{q},\mathbf{p})$ corresponding to apparent singularities $\mathbf{q}$ and their dual $\mathbf{p}$ on the spectral curve  ($\det(p_i I_2-\td{L}(q_i))=0$). In particular, we provide (see Definition \ref{DefinitionSymplecticForm}) the explicit expression of the associated fundamental two-form $\Omega$ defined by Yamakawa \cite{Yamakawa2019FundamentalTwoForms}
\beqq \Omega=\sum_{i} d q_i\wedge dp_i +\sum_{k}  d \text{Ham}_k(\mathbf{q},\mathbf{p},\mathbf{t})\wedge d t_k\eeqq
through the explicit expression of the time-dependent Hamiltonians in Theorem \ref{HamTheorem}.

Note that our results differ from those of \cite{Yamakawa2017TauFA,Yamakawa2019FundamentalTwoForms} since the strategy is different. Indeed, D. Yamakawa uses the local diagonal gauge to derive his formulas. This strategy is natural from the geometric interpretation of isomonodromic deformations developed by P. Boalch \cite{Boalch2001} or to make connections with isospectral deformations. More precisely, D. Yamakawa first defines the isospectral Hamiltonians using the standard intrinsic residue-trace formulas at each pole. He then imposes some conditions on the exterior derivatives with respect to irregular times (Lemma $4.1$ of \cite{Yamakawa2017TauFA}), i.e. implicitly selects special Darboux coordinates, so that the isospectral Hamiltonians match with the isomonodromic Hamiltonians. This strategy consisting in using special Darboux coordinates to identify the isospectral invariants with the isomonodromic Hamiltonians for non-Fuchsian singularities has also been used in \cite{Darboux_coord93,BertolaHarnadHurtubise2022,Chiba}. However, if the construction is geometrically very interesting and efficient, it does not provide explicit Darboux coordinates nor expressions for the Lax matrices and it requires substantial computations to apply (because one needs to compute all special local gauge matrices and connect them together) on examples. In our case, the existence of local diagonalizations remains the starting point and provides the natural framework for isomonodromic deformations. Our strategy is to pass to the oper (companion) gauge determined by the
chosen cyclic component of the normalized representative, on the
generic chart where this component is cyclic. It turns out that there are natural Darboux coordinates in this gauge, namely the apparent singularities and their dual partners on the spectral curve, that we shall use to express the Lax matrices and the Hamiltonian evolutions following the standard and historical approach of Garnier and Jimbo-Miwa-Ueno \cite{Garnier1912,Garnier,JimboMiwaUeno}. In fact, the two strategies may be seen as solutions to the same problem, namely that for non-Fuchsian singularities isospectral invariants do not generally match with isomonodromic Hamiltonians. Therefore a first strategy is to select the Darboux coordinates specifically so that the identification remains valid \cite{BertolaHarnadHurtubise2022,Yamakawa2017TauFA}. The strategy used here is to choose apparent-singularity Darboux coordinates and derive the isomonodromic Hamiltonians directly. Theorem \ref{HamTheoremReduced} expresses them through accessory coefficients whose trace-residue reconstruction permits a comparison with \cite{MartaPaper2022}, subject to the gauge and coordinate distinctions in Remark \ref{RemarkHInTermsOfTraceSquared}. Let us finally mention that these two strategies have been merged for the $\mathfrak{gl}_2$ case in  \cite{MarchalAlameddineIsospectralIsomono2023} during the final preparation of this paper by providing a time-dependent and birational change of coordinates relating both sets of Darboux coordinates. Up to this change of coordinates, the non-autonomous Hamiltonians and the fundamental symplectic two-form (Theorems \ref{MainTheotau0} and \ref{HamTheoremReduced}) computed in this paper should identify with those of Yamakawa \cite{Yamakawa2017TauFA,Yamakawa2019FundamentalTwoForms}.

Note also that the oper gauge in which computations are carried out in this article, is equivalent to the so-called ``quantum curve''  or scalar form associated to the meromorphic connection and is currently a key feature associated to the quantization of moduli spaces. It can be seen as a generalization of the reduction of Schlesinger systems to Garnier systems in the presence of irregular singularities. The choice of apparent singularities as (half of) Darboux coordinates is also canonical since it extends the works of Jimbo-Miwa-Ueno and H. Chiba. Finally, our formulas allow for an immediate use in practice without additional computation and we are able to recover the Painlev\'{e} cases in a few lines. In the end, we believe that the two strategies are complementary and have their own interests. Although it is beyond the scope of this article to compare both approaches in detail, we shall discuss this question in Section \ref{SectionOutlooks}. 

Moreover, let us also stress that one of the main interests of the method developed in this article is that it bypasses the isospectral deformations setup that was used in previous papers to handle isomonodromic deformations. Nevertheless, our main results (Theorem \ref{HamTheoremReduced}) combined with those of \cite{BertolaHarnadHurtubise2022} indicate that both formalisms are related in a very explicit way that has been very recently investigated in \cite{MarchalAlameddineIsospectralIsomono2023}.

\medskip

To sum up, the main results obtained in this article are the following:
\begin{itemize} \item A general expression of the Lax pairs in terms of apparent singularities $\left(q_i\right)_{1\leq i\leq g}$ and their dual coordinates $\left(p_i\right)_{1\leq i\leq g}$ in Propositions \ref{PropLaxMatrix}, \ref{PropA12Form} and \ref{Propcalpha}.
\item \sloppy{A general expression of the evolutions of the Darboux coordinates $\left(q_i,p_i\right)_{1\leq i\leq g}$ under any irregular time or any position of the poles in Theorem \ref{HamTheorem} and a proof that these evolutions are Hamiltonian with explicit expressions for the corresponding Hamiltonians. This provides an explicit expression of the fundamental two-form $\Omega$ in a canonical form in Definition \ref{DefinitionSymplecticForm}.}
\item The explicit expressions of the Lax matrices and the Hamiltonians give a birational map between the symplectic Ehresmann connection and the Jimbo-Miwa-Ueno/Boalch symplectic isomonodromy connection in Theorem \ref{TheoremCorrespondence}.
\item A reduction of the full space of isomonodromic deformations to a subspace of non-trivial deformations (of dimension $g$) preserving the fundamental symplectic two-form $\Omega$ (Theorem \ref{MainTheotau0}). This reduction is equivalent to a map between the set of irregular times and locations of the poles to a set of trivial and non-trivial isomonodromic times that is provided in Definitions \ref{DefTrivialrgeq3}, \ref{DefTrivialrequal2}, \ref{DefTrivialrequal1}, \ref{DefTrivialrequal1n1} depending on the degree of the pole at infinity and the number of poles. In particular, this reduction provides a compatible nonautonomous Hamiltonian system with $g$ non-trivial isomonodromic times.
\item Simple formulas of the Hamiltonians in terms of the non-trivial isomonodromic times after a canonical choice of the trivial times are provided in Section \ref{SectionMainResult}. In particular, we obtain that the Hamiltonians are (time-dependent) linear combinations of the accessory coefficients (which are independent of the deformation direction) in Theorem \ref{HamTheoremReduced}. The coefficients of these linear combinations are explicit. Remark \ref{RemarkHInTermsOfTraceSquared} discusses the comparison with \cite{MartaPaper2022}.
\item Application of the general construction to the Painlev\'{e} equations and Fuchsian systems is presented in Section \ref{SectionExemples}. Moreover, the case of Chiba's genus-two system $(P_{\mathrm{II}-2})_2$ (which is a case of genus $g=2$) is also presented in Section \ref{SectionSecondElementP2} and recovers results of H. Chiba \cite{Chiba}.
\item The connection with the quantization of classical spectral curves via topological recursion is presented as a by-product in Section \ref{SectionTR}.
\end{itemize}

The paper is meant to be self-contained and all details of each proof are presented in appendices for completeness. Maple files for each of the examples are available at \url{http://math.univ-lyon1.fr/~marchal/AdditionalRessources/index.html} and can be used as a check for the general formulas presented in this article.

\section{Meromorphic connections, gauges and Darboux coordinates} \label{Section2Mero}

\subsection{Meromorphic connections and irregular times}

The space of $\mathfrak{gl}_2(\mathbb{C})$ meromorphic connections has been studied from many different perspectives. In the present article, we shall mainly follow the point of view of the Montr\'{e}al group \cite{Darboux_coord93,HarnadHurtubise1997} together with some insights from the work of P. Boalch \cite{Boalch2001}. Let us first define the space we shall study.

\begin{definition}[Space of rational connections]
Let $n \in \mathbb{N}$ and $\{X_i\}_{i=1}^n$ be $n$ distinct points in the complex plane. Let us denote $\mathcal{R} := \{\infty,X_1,\dots,X_n\}$ and $\mathcal{R}_0 := \{X_1,\dots,X_n\}$. For any $\mathbf{r}:=(r_\infty,r_1,\dots,r_n)\in \left(\mathbb{N}\setminus\{0\}\right)^{n+1}$, let us define
\small{\beq
F_{\mathcal{R},\mathbf{r}} := \left\{\hat{L}(\lambda) = \sum_{k=1}^{r_\infty-1} \hat{L}^{[\infty,k]} \lambda^{k-1} + \sum_{s=1}^n \sum_{k=0}^{r_s-1} \frac{\hat{L}^{[X_s,k]}}{(\lambda-X_s)^{k+1}}
\,\, \text{ with }\,\, \{\hat{L}^{[p,k]}\} \in \left(\mathfrak{gl}_2\right)^{r-1}\right\}/\text{GL}_2(\mathbb{C}) 
\eeq}
\normalsize{where} $r = r_\infty + \underset{s=1}{\overset{n}{\sum}} r_s$ and $\text{GL}_2(\mathbb{C})$ acts simultaneously by conjugation on all coefficients $\{\hat{L}^{[p,k]}\}_{p,k}$.

Laurent coefficients outside their stated index ranges are understood
to be zero; for example, $\hat L^{[X_s,1]}=0$ when $r_s=1$.
The same convention applies to the normalized representatives below.
Empty sums are omitted. A term multiplied by an indicator
$\delta_{\mathcal C}$ is included only when the condition
$\mathcal C$ holds, before evaluating its indexed factors.

Let $\hat F_{\mathcal R,\mathbf r}$ denote the locus on which the
leading polar coefficient at every pole is regular semisimple. We label
its two eigenvalues. Lower Laurent coefficients and lower irregular times
are allowed to coincide on the two sheets. At a regular singular point
we additionally assume nonresonance whenever a log-free diagonal local
normal form is used: the difference of the two residue eigenvalues is
not an integer (after introducing $\hbar\ne0$, divide that difference
by $\hbar$). The normalization and Darboux coordinates below are used on
the generic open chart where their denominators and normalization
coefficients are nonzero.
\end{definition}

\begin{remark}In the present article, we shall assume that $\infty$ is always a pole following the standard convention. Of course, one may always use a change of coordinates in order to remove this assumption.   
\end{remark}

$F_{\mathcal{R},\mathbf{r}}$ can be given a Poisson structure  inherited from the Poisson structure of a corresponding loop algebra \cite{Harnad:1993hw,HarnadHurtubise,woodhouse2007duality}. On the generic locus where the simultaneous stabilizer consists only of scalar matrices, its dimension is
\beq
\dim F_{\mathcal{R},\mathbf{r}} = 4r-7.
\eeq	
The one-coefficient case $r=2$ is excluded from this dimension
count: its generic conjugacy quotient has dimension two.
The space $F_{\mathcal{R},\mathbf{r}}$ has been intensively studied from the point of view of isospectral and isomonodromic deformations. Following the works of P. Boalch \cite{Boalch2001,Boalch2012,Boalch2022} and of D. Yamakawa \cite{Yamakawa2017TauFA,Yamakawa2019FundamentalTwoForms}, one can use the Poisson structure on $F_{\mathcal{R},\mathbf{r}}$ in order to describe it as a bundle whose fibers are symplectic leaves obtained by fixing the irregular type and monodromies of $\hat{L}(\lambda)$. Let us briefly review this perspective and use it to define local coordinates on $F_{\mathcal{R},\mathbf{r}}$ trivializing the fibration.

In this article, \textbf{we shall restrict to $\hat{F}_{\mathcal{R},\mathbf{r}}$ to simplify the presentation} but the present setup may be adapted to $F_{\mathcal{R},\mathbf{r}}$ using ramified covers and Puiseux series.\footnote{In particular, the case of the Painlev\'{e} $1$ equation and hierarchy for which $\check{L}^{[\infty,r_\infty-1]}$ cannot be diagonalized has been investigated in \cite{MarchalAlameddineP1Hierarchy2023} and provides results similar to the ones presented in this article, in other words, this article deals with the \textit{generic} case only.}

\medskip

For any pole $p \in \mathcal{R}$, let us define a local coordinate
\beq
\forall\, p \in \mathcal{R} \, , \; 
z_p(\lambda):= \left\{ \begin{array}{l}
(\lambda-p) \qquad \text{if} \qquad p \in \{X_1,\dots,X_n\} \cr
\lambda^{-1} \qquad \text{if} \qquad p = \infty. \cr
\end{array}
\right.
\eeq

Given $\hat{L}(\lambda)$ in an orbit of $\hat{F}_{\mathcal{R},\mathbf{r}}$, let $\hat{\Psi}(\lambda)$ be a wave matrix solution to the linear differential system
\beqq
\partial_\lambda \hat{\Psi}(\lambda) = \hat{L}(\lambda)  \hat{\Psi}(\lambda).
\eeqq
For each pole $p\in\mathcal R$, the regular-semisimple formal
classification, with the preceding nonresonance hypothesis at regular
singular points, gives a formal gauge
$G_p\in\mathrm{GL}_2(\mathbb C[[z_p]])$ and a local formal fundamental
matrix in the following normal form. These formal series need not
converge at irregular singularities; analytic realizations are sectorial.
\footnotesize{\beq \label{Birkhoff} \Psi_p(\lambda)=\Psi_p^{(\text{reg})}(\lambda)\, \diag\left(\exp\left(- \sum_{k=1}^{r_p-1}\frac{t_{p^{(1)},k}}{k z_p(\lambda)^{\,k}}  +t_{p^{(1)},0}\ln z_p(\lambda)\right) , \exp\left(- \sum_{k=1}^{r_p-1}\frac{t_{p^{(2)},k}}{k z_p(\lambda)^{\,k}}  +t_{p^{(2)},0}\ln z_p(\lambda)\right) \right) 
\eeq} 
\normalsize{Here} $\Psi_p^{(\mathrm{reg})}$ is an invertible formal
regular factor. The exact diagonal normal form can be chosen with this
factor constant, by incorporating it into $G_p$. In that convention the
following identity is an identity of formal connections:
\beq \label{DiagoCondition}
G_p \,\hat{L}(\lambda) d\lambda\, G_p^{-1} + (\partial_\lambda G_p)G_p^{-1}d\lambda = dQ_p(z_p) + \Lambda_p \frac{dz_p}{z_p} 
 \,\,\text{where}\,\, Q_p(z_p) = \sum_{k=1}^{r_p-1} \frac{Q_{p,k}}{z_p^{\, k}}
\eeq 
with 
\beq 
Q_{p,k} =  \diag \left(- \frac{t_{p^{(1)},k}}{k},- \frac{t_{p^{(2)},k}}{k}\right) \,\,\text{ and }\, 
\Lambda_p = \diag(t_{p^{(1)},0},t_{p^{(2)},0})\,,\, \forall \, p\in \mathcal{R} 
\eeq
\sloppy{for some complex numbers  $\left(t_{p^{(1)},k},t_{p^{(2)},k}\right)_{p\in \mathcal{R},0\leq k\leq r_p-1}$. $Q_p(z_p)$ is called the irregular type of $\hat{L}$ at $p$ and $ \Lambda_p$ its residue (also called exponent of formal monodromy). Equation \eqref{Birkhoff} is known as the Birkhoff factorization or formal normal solution or Turrittin--Levelt fundamental form \cite{Birkhoff,Wasowbook}. We shall denote $\mathbf{t}=\left(t_{p^{(i)},k}\right)_{1\leq i\leq 2,p\in \mathcal{R},1\leq k\leq r_p-1}$ the irregular times while the residues $\mathbf{t}_0:=\left(t_{p^{(i)},0}\right)_{1\leq i\leq 2,p\in \mathcal{R}}$ will be referred to as monodromies by abuse of language.}
\begin{remark}
    The relation of the present setup to the study of moduli spaces, seen as a set of isomorphism classes consisting of a generic connection over a trivialisable holomorphic vector bundle with a compatible framing, may be easily seen. However we shall not insist on this point since it is beyond the scope of this article and is left for future investigations.
\end{remark}
\begin{remark} \sloppy{In the literature, the set of irregular times $\mathbf{t}=\left(t_{p^{(i)},k}\right)_{1\leq i\leq 2,p\in \mathcal{R},1\leq k\leq r_p-1}$ is referred to as ``spectral times'' or ``KP times''. This terminology originates from the study of isospectral systems and does not include the monodromy parameters $\left(t_{p^{(i)},0}\right)_{1\leq i\leq 2, p\in \mathcal{R}}$.} 
\end{remark}

\subsection{Representative normalized at infinity \label{SectionNormalizationInfinity}}
Fixing the irregular times and formal monodromies does not determine
the connection uniquely. On the open locus where the off-diagonal
coefficient used for normalization is nonzero, and after labeling the
leading eigenvalues, each orbit has a unique normalized representative
with diagonal leading coefficient at infinity and the following
subleading normalization:
\beq
\Res_{\lambda \to \infty} \hat{L}(\lambda) \lambda^{-(r_\infty-2)}=-\begin{pmatrix} \beta_{r_\infty-2} & 1\\ \delta_{r_\infty-2} & \gamma_{r_\infty-2} \end{pmatrix}.
\eeq
In particular, since the Poisson structure is independent of the choice of representative in the orbit (\cite{Harnad:1993hw,HarnadHurtubise,woodhouse2007duality,Yamakawa2017TauFA}), one can identify $\hat{F}_{\mathcal{R},\mathbf{r}}$ with the space of such representatives:
\begin{itemize} \item If $r_\infty\geq 3$:
\small{\beq\label{NormalizationInfty}
\begin{array}{ll}
\hat{F}_{\mathcal{R},\mathbf{r}} := &\left\{  \tilde{L}(\lambda) = {\displaystyle \sum_{k=1}^{r_\infty-1}} \tilde{L}^{[\infty,k]} \lambda^{k-1} + {\displaystyle \sum_{s=1}^n \sum_{k=0}^{r_s-1}} \frac{\tilde{L}^{[X_s,k]}}{(\lambda-X_s)^{k+1}}
\,\text{ such that }\,  \right.\cr
& \, \tilde L^{[\infty,r_\infty-1]},\quad (\tilde L^{[X_s,r_s-1]})_{1\le s\le n}\quad \text{are regular semisimple, and}\cr
& \left.  \, \tilde{L}^{[\infty,r_\infty-1]}=\text{diag}(-t_{\infty^{(1)},r_\infty-1},-t_{\infty^{(2)},r_\infty-1})  \text{ and }\,  \right. \cr
&\left. \, \tilde{L}^{[\infty,r_\infty-2]} =\begin{pmatrix} \beta_{r_\infty-2} & 1 \\ \delta_{r_\infty-2} & \gamma_{r_\infty-2} \end{pmatrix}, \, (\beta_{r_\infty-2},\delta_{r_\infty-2},\gamma_{r_\infty-2}) \in \mathbb{C}^3
\right\}
\end{array}
\eeq}
\item If $r_\infty=2$:
\small{\beq\begin{array}{ll}
\hat{F}_{\mathcal{R},\mathbf{r}} := &\left\{  \tilde{L}(\lambda) = {\displaystyle \tilde{L}^{[\infty,1]}} + {\displaystyle \sum_{s=1}^n \sum_{k=0}^{r_s-1}} \frac{\tilde{L}^{[X_s,k]}}{(\lambda-X_s)^{k+1}}
\,\text{ such that }\,  \right.\cr
& \, \tilde L^{[\infty,1]},\quad (\tilde L^{[X_s,r_s-1]})_{1\le s\le n}\quad \text{are regular semisimple, and}\cr
& \left.  \, \tilde{L}^{[\infty,1]}=\text{diag}(-t_{\infty^{(1)},1},-t_{\infty^{(2)},1})  \text{ and } {\displaystyle\sum_{s=1}^n}\tilde{L}^{[X_s,0]} =\begin{pmatrix} \beta_{0} & 1 \\ \delta_0 & \gamma_0 \end{pmatrix}, \, (\beta_0,\delta_0,\gamma_0) \in \mathbb{C}^3
\right\}
\end{array}
\eeq}
\item If $r_\infty=1$:
\small{\beq\begin{array}{ll}
\hat{F}_{\mathcal{R},\mathbf{r}} := &\left\{  \tilde{L}(\lambda) =  {\displaystyle \sum_{s=1}^n \sum_{k=0}^{r_s-1}} \frac{\tilde{L}^{[X_s,k]}}{(\lambda-X_s)^{k+1}}
\,\text{ such that }\,  \right.\cr
& \, -\sum_{s=1}^n\tilde L^{[X_s,0]},\quad (\tilde L^{[X_s,r_s-1]})_{1\le s\le n}\quad \text{are regular semisimple, and}\cr
& \left.  \, {\displaystyle\sum_{s=1}^n\tilde{L}^{[X_s,0]}}=\text{diag}(-t_{\infty^{(1)},0},-t_{\infty^{(2)},0})  \text{ and }\,\right. \cr 
&\left. \, {\displaystyle\sum_{s=1}^n}\tilde{L}^{[X_s,1]}+{\displaystyle\sum_{s=1}^n} X_s\tilde{L}^{[X_s,0]} =\begin{pmatrix} \beta_{-1} & 1 \\ \delta_{-1} & \gamma_{-1} \end{pmatrix}, \, (\beta_{-1},\delta_{-1},\gamma_{-1}) \in \mathbb{C}^3
\right\}
\end{array}
\eeq}
\end{itemize}


\normalsize{In} the following, we shall use the notation $\tilde{L}(\lambda)$ whenever we consider such a representative and we shall call it a representative ``normalized at infinity'' to stress that the $\text{GL}_2(\mathbb{C})$ global conjugation action has been used to select a representative of the orbit specifically normalized at infinity.

\begin{remark}\label{RemarkCoeff} The choice of normalization at infinity implies that coefficients $\beta_{r_\infty-2}$ and $\gamma_{r_\infty-2}$ are directly connected to the irregular times $-t_{\infty^{(1)},r_\infty-2}$ and $-t_{\infty^{(2)},r_\infty-2}$ for $r_\infty\geq 2$. Indeed, let us note that the diagonalization of the singular part given by \eqref{DiagoCondition} implies that
\small{\bea\label{SubLeadingCoeff} &&\det (\td{L}(\lambda)+G_\infty^{-1} \partial_\lambda G_\infty)\overset{\lambda\to \infty}{=}\left(\sum_{k=0}^{r_\infty-1} t_{\infty^{(1)},k}\lambda^{k-1}+ O(\lambda^{-2})\right)\left(\sum_{k=0}^{r_\infty-1} t_{\infty^{(2)},k}\lambda^{k-1}+ O(\lambda^{-2})\right)\cr
&&=t_{\infty^{(1)},r_\infty-1}t_{\infty^{(2)},r_\infty-1}\lambda^{2r_\infty-4}+(t_{\infty^{(1)},r_\infty-1}t_{\infty^{(2)},r_\infty-2}+t_{\infty^{(2)},r_\infty-1}t_{\infty^{(1)},r_\infty-2} )\lambda^{2r_\infty-5}\cr
&&+O(\lambda^{2r_\infty-6}).\cr
&&
\eea}
\normalsize{The} l.h.s. is of the form:
\begin{itemize}\item If $r_\infty\geq 3$:
\footnotesize{\bea \label{Idrinftygeq3} &&\det (\td{L}(\lambda)+G_\infty^{-1} \partial_\lambda G_\infty)\overset{\lambda\to \infty}{=}\cr
&&\begin{vmatrix} -t_{\infty^{(1)},r_\infty-1}\lambda^{r_\infty-2}+\beta_{r_\infty-2}\lambda^{r_\infty-3}+O(\lambda^{r_\infty-4})&\lambda^{r_\infty-3}+O(\lambda^{r_\infty-4})\\
\delta_{r_\infty-2} \lambda^{r_\infty-3}+O(\lambda^{r_\infty-4})&- t_{\infty^{(2)},r_\infty-1}\lambda^{r_\infty-2}+\gamma_{r_\infty-2}\lambda^{r_\infty-3}+O(\lambda^{r_\infty-4})\end{vmatrix}\cr
&&=t_{\infty^{(1)},r_\infty-1}t_{\infty^{(2)},r_\infty-1}\lambda^{2r_\infty-4}-(\gamma_{r_\infty-2} t_{\infty^{(1)},r_\infty-1}+\beta_{r_\infty-2} t_{\infty^{(2)},r_\infty-1})\lambda^{2r_\infty-5}+O(\lambda^{2r_\infty-6}).\cr
&&
\eea}
\normalsize{Moreover}, the diagonalization of the singular part also implies that for $r_\infty\geq 3$:
\small{\beq \label{Tracerinftygeq3}\Tr \td{L}(\lambda)\overset{\lambda\to \infty}{=}-(t_{\infty^{(1)},r_\infty-1}+t_{\infty^{(2)},r_\infty-1})\lambda^{r_\infty-2}-(t_{\infty^{(1)},r_\infty-2}+t_{\infty^{(2)},r_\infty-2})\lambda^{r_\infty-3}+ O(\lambda^{r_\infty-4}). \eeq}
\normalsize{Identifying} \eqref{Idrinftygeq3} with the r.h.s. of \eqref{SubLeadingCoeff} and using \eqref{Tracerinftygeq3}, we get 
\beq  \beta_{r_\infty-2}=-t_{\infty^{(1)},r_\infty-2} \, \text{ and }\, \gamma_{r_\infty-2}=-t_{\infty^{(2)},r_\infty-2}.\eeq
\item If $r_\infty=2$, we observe that the matrix $G_\infty(\lambda)$ is of the form $G_{\infty}(\lambda)=G_0+G_1\lambda^{-1}+O(\lambda^{-2})$ with $G_0$ diagonal (in order to preserve the diagonal leading coefficient). Consequently $G_{\infty}^{-1}\partial_\lambda G_\infty=O(\lambda^{-2})$. The diagonalization \eqref{DiagoCondition} implies
\beq\label{SubLeadingCoeff2} \det (\td{L}(\lambda)+G_\infty^{-1} \partial_\lambda G_\infty)\overset{\lambda\to \infty}{=}t_{\infty^{(1)},1}t_{\infty^{(2)},1}+(t_{\infty^{(1)},1}t_{\infty^{(2)},0}+t_{\infty^{(2)},1}t_{\infty^{(1)},0}) \lambda^{-1}+O(\lambda^{-2}).\eeq
Since $G_{\infty}^{-1}\partial_\lambda G_\infty=O(\lambda^{-2})$, the l.h.s. is of the form
\small{\bea \label{Idrinftyequal2} \det (\td{L}(\lambda)+G_\infty^{-1} \partial_\lambda G_\infty)&\overset{\lambda\to \infty}{=}&\begin{vmatrix} -t_{\infty^{(1)},1}+\beta_0\lambda^{-1}+O(\lambda^{-2})&\lambda^{-1}+O(\lambda^{-2})\\
\delta_0 \lambda^{-1}+O(\lambda^{-2})& -t_{\infty^{(2)},1}+\gamma_0\lambda^{-1}+O(\lambda^{-2})\end{vmatrix}\cr
&=&t_{\infty^{(1)},1}t_{\infty^{(2)},1}-(\gamma_0 t_{\infty^{(1)},1}+\beta_0 t_{\infty^{(2)},1}) \lambda^{-1}+O(\lambda^{-2}).
\eea}
\normalsize{Moreover}, the diagonalization of the singular part also implies that for $r_\infty=2$:
\beq \label{Tracerinftyequal2}\Tr \td{L}(\lambda)\overset{\lambda\to \infty}{=}-(t_{\infty^{(1)},1}+t_{\infty^{(2)},1})- (t_{\infty^{(1)},0}+t_{\infty^{(2)},0})\lambda^{-1}+  O(\lambda^{-2}). \eeq
Identifying \eqref{Idrinftyequal2} with the r.h.s. of \eqref{SubLeadingCoeff2} and using \eqref{Tracerinftyequal2} we get 
\beq \beta_0=-t_{\infty^{(1)},0} \,\text{ and }\,\gamma_0=-t_{\infty^{(2)},0}. \eeq
\item If $r_\infty=1$ we simply have
\beq \beta_{-1}=\sum_{s=1}^n \left[\td{L}^{[X_s,1]}+X_s\td{L}^{[X_s,0]}\right]_{1,1}\,\,, \,\,\gamma_{-1}=\sum_{s=1}^n \left[\td{L}^{[X_s,1]}+X_s\td{L}^{[X_s,0]}\right]_{2,2}. \eeq
\end{itemize}
\end{remark}

\subsection{Darboux coordinates and general isomonodromic deformations}

Let $\hat{F}_{\mathcal{R},\mathbf{r},\mathbf{t_0}} \subset \hat{F}_{\mathcal{R},\mathbf{r}}$ be the subset of $\hat{F}_{\mathcal{R},\mathbf{r}}$ composed of $\hat{L}(\lambda)$ with fixed monodromies $\mathbf{t_0}$.

Let us first recall \cite{Harnad:1993hw,HarnadHurtubise,woodhouse2007duality,Boalch2012} that the set 
\beq
\hat{\mathcal{M}}_{\mathcal{R},\mathbf{r},\mathbf{t},\mathbf{t_0}} :=\left\{\hat{L}(\lambda) \in \hat{F}_{\mathcal{R},\mathbf{r}}\,\,/\,\, \hat{L}(\lambda) \,\text{ has irregular times } \mathbf{t}  \,\text{ and monodromies } \mathbf{t_0}\right\}
\eeq
is a symplectic manifold of dimension
\beq
\dim \hat{\mathcal{M}}_{\mathcal{R},\mathbf{r},\mathbf{t},\mathbf{t}_0} = 4r-7 - (2r-1) = 2g \,\, \text{ where }\,\, g:= r-3
\eeq
is the genus of a generic smooth spectral curve defined by the algebraic equation $\det(yI_2-\hat{L}(\lambda) ) = 0$. \textbf{In the rest of the paper we shall assume that $\mathbf{g>0}$ corresponding to a non-trivial symplectic structure}. For a generic value of the irregular times $\mathbf{t}$ and monodromies $\mathbf{t}_0$, the Montr\'{e}al group introduced a set of local Darboux coordinates $\left(q_i,p_i\right)_{1\leq i\leq g}$ on $\hat{\mathcal{M}}_{\mathcal{R},\mathbf{r},\mathbf{t},\mathbf{t}_0}$ which can be obtained as follows. We work on the open chart where
the apparent zeros $q_1,\ldots,q_g$ are simple, pairwise distinct and
disjoint from $X_1,\ldots,X_n$, and every normalization denominator is
nonzero. 

\bigskip

Let $\tilde{L}(\lambda) \in \hat{\mathcal{M}}_{\mathcal{R},\mathbf{r},\mathbf{t},\mathbf{t}_0}$ be a representative of the form described in Section \ref{SectionNormalizationInfinity}. The entry $\left[\tilde{L}(\lambda)\right]_{1,2}$ is a rational function of $\lambda$ with $g$ zeros denoted $\left(q_i\right)_{1\leq i\leq g}$:
\beq\label{Conditionqi}
\forall\, i\in \llbracket 1,g\rrbracket \, : \; \left[\tilde{L}(q_i)\right]_{1,2} = 0.
\eeq
This defines half of the spectral Darboux coordinates. The second half is obtained by evaluating the entry $\left[\tilde{L}(\lambda)\right]_{1,1}$  at $\lambda = q_i$:
\beq \label{Conditionpi}
\forall\, i\in \llbracket 1,g\rrbracket \, : \; p_i:=\left[\tilde{L}(q_i)\right]_{1,1}.
\eeq

Let us remark that, by definition, the pair $\left(q_i,p_i\right)$ defines a point on the spectral curve for any $i \in \llbracket 1,g \rrbracket$:
\beq 
\forall\, i\in \llbracket 1,g\rrbracket \, : \; \det(p_i\, I_2-\tilde{L}(q_i)) = 0.
\eeq

Thus, we have obtained a local description of the space $\hat{F}_{\mathcal{R},\mathbf{r},\mathbf{t}_0}$ as a trivial bundle $\hat{F}_{\mathcal{R},\mathbf{r},\mathbf{t}_0} \to B$ where the base $B$ is the set of irregular times satisfying the condition
\beq
B:=\left\{\mathbf t\in\mathbb C^{2(r-n-1)}:
t_{p^{(1)},r_p-1}\ne t_{p^{(2)},r_p-1}
\text{ for every }p\in\mathcal R\text{ with }r_p\ge2\right\}.
\eeq
In particular, the fiber above a point $\mathbf{t} \in B$ is $\hat{\mathcal{M}}_{\mathcal{R},\mathbf{r},\mathbf{t},\mathbf{t}_0}$ which can be equipped with spectral Darboux coordinates $\left(q_i,p_i\right)_{1\leq i\leq g}$.

\medskip

The space $B$ is a space of isomonodromic deformations meaning that any vector field $\partial_t \in T_{\mathbf t} B$ gives rise to a deformation of $\tilde{L}(\lambda)$ preserving its generalized monodromy data.
In addition to deformations with respect to a vector $\partial_t \in T_{\mathbf t} B$  we shall also consider the standard deformations with respect to the position $\left(X_i\right)_{1\leq i\leq n}$ of the finite poles.
%

There exist different equivalent ways to characterize the property of being an isomonodromic vector field. One of them \cite{HarnadHurtubise} is the existence of a compatible system, referred to as a Lax pair, of the form
\beq \label{CompatGeneralCase}
\left\{
\begin{array}{l}
\partial_\lambda \tilde{\Psi}(\lambda,\mathbf{t}) = \tilde{L}(\lambda) \tilde{\Psi}(\lambda,\mathbf{t})  \cr
\partial_t \tilde{\Psi}(\lambda,\mathbf{t}) = \tilde{A}_t(\lambda) \tilde{\Psi}(\lambda,\mathbf{t})  \cr
\end{array}
\right.
\eeq
where $ \tilde{A}_t(\lambda) $ is a rational function of $\lambda$ with poles dominated by the poles of $\tilde{L}(\lambda)$. In the rest of the paper, we shall extensively study the compatibility condition (zero-curvature equation)
\beq
\partial_\lambda \tilde{A}_t(\lambda) - \partial_t \tilde{L}(\lambda) + \left[\tilde{A}_t(\lambda),\tilde{L}(\lambda)\right] = 0.
\eeq
obtained from \eqref{CompatGeneralCase}.

\subsection{Scalar differential equation, oper gauge and local diagonal gauges}

Let us now consider an orbit in $\hat{F}_{\mathcal{R},\mathbf{r}}$ and a representative $\td{L}(\lambda)$ of this orbit normalized at infinity as described above. Let $\td{\Psi}(\lambda)$ be a wave matrix solution to the linear system
\beq  \partial_\lambda \td{\Psi}(\lambda)=\td{L}(\lambda)\td{\Psi}(\lambda).\eeq

Finite jets of the formal normalizing gauges can be realized by
holomorphic gauges that diagonalize the displayed principal parts.
This finite-jet assertion does not assert convergence of the exact formal
normal form. In these local gauges the principal parts are
 \bea
 \td{L}_{\infty}&\overset{\lambda\to \infty}{=}& \text{diag}\left(-\underset{k=0}{\overset{r_\infty-1}{\sum}} t_{\infty^{(1)},k}\lambda^{k-1},-\underset{k=0}{\overset{r_\infty-1}{\sum}} t_{\infty^{(2)},k} \lambda^{k-1}\right) +o(\lambda^{-1}),  \cr
 \td{L}_{X_s}&\overset{\lambda\to X_s}{=}& \text{diag}\left(\underset{k=0}{\overset{r_s-1}{\sum}} t_{X_s^{(1)},k}(\lambda-X_s)^{-k-1},\underset{k=0}{\overset{r_s-1}{\sum}} t_{X_s^{(2)},k}(\lambda-X_s)^{-k-1}\right) +O(1)\cr
&&
\eea
The corresponding formal fundamental matrices have the form
\eqref{Birkhoff}; at irregular poles their analytic realizations are
sectorial. Regular semisimplicity of the leading polar coefficient implies
an unramified formal type, but unramifiedness alone does not imply regular
semisimplicity.

\medskip

Moreover, the differential system $\partial_\lambda \td{\Psi}(\lambda)=\td{L}(\lambda)\td{\Psi}(\lambda)$ may be written as a scalar differential equation for $\td{\Psi}_{1,1}$ that is equivalent to a companion-like (oper) matrix system. More precisely, defining 
\beq \label{CompanionMatrix0} \Psi(\lambda)=G(\lambda) \td{\Psi}(\lambda) \,\,\text{with}\,\, G(\lambda)=\begin{pmatrix} 1&0\\ \td{L}_{1,1}& \td{L}_{1,2}\end{pmatrix}\eeq
we end up with the fact that $\Psi$ is a solution of the companion-like system
\beq \label{CompanionMatrix}\partial_\lambda \Psi(\lambda)=L(\lambda)\Psi(\lambda)\,\,\text{with}\,\, L(\lambda)=\begin{pmatrix}0&1\\ L_{2,1}&L_{2,2}\end{pmatrix}\eeq 
given by
\bea \label{LInTermsOfTdL} L_{2,1}&=&-\det \td{L}+\partial_\lambda\td{L}_{1,1}-\td{L}_{1,1}\frac{\partial_\lambda\td{L}_{1,2}}{\td{L}_{1,2}},\cr
L_{2,2}&=&\Tr\, \td{L} +\frac{\partial_\lambda\td{L}_{1,2}}{\td{L}_{1,2}}.
\eea
Note in particular that the first line of $\Psi$ and $\td{\Psi}$ is obviously the same: $\Psi_{1,1}=\td{\Psi}_{1,1}\overset{\text{def}}{=}\psi_1$ and $\Psi_{1,2}=\td{\Psi}_{1,2}\overset{\text{def}}{=}\psi_2$. Hence, we immediately obtain that
\beq \Psi(\lambda)=\begin{pmatrix}\td{\Psi}_{1,1}(\lambda)& \td{\Psi}_{1,2}(\lambda)\\ \partial_\lambda \td{\Psi}_{1,1}(\lambda)& \partial_\lambda \td{\Psi}_{1,2}(\lambda) \end{pmatrix}=\begin{pmatrix}\psi_{1}(\lambda)& \psi_2(\lambda)\\ \partial_\lambda \psi_{1}(\lambda)& \partial_\lambda \psi_{2}(\lambda) \end{pmatrix}.\eeq
The companion-like system \eqref{CompanionMatrix} is equivalent to saying that $\psi_1$ and $\psi_2$ satisfy the linear ODE:
\beq \left(\left[\partial_{\lambda}\right]^2 -L_{2,2}(\lambda)\partial_\lambda -L_{2,1}(\lambda)\right)\psi_i=0 \,\, ,\,\, i\in \{1,2\}\eeq 
which is sometimes referred to as the ``quantum curve''. The companion gauge is also referred to as the ``oper gauge'' in part of the literature.

\subsection{Introduction of a scaling parameter $\hbar$}\label{SectionIntrohbar}

In order to make the connection with formal $\hbar$-transseries appearing in the quantization of classical spectral curves via topological recursion of \cite{EO07}, we shall also introduce a formal $\hbar$ parameter by a simple rescaling of the location of the poles, monodromies and irregular times. We shall perform the following rescaling\footnote{The change on $\Psi$ is made so that the normalization at infinity is preserved by the rescaling.}:
\bea t_{\infty^{(i)},k} &\to& \hbar^{k-1} t_{\infty^{(i)},k} 
 \,\,,\,\,  \forall \, (i,k)\in \llbracket 1,2\rrbracket\times \llbracket 0, r_\infty-1\rrbracket ,\cr
 t_{X_s^{(i)},k}&\to&\hbar^{-1-k}t_{X_s^{(i)},k} 
 \,\,,\,\,  \forall \, (i,s,k)\in \llbracket 1,2\rrbracket\times\llbracket 1,n\rrbracket\times  \llbracket 0, r_s-1\rrbracket ,\cr
X_s&\to& \hbar^{-1}X_s 
\,\,,\,\, \forall \, s\in \llbracket 1,n\rrbracket ,\cr
\lambda&\to& \hbar^{-1}\lambda \cr
\Psi&\to&\text{diag}\left(\hbar^{-\frac{r_\infty-3}{2}},\hbar^{\frac{r_\infty-3}{2}} \right)\Psi. 
\eea
In particular the differential system now reads
\beq \hbar\partial_\lambda \td{\Psi}(\lambda,\hbar)=\td{L}(\lambda,\hbar)\td{\Psi}(\lambda,\hbar).\eeq
Such differential systems are known as $\hbar$-connections. The problem of Abelianisation of $\hbar$-connections is presently a very active domain in relation with exact WKB, Borel resummation, etc. This provides the main motivation to introduce the formal $\hbar$ parameter in the rest of the paper. However, for readers uneasy with this additional $\hbar$ parameter, \textbf{we stress that $\hbar$ may be fixed to $1$ in the rest of the paper} except for Section \ref{SectionTR}.

\subsection{Explicit expression for the gauge transformations and Lax matrices}\label{SectionGaugeTransfo}
In order to relate the matrices $\Psi$ to $\td{\Psi}$, we shall introduce an intermediate wave matrix $\check{\Psi}$. Let us define the following gauge transformations. 

\begin{proposition}\label{GaugeTransfoProp} The matrices are related by the gauge transformations
\bea \label{GaugeTransfo} 
\td{\Psi}(\lambda,\hbar)&=&G_1(\lambda,\hbar)\check{\Psi}(\lambda,\hbar) \, \text{ with }\, G_1(\lambda,\hbar)=\begin{pmatrix} 1&0\\ t_{\infty^{(1)},r_\infty-1} \lambda+\eta_0& 1\end{pmatrix},\cr
\check{\Psi}(\lambda,\hbar)&=&J(\lambda,\hbar) \Psi(\lambda,\hbar)\, \text{ with }\, J(\lambda,\hbar)=\begin{pmatrix}1 &0\\
\frac{Q(\lambda,\hbar)}{\underset{j=1}{\overset{g}{\prod}}(\lambda-q_j)}& \frac{\underset{s=1}{\overset{n}{\prod}} (\lambda-X_s)^{r_s}}{\underset{j=1}{\overset{g}{\prod}}(\lambda-q_j)}
\end{pmatrix}\cr
&&
\eea
where $Q$ is the unique polynomial in $\lambda$ of degree at most $g-1$ such that (with the convention that empty products are set to $1$) 
\beq Q(q_i,\hbar)=-p_i \prod_{s=1}^{n}(q_i-X_s)^{r_s} \,,\, \forall \, i\in \llbracket 1,g\rrbracket ,
\eeq
i.e.
\beq \label{DefQ2}  Q(\lambda,\hbar)= -\sum_{i=1}^g p_i \prod_{s=1}^{n}(q_i-X_s)^{r_s} \prod_{j\neq i}\frac{\lambda-q_j}{q_i-q_j},\eeq
and $\eta_0$ is given by
\bea \eta_0&=&t_{\infty^{(1)},r_\infty-2}+t_{\infty^{(1)},r_\infty-1}\left(\sum_{j=1}^gq_j -\sum_{s=1}^nr_sX_s\right) \,\, \text{ if }r_\infty\geq 2,\cr
\eta_0&=&t_{\infty^{(1)},0}\left(\sum_{j=1}^gq_j -\sum_{s=1}^nr_sX_s\right)-\sum_{s=1}^n\left[X_s\td{L}^{[X_s,0]}+\td{L}^{[X_s,1]}\right]_{1,1} \,\, \text{ if }r_\infty=1.\cr
&&
\eea
\end{proposition}

\begin{proof}The proof is based on the fact that 
\beq \td{G}(\lambda)=\left(G_1(\lambda,\hbar)J(\lambda,\hbar)\right)^{-1}=\begin{pmatrix} 1&0\\ \frac{-Q(\lambda)-(t_{\infty^{(1)},r_\infty-1}\lambda+\eta_0)\underset{j=1}{\overset{g}{\prod}}(\lambda-q_j)}{\underset{s=1}{\overset{n}{\prod}}(\lambda-X_s)^{r_s}}&\frac{\underset{j=1}{\overset{g}{\prod}}(\lambda-q_j)}{ \underset{s=1}{\overset{n}{\prod}}(\lambda-X_s)^{r_s}}\end{pmatrix}\eeq
recovers the matrix \eqref{CompanionMatrix0}. It is given in detail in Appendix \ref{AppendixProofExplicitGaugeTransfo}. 
\end{proof}

We shall also introduce a matrix $\Psi_{\mathrm{GMR}}$ whose gauge
makes the trace-residue reconstruction particularly simple
(Proposition \ref{MartaConnection}). The notation is motivated by the
comparison with \cite{MartaPaper2022}; it does not identify the two
connections in the same gauge, as explained in
Remark \ref{RemarkHInTermsOfTraceSquared}.
\beq\label{DefPsic}
\Psi_{{\text{GMR}}}(\lambda,\hbar):=G_{{\text{GMR}}}(\lambda,\hbar)\check{\Psi}(\lambda,\hbar) \, \text{ with }\, G_{{\text{GMR}}}(\lambda,\hbar):=\begin{pmatrix} 1&0\\ -\frac{Q(\lambda,\hbar)}{\underset{i=1}{\overset{g}{\prod}} (\lambda-q_i)}& 1\end{pmatrix}.
\eeq

The gauge transformation is such that $\td{\Psi}(\lambda,\hbar)$, $\check{\Psi}(\lambda,\hbar)$, $\Psi(\lambda,\hbar)$ and $\Psi_{{\text{GMR}}}(\lambda,\hbar)$ satisfy the Lax systems,
\bea \hbar \partial_\lambda \td{\Psi}(\lambda,\hbar)&=&\td{L}(\lambda,\hbar) \td{\Psi}(\lambda,\hbar),\cr
\hbar \partial_\lambda \check{\Psi}(\lambda,\hbar)&=&\check{L}(\lambda,\hbar) \check{\Psi}(\lambda,\hbar),\cr
\hbar \partial_\lambda \Psi(\lambda,\hbar)&=&L(\lambda,\hbar) \Psi(\lambda,\hbar),\cr
\hbar \partial_\lambda \Psi_{{\text{GMR}}}(\lambda,\hbar)&=&L_{{\text{GMR}}}(\lambda,\hbar) \Psi_{{\text{GMR}}}(\lambda,\hbar).
\eea

\begin{remark} The intermediate matrices $\check{\Psi}$ and $\check{L}$ have been used in \cite{MOsl2,Quantization_2021}. In this gauge there are no apparent singularities at $\lambda=q_i$ but the normalization at infinity does not match with the choice of the representative defined in Section \ref{SectionNormalizationInfinity}, hence the necessity of the additional gauge transformation $G_1$.
\end{remark}

Note that by definition, the entries of $\check{L}$ are related to those of $L$ by
\bea \label{CheckLEquations}\check{L}_{1,1}(\lambda,\hbar)&=&-\frac{Q(\lambda,\hbar)}{\underset{s=1}{\overset{n}{\prod}}(\lambda-X_s)^{r_s}},\cr
\check{L}_{1,2}(\lambda,\hbar)&=&\frac{\underset{j=1}{\overset{g}{\prod}}(\lambda-q_j)}{\underset{s=1}{\overset{n}{\prod}}(\lambda-X_s)^{r_s}},\cr
\check{L}_{2,2}(\lambda,\hbar)&=&P_1(\lambda)+\frac{Q(\lambda,\hbar)}{\underset{s=1}{\overset{n}{\prod}}(\lambda-X_s)^{r_s}},\cr
\check{L}_{2,1}(\lambda,\hbar)&=&\hbar \frac{ \partial \bigg(\frac{Q(\lambda,\hbar)}{\underset{j=1}{\overset{g}{\prod}} (\lambda-q_j)}\bigg)}{\partial \lambda} +L_{2,1}(\lambda,\hbar) \frac{\underset{s=1}{\overset{n}{\prod}} (\lambda-X_s)^{r_s}}{\underset{j=1}{\overset{g}{\prod}} (\lambda - q_j)}\cr 
&& - P_1(\lambda) \frac{Q(\lambda,\hbar)}{\underset{j=1}{\overset{g}{\prod}} (\lambda-q_j)} - \frac{Q(\lambda,\hbar)^2}{\underset{j=1}{\overset{g}{\prod}}(\lambda-q_j) \, {\underset{s=1}{\overset{n}{\prod}} (\lambda-X_s)^{r_s}}}
\eea 
where we have defined
\bea \label{P1Coeffs}P_{\infty,k}^{(1)}&=& -(t_{\infty^{(1)},k+1}+t_{\infty^{(2)},k+1}) \,\,,\,\, \forall\, k\in\llbracket 0,r_\infty-2 \rrbracket,\cr
P_{X_s,k}^{(1)}&=&t_{X_s^{(1)},k-1}+t_{X_s^{(2)},k-1}\,\,,\,\, \forall\, s\in \llbracket 1,n\rrbracket\,,\, k\in\llbracket 1,r_s \rrbracket
\eea
and regrouped them into the rational function $P_1$:
\beq P_1(\lambda)=\sum_{j=0}^{r_\infty-2} P_{\infty,j}^{(1)}\lambda^j +\sum_{s=1}^n \sum_{j=1}^{r_s}P_{X_s,j}^{(1)}(\lambda-X_s)^{-j}.\eeq

Similarly, the entries of $\td{L}$ are related to those of $\check{L}$ by
\bea\label{TdLEquations} \td{L}_{1,1}(\lambda,\hbar)&=&\check{L}_{1,1}(\lambda,\hbar)-(t_{\infty^{(1)},r_\infty-1} \lambda+\eta_0)\check{L}_{1,2}(\lambda,\hbar)\cr 
\td{L}_{1,2}(\lambda,\hbar)&=&\check{L}_{1,2}(\lambda,\hbar)\cr
\td{L}_{2,1}(\lambda,\hbar)&=&\check{L}_{2,1}(\lambda,\hbar)-(t_{\infty^{(1)},r_\infty-1} \lambda+\eta_0)^2\check{L}_{1,2}(\lambda,\hbar)\cr
&&+(t_{\infty^{(1)},r_\infty-1} \lambda+\eta_0)(\check{L}_{1,1}(\lambda,\hbar)-\check{L}_{2,2}(\lambda,\hbar)) +\hbar t_{\infty^{(1)},r_\infty-1}\cr
 \td{L}_{2,2}(\lambda,\hbar)&=&\check{L}_{2,2}(\lambda,\hbar)+(t_{\infty^{(1)},r_\infty-1} \lambda+\eta_0)\check{L}_{1,2}(\lambda,\hbar).\cr
&&
\eea
Finally, let us note that the previous results imply the following proposition
\begin{proposition}\label{MartaConnection}We have
\beq L_{\text{GMR}}(\lambda,\hbar)=\begin{pmatrix}0&\frac{\underset{i=1}{\overset{g}{\prod}} (\lambda-q_i)}{\underset{s=1}{\overset{n}{\prod}}(\lambda-X_s)^{r_s}}
\\ \frac{\underset{s=1}{\overset{n}{\prod}}(\lambda-X_s)^{r_s}}{\underset{i=1}{\overset{g}{\prod}} (\lambda-q_i)}L_{2,1}(\lambda,\hbar)&P_1(\lambda)  \end{pmatrix} \eeq
so that
\beq L_{2,1}(\lambda,\hbar)=-\det(L_{\text{GMR}}(\lambda,\hbar))=\frac{1}{2}\Tr(L_{\text{GMR}}(\lambda,\hbar)^2)-\frac{1}{2}P_1(\lambda)^2. \eeq
\end{proposition}

\begin{proof}The proof follows from direct computations.
\end{proof}

We shall also define the corresponding Wronskians and obtain their explicit expressions:
In all primitives below, $\lambda_0$ is a nonsingular base point;
paths and branches of logarithms are fixed on a simply connected domain.
\begin{definition}[Wronskians]Let us define $W(\lambda,\hbar)=\det \Psi(\lambda,\hbar)$, $\check{W}(\lambda,\hbar)=\det \check{\Psi}(\lambda,\hbar)$ and $\td{W}(\lambda,\hbar)=\det \td{\Psi}(\lambda,\hbar)$ the Wronskians associated to the corresponding wave matrices. They are given by
\bea \td{W}(\lambda)&=&\td{W}_0 \exp\left(\frac{1}{\hbar}\int_{\lambda_0}^\lambda P_1(\td{\lambda})\,d\td{\lambda} \right),\cr
\check{W}(\lambda)&=&\td{W}(\lambda)=\td{W}_0 \exp\left(\frac{1}{\hbar}\int_{\lambda_0}^\lambda P_1(\td{\lambda})\,d\td{\lambda} \right)\cr
W(\lambda)&=&W_0\frac{\underset{i=1}{\overset{g}{\prod}} (\lambda-q_i)}{\underset{s=1}{\overset{n}{\prod}} (\lambda-X_s)^{r_s}} \exp\left(\frac{1}{\hbar}\int_{\lambda_0}^\lambda P_1(\td{\lambda})d\td{\lambda}  
\right). \eea 
\end{definition}

Combining the gauge transformations $\left(G_p\right)_{p\in \mathcal{R}}$, $G_1$ and $J$, we get the following proposition.

\begin{proposition}\label{PropPsiAsymp}
At each pole choose a local formal fundamental matrix, or a sectorial
fundamental matrix with that formal asymptotic expansion. In a common
domain a local analytic matrix is related to a fixed global fundamental
matrix by right multiplication by a connection matrix $C_p$; sector
changes introduce the corresponding Stokes matrices. In the formulas
below $\psi_1,\psi_2$ denote the first-row entries of this local matrix,
not two fixed global columns at every pole. We assume the first-row
components of the leading finite-pole eigenvectors are nonzero, as on
the generic cyclic chart. The local expansions are
\small{\bea \label{PsiAsymptotics0}\psi_1(\lambda)&\overset{\lambda\to \infty}{=}&\exp\left(-\frac{1}{\hbar}\sum_{k=1}^{r_\infty-1} \frac{t_{\infty^{(1)},k}}{k}\lambda^k -\frac{1}{\hbar}t_{\infty^{(1)},0}\ln \lambda+ A_{\infty^{(1)},0} +O\left(\lambda^{-1}\right)\right),\cr
\psi_2(\lambda)&\overset{\lambda\to \infty}{=}&\exp\left(-\frac{1}{\hbar}\sum_{k=1}^{r_\infty-1} \frac{t_{\infty^{(2)},k}}{k}\lambda^k -\frac{1}{\hbar}t_{\infty^{(2)},0}\ln \lambda -\ln \lambda+ A_{\infty^{(2)},0}+O\left(\lambda^{-1}\right)\right),\cr
\psi_1(\lambda)&\overset{\lambda\to X_s}{=}&\exp\left(-\frac{1}{\hbar}\sum_{k=1}^{r_s-1} \frac{t_{X_s^{(1)},k}}{k}(\lambda-X_s)^{-k} +\frac{1}{\hbar}t_{X_s^{(1)},0}\ln(\lambda-X_s)+ A_{X_s^{(1)},0}+O\left(\lambda-X_s\right)\right),\cr
\psi_2(\lambda)&\overset{\lambda\to X_s}{=}&\exp\left(-\frac{1}{\hbar}\sum_{k=1}^{r_s-1} \frac{t_{X_s^{(2)},k}}{k}(\lambda-X_s)^{-k}+\frac{1}{\hbar}t_{X_s^{(2)},0}\ln(\lambda-X_s) + A_{X_s^{(2)},0}+O\left(\lambda-X_s\right)\right)\cr
&&
\eea}
\normalsize{for} all $s\in \llbracket 1,n\rrbracket$. 
\end{proposition}

It is important to remark that the logarithmic terms in the expansions around $\infty$ differ in $\psi_1$ and $\psi_2$ by a shift $+\hbar$ due to the gauge transformations. The reason is presented in the proof given in Appendix \ref{AppendixProofPropPsiAsymp}.

\begin{remark}Note that we recover the shift by $\hbar$ in the asymptotic expansion of $\psi_2$ at infinity that was arising in \cite{Quantization_2021}. In \cite{Quantization_2021}, the shift arises from the choice of $\infty^{(1)}$ as a base point for integrating the Eynard-Orantin differentials. The present setup shows that this choice maps directly to the choice of normalization of the Lax matrix at infinity, i.e. a choice of representative of the orbit in $\hat{F}_{\mathcal{R},\mathbf{r}}$.
\end{remark}

\subsection{Explicit expression for the Lax matrix $L$}

We shall now write an explicit formula for the Lax matrix $L$. In order to do so, we introduce the following quantities.

\begin{definition}We define:
\bea \label{P2Coeffs} P_{\infty,2r_\infty-4-k}^{(2)}&=&\sum_{j=0}^k t_{\infty^{(1)},r_\infty-1-j}t_{\infty^{(2)},r_\infty-1 -(k-j)} \,\, ,\,\, \forall \,k\in \llbracket 0, r_{\infty}-1\rrbracket,\cr
P_{X_s,2r_s-k}^{(2)}&=&\sum_{j=0}^k t_{X_s^{(1)},r_s-1-j}t_{X_s^{(2)},r_s-1 -(k-j)} \,\, ,\,\, \forall\, s\in \llbracket 1,n \rrbracket\,\, ,\,\, \forall \,k\in \llbracket 0, r_s-1\rrbracket\cr
&&
\eea
and regroup the previous quantities to define rational functions $P_2$ and $\td{P}_2$ by
\bea\label{DeftdP2}
P_2(\lambda)&=&\sum_{j=0}^{2r_\infty-4} P_{\infty,j}^{(2)} \lambda^{j} +\sum_{s=1}^n\sum_{j=1}^{2r_s} \frac{P_{X_s,j}^{(2)}}{(\lambda-X_s)^j},\cr
\td{P}_2(\lambda)&=&\sum_{j=\text{max}(0,r_\infty-3)}^{2r_\infty-4} P_{\infty,j}^{(2)} \lambda^{j} +\sum_{s=1}^n\sum_{j=r_s+1}^{2r_s} \frac{P_{X_s,j}^{(2)}}{(\lambda-X_s)^j}.
\eea
\end{definition}

\begin{remark}
Note that the coefficients $\left(P^{(2)}_{\infty,k}\right)_{0\leq k \leq r_\infty-4}$ and $\left(P^{(2)}_{X_s,k}\right)_{1\leq k\leq r_s}$ for all $s\in \llbracket 1, n\rrbracket$ remain undetermined at this stage. As we will see below, they shall correspond to the Hamiltonians. On the contrary, the coefficients of $\td{P}_2,$ that are completely determined by the monodromies and irregular times, are often referred to as the Casimirs to remind their origin in the associated isospectral system.
\end{remark}

Studying the asymptotics at each pole using \eqref{PsiAsymptotics0}, we obtain the general form of the Lax matrix $L(\lambda,\hbar)$.

\begin{proposition}[General form of the Lax matrix]\label{PropLaxMatrix}The Lax matrix has the following entries
\bea L_{1,1}(\lambda)&=&0,\cr
L_{1,2}(\lambda)&=&1,\cr
L_{2,1}(\lambda)
&=& -\td{P}_2(\lambda) +\sum_{j=0}^{r_\infty-4}H_{\infty,j}\lambda^j+\sum_{s=1}^n\sum_{j=1}^{r_s}H_{X_s,j}(\lambda-X_s)^{-j}\cr
&&-\hbar t_{\infty^{(1)},r_\infty-1}\lambda^{r_\infty-3}\delta_{r_\infty\geq 3}-\sum_{j=1}^{g} \frac{\hbar p_j}{\lambda-q_j},\cr
L_{2,2}(\lambda)&=&P_1(\lambda)+\sum_{j=1}^{g} \frac{\hbar}{\lambda-q_j} -\sum_{s=1}^n \frac{\hbar r_s}{\lambda-X_s}.
\eea
Moreover for $r_\infty=2$, we have from Remark \ref{RemarkLrinftyequal2},
\bea \label{ConditionsAddrinftyequal2} P_{\infty,0}^{(2)}&=&t_{\infty^{(1)},1}t_{\infty^{(2)},1},\cr
\sum_{s=1}^n H_{X_s,1}&=& \hbar \sum_{j=1}^g p_j-(t_{\infty^{(1)},1}t_{\infty^{(2)},0}+t_{\infty^{(2)},1}t_{\infty^{(1)},0}+\hbar t_{\infty^{(1)},1}) 
\eea
while, for $r_\infty=1$, we have
\bea \label{ConditionsAddrinftyequal1}\sum_{s=1}^n H_{X_s,1}&=& \hbar \sum_{j=1}^g p_j,\cr
\sum_{s=1}^n X_s H_{X_s,1}+\sum_{s=1}^n H_{X_s,2}\delta_{r_s\geq 2} -\sum_{s=1}^nt_{X_s^{(1)},0}t_{X_s^{(2)},0}\delta_{r_s=1}&=&\hbar \sum_{j=1}^g q_j p_j 
-t_{\infty^{(1)},0}(t_{\infty^{(2)},0}+\hbar).\cr
&& 
\eea
\end{proposition}

\begin{remark}\label{RemarkAdditionalConstraints} For $r_\infty=1$, the behavior at infinity implies some additional relations.
\bea \label{rinftyequal1SpecialRelations} 0&=&\sum_{s=1}^n P_{X_s,1}^{(2)},\cr
t_{\infty^{(1)},0}t_{\infty^{(2)},0}&=&\sum_{s=1}^n P_{X_s,2}^{(2)}+ \sum_{s=1}^n X_s P_{X_s,1}^{(2)}.
\eea
For $r_\infty=2$, we have the additional constraint
\beq\label{rinftyequal2SpecialRelations} \sum_{s=1}^n P_{X_s,1}^{(2)}=t_{\infty^{(1)},1} t_{\infty^{(2)},0}+t_{\infty^{(2)},1} t_{\infty^{(1)},0}.\eeq
\end{remark}

\begin{proof}Details of the computation at each pole are presented in Appendix \ref{AppendixLax}.
\end{proof}

Proposition \ref{PropLaxMatrix} combined with Proposition \ref{MartaConnection} immediately implies that the coefficients $\left(H_{\infty,j}\right)_{0\leq j\leq r_\infty-4}$ and $\left(H_{X_s,j}\right)_{1\leq s\leq n,1\leq j\leq r_s}$ may be recovered as suitable residues of $\Tr(L_{\text{GMR}}^2)$.

\begin{proposition}\label{HInTermsOfTraceSquared}We have
\small{\bea H_{\infty,j}&=&-\Res_{\lambda\to \infty} \left(\frac{1}{2}\Tr(L_{\text{GMR}}(\lambda,\hbar)^2)-\frac{1}{2}P_1(\lambda)^2\right)\lambda^{-j-1} \,,\,\forall \, j\in \llbracket 0,r_\infty-4\rrbracket\cr
H_{X_s,j}&=&\Res_{\lambda\to X_s} \left(\frac{1}{2}\Tr(L_{\text{GMR}}(\lambda,\hbar)^2)-\frac{1}{2}P_1(\lambda)^2\right)(\lambda-X_s)^{j-1}\,,\,\forall \, (s,j)\in \llbracket 1,n\rrbracket\times \llbracket 1,r_s\rrbracket.\cr
&&
\eea}
\end{proposition}

\normalsize{We} shall discuss the importance of this result and the connection with results of \cite{MartaPaper2022} in Remark \ref{RemarkHInTermsOfTraceSquared}.

\medskip

In addition, we get an alternative expression for the coefficient $\eta_0$ used in the gauge transformation \eqref{GaugeTransfo}. 

\begin{proposition}\label{Consistencyg} The coefficient $\eta_0$ is also given by
\bea \label{gspecial} 
\eta_0&=& t_{\infty^{(1)},r_\infty-2}+t_{\infty^{(1)},r_\infty-1}\left(\sum_{j=1}^gq_j-\sum_{s=1}^{n} r_s X_s\right) \,\,,\,\, \text{ if } \, r_\infty\geq 2\cr
\eta_0&=&\frac{1}{t_{\infty^{(1)},0}-t_{\infty^{(2)},0}}\Big[ -\sum_{s=1}^n(2X_sP_{X_s,2}^{(2)}\delta_{r_s=1}+P_{X_s,3}^{(2)}\delta_{r_s=2})\cr
&& + \sum_{s=1}^n(X_s^2H_{X_s,1}+2X_sH_{X_s,2}\delta_{r_s\geq 2}+H_{X_s,3}\delta_{r_s\geq 3})-\hbar \sum_{j=1}^{g}p_jq_j^2\cr
&&-t_{\infty^{(1)},0}\underset{s=1}{\overset{n}{\sum}} \left((t_{X_s^{(1)},1}+t_{X_s^{(2)},1})\delta_{r_s\geq 2}+ X_s(t_{X_s^{(1)},0}+t_{X_s^{(2)},0})\right)\cr
&&+t_{\infty^{(1)},0}(t_{\infty^{(1)},0}-t_{\infty^{(2)},0}-\hbar)\left(\sum_{j=1}^g q_j-\sum_{s=1}^n r_s X_s\right)
\Big] \,\, \, \text{ if }\,\, r_\infty=1 .
\eea 
\end{proposition}

\begin{proof}The proof is given in Appendix \ref{AppendixGaugeTransfo}.
\end{proof}

Finally, let us make a comment about gauge choices that shall be used when considering general isomonodromic deformations.

\begin{remark} Fixing the matrix $\td{L}$ does not determine uniquely the gauge. Indeed, gauge transformations of the form $\hat{\Psi}=G\,\td{\Psi}$ with $G=u(\mathbf{t}) I_2$ (i.e. a gauge matrix independent of $\lambda$ and proportional to the identity matrix) change neither $\td{L}$ nor $L$. On the contrary, these gauge transformations may change the auxiliary matrices $\td{A}_{t}$ defined in Section \ref{SectionAuxi} because of the term $\hbar\partial_t[G] G^{-1}= \hbar\partial_t[u(\mathbf{t})] u(\mathbf{t})^{-1}$. This gauge freedom is discussed in Proposition \ref{PropositionTrace} where it is proven that this gauge freedom is irrelevant regarding isomonodromic deformations.
\end{remark}

\subsection{Summary of the construction}
Using the geometric considerations and the existence of the local gauge transformations we have been able to study the asymptotics of the wave functions $\psi_1$ and $\psi_2$ around each pole  (Proposition \ref{PropPsiAsymp}) and deduce from it the general form of the associated Lax matrices in terms of spectral Darboux coordinates, monodromies and irregular times. In particular, we have introduced $3$ different gauges with explicit expressions of the gauge matrices connecting them. Their main features are
\begin{itemize}\item The gauge $\td{\Psi}$ corresponds to the gauge in which $\td{L}$ is the unique representative of $\hat{F}_{\mathcal{R},\mathbf{r}}$. In particular, $\td{L}$ has only pole singularities in $\mathcal{R}$ and is properly normalized at infinity.
\item The gauge $\check{\Psi}$ is an intermediate gauge for which $\check{L}$ has only pole singularities in $\mathcal{R}$ but is not normalized in a specific way at infinity. It is the gauge used in \cite{MOsl2,Quantization_2021} and may be used as a starting point for other choices of normalization.
\item The companion gauge $\Psi$ in which $L$ is companion-like. In this gauge, $L$ has pole singularities in $\mathcal{R}$ but also apparent singularities at $\lambda\in\left\{(q_i)_{1\leq i\leq g}\right\}$. However, since $L$ is companion-like, there are only two non-trivial entries and their general form is given by Proposition \ref{PropLaxMatrix}. As we shall see below, this gauge (which is directly equivalent to the quantum curve satisfied by $\psi_1$ and $\psi_2$) is very convenient for the computation of the compatibility equations.
\end{itemize}

\section{Classical spectral curve and connection with topological recursion}\label{SectionTR}
Before turning to isomonodromic deformations, let us briefly mention the connection of the present setup with the classical spectral curve and the topological recursion. Although the connection with topological recursion is interesting for applications in mathematical physics, one does not need it to obtain the isomonodromic deformations and Hamiltonian systems that shall be built below. Thus, readers with no interest in topological recursion or in WKB expansions may skip the content of this section.  

\medskip

Let us first recall how one may obtain the classical spectral curve from a Lax system. When dealing with a Lax system of the form 
\beq \hbar \partial_\lambda \Psi(\lambda,\hbar)=L(\lambda,\hbar) \Psi(\lambda,\hbar),\eeq
it is standard to define the ``classical spectral curve'' as $\underset{\hbar \to 0}{\lim} \det(y I_2-L(\lambda,\hbar))=0$. It is important to note that the classical spectral curve is unaffected by the gauge transformations $\Psi(\lambda,\hbar)\to G(\lambda,\hbar)\Psi(\lambda,\hbar)$ with both $G(\lambda,\hbar)$ and $G(\lambda,\hbar)^{-1}$ regular at $\hbar=0$. Indeed, the conjugation of the Lax matrix does not change the characteristic polynomial and the additional term $\hbar (\partial_\lambda G) G^{-1}$ disappears in the limit $\hbar \to 0$. In other words,
\beq \underset{\hbar \to 0}{\lim} \det(y I_2-L(\lambda,\hbar))=\underset{\hbar \to 0}{\lim} \det(y I_2-\check{L}(\lambda,\hbar))=\underset{\hbar \to 0}{\lim} \det(y I_2-\td{L}(\lambda,\hbar)).\eeq
In our case, the general expression of the matrix $L(\lambda,\hbar)$ implies that the classical spectral curve is
\beq \label{ClassicalSpectralCurve} y^2-P_1(\lambda,\hbar=0)y+P_2(\lambda,\hbar=0)=0.\eeq
For generic coefficients, its smooth compactification defines a Riemann surface $\Sigma$ of genus $g=r_\infty-3+\underset{s=1}{\overset{n}{\sum}}r_s$ whose coefficients are determined by \eqref{P1Coeffs} and \eqref{P2Coeffs}. Note that only $g$ coefficients remain undetermined at this stage. On a smooth local chart with a fixed homology basis, they can be parametrized by the filling fractions $\left(\epsilon_i\right)_{1\leq i\leq g}$ naturally associated to some period integrals on this Riemann surface. The asymptotic expansions of the differential form $y\,dx$ at each pole are directly related to the asymptotics of the wave functions \eqref{PsiAsymptotics} since we have
\bea y(z)&\overset{z\to \infty^{(i)}}{=}&-\sum_{k=0}^{r_\infty-1} t_{\infty^{(i)},k} x(z)^{k-1}+ O\left((x(z))^{-2}\right),\cr
y(z)&\overset{z\to X_s^{(i)}}{=}&\sum_{k=0}^{r_s-1} t_{X_s^{(i)},k} (x(z)-X_s)^{-k-1}+ O\left(1\right).
\eea
Finally, remark that the shift by $\hbar$ of $t_{\infty^{(2)},0}$ vanishes in the limit $\hbar \to 0$. In particular, the study of the residues of the classical spectral curve implies that
\beq \label{SumResidues} 0= t_{\infty^{(1)},0}+t_{\infty^{(2)},0}+\sum_{s=1}^n (t_{X_s^{(1)},0}+t_{X_s^{(2)},0}).
\eeq

\medskip

Let us now discuss the connection of the present work with the Chekhov-Eynard-Orantin topological recursion \cite{CE061,EO07,EORev}. Recent works \cite{MOsl2,Quantization_2021} have shown how to quantize the classical spectral curve using topological recursion. Indeed, applying the topological recursion to the classical spectral curve \eqref{ClassicalSpectralCurve} generates Eynard-Orantin differentials $\left(\om_{h,n}\right)_{h\geq 0, n\geq 0}$ that can be regrouped into formal $\hbar$-transseries to define formal wave functions $\left(\psi_{1}^{\text{TR}},\psi_{2}^{\text{TR}}\right)$ that satisfy a quantum curve, i.e. a second-order linear ODE with pole singularities in $\mathcal{R}$ and apparent singularities at $\lambda=q_i$ and whose $\hbar\to 0$ limit recovers the classical spectral curve. To identify this quantum curve with the present oper, one must match
the quantum accessory data and the transseries parameters, in addition
to the classical spectral curve, and use compatible normalizations.
If these data are matched so that the two scalar equations coincide,
their fundamental matrices satisfy 
\beq \Psi(\lambda,\hbar)= \begin{pmatrix} \psi_{1}^{\text{TR}}(\lambda,\hbar)& \psi_{2}^{\text{TR}}(\lambda,\hbar) \\ \hbar \partial_\lambda \psi_{1}^{\text{TR}}(\lambda,\hbar)&\hbar \partial_\lambda \psi_{2}^{\text{TR}}(\lambda,\hbar)\end{pmatrix}C\eeq
where $C$ is a constant (independent of $\lambda$) matrix. Thus the equality is conditional on matching the full quantum
equation and the solution data; equality of classical spectral curves
alone does not identify the full Lax system. However the price to pay in this point of view is the mandatory introduction of the formal parameter $\hbar$ to define the formal $\hbar$-transseries and then $\left(\psi_{1}^{\text{TR}},\psi_{2}^{\text{TR}}\right)$. As explained in Section \ref{SectionIntrohbar}, this formal parameter can be removed by proper rescaling at the level of the Lax system but it is unclear how the topological recursion wave functions may be defined after this rescaling since there is no longer a formal parameter to define the series. This issue is closely related to the analytical meaning that might be given to the formal $\hbar$-transseries. In particular, it is presently unclear how to resum analytically the $\hbar$-transseries to obtain non-formal quantities but current works are in progress to tackle this problem. In particular, the main issue at stake for the content of this paper is the following: even if the formal $\hbar$-transseries wave functions may have analytic realizations in neighborhoods of $\hbar=0$ (using for example works of N. Nikolaev \cite{nikolaev_2022,nikolaev_2022_2} or work of O. Costin and R. D. Costin \cite{Costin-P1,CostinCostin01}), it is unknown if one may extend these analytical objects up to $\hbar=1$ that corresponds to the natural value of the parameter from the geometric perspective.

\begin{remark}[Topological Type property]The construction of the Lax pairs and the Hamiltonian systems presented in this paper is independent of the type of solutions that one may look for. According to \cite{Quantization_2021}, the most general solutions of the Lax system are expected to be $\hbar$-transseries. However, one may look for simpler solutions. Of particular interest are formal power series solutions of the Hamiltonian systems:
\beq \hat{q}_i(\boldsymbol{\tau},\hbar)=\sum_{k=0}^{\infty} q_i^{(k)}(\boldsymbol{\tau}) \hbar^k \,\,,\,\, \hat{p}_i(\boldsymbol{\tau},\hbar)= \sum_{k=0}^{\infty} p_i^{(k)}(\boldsymbol{\tau}) \hbar^k \,\,,\,\, \forall \, i\in \llbracket 1,g\rrbracket\eeq
Away from poles and turning points, the associated Lax system admits local formal WKB fundamental matrices of the form
\beq\hat{\Psi}(\lambda,\boldsymbol{\tau},\hbar)
=\left(\sum_{k=0}^{\infty}\Psi_k(\lambda,\boldsymbol{\tau})\hbar^k\right)
\begin{pmatrix}e^{S_1(\lambda,\boldsymbol{\tau})/\hbar}&0\\
0&e^{S_2(\lambda,\boldsymbol{\tau})/\hbar}\end{pmatrix},\eeq
where $\det\Psi_0\ne0$ and $\partial_\lambda S_i=y_i$ on the two local sheets of the classical spectral curve. The topological type property provides a framework for reconstructing formal correlation functions by topological recursion; see \cite{BergereBorotEynard}. The reconstruction theorem of \cite{MOsl2}
applies to the formal systems constructed there, with the stated
$\hbar$-scaling and genericity assumptions. A useful ingredient is an
independent deformation whose traceless auxiliary matrix has the form
\[
\td{A}_\tau(\lambda,\hbar)
=\frac{B_1\lambda+B_0}{p(\lambda)},
\]
where $B_0,B_1$ are independent of $\lambda$ and $p$ is a nonzero
scalar polynomial. In the present canonical trace-free normalization,
suitable directions are $\tau_{\infty,1}$ when $r_\infty\ge4$,
$\tau_{X_s,1}$ when it is independent, and an independent position
$\td{X}_s$ of a regular singularity ($r_s=1$).
The infinity-time auxiliary matrix is polynomial of degree at most one.
For either of these finite-pole directions, the auxiliary matrix has
at most a simple pole at $X_s$ and is bounded at infinity, so it has
the form $A_\infty+A_s/(\lambda-X_s)$.
The highest-index times do not generally have this property:
in the genus-two example, $\td{A}_{\tau_{\infty,2}}$ is
quadratic in $\lambda$, whereas
$\td{A}_{\tau_{\infty,1}}$ is linear.
Application of the topological-type theorem additionally requires
matching the formal solution branch and the hypotheses of
\cite{MOsl2}; the auxiliary-matrix shape alone is not a verification
of those hypotheses for every solution. When these conditions hold,
the corresponding formal correlation functions, defined by the determinantal formulas of \cite{bergre2009determinantal}, and tau-function are
reconstructed by topological recursion as described in that reference.
For a related identification of a matrix-model partition function with an
isomonodromic tau-function in the two-matrix model, see \cite{BertolaMarchal2008}.
\end{remark}

\section{General isomonodromic deformations and auxiliary matrices}\label{SectionAuxi}
\subsection{Definition of general isomonodromic deformations}
Section \ref{Section2Mero} provides a natural set of parameters for which we may consider deformations, namely the irregular times $\left(t_{\infty^{(i)},k}\right)_{1\leq k\leq r_\infty-1, 1\leq i\leq 2}$ and $\left(t_{X_s^{(i)},k}\right)_{1\leq i\leq 2, 1\leq s\leq n, 1\leq k\leq r_s-1}$ and the location of the poles $\left(X_s\right)_{1\leq s\leq n}$. In order to study deformations with respect to these parameters we introduce the following definition.

\begin{definition}\label{DefGeneralDeformationsDefinition} We define the following general deformation operators.
\beq \label{GeneralDeformationsDefinition}\mathcal{L}_{\boldsymbol{\alpha}}=\hbar \sum_{i=1}^2\sum_{k=1}^{r_\infty-1} \alpha_{\infty^{(i)},k} \partial_{t_{\infty^{(i)},k}}+\hbar \sum_{i=1}^2\sum_{s=1}^n\sum_{k=1}^{r_s-1} \alpha_{X_s^{(i)},k} \partial_{t_{X_s^{(i)},k}}+ \hbar\sum_{s=1}^n \alpha_{X_s} \partial_{X_s}\eeq
where we define the vector $\boldsymbol{\alpha}\in \mathbb{C}^{2r_{\infty}-2+2\underset{s=1}{\overset{n}{\sum}}r_s -2n+n}=\mathbb{C}^{2g+4-n}$ by
\beq \boldsymbol{\alpha}= \sum_{i=1}^2\sum_{k=1}^{r_\infty-1} \alpha_{\infty^{(i)},k}\mathbf{e}_{\infty^{(i)},k}+\sum_{i=1}^2\sum_{s=1}^n\sum_{k=1}^{r_s-1} \alpha_{X_s^{(i)},k} \mathbf{e}_{X_s^{(i)},k}+\sum_{s=1}^n \alpha_{X_s}\mathbf{e}_{X_s}.\eeq
\end{definition}

It is important to notice that we do not consider deformations with respect to either $\left(t_{\infty^{(i)},0}\right)_{1\leq i\leq 2}$ or $\left(t_{X_s^{(i)},0}\right)_{1\leq i\leq 2,1\leq s\leq n}$ since that would affect the monodromy data at each pole. Thus, deformations defined by Definition \ref{DefGeneralDeformationsDefinition} shall be seen as ``general isomonodromic deformations'' in $\hat{F}_{\mathcal{R},\mathbf{r}}$ (\cite{Yamakawa2017TauFA}). Moreover, we stress that the coefficients of the vector $\boldsymbol{\alpha}$ are allowed to depend on the position of the poles $X_s$ and on all the coefficients of the singular parts of the wave functions (including $\left(t_{\infty^{(i)},0}\right)_{1\leq i\leq 2}$ and $\left(t_{X_s^{(i)},0}\right)_{1\leq i\leq 2,1\leq s\leq n}$).

\medskip
 
Associated to a vector $\boldsymbol{\alpha}$ are general auxiliary Lax matrices $\td{A}_{\boldsymbol{\alpha}}(\lambda)$, $\check{A}_{\boldsymbol{\alpha}}(\lambda)$  and $A_{\boldsymbol{\alpha}}(\lambda)$ defined by
\bea \td{A}_{\boldsymbol{\alpha}}(\lambda)=\mathcal{L}_{\boldsymbol{\alpha}}[\td{\Psi}(\lambda)] \td{\Psi}^{-1}(\lambda)
 \,\, &\Leftrightarrow&\,\, \mathcal{L}_{\boldsymbol{\alpha}}[\td{\Psi}(\lambda)] =\td{A}_{\boldsymbol{\alpha}}(\lambda) \td{\Psi}(\lambda)\cr
\check{A}_{\boldsymbol{\alpha}}(\lambda)=\mathcal{L}_{\boldsymbol{\alpha}}[\check{\Psi}(\lambda)] \check{\Psi}^{-1}(\lambda)
 \,\, &\Leftrightarrow&\,\, \mathcal{L}_{\boldsymbol{\alpha}}[\check{\Psi}(\lambda)] =\check{A}_{\boldsymbol{\alpha}}(\lambda) \check{\Psi}(\lambda)\cr
A_{\boldsymbol{\alpha}}(\lambda)=\mathcal{L}_{\boldsymbol{\alpha}}[\Psi(\lambda)] \Psi^{-1}(\lambda) \,\, &\Leftrightarrow&\,\, \mathcal{L}_{\boldsymbol{\alpha}}[\Psi(\lambda)] =A_{\boldsymbol{\alpha}}(\lambda) \Psi(\lambda).
\eea

In particular, $\td{A}_{\boldsymbol{\alpha}}(\lambda)$ and $\check{A}_{\boldsymbol{\alpha}}(\lambda)$ are rational functions of $\lambda$ with only possible poles in $\mathcal{R}$ while $A(\lambda)$ may also have additional poles at $\{q_1,\dots,q_g\}$ (\cite{HURTUBISE20081394,Harnad2007}). Note that $\left(L(\lambda),A_{\boldsymbol{\alpha}}(\lambda)\right)$, $\left(\check{L}(\lambda),\check{A}_{\boldsymbol{\alpha}}(\lambda)\right)$ and $\left(\td{L}(\lambda),\td{A}_{\boldsymbol{\alpha}}(\lambda)\right)$ provide equivalent Lax pairs (i.e. corresponding to the same isomonodromic deformations and providing the same Hamiltonian system) but expressed in three different gauges. The corresponding compatibility equations are
\bea \label{CompatibilityEquation}\mathcal{L}_{\boldsymbol{\alpha}}[L]&=&[A_{\boldsymbol{\alpha}},L]+\hbar\partial_\lambda A_{\boldsymbol{\alpha}}\cr
\mathcal{L}_{\boldsymbol{\alpha}}[\check{L}]&=&[\check{A}_{\boldsymbol{\alpha}},\check{L}]+\hbar\partial_\lambda \check{A}_{\boldsymbol{\alpha}}\cr
\mathcal{L}_{\boldsymbol{\alpha}}[\td{L}]&=&[\td{A}_{\boldsymbol{\alpha}},\td{L}]+\hbar\partial_\lambda \td{A}_{\boldsymbol{\alpha}}.
\eea

We shall now use the asymptotic expansions of the wave matrices in order to obtain information on the general form of the auxiliary matrices. Then, we shall use the compatibility equations in order to determine the evolutions of the Darboux coordinates under general isomonodromic deformations and prove that these evolutions are Hamiltonian. 

\subsection{General form of the auxiliary matrix $A_{\boldsymbol{\alpha}}(\lambda,\hbar)$}

Using compatibility equations one may easily obtain two of the entries of $A_{\boldsymbol{\alpha}}(\lambda)$. Indeed, since $L$ is a companion-like matrix, compatibility equations \eqref{CompatibilityEquation} imply that
\bea \label{TrivialEntriesA}\left[A_{\boldsymbol{\alpha}}(\lambda)\right]_{2,1}&=&\hbar \partial_{\lambda} \left[A_{\boldsymbol{\alpha}}(\lambda)\right]_{1,1}+\left[A_{\boldsymbol{\alpha}}(\lambda)\right]_{1,2}L_{2,1}(\lambda),\cr
\left[A_{\boldsymbol{\alpha}}(\lambda)\right]_{2,2}&=&\hbar \partial_{\lambda} \left[A_{\boldsymbol{\alpha}}(\lambda)\right]_{1,2}+\left[A_{\boldsymbol{\alpha}}(\lambda)\right]_{1,1}+\left[A_{\boldsymbol{\alpha}}(\lambda)\right]_{1,2}L_{2,2}(\lambda),
\eea
so that only the first line of $A_{\boldsymbol{\alpha}}(\lambda)$ remains unknown at this stage. The other two entries of the compatibility equation \eqref{CompatibilityEquation} lead to
\bea \label{Compat}\mathcal{L}_{\boldsymbol{\alpha}}[L_{2,1}(\lambda)]&=&\hbar^2 \frac{\partial^2 \left[A_{\boldsymbol{\alpha}}(\lambda)\right]_{1,1}}{\partial \lambda^2} +2\hbar L_{2,1}(\lambda)\, \partial_\lambda \left[A_{\boldsymbol{\alpha}}(\lambda)\right]_{1,2}+\hbar \left[A_{\boldsymbol{\alpha}}(\lambda)\right]_{1,2} \, \partial_\lambda L_{2,1}(\lambda)\cr
&&- \hbar L_{2,2}(\lambda)\, \partial_{\lambda} \left[A_{\boldsymbol{\alpha}}(\lambda)\right]_{1,1},\cr
\mathcal{L}_{\boldsymbol{\alpha}}[L_{2,2}(\lambda)]&=&\hbar^2 \frac{\partial^2 \left[A_{\boldsymbol{\alpha}}(\lambda)\right]_{1,2}}{\partial \lambda^2} +2\hbar \partial_\lambda \left[A_{\boldsymbol{\alpha}}(\lambda)\right]_{1,1}+\hbar L_{2,2}(\lambda)\,\partial_\lambda \left[A_{\boldsymbol{\alpha}}(\lambda)\right]_{1,2}\cr
&&+ \hbar \left[A_{\boldsymbol{\alpha}}(\lambda)\right]_{1,2}\,\partial_{\lambda}L_{2,2}(\lambda)
\eea
that shall be used later to determine the evolution equations for $\left(q_i,p_i\right)_{1\leq i\leq g}$. Before studying the compatibility equations, let us observe that the asymptotic expansions of the wave matrix $\Psi$ at each pole allow us to determine the general form of the auxiliary matrix $A_{\boldsymbol{\alpha}}(\lambda,\hbar)$. Indeed, we get the following results.

\begin{proposition}\label{PropAsymptoticExpansionA12} The asymptotic expansions of entry $\left[A_{\boldsymbol{\alpha}}(\lambda)\right]_{1,2}$ at each pole are given by
\bea \left[A_{\boldsymbol{\alpha}}(\lambda)\right]_{1,2}&\overset{\lambda\to \infty}{=}&\sum_{i=-1}^{r_\infty-3} \frac{\nu^{(\boldsymbol{\alpha})}_{\infty,i}}{\lambda^i} +O\left(\lambda^{-(r_\infty-2)}\right),\cr
\left[A_{\boldsymbol{\alpha}}(\lambda)\right]_{1,2}&\overset{\lambda\to X_s}{=}&\sum_{i=0}^{r_s-1} \nu^{(\boldsymbol{\alpha})}_{X_s,i}(\lambda-X_s)^i +O\left((\lambda-X_s)^{r_s}\right)
\eea
Coefficients $\left(\nu^{(\boldsymbol{\alpha})}_{\infty,k}\right)_{-1\leq k\leq r_\infty-3}$ and $\left(\nu^{(\boldsymbol{\alpha})}_{X_s,i}\right)_{1\leq s\leq n,0\leq i\leq r_s-1}$ are determined by 
\beq \label{RelationNuAlphas} \forall \, s\in\llbracket 1,n\rrbracket \,:\,\nu^{(\boldsymbol{\alpha})}_{X_s,0}=-\alpha_{X_s}\, \text{ and }\,  M_s\begin{pmatrix}  \nu^{(\boldsymbol{\alpha})}_{X_s,1}\\ \vdots\\ \vdots \\ \nu^{(\boldsymbol{\alpha})}_{X_s,r_s-1}\end{pmatrix}=\begin{pmatrix} 
-\frac{\alpha_{X_s^{(1)},r_s-1}-\alpha_{X_s^{(2)},r_s-1}}{r_s-1}\\
\vdots\\
\vdots\\
-\frac{\alpha_{X_s^{(1)},1}-\alpha_{X_s^{(2)},1}}{1}
\end{pmatrix}\eeq
For $r_s=1$, the matrix $M_s$ and its associated positive-index
vectors are empty and are omitted throughout; only
$\nu^{(\boldsymbol{\alpha})}_{X_s,0}=-\alpha_{X_s}$ remains in
\eqref{RelationNuAlphas}. For $r_s\ge2$, $M_s$ is the
$(r_s-1)\times(r_s-1)$ lower triangular Toeplitz matrix, independent
of the deformation vector $\boldsymbol{\alpha}$, given by

\footnotesize{\beq\label{MatrixMs} M_s=\begin{pmatrix}(t_{X_s^{(1)},r_s-1}-t_{X_s^{(2)},r_s-1})&0&\dots& &\dots &0\\
(t_{X_s^{(1)},r_s-2}-t_{X_s^{(2)},r_s-2})&(t_{X_s^{(1)},r_s-1}-t_{X_s^{(2)},r_s-1})& 0& & &\vdots\\
\vdots & \ddots&\ddots &\ddots  & &\vdots\\
\vdots &\ddots&\ddots&\ddots&0&\vdots\\
(t_{X_s^{(1)},2}-t_{X_s^{(2)},2})&\ddots &\ddots&\ddots& (t_{X_s^{(1)},r_s-1}-t_{X_s^{(2)},r_s-1})&0\\
(t_{X_s^{(1)},1}-t_{X_s^{(2)},1})&(t_{X_s^{(1)},2}-t_{X_s^{(2)},2})& \dots & & (t_{X_s^{(1)},r_s-2}-t_{X_s^{(2)},r_s-2})& (t_{X_s^{(1)},r_s-1}-t_{X_s^{(2)},r_s-1})
 \end{pmatrix}.
\eeq}\normalsize{}
The situation at $\infty$ is similar but depends on the value of $r_\infty$.
\begin{itemize}
\item For $r_\infty\geq 3$ we have
\beq \label{RelationNuAlphaInfty} M_\infty\begin{pmatrix} \nu^{(\boldsymbol{\alpha})}_{\infty,-1}\\ \nu^{(\boldsymbol{\alpha})}_{\infty,0}\\ \vdots \\ \nu^{(\boldsymbol{\alpha})}_{\infty,r_\infty-3}\end{pmatrix}=\begin{pmatrix} \frac{\alpha_{\infty^{(1)},r_\infty-1}-\alpha_{\infty^{(2)},r_\infty-1}}{r_\infty-1} \\\frac{\alpha_{\infty^{(1)},r_\infty-2}-\alpha_{\infty^{(2)},r_\infty-2}}{r_\infty-2}\\ \vdots \\ \frac{\alpha_{\infty^{(1)},1}-\alpha_{\infty^{(2)},1}}{1}\end{pmatrix}\eeq

where $M_\infty$ is a lower triangular Toeplitz matrix of size $(r_\infty-1)\times(r_\infty-1)$ independent of the deformation $\boldsymbol{\alpha}$,

\footnotesize{\beq\label{MatrixMInfty} M_\infty=\begin{pmatrix}(t_{\infty^{(1)},r_\infty-1}-t_{\infty^{(2)},r_\infty-1})&0&\dots &\dots &0\\
\vdots &\ddots &\ddots  & &\vdots\\
\vdots &&\ddots&0&0\\
(t_{\infty^{(1)},2}-t_{\infty^{(2)},2}) &\dots&& (t_{\infty^{(1)},r_\infty-1}-t_{\infty^{(2)},r_\infty-1})&0\\
(t_{\infty^{(1)},1}-t_{\infty^{(2)},1})& \dots & & (t_{\infty^{(1)},r_\infty-2}-t_{\infty^{(2)},r_\infty-2})& (t_{\infty^{(1)},r_\infty-1}-t_{\infty^{(2)},r_\infty-1})
 \end{pmatrix}.
\eeq}\normalsize{}

\item For $r_\infty=2$, $M_\infty$ is a $1\times 1$ matrix and
\beq \label{RelationNuAlphaInftyrinftyequal2} 
M_\infty \nu^{(\boldsymbol{\alpha})}_{\infty,-1}=(t_{\infty^{(1)},1}-t_{\infty^{(2)},1})\nu^{(\boldsymbol{\alpha})}_{\infty,-1}=\alpha_{\infty^{(1)},1}-\alpha_{\infty^{(2)},1}.
\eeq
Note in particular that $\nu^{(\boldsymbol{\alpha})}_{\infty,0}$ is not determined.
\item For $r_\infty=1$, $M_\infty$ is not defined and neither $\nu^{(\boldsymbol{\alpha})}_{\infty,-1}$ nor $\nu^{(\boldsymbol{\alpha})}_{\infty,0}$ is determined.
\end{itemize}
\end{proposition}

\begin{proof}The proof is given in Appendix \ref{AppendixExpansionA}.
\end{proof}

The previous proposition may be used to determine the general form of the entry $\left[A_{\boldsymbol{\alpha}}(\lambda)\right]_{1,2}$.

\begin{proposition}\label{PropA12Form} Entry $\left[A_{\boldsymbol{\alpha}}(\lambda)\right]_{1,2}$ is given by
\beq \label{ExpressionA12} \left[A_{\boldsymbol{\alpha}}(\lambda)\right]_{1,2}=\nu^{(\boldsymbol{\alpha})}_{\infty,-1}\lambda+\nu^{(\boldsymbol{\alpha})}_{\infty,0} + \sum_{j=1}^g \frac{\mu^{(\boldsymbol{\alpha})}_j}{\lambda-q_j}.\eeq
Coefficients $\left(\mu^{(\boldsymbol{\alpha})}_j\right)_{1\leq j\leq g}$ are determined by the linear system
\beq \label{RelationNuMuMatrixForm} \mathbf{V}\begin{pmatrix}\mu^{(\boldsymbol{\alpha})}_1\\ \vdots\\ \vdots\\\mu^{(\boldsymbol{\alpha})}_g\end{pmatrix}\overset{\text{def}}{=}\begin{pmatrix} V_\infty\\  V_1 \\ \vdots \\\vdots \\V_n\end{pmatrix}\begin{pmatrix} \mu^{(\boldsymbol{\alpha})}_1\\ \vdots\\\vdots\\ \mu^{(\boldsymbol{\alpha})}_g\end{pmatrix}= \begin{pmatrix} \boldsymbol{\nu}^{(\boldsymbol{\alpha})}_\infty\\\boldsymbol{\nu}^{(\boldsymbol{\alpha})}_{X_1} \\ \vdots\\  \boldsymbol{\nu}^{(\boldsymbol{\alpha})}_{X_n}\end{pmatrix} \eeq
where $V_\infty$ is a $\max(0,r_\infty-3)\times g$ matrix (empty for $r_\infty\leq3$) and $\left(V_s\right)_{1\leq s\leq n}$ are $r_s\times g$ matrices given by
\beq \label{DefVinfty}V_\infty=\begin{pmatrix}1&1 &\dots &\dots &1\\
q_1& q_2&\dots &\dots& q_{g}\\
\vdots & & & & \vdots\\
\vdots & & & & \vdots\\
q_1^{r_\infty-4}& q_2^{r_\infty-4} &\dots & \dots& q_{g}^{r_\infty-4}\end{pmatrix}\,,\, 
V_s=\begin{pmatrix}\frac{1}{q_1-X_s}& \dots &\dots& \frac{1}{q_g-X_s}\\
\frac{1}{(q_1-X_s)^2}& \dots &\dots& \frac{1}{(q_g-X_s)^2}\\
\vdots & & & \vdots\\
\vdots & & & \vdots\\
\frac{1}{(q_1-X_s)^{r_s}}& \dots &\dots& \frac{1}{(q_g-X_s)^{r_s}}\\
\end{pmatrix}
\eeq

and $\boldsymbol{\nu}^{(\boldsymbol{\alpha})}_\infty\in \mathbb{C}^{\max(0,r_\infty-3)}$, $\boldsymbol{\nu}^{(\boldsymbol{\alpha})}_{X_s}\in \mathbb{C}^{r_s}$ are vectors given by
\beq \label{DefRHSmunu} \boldsymbol{\nu}^{(\boldsymbol{\alpha})}_\infty=\begin{pmatrix}\nu^{(\boldsymbol{\alpha})}_{\infty,1}\\ \nu^{(\boldsymbol{\alpha})}_{\infty,2}\\\vdots \\\vdots \\ \nu^{(\boldsymbol{\alpha})}_{\infty,r_\infty-3}\end{pmatrix} \,\,,\,\, \boldsymbol{\nu}^{(\boldsymbol{\alpha})}_{X_s}=\begin{pmatrix} -\nu^{(\boldsymbol{\alpha})}_{{X_s},0}+\nu^{(\boldsymbol{\alpha})}_{\infty,0}+\nu^{(\boldsymbol{\alpha})}_{\infty,-1}X_s\\ -\nu^{(\boldsymbol{\alpha})}_{{X_s},1}+\nu^{(\boldsymbol{\alpha})}_{\infty,-1}\\-\nu^{(\boldsymbol{\alpha})}_{{X_s},2}\\   \vdots \\ -\nu^{(\boldsymbol{\alpha})}_{{X_s},r_s-1}\end{pmatrix}.\eeq
\end{proposition}

\begin{proof}The proof is given in Appendix \ref{ProofEntryA12}.
\end{proof}

Note that for $r_\infty\leq 3$, $V_\infty$ is not defined and shall not be written in \eqref{RelationNuMuMatrixForm}. Thus, for $r_\infty\leq 2$, the previous linear system may look overdetermined since there are more rows than columns. This is not the case because of the following remark.

\begin{remark}\label{RemarkSpecialCases} For $r_\infty\leq 2$, one has to remember that the quantities $\nu^{(\boldsymbol{\alpha})}_{\infty,-1}$ and $\nu^{(\boldsymbol{\alpha})}_{\infty,0}$ may not be determined by Proposition \ref{PropAsymptoticExpansionA12}. For $r_\infty=2$, only $\nu^{(\boldsymbol{\alpha})}_{\infty,-1}$ is determined by Proposition \ref{PropAsymptoticExpansionA12} while for $r_\infty=1$ neither $\nu^{(\boldsymbol{\alpha})}_{\infty,-1}$ nor $\nu^{(\boldsymbol{\alpha})}_{\infty,0}$ is determined by Proposition \ref{PropAsymptoticExpansionA12}. In these cases, one must use the extra rows of \eqref{RelationNuMuMatrixForm} to determine these additional unknown coefficients and then use the information to obtain $\left(\mu^{(\boldsymbol{\alpha})}_j\right)_{1\leq j\leq g}$.
\end{remark}

\begin{remark}For $r_\infty\ge3$, put $m=r_\infty-3$.
With the displayed row order, the square matrix $\mathbf V$ has determinant
\beq \det \mathbf{V}=(-1)^{\frac{m(m-1)}{2}}\frac{\underset{1\leq i<j\leq g}{\prod}(q_i-q_j)}{\underset{i=1}{\overset{g}{\prod}}\underset{s=1}{\overset{n}{\prod}}(q_i-X_s)^{r_s}}\,\,\underset{1\leq s'<s\leq n}{\prod} (X_s-X_{s'})^{r_sr_{s'}}.\eeq
In particular, it is non-zero as soon as $(q_1,\dots,q_g,X_1,\dots,X_n)$ are all distinct. For $r_\infty\in\{1,2\}$ this block is rectangular:
one must use the augmented systems of Remark \ref{RemarkSpecialCases}.
\end{remark}

Finally, one can obtain the general form of entry $\left[A_{\boldsymbol{\alpha}}(\lambda)\right]_{1,1}$. We get the following proposition.

\begin{proposition}\label{Propcalpha}The entry $\left[A_{\boldsymbol{\alpha}}(\lambda)\right]_{1,1}$ is given by
\beq \label{ExpressionA11} \left[A_{\boldsymbol{\alpha}}(\lambda)\right]_{1,1}=\sum_{i=0}^{r_\infty-1}c^{(\boldsymbol{\alpha})}_{\infty,i}\lambda^i+\sum_{s=1}^n\sum_{i=1}^{r_s-1}c^{(\boldsymbol{\alpha})}_{{X_s},i}(\lambda-X_s)^{-i}+\sum_{j=1}^g\frac{\rho^{(\boldsymbol{\alpha})}_j}{\lambda-q_j}.\eeq
with 
\beq \forall\, j\in \llbracket 1,g\rrbracket \,:\, \rho^{(\boldsymbol{\alpha})}_j=-\mu^{(\boldsymbol{\alpha})}_j p_j\eeq
Coefficients $\left(c^{(\boldsymbol{\alpha})}_{\infty,k}\right)_{1\leq k\leq r_\infty-1}$ and $\left(c^{(\boldsymbol{\alpha})}_{X_s,k}\right)_{1\leq s\leq n, 1\leq k\leq r_s-1}$ are determined by
\tiny{\beq \label{Relationckalphainfty}M_\infty\begin{pmatrix}c^{(\boldsymbol{\alpha})}_{\infty, r_\infty-1} \\ c^{(\boldsymbol{\alpha})}_{\infty, r_\infty-2}\\ \vdots\\ c^{(\boldsymbol{\alpha})}_{\infty, i}\\ \vdots\\ c^{(\boldsymbol{\alpha})}_{\infty, 1}\end{pmatrix}=\begin{pmatrix}\frac{t_{\infty^{(2)},r_\infty-1}\alpha_{\infty^{(1)},r_\infty-1}-t_{\infty^{(1)},r_\infty-1}\alpha_{\infty^{(2)},r_\infty-1}}{r_\infty-1}\\
\frac{t_{\infty^{(2)},r_\infty-1}\alpha_{\infty^{(1)},r_\infty-2}-t_{\infty^{(1)},r_\infty-1}\alpha_{\infty^{(2)},r_\infty-2}}{r_\infty-2}+\frac{t_{\infty^{(2)},r_\infty-2}\alpha_{\infty^{(1)},r_\infty-1}-t_{\infty^{(1)},r_\infty-2}\alpha_{\infty^{(2)},r_\infty-1}}{r_\infty-1}\\
\vdots\\
\underset{k=i}{\overset{r_\infty-1}{\sum}} \frac{t_{\infty^{(2)},r_\infty-1+i-k}\alpha_{\infty^{(1)},k}-t_{\infty^{(1)},r_\infty-1+i-k}\alpha_{\infty^{(2)},k}}{k}\\
\vdots\\
\frac{t_{\infty^{(2)},r_\infty-1}\alpha_{\infty^{(1)},1}-t_{\infty^{(1)},r_\infty-1}\alpha_{\infty^{(2)},1}}{1}
+\dots+ \frac{t_{\infty^{(2)},1}\alpha_{\infty^{(1)},r_\infty-1}-t_{\infty^{(1)},1}\alpha_{\infty^{(2)},r_\infty-1}}{r_\infty-1}
 \end{pmatrix}
\eeq}\normalsize{}
with the matrix $M_\infty$ given by \eqref{MatrixMInfty}. Similarly, for all $s\in \llbracket 1,n\rrbracket$,
\footnotesize{\beq \label{Relationckalphas}M_s\begin{pmatrix}c^{(\boldsymbol{\alpha})}_{{X_s}, r_s-1} \\ c^{(\boldsymbol{\alpha})}_{{X_s}, r_s-2}\\ \vdots\\ c^{(\boldsymbol{\alpha})}_{{X_s},i}\\ \vdots\\ c^{(\boldsymbol{\alpha})}_{{X_s}, 1}\end{pmatrix}=\begin{pmatrix}\frac{t_{X_s^{(2)},r_s-1}\alpha_{X_s^{(1)},r_s-1}-t_{X_s^{(1)},r_s-1}\alpha_{X_s^{(2)},r_s-1}}{r_s-1}\\
\frac{t_{X_s^{(2)},r_s-1}\alpha_{X_s^{(1)},r_s-2}-t_{X_s^{(1)},r_s-1}\alpha_{X_s^{(2)},r_s-2}}{r_s-2}+\frac{t_{X_s^{(2)},r_s-2}\alpha_{X_s^{(1)},r_s-1}-t_{X_s^{(1)},r_s-2}\alpha_{X_s^{(2)},r_s-1}}{r_s-1}\\
\vdots\\
\underset{k=i}{\overset{r_s-1}{\sum}} \frac{t_{X_s^{(2)},r_s-1+i-k}\alpha_{X_s^{(1)},k}-t_{X_s^{(1)},r_s-1+i-k}\alpha_{X_s^{(2)},k}}{k}\\
\vdots\\
\frac{t_{X_s^{(2)},r_s-1}\alpha_{X_s^{(1)},1}-t_{X_s^{(1)},r_s-1}\alpha_{X_s^{(2)},1}}{1}
+\dots+ \frac{t_{X_s^{(2)},1}\alpha_{X_s^{(1)},r_s-1}-t_{X_s^{(1)},1}\alpha_{X_s^{(2)},r_s-1}}{r_s-1}
 \end{pmatrix}
\eeq}\normalsize{}
with the matrices $\left(M_s\right)_{1\leq s\leq n}$ given by \eqref{MatrixMs}.
\end{proposition}

\begin{proof}The proof is given in Appendix \ref{AppendixA11}.
\end{proof}

In the previous propositions, one may easily observe that quantities like $M_\infty$, $\left(M_s\right)_{1\leq s\leq n}$, $\left(\nu^{(\boldsymbol{\alpha})}_{\infty,k}\right)_{1\leq k\leq r_\infty-3}$, $\left(\nu^{(\boldsymbol{\alpha})}_{{X_s},k}\right)_{1\leq s\leq n, 1\leq k\leq r_s-1}$, $\left(c^{(\boldsymbol{\alpha})}_{{X_s},k}\right)_{1\leq s\leq n,1\leq k\leq r_s-1}$, $\left(c^{(\boldsymbol{\alpha})}_{\infty, k}\right)_{1\leq k\leq r_\infty-1}$ are independent of Darboux coordinates and are only determined by the monodromies, the irregular times and the coefficients of the deformation $\boldsymbol{\alpha}$. On the contrary, quantities like $\left(\mu^{(\boldsymbol{\alpha})}_j\right)_{1\leq j\leq g}$, $\left(\rho^{(\boldsymbol{\alpha})}_j\right)_{1\leq j\leq g}$, $V_\infty$ and $(V_s)_{1\leq s\leq n}$ depend on the Darboux coordinates. The situation is more complicated for $\nu^{(\boldsymbol{\alpha})}_{\infty,-1}$ and $\nu^{(\boldsymbol{\alpha})}_{\infty,0}$ since they do not depend on the Darboux coordinates for $r_\infty\geq 3$ but depend on them for $r_\infty\leq 2$ according to Remark \ref{RemarkSpecialCases}.

\begin{remark}
The scalar coefficient $c^{(\boldsymbol{\alpha})}_{\infty,0}$ in
$[A_{\boldsymbol{\alpha}}]_{1,1}$ depends on the scalar normalization of
the wave matrix. A scalar gauge independent of $\lambda$ leaves
$\td{L}_{1,2}$ unchanged and therefore cannot be fixed by normalizing
that entry. In the Hamiltonian convention used below, the scalar coefficient
$c^{(\boldsymbol{\alpha})}_{\infty,0}$ is omitted: the polynomial $c$-sum
starts at $k=1$. This convention does not require a particular scalar
normalization of the auxiliary matrix.
\end{remark}

\section{General Hamiltonian evolutions}
The previous sections provide the general form of the matrices $L(\lambda,\hbar)$ and $A_{\boldsymbol{\alpha}}(\lambda)$ through Propositions \ref{PropLaxMatrix}, \ref{PropAsymptoticExpansionA12}, \ref{PropA12Form}, \ref{Propcalpha} and equation \eqref{TrivialEntriesA}. As we shall see below, inserting this previous knowledge into the compatibility equations \eqref{Compat} provides the evolutions of the Darboux coordinates.

\medskip 
The first step is to look at order $(\lambda-q_j)^{-2}$ in $\mathcal{L}_{\boldsymbol{\alpha}}[L_{2,2}(\lambda)]$. We obtain, for all $j\in \llbracket 1, g\rrbracket$:
\beq \label{Lqj}\mathcal{L}_{\boldsymbol{\alpha}}[q_j]=2\mu^{(\boldsymbol{\alpha})}_j\left(p_j -\frac{1}{2}P_1(q_j)+\frac{1}{2} \sum_{s=1}^n\frac{\hbar r_s}{q_j-X_s}\right)-\hbar \nu^{(\boldsymbol{\alpha})}_{\infty,0} -\hbar \nu^{(\boldsymbol{\alpha})}_{\infty,-1} q_j-\hbar \sum_{i\neq j}\frac{\mu^{(\boldsymbol{\alpha})}_j+\mu^{(\boldsymbol{\alpha})}_i}{q_j-q_i}.\eeq

The next step is to determine the coefficients $\left(H_{\infty,j}\right)_{0\leq j\leq r_\infty-4}$ and $\left(H_{X_s,j}\right)_{1\leq s\leq n, 1\leq j\leq r_s}$ that remain unknown in $L_{2,1}(\lambda)$. To achieve this task, we look at order $(\lambda-q_j)^{-2}$ in $\mathcal{L}_{\boldsymbol{\alpha}}[L_{2,1}(\lambda)]$ using \eqref{Compat}. We obtain

\footnotesize{
\bea -\hbar p_j \mathcal{L}_{\boldsymbol{\alpha}}[q_j]&=&-2\hbar \mu^{(\boldsymbol{\alpha})}_j\left(
-\td{P}_2(q_j) +\sum_{k=0}^{r_\infty-4}H_{\infty,k}q_j^k+\sum_{s=1}^n\sum_{k=1}^{r_s}H_{X_s,k}(q_j-X_s)^{-k}-\hbar t_{\infty^{(1)},r_\infty-1}q_j^{r_\infty-3}\delta_{r_\infty\geq 3}-\sum_{i\neq j} \frac{\hbar p_i}{q_j-q_i}\right)\cr
&&+\hbar^2 p_j\left( \nu^{(\boldsymbol{\alpha})}_{\infty,-1} q_j+\nu^{(\boldsymbol{\alpha})}_{\infty,0}+\sum_{i\neq j} \frac{\mu^{(\boldsymbol{\alpha})}_i}{q_j-q_i}\right)-\hbar \mu^{(\boldsymbol{\alpha})}_j p_j\left(P_1(q_j)- \sum_{s=1}^n\frac{ \hbar r_s}{q_j-X_s}+\sum_{i\neq j} \frac{\hbar}{q_j-q_i}\right).\cr 
&&
\eea}
\normalsize{}
Inserting \eqref{Lqj} provides, for all $j\in \llbracket 1,g\rrbracket$,
\bea \label{DefCi}\sum_{k=0}^{r_\infty-4}H_{\infty,k}q_j^k+\sum_{s=1}^n\sum_{k=1}^{r_s}H_{X_s,k}(q_j-X_s)^{-k}&=&p_j^2 -P_1(q_j)p_j+p_j\sum_{s=1}^n \frac{ \hbar r_s}{q_j-X_s} +\td{P}_2(q_j)\cr
&&+\hbar \sum_{i\neq j}\frac{p_i-p_j}{q_j-q_i}+\hbar t_{\infty^{(1)},r_\infty-1}q_j^{r_\infty-3}\delta_{r_\infty\geq 3}\cr
&&
\eea
where it is obvious that the r.h.s. is independent of the deformation vector $\boldsymbol{\alpha}$. The last relation can be rewritten into a matrix form.

\begin{proposition}\label{PropDefCi2} We have 
\footnotesize{\beq \label{DefCi2}
\begin{pmatrix} V_\infty^{t}& V_1^t&\dots &V_n^{t}\end{pmatrix}\begin{pmatrix}\mathbf{H}_{\infty}\\\mathbf{H}_{X_1}\\ \vdots \\ \mathbf{H}_{X_n} \end{pmatrix}=\begin{pmatrix} p_1^2- P_1(q_1)p_1 + p_1\underset{s=1}{\overset{n}{\sum}} \frac{\hbar r_s}{q_1-X_s}+\td{P}_2(q_1)+\hbar \underset{i\neq 1}{\sum}\frac{p_i-p_1}{q_1-q_i}+\hbar t_{\infty^{(1)},r_\infty-1}q_1^{r_\infty-3}\delta_{r_\infty\geq 3}\\
\vdots\\
\vdots\\
p_g^2- P_1(q_g)p_g + p_g\underset{s=1}{\overset{n}{\sum}} \frac{\hbar r_s}{q_g-X_s}+\td{P}_2(q_g)+\hbar \underset{i\neq g}{\sum}\frac{p_i-p_g}{q_g-q_i}+\hbar t_{\infty^{(1)},r_\infty-1}q_g^{r_\infty-3}\delta_{r_\infty\geq 3}
\end{pmatrix}
\eeq}
\normalsize{with} $\mathbf{H}_{\infty}=(H_{\infty,0},\dots H_{\infty,r_\infty-4})^t$ (null if $r_\infty \leq 3$), and, for all $s\in \llbracket 1,n\rrbracket$, $\mathbf{H}_{X_s}=(H_{X_s,1},\dots, H_{X_s,r_s})^t$. We recall here that the matrices $V_\infty$ and $(V_s)_{1\leq s\leq n}$ are defined by \eqref{DefVinfty}.

For $r_\infty=2$, the previous linear system has to be complemented with the additional relation \eqref{ConditionsAddrinftyequal2} 
\beq \label{rinfty2Special}\sum_{s=1}^n H_{X_s,1}= \hbar \sum_{j=1}^g p_j-(t_{\infty^{(1)},1}t_{\infty^{(2)},0}+t_{\infty^{(2)},1}t_{\infty^{(1)},0}+\hbar t_{\infty^{(1)},1}) . \eeq
For $r_\infty=1$, the previous linear system has to be complemented with the two additional relations \eqref{ConditionsAddrinftyequal1}
\bea\label{rinfty1Special} 
\sum_{s=1}^n X_s H_{X_s,1}+\sum_{s=1}^n H_{X_s,2}\delta_{r_s\geq 2}&=&\hbar \sum_{j=1}^g q_j p_j +\sum_{s=1}^n t_{X_s^{(1)},0}t_{X_s^{(2)},0}\delta_{r_s=1} -t_{\infty^{(1)},0}(t_{\infty^{(2)},0}+\hbar),\cr
\sum_{s=1}^n H_{X_s,1}&=& \hbar \sum_{j=1}^g p_j.
\eea
\end{proposition}

Finally, in order to obtain the evolution equation for $p_j$ we look at order $(\lambda-q_j)^{-1}$ of the entry $\mathcal{L}_{\boldsymbol{\alpha}}[L_{2,1}(\lambda)]$. We get, for all $j\in \llbracket 1,g\rrbracket$,
\footnotesize{\bea \label{Lpj} \mathcal{L}_{\boldsymbol{\alpha}}[p_j]&=&\hbar \sum_{i\neq j}\frac{(\mu^{(\boldsymbol{\alpha})}_i+\mu^{(\boldsymbol{\alpha})}_j)(p_i-p_j)}{(q_j-q_i)^2} \cr
&&+\mu^{(\boldsymbol{\alpha})}_j\Big(p_j P_1'(q_j)+ p_j\sum_{s=1}^n \frac{\hbar r_s}{(q_j-X_s)^2}-\td{P}_2'(q_j)+\sum_{k=1}^{r_\infty-4}kH_{\infty,k}q_j^{k-1}-\sum_{s=1}^n\sum_{k=1}^{r_s}kH_{X_s,k}(q_j-X_s)^{-k-1}  \cr
&&-\hbar (r_\infty-3)t_{\infty^{(1)},r_\infty-1}q_j^{r_\infty-4}\delta_{r_\infty\geq 3}\Big)\cr
&&+\hbar \nu^{(\boldsymbol{\alpha})}_{\infty,-1}p_j+\hbar \sum_{k=1}^{r_\infty-1}kc^{(\boldsymbol{\alpha})}_{\infty,k}q_j^{k-1}-\hbar\sum_{s=1}^n\sum_{k=1}^{r_s-1}kc^{(\boldsymbol{\alpha})}_{{X_s},k}(q_j-X_s)^{-k-1}.
\eea}
\normalsize{} 

Thus, we have obtained the general evolutions for $(p_j,q_j)_{1\leq j\leq g}$ through \eqref{Lqj} and \eqref{Lpj}. It turns out that these evolutions are Hamiltonian.

\begin{theorem}[Hamiltonian evolution] \label{HamTheorem} Defining 
\bea \label{DefHam} \text{Ham}^{(\boldsymbol{\alpha})}(\mathbf{q},\mathbf{p})
&=&\sum_{k=0}^{r_\infty-4} \nu_{\infty,k+1}^{\boldsymbol{(\alpha)}}H_{\infty,k}-\sum_{s=1}^n\sum_{k=2}^{r_s}\nu_{X_s,k-1}^{\boldsymbol{(\alpha)}}H_{X_s,k}+\sum_{s=1}^n \alpha_{X_s}H_{X_s,1}\cr
&&-\hbar \sum_{j=1}^g\left[\sum_{k=1}^{r_\infty-1}c^{(\boldsymbol{\alpha})}_{\infty,k}q_j^{k}+\sum_{s=1}^n\sum_{k=1}^{r_s-1}c^{(\boldsymbol{\alpha})}_{{X_s},k}(q_j-X_s)^{-k}\right]\cr
&&+\nu_{\infty,-1}^{\boldsymbol{(\alpha)}}\sum_{s=1}^n\left(X_s H_{X_s,1}+H_{X_s,2}\delta_{r_s\geq 2}\right)+\nu_{\infty,0}^{\boldsymbol{(\alpha)}}\sum_{s=1}^n H_{X_s,1}\cr
&&-  \delta_{r_\infty\in \{1,2\}}\left(\sum_{s=1}^n H_{X_s,1}-\hbar \sum_{j=1}^g p_j\right) \nu_{\infty,0}^{(\boldsymbol{\alpha})}\cr
&&-\delta_{r_\infty=1}\left(\sum_{s=1}^n X_s H_{X_s,1}+\sum_{s=1}^n H_{X_s,2}\delta_{r_s\geq 2} -\hbar \sum_{j=1}^g q_j p_j\right)\nu_{\infty,-1}^{(\boldsymbol{\alpha})}\cr
 &&-\hbar \nu^{(\boldsymbol{\alpha})}_{\infty,0}\sum_{j=1}^{g} p_j-\hbar \nu^{(\boldsymbol{\alpha})}_{\infty,-1}\sum_{j=1}^g q_jp_j,
\eea
the evolutions for $j\in \llbracket 1,g\rrbracket$,
\small{\bea 
\mathcal{L}_{\boldsymbol{\alpha}}[q_j]&=&2\mu^{(\boldsymbol{\alpha})}_j\left(p_j -\frac{1}{2}P_1(q_j)+\frac{1}{2} \sum_{s=1}^n \frac{ \hbar r_s}{q_j-X_s}\right)-\hbar \nu^{(\boldsymbol{\alpha})}_{\infty,0} -\hbar \nu^{(\boldsymbol{\alpha})}_{\infty,-1} q_j-\hbar \sum_{i\neq j}\frac{\mu^{(\boldsymbol{\alpha})}_j+\mu^{(\boldsymbol{\alpha})}_i}{q_j-q_i},\cr
\mathcal{L}_{\boldsymbol{\alpha}}[p_j]&=&\hbar \sum_{i\neq j}\frac{(\mu^{(\boldsymbol{\alpha})}_i+\mu^{(\boldsymbol{\alpha})}_j)(p_i-p_j)}{(q_j-q_i)^2} \cr
&&+\mu^{(\boldsymbol{\alpha})}_j\Big(p_j P_1'(q_j)+ p_j \sum_{s=1}^n \frac{\hbar r_s}{(q_j-X_s)^2}-\td{P}_2'(q_j)+\sum_{k=1}^{r_\infty-4}kH_{\infty,k}q_j^{k-1}-\sum_{s=1}^n\sum_{k=1}^{r_s}kH_{X_s,k}(q_j-X_s)^{-k-1}  \cr
&&-\hbar (r_\infty-3)t_{\infty^{(1)},r_\infty-1}q_j^{r_\infty-4}\delta_{r_\infty\geq 3}\Big)\cr
&&+\hbar \nu^{(\boldsymbol{\alpha})}_{\infty,-1}p_j+\hbar \sum_{k=1}^{r_\infty-1}kc^{(\boldsymbol{\alpha})}_{\infty,k}q_j^{k-1}-\hbar\sum_{s=1}^n\sum_{k=1}^{r_s-1}kc^{(\boldsymbol{\alpha})}_{{X_s},k}(q_j-X_s)^{-k-1}\cr
&&
\eea}
\normalsize{are} Hamiltonian in the sense that 
\beq \forall\, j\in \llbracket 1,g\rrbracket \,:\, \mathcal{L}_{\boldsymbol{\alpha}}[q_j]=\frac{\partial \text{Ham}^{(\boldsymbol{\alpha})}(\mathbf{q},\mathbf{p})}{\partial p_j}\, \text{ and }\, \mathcal{L}_{\boldsymbol{\alpha}}[p_j]=-\frac{\partial \text{Ham}^{(\boldsymbol{\alpha})}(\mathbf{q},\mathbf{p})}{\partial q_j}.\eeq
Quantities involved in the Hamiltonian evolution are defined by Propositions \ref{PropAsymptoticExpansionA12}, \ref{PropA12Form}, \ref{Propcalpha} and \ref{PropDefCi2}.
\end{theorem}

\begin{proof}The proof is given in Appendix \ref{AppendixA}. \end{proof}

\begin{remark}
The Hamiltonian may also be rewritten as
\bea
\text{Ham}^{(\boldsymbol{\alpha})}(\mathbf{q},\mathbf{p})&=&-\frac{\hbar}{2}\displaystyle{\sum_{\substack{(i,j)\in \llbracket 1,g\rrbracket^2 \\ i\neq j }}} \frac{(\mu^{(\boldsymbol{\alpha})}_i+\mu^{(\boldsymbol{\alpha})}_j)(p_i-p_j)}{q_i-q_j} -\hbar \sum_{j=1}^{g} (\nu^{(\boldsymbol{\alpha})}_{\infty,0} p_j+\nu^{(\boldsymbol{\alpha})}_{\infty,-1}q_jp_j) \cr
&&+\sum_{j=1}^g \mu^{(\boldsymbol{\alpha})}_j p_j \sum_{s=1}^n \frac{\hbar r_s}{q_j-X_s}\cr
&&+\sum_{j=1}^{g}\mu^{(\boldsymbol{\alpha})}_j\left[p_j^2-P_1(q_j)p_j +\td{P}_2(q_j) +\hbar t_{\infty^{(1)},r_\infty-1}q_j^{r_\infty-3}\delta_{r_\infty\geq 3}\right]\cr
&&-\hbar \sum_{j=1}^g\left[\sum_{k=1}^{r_\infty-1}c^{(\boldsymbol{\alpha})}_{\infty,k}q_j^{k}+\sum_{s=1}^n\sum_{k=1}^{r_s-1}c^{(\boldsymbol{\alpha})}_{{X_s},k}(q_j-X_s)^{-k}\right]\cr
&&+  \delta_{r_\infty=2}(t_{\infty^{(1)},1}t_{\infty^{(2)},0}+t_{\infty^{(2)},1}t_{\infty^{(1)},0}+\hbar t_{\infty^{(1)},1}) \nu_{\infty,0}^{(\boldsymbol{\alpha})}\cr
&&-\delta_{r_\infty=1}\left(\sum_{s=1}^n t_{X_s^{(1)},0}t_{X_s^{(2)},0}\delta_{r_s=1} -t_{\infty^{(1)},0}(t_{\infty^{(2)},0}+\hbar)\right)\nu_{\infty,-1}^{(\boldsymbol{\alpha})}.\eea
\end{remark}

\begin{remark}It is important to notice that the coefficients $\left(H_{\infty,k}\right)_{0\leq k\leq r_\infty-4}$ and $\left(H_{X_s,k}\right)_{1\leq s\leq n, 1\leq k\leq r_s}$ determined by Proposition \ref{PropDefCi2} are independent of the chosen deformation direction. They are coefficients of the Lax matrix $L$ and generally depend on the phase coordinates and deformation times. Consequently, the dependence of the Hamiltonian regarding deformations appears in the coefficients $\left(\nu^{(\boldsymbol{\alpha})}_{\infty,k}\right)_{-1\leq k\leq r_\infty-3}$ and $\left(\nu^{(\boldsymbol{\alpha})}_{X_s,k}\right)_{1\leq s\leq n,1\leq k\leq r_s-1}$ as well as through the $c^{(\boldsymbol\alpha)}$ coefficients and
the moving-pole components $\alpha_{X_s}$.  
\end{remark}

The Hamiltonian equations of Theorem \ref{HamTheorem} define a
symplectic structure on each fixed-time fiber. The associated closed
two-form on the total space is defined as follows.

\begin{definition}[Fundamental two-form on $\hat{\mathcal{F}}_{\mathcal{R},\mathbf{r},\mathbf{t}_0}$]\label{DefinitionSymplecticForm}
Choose the time-only normalization described in
Remark \ref{RemarkFlatHamiltonianNormalization} and define
\bea \Omega&:=&\hbar \sum_{j=1}^{g} dq_j\wedge dp_j -\sum_{k=1}^{r_\infty-1}\sum_{i=1}^2 dt_{\infty^{(i)},k}\wedge d\text{Ham}^{(\mathbf{e}_{\infty^{(i)},k})}\cr
&&-  \sum_{s=1}^n\sum_{k=1}^{r_s-1}\sum_{i=1}^2 dt_{X_s^{(i)},k}\wedge d\text{Ham}^{(\mathbf{e}_{X_s^{(i)},k})}-\sum_{s=1}^n dX_s\wedge d\text{Ham}^{(\mathbf{e}_{X_s})}.\eea
The form $\Omega$ is closed. For $\hbar\ne0$, its restriction to a
fixed-time fiber is the symplectic form $\hbar\underset{j}{\sum} dq_j\wedge dp_j$.
With the chosen flat normalization, $\Omega$ has rank $2g$ and its
kernel is the horizontal isomonodromy distribution: it is a
presymplectic form on the total space.
\end{definition}

\begin{remark}[Local normalization of the Hamiltonians]
\label{RemarkFlatHamiltonianNormalization}
Let $(t_a)_{1\le a\le m}$ denote independent local deformation
coordinates, including pole positions, and write $H_a$ for the
Hamiltonian of $\hbar\partial_{t_a}$. We use the bracket
\[
\{F,G\}=\sum_{j=1}^g
\left(\frac{\partial F}{\partial q_j}\frac{\partial G}{\partial p_j}
-\frac{\partial F}{\partial p_j}\frac{\partial G}{\partial q_j}\right).
\]
Compatibility of the Hamiltonian vector fields implies that
\begin{equation}\label{DefHamiltonianCurvature}
C_{ab}=\hbar(\partial_{t_a}H_b-\partial_{t_b}H_a)-\{H_a,H_b\}
\end{equation}
is independent of the phase coordinates. The Jacobi identity then
implies that $\underset{a<b}{\sum}C_{ab}\,dt_a\wedge dt_b$ is closed.
On a star-shaped local time chart with base point $\mathbf t^\ast$,
a flat choice for $\hbar\ne0$ is
\begin{align}\label{FlatHamiltonianPrescription}
f_b(\mathbf t)&=-\frac1\hbar\int_0^1\rho
\sum_{a=1}^m(t_a-t_a^\ast)
C_{ab}\bigl(\mathbf t^\ast+\rho(\mathbf t-\mathbf t^\ast)\bigr)\,d\rho,
\nonumber\\
H_b^{\mathrm{flat}}&=H_b+f_b(\mathbf t).
\end{align}
Indeed, the homotopy formula for closed differential forms gives
$\hbar(\partial_{t_a}f_b-\partial_{t_b}f_a)=-C_{ab}$.
These additions leave every Hamiltonian vector field unchanged.
The accessory coefficients determined by Proposition \ref{PropDefCi2}
are not altered. Whenever the total-space kernel of $\Omega$ is used,
the symbols for its Hamiltonians denote flat representatives.
The formulas for the vector fields themselves do not require this
additional choice. An explicit example is given in
\eqref{GenusTwoFlatHamiltonians}.
\end{remark}
There are $m=2r-2-n=2g+4-n$ deformation parameters, including pole
positions, so the total space has dimension $2g+m=4g+4-n$.
The reduction in Theorem \ref{MainTheotau0} expresses the fundamental
two-form in coordinates invariant under the trivial deformation
directions. It does not turn a nondegenerate form on a larger space into
one on a smaller space. All statements below about its total-space kernel
use the compatible time-only normalization just specified.

\begin{theorem}\label{TheoremCorrespondence}
On the generic chart described above, the unordered spectral Darboux
data $\{(q_j,p_j)\}_{j=1}^g$ give a birational description of the
normalized connections and identify their Hamiltonian and Lax
isomonodromy connections. In other words, there exists a birational map between the symplectic Ehresmann
connection and the Jimbo-Miwa-Ueno/Boalch symplectic isomonodromy connection.
\end{theorem}

\begin{proof} The symplectic Ehresmann connection is characterized by the explicit time-dependent Hamiltonians given by Theorem \ref{HamTheorem} where coefficients are determined as rational functions in terms of irregular times, monodromies, positions of finite poles and Darboux coordinates $\left(q_i,p_i\right)_{1\leq i\leq g}$ in Propositions \ref{PropAsymptoticExpansionA12}, \ref{PropA12Form}, \ref{Propcalpha} and \ref{PropDefCi2}. On the other hand, the entries of the Lax matrices in the oper gauge $\left(L(\lambda),A_{\boldsymbol{\alpha}}(\lambda)\right)$ are determined by \eqref{TrivialEntriesA}, Propositions \ref{PropLaxMatrix}, \ref{PropA12Form}, \ref{Propcalpha} and \ref{PropDefCi2} and are also rational functions of the irregular times, monodromies, positions of finite poles and Darboux coordinates. Finally, one may obtain the entries of the Lax matrices in the geometric gauge $\left(\td{L}(\lambda),\td{A}_{\boldsymbol{\alpha}}(\lambda)\right)$ using the previous expressions and the expression of the gauge transformation $\td{\Psi}(\lambda)=G(\lambda)\Psi(\lambda)$ given in Proposition \ref{GaugeTransfoProp} (and also Proposition \ref{Consistencyg} for the value of $\eta_0$) using the standard formula 
\beaa \td{L}(\lambda)&=&G(\lambda)L(\lambda)G(\lambda)^{-1}+\hbar (\partial_\lambda G(\lambda))G(\lambda)^{-1},\cr
\td{A}_{\boldsymbol{\alpha}}(\lambda)&=&G(\lambda)A_{\boldsymbol{\alpha}}(\lambda)G(\lambda)^{-1}+(\mathcal{L}_{\boldsymbol{\alpha}} G(\lambda))G(\lambda)^{-1} .
\eeaa
These reconstruction formulas are rational and invariant under
simultaneous permutations of the pairs $(q_j,p_j)$. Conversely, from a
normalized connection one recovers the monic apparent polynomial
$D(\lambda)=\underset{j=1}{\overset{g}{\prod}}(\lambda-q_j)$ by clearing the fixed pole denominator
in $\td{L}_{1,2}$. On the chart $D(X_s)\neq 0$, reduction of
$\td{L}_{1,1}$ modulo $D$ is well defined, using the inverse of its
denominator in $\mathbb C[\lambda]/(D)$, and gives the unique polynomial
$R$ of degree at most $g-1$ with $R(q_j)=p_j$. Its coefficients are
rational functions of the connection coefficients. The coefficients of
$(D,R)$ describe the unordered pairs rationally, furnishing the inverse
on the symmetric quotient. The reconstruction formulas and the
compatibility equations identify the two deformation connections.
\end{proof} 

It is worth noticing that Theorems \ref{HamTheorem} and \ref{TheoremCorrespondence} indicate that the unordered Darboux data define a birational chart on the generic locus of the moduli space, while the ordered coordinates define a chart on its ordered-root cover. These results were already proved in the Fuchsian case in \cite{DubrovinMazzocco2007} (Theorem A.2) and in \cite{Szabo2013} (Theorem 1). They were later extended to rank two connections with arbitrary unramified poles in \cite{KomyoLoraySaitoSzabo2023} (Lemma 7 and Theorem A).

\section{Decomposition and reduction of the space of isomonodromic deformations}\label{SectionSympRed}
\subsection{Subspaces of trivial and non-trivial deformations}
So far we have considered general isomonodromic deformations with respect to all irregular times by considering $\mathcal{L}_{\boldsymbol{\alpha}}$ characterized by a general vector $\boldsymbol{\alpha}\in \mathbb{C}^{2g+4-n}$. However, as we will see below, there exists a subspace of deformations of dimension $g+4-n$ for which the evolutions of the Darboux coordinates are trivial, leaving thus only a non-trivial subspace of deformations of dimension $g$. These non-trivial deformations shall later be mapped to $g$ isomonodromic times whose expressions will be explicit in terms of the initial irregular times providing a compatible nonautonomous Hamiltonian system. After the trace shift and reduction, the non-trivial evolution equations depend on the differences $\left(t_{p^{(1)},k}-t_{p^{(2)},k}\right)_{p\in \mathcal{R},k\geq 0}$ and are independent of the sums $\left(t_{p^{(1)},k}+t_{p^{(2)},k}\right)_{p\in \mathcal{R},k\geq 0}$. For $k=0$, this independence is understood via the scalar-twist identification of Theorem \ref{TheoMonodromies}. This recovers the standard result that considering meromorphic connections in $\mathfrak{gl}_2(\mathbb{C})$ or in $\mathfrak{sl}_2(\mathbb{C})$ is essentially the same at the level of the Hamiltonian systems because one may remove the trace of the Lax matrices by normalizing the wave matrix using $\td{\Psi}\to \td{\Psi} (\det \td{\Psi})^{-\frac{1}{2}}$.\footnote{Although the reduced evolution equations are independent of these sums, the unreduced Hamiltonians and Lax matrices can depend explicitly on them; this dependence also depends on the normalization at infinity of $\td{L}$.} This remark provides a subspace of trivial deformations of dimension $g+2-n$. Finally the remaining two trivial deformations correspond to the remaining degrees of freedom in the action of the M\"{o}bius transformations (keeping in mind that we assumed that $\infty$ is always a pole, thus fixing automatically one degree of freedom). Indeed, one may rescale $\lambda\to T_2\lambda+T_1$ without changing the symplectic structure. However, this rescaling is non-trivial for our choice of Darboux coordinates and for the irregular times and one needs to keep track of this rescaling to obtain invariant quantities. These two degrees of freedom may be used to fix either leading coefficients at infinity or some positions of the finite poles depending on the values of $r_\infty$ and $n$. More precisely, according to the expression of the Hamiltonian in Theorem \ref{HamTheorem}, the most convenient choice is \textbf{to fix the leading coefficients at infinity first and then only some locations of some finite poles when $r_\infty\leq 2$}.\footnote{Note that this is not the convention used by Jimbo-Miwa in the case of the Painlev\'{e} equations \cite{JimboMiwa} since they chose to fix the location of finite poles first and then only leading coefficients at infinity if needed (when $n\leq 1$). This will provide a different Lax pair for the Painlev\'{e} $4$ equation but for completeness, we shall also provide the corresponding choice to obtain the standard Jimbo-Miwa Lax pair.} In the present paper, we choose 
\begin{itemize} \item For $r_\infty\geq 3$, we fix the leading coefficient $t_{\infty^{(1)},r_\infty-1}$ (conventionally set to $1$) and sub-leading coefficient $t_{\infty^{(1)},r_\infty-2}$ (conventionally set to $0$).
\item For $r_\infty=2$, we fix the leading coefficient $t_{\infty^{(1)},1}$ (conventionally set to $1$) and the location of a finite pole $X_1$ (conventionally set to $0$)
\item For $r_\infty=1$ and $n\geq 2$, we fix the location of the finite poles $X_1$ and $X_2$ (conventionally set to $0$ and $1$)
\item For $r_\infty=1$ and $n=1$, we fix the location of the finite pole $X_1$ (conventionally set to $0$) and the leading coefficient $t_{X_1^{(1)},r_1-1}$ (conventionally set to $1$)
\end{itemize}

\medskip

Let us stress that this symplectic reduction was already known from the early works of Fuchs and Garnier (See \cite{FromGaussToPainleve}) in the case of Fuchsian singularities. However, if the strategy remains the same, we need to keep track of the transformations $\lambda \to a\lambda+b$ and $\td{\Psi}\to \td{\Psi}(\det \td{\Psi})^{-\frac{1}{2}}$ on all quantities involved not only in the Hamiltonian systems but also in the entries of the Lax matrices. In order to decompose the space of isomonodromic deformations, we introduce the following vectors.  

\begin{definition}\label{TrivialVectors} We define the following vectors and their corresponding deformations.
\begin{itemize} \item For all $k\in \llbracket 1,r_\infty-1\rrbracket$, the vectors $\mathbf{v}_{\infty,k}$,
\bea \label{Defvinftyk} \mathcal{L}_{\mathbf{v}_{\infty,k}}&=&\hbar\partial_{t_{\infty^{(1)},k}} + \hbar\partial_{t_{\infty^{(2)},k}}\,\,\Leftrightarrow \,\, \cr
&&\alpha_{\infty^{(i)},r}=\delta_{r,k}\,,\, \alpha_{X_s^{(i)},m}=0\,,\, \alpha_{X_s}=0 \,,\,\cr
&& \forall s\in \llbracket 1,n\rrbracket, i\in \llbracket 1,2\rrbracket\,,\, r\in\llbracket 1,r_\infty-1\rrbracket, m\in \llbracket 1,r_s-1\rrbracket.
\eea
\item For all $s\in \llbracket 1,n\rrbracket$ and all $k\in \llbracket 1, r_s-1\rrbracket$, the vectors $\mathbf{v}_{X_s,k}$,
\bea \label{DefvXsk} \mathcal{L}_{\mathbf{v}_{X_s,k}}&=&\hbar\partial_{t_{X_s^{(1)},k}} + \hbar\partial_{t_{X_s^{(2)},k}} \,\,\Leftrightarrow \,\,\cr
&& \alpha_{\infty^{(i)},r}=0\,,\, \alpha_{X_{s'}^{(i)},m}=\delta_{s',s}\delta_{m,k}\,,\, \alpha_{X_{s'}}=0 \,,\, \cr
&&\forall s'\in \llbracket 1,n\rrbracket, i\in \llbracket 1,2\rrbracket\,,\, r\in\llbracket 1,r_\infty-1\rrbracket, m\in \llbracket 1,r_{s'}-1\rrbracket.
\eea
\item For all $k\in\llbracket 1, r_\infty-1\rrbracket$, the vectors $\mathbf{u}_{\infty,k}$,
\bea \label{Defuinftyk} \mathcal{L}_{\mathbf{u}_{\infty,k}}&=&\hbar \sum_{j=1}^{k}j(t_{\infty^{(1)},r_\infty-1-k+j} \partial_{t_{\infty^{(1)},j}} +t_{\infty^{(2)},r_\infty-1-k+j} \partial_{t_{\infty^{(2)}, j}}) \,\,\Leftrightarrow \,\,\cr
&& \alpha_{\infty^{(i)},r}=r t_{\infty^{(i)},r_\infty-1-k+r}\delta_{1\leq r\leq k}\,,\, \alpha_{X_s^{(i)},m}=0\,,\, \alpha_{X_s}=0 \,,\,\cr
&& \forall s\in \llbracket 1,n\rrbracket, i\in \llbracket 1,2\rrbracket\,,\, r\in\llbracket 1,r_\infty-1\rrbracket, m\in \llbracket 1,r_s-1\rrbracket.
\eea
\item For all $s\in\llbracket 1, n\rrbracket$ and all $k\in \llbracket 1, r_s-1\rrbracket$, the vectors $\mathbf{u}_{X_s,k}$,
\bea \label{DefuXsk} \mathcal{L}_{\mathbf{u}_{X_s,k}}&=&\hbar \sum_{j=1}^{k}j(t_{X_s^{(1)},r_s-1-k+j} \partial_{t_{X_s^{(1)},j}} +t_{X_s^{(2)},r_s-1-k+j} \partial_{t_{X_s^{(2)}, j}}) \,\,\Leftrightarrow \,\,\cr
&& \alpha_{X_{s'}^{(i)},m}=
\begin{cases}
m t_{X_s^{(i)},r_s-1-k+m},&\text{if }s'=s\text{ and }1\le m\le k,\\
0,&\text{otherwise},
\end{cases}\nonumber\cr
&&\alpha_{\infty^{(i)},r}=0,\quad \alpha_{X_{s'}}=0,\cr
&& \forall s'\in\llbracket1,n\rrbracket,\quad
i\in\llbracket1,2\rrbracket,\quad r\in\llbracket1,r_\infty-1\rrbracket,\quad
m\in\llbracket1,r_{s'}-1\rrbracket.
\eea
\item The vector $\mathbf{a}$,
\bea \label{Defa} \mathcal{L}_{\mathbf{a}}&=& \mathcal{L}_{\mathbf{u}_{\infty,r_\infty-1}}-\sum_{s=1}^n\mathcal{L}_{\mathbf{u}_{X_s,r_s-1}}- \hbar \sum_{s=1}^n X_s \partial_{X_s}\cr
&=&\hbar\sum_{r=1}^{r_\infty-1}r (t_{\infty^{(1)},r}\partial_{ t_{\infty^{(1)},r}}+t_{\infty^{(2)},r}\partial_{ t_{\infty^{(2)},r}})\cr
&&-\hbar\sum_{s=1}^n 
\sum_{r=1}^{r_s-1}r (t_{X_s^{(1)},r}\partial_{ t_{X_s^{(1)},r}}+t_{X_s^{(2)},r}\partial_{ t_{X_s^{(2)},r}})- \hbar \sum_{s=1}^n X_s \partial_{X_s}\cr
&&\Leftrightarrow \,\, 
\alpha_{X_{s}^{(i)},m}=-m t_{X_{s}^{(i)},m} \,,\,
\alpha_{\infty^{(i)},r}=r t_{\infty^{(i)},r}\,,\, \alpha_{X_s}=-X_s \,,\,\cr
&& \forall s\in \llbracket 1,n\rrbracket, i\in \llbracket 1,2\rrbracket\,,\, r\in\llbracket 1,r_\infty-1\rrbracket, m\in \llbracket 1,r_{s}-1\rrbracket . \cr
&&
\eea
\item The vector $\mathbf{b}$,
\bea \label{Defb} \mathcal{L}_{\mathbf{b}}&=& \mathcal{L}_{\mathbf{u}_{\infty,r_\infty-2}}- \hbar \sum_{s=1}^n \partial_{X_s}\cr
&=&\hbar\sum_{r=1}^{r_\infty-2}r (t_{\infty^{(1)},r+1}\partial_{ t_{\infty^{(1)},r}}+t_{\infty^{(2)},r+1}\partial_{ t_{\infty^{(2)},r}})- \hbar \sum_{s=1}^n \partial_{X_s}\cr
&&\Leftrightarrow \,\, \alpha_{X_{s}^{(i)},m}=0\,,\,
 \alpha_{\infty^{(i)},r}=r t_{\infty^{(i)},r+1}\delta_{1\leq r\leq r_\infty-2}\,,\, \alpha_{X_s}=-1 \,,\,\cr
&& \forall s\in \llbracket 1,n\rrbracket, i\in \llbracket 1,2\rrbracket\,,\, r\in\llbracket 1,r_\infty-1\rrbracket, m\in \llbracket 1,r_{s}-1\rrbracket .
\eea
\item For all $s\in\llbracket 1, n\rrbracket$, the vectors $\mathbf{w}_s$,
\bea \label{Defuws} \mathcal{L}_{\mathbf{w}_{s}}&=&\hbar \partial_{X_s}\,\,\Leftrightarrow \,\,\cr
&& \alpha_{X_{s'}}=\delta_{s',s} \,,\,\alpha_{X_{s'}^{(i)},m}=0 \,,\, \alpha_{\infty^{(i)},r}=0\,,\, \cr
&& \forall s'\in \llbracket 1,n\rrbracket, i\in \llbracket 1,2\rrbracket\,,\, r\in\llbracket 1,r_\infty-1\rrbracket, m\in \llbracket 1,r_{s'}-1\rrbracket .\cr
&&
\eea
\end{itemize}
\end{definition} 

In \eqref{Defa} and \eqref{Defb}, the notation
$\mathcal{L}_{\mathbf u_{p,k}}=0$ is understood when $k\le0$,
consistently with the empty sums in the expanded formulas.

For $r_\infty\ge3$, the vectors
$(\mathbf u_{\infty,k})_{1\le k\le r_\infty-3}$,
$(\mathbf u_{X_s,k})_{s,\,1\le k\le r_s-1}$,
all positive-index trace vectors $\mathbf v$, $\mathbf a$, $\mathbf b$,
and $(\mathbf w_s)_{1\le s\le n}$ form a basis of the deformation
space. For $r_\infty=2$, omit the empty infinity $\mathbf u$ family
and omit $\mathbf w_1$. For $r_\infty=1$, $n\ge2$, omit
$\mathbf w_1,\mathbf w_2$ and the empty infinity families. For
$r_\infty=1$, $n=1$, omit $\mathbf w_1$ and
$\mathbf u_{X_1,r_1-1}$. These choices are bases on the generic chart:
the eliminated components are recovered using $\mathbf a,\mathbf b$,
since the relevant leading-time difference, or $X_2-X_1$, is nonzero.
In every case the number of vectors is $2g+4-n$.

\begin{remark}\label{RemarkTransfo}\sloppy{The previous definitions are chosen so that the vector fields $\left(\mathcal{L}_{\mathbf{v}_{\infty,k}}\right)_{1\leq k\leq r_\infty-1}$ and $\left(\mathcal{L}_{\mathbf{v}_{X_s,k}}\right)_{s\leq n, 1\leq k\leq r_s-1}$ correspond to the transformation $\td{\Psi}\to \td{\Psi}(\det \td{\Psi})^{-\frac{1}{2}}$, i.e. to go from connections on $\mathfrak{gl}_2(\mathbb{C})$ to connections on $\mathfrak{sl}_2(\mathbb{C})$. On the contrary, the vector field $\mathcal{L}_{\mathbf{a}}$ keeps track of the dilatation $\lambda\to c\lambda$ while $\mathcal{L}_{\mathbf{b}}$ keeps track of the translation $\lambda\to \lambda+c$.}
\end{remark}

The previous vectors satisfy the following proposition.
\begin{proposition}\label{TrivialSubspace} For a given $j\in \llbracket 1 ,r_\infty-1\rrbracket$, we have
\bea 0&=&\nu^{(\mathbf{v}_{\infty,j})}_{\infty,k} \,\,,\,\forall \,k \in \llbracket -1 ,r_{\infty}-3\rrbracket ,\cr
0&=&\nu^{(\mathbf{v}_{\infty,j})}_{{X_s},k} \,\,,\,\forall \,(s,k) \in \llbracket 1,n\rrbracket\times\llbracket 1,r_{s}-1\rrbracket,\cr
-\frac{1}{j}\delta_{k,j}&=&c^{(\mathbf{v}_{\infty,j})}_{\infty,k}\,\,,\,\forall \,k \in \llbracket 1 ,r_{\infty}-1\rrbracket,\cr
0&=&c^{(\mathbf{v}_{\infty,j})}_{{X_s},k}\,\,,\,\forall \,(s,k) \in \llbracket 1,n\rrbracket\times\llbracket 1,r_{s}-1\rrbracket.
\eea
Let $s\in \llbracket 1,n\rrbracket$ and $j\in \llbracket 1 ,r_s-1\rrbracket$, then we have
\bea 0&=&\nu^{(\mathbf{v}_{X_s,j})}_{\infty,k} \,\,,\,\forall \,k \in \llbracket -1 ,r_{\infty}-3\rrbracket ,\cr
0&=&\nu^{(\mathbf{v}_{X_s,j})}_{{X_{s'}},k} \,\,,\,\forall \,(s',k) \in \llbracket 1,n\rrbracket\times\llbracket 1,r_{s'}-1\rrbracket,\cr
0&=&c^{(\mathbf{v}_{X_s,j})}_{\infty,k}\,\,,\,\forall \,k \in \llbracket 1 ,r_{\infty}-1\rrbracket,\cr
-\frac{1}{j}\delta_{s',s}\delta_{k,j}&=&c^{(\mathbf{v}_{X_s,j})}_{{X_{s'}},k}\,\,,\,\forall \,(s',k) \in \llbracket 1,n\rrbracket\times\llbracket 1,r_{s'}-1\rrbracket.
\eea
\end{proposition}

\begin{proof}The proof is given in Appendix \ref{AppendixB}.
\end{proof}

We also have,

\begin{proposition}\label{TrivialSubspace2} For a given $j\in \llbracket 1 ,r_\infty-1\rrbracket$, we have
\bea \delta_{r_\infty-2-j,k}&=&\nu^{(\mathbf{u}_{\infty,j})}_{\infty,k} \,\,,\,\forall \,k \in \llbracket -1 ,r_{\infty}-3\rrbracket , \cr
0&=&\nu^{(\mathbf{u}_{\infty,j})}_{{X_s},k} \,\,,\,\forall \,(s,k) \in \llbracket 1,n\rrbracket\times\llbracket 1,r_{s}-1\rrbracket,\cr
0&=&c^{(\mathbf{u}_{\infty,j})}_{\infty,k}\,\,,\,\forall \,k \in \llbracket 1 ,r_{\infty}-1\rrbracket,\cr
0&=&c^{(\mathbf{u}_{\infty,j})}_{{X_s},k}\,\,,\,\forall \,(s,k) \in \llbracket 1,n\rrbracket\times\llbracket 1,r_{s}-1\rrbracket.
\eea
Let $s\in \llbracket 1,n\rrbracket$ and $j\in \llbracket 1 ,r_s-1\rrbracket$, then we have
\bea 0&=&\nu^{(\mathbf{u}_{X_s,j})}_{\infty,k} \,\,,\,\forall \,k \in \llbracket -1 ,r_{\infty}-3\rrbracket ,\cr
-\delta_{s',s}\delta_{k,r_s-j}&=&\nu^{(\mathbf{u}_{X_s,j})}_{{X_{s'}},k} \,\,,\,\forall \,(s',k) \in \llbracket 1,n\rrbracket\times\llbracket 1,r_{s'}-1\rrbracket,\cr
0&=&c^{(\mathbf{u}_{X_s,j})}_{\infty,k}\,\,,\,\forall \,k \in \llbracket 1 ,r_{\infty}-1\rrbracket,\cr
0&=&c^{(\mathbf{u}_{X_s,j})}_{{X_{s'}},k}\,\,,\,\forall \,(s',k) \in \llbracket 1,n\rrbracket\times\llbracket 1,r_{s'}-1\rrbracket.
\eea
\end{proposition}

\begin{proof}The proof is given in Appendix \ref{AppendixC}.
\end{proof}

Finally, vectors $\mathbf{a}$ and $\mathbf{b}$ satisfy,

\begin{proposition}\label{TrivialSubspace3}The vectors $\mathbf{a}$ and $\mathbf{b}$ are chosen so that
\bea \nu^{(\mathbf{a})}_{{X_s},0}&=&X_s \,\,,\,\, \forall \,  s\in \llbracket 1, n\rrbracket\, ,\cr
\nu^{(\mathbf{a})}_{{X_s},k}&=&\delta_{k,1} \,\,,\,\, \forall \,  s\in \llbracket 1, n\rrbracket\, k\in \llbracket 1, r_s-1\rrbracket ,\cr
\nu^{(\mathbf{a})}_{\infty,k}&=&\delta_{k,-1}\,\,,\,\, \forall \, k\in \llbracket -1, r_\infty-3\rrbracket,\cr
c^{(\mathbf{a})}_{\infty,k}&=&0\,\,,\,\, \forall \, k\in \llbracket 1, r_\infty-1\rrbracket,\cr
c^{(\mathbf{a})}_{{X_s},k}&=&0\,\,,\,\, \forall \,  s\in \llbracket 1, n\rrbracket\,,\, k\in \llbracket 1, r_s-1\rrbracket,\cr
\nu^{(\mathbf{b})}_{{X_s},0}&=&1 \,\,,\,\, \forall \, s\in \llbracket 1, n\rrbracket,\cr
\nu^{(\mathbf{b})}_{\infty,k}&=&\delta_{k,0}\,\,,\,\, \forall \, k\in \llbracket -1, r_\infty-3\rrbracket,\cr
\nu^{(\mathbf{b})}_{{X_s},k}&=&0\,\,,\,\, \forall \,  s\in \llbracket 1, n\rrbracket\,,\, k\in \llbracket 1, r_s-1\rrbracket,\cr
c^{(\mathbf{b})}_{\infty,k}&=&0\,\,,\,\, \forall \, k\in \llbracket 1, r_\infty-1\rrbracket,\cr
c^{(\mathbf{b})}_{{X_s},k}&=&0\,\,,\,\, \forall \,  s\in \llbracket 1, n\rrbracket\,,\, k\in \llbracket 1, r_s-1\rrbracket.
\eea
\end{proposition}

\begin{proof}The proof mainly follows from Proposition \ref{TrivialSubspace2}. For completeness, we detail it in Appendix \ref{AppendixD}.
\end{proof}

The previous propositions imply that we get trivial deformations along vectors $(\mathbf{v}_{\infty,k})_{1\leq k\leq r_\infty-1}$, $(\mathbf{v}_{X_s,k})_{1\leq s\leq n, 1\leq k\leq r_s-1}$, $\mathbf{a}$ and $\mathbf{b}$.

\begin{theorem}\label{TheoremTrivialSubspace1}We have the following trivial deformations.
\bea 0&=&\mathcal{L}_{\mathbf{v}_{\infty,k}}[q_j]  \,\,,\,\ \forall \,(j,k)\in \llbracket 1,g\rrbracket\times\llbracket 1,r_\infty-1\rrbracket,\cr
-\hbar q_j^{k-1}&=&\mathcal{L}_{\mathbf{v}_{\infty,k}}[p_j] \,\,,\,\ \forall \,(j,k)\in \llbracket 1,g\rrbracket\times\llbracket 1,r_\infty-1\rrbracket,\cr
0&=&\mathcal{L}_{\mathbf{v}_{X_s,k}}[q_j]\,\,,\,\ \forall \,(j,s,k)\in \llbracket 1,g\rrbracket\times \llbracket 1,n\rrbracket\times\llbracket 1,r_s-1\rrbracket,\cr
\hbar (q_j-X_{s})^{-k-1}&=&\mathcal{L}_{\mathbf{v}_{X_s,k}}[p_j] \,\,,\,\ \forall \,(j,s,k)\in \llbracket 1,g\rrbracket\times \llbracket 1,n\rrbracket\times\llbracket 1,r_s-1\rrbracket,\cr
-\hbar q_j&=&\mathcal{L}_{\mathbf{a}}[q_j] \,\,,\,\ \forall \,j\in \llbracket 1,g\rrbracket,\cr
-\hbar&=&\mathcal{L}_{\mathbf{b}}[q_j]\,\,,\,\ \forall \,j\in \llbracket 1,g\rrbracket,\cr
\hbar p_j&=&\mathcal{L}_{\mathbf{a}}[p_j]\,\,,\,\ \forall \,j\in \llbracket 1,g\rrbracket,\cr
0&=&\mathcal{L}_{\mathbf{b}}[p_j] \,\,,\,\ \forall \,j\in \llbracket 1,g\rrbracket.
\eea
\end{theorem}

\begin{proof}The proof is given in Appendix \ref{AppendixE}.
\end{proof}

Theorem \ref{TheoremTrivialSubspace1} is coherent with the action on the Darboux coordinates of the transformations equivalent to the vector fields mentioned in Remark \ref{RemarkTransfo}. Note that the subspace of trivial deformations is of dimension $r_\infty-1+\underset{s=1}{\overset{n}{\sum}}r_s -n+2=g-n+4$ leaving a subspace of non-trivial deformations of dimension $g$. 

\subsection{Definition of trivial times and isomonodromic times}\label{SectionDefTrivial}
The nontrivial deformation space has dimension $g$. For
$r_\infty\ge3$ it is represented by
$(\mathbf u_{\infty,k})_{1\le k\le r_\infty-3}$,
$(\mathbf u_{X_s,k})_{s,\,1\le k\le r_s-1}$ and
$(\mathbf w_s)_{1\le s\le n}$. For $r_\infty=2$ retain only
$s\ge2$ in the position family; for $r_\infty=1$, $n\ge2$, retain
only $s\ge3$. For $r_\infty=1$, $n=1$, retain the finite-time
directions $1\le k\le r_1-2$ and no position direction.
The following local change of time coordinates consists of $g+4-n$
trivial times and $g$ nontrivial times. The invariant Darboux
coordinates are independent of the former. This choice is not unique.
As we will see below, the symplectic reduction allows us to take $P_1=0$, i.e. to remove the trace of $\td{L}$, which is a canonical choice in the context of integrable systems. However, the situation is more involved for translations and dilatations because the canonical choice for the corresponding trivial times depends on the value of $r_\infty$ and of the number of finite poles. However, it is possible to choose them wisely (in the next four subsections) so that they share the same main properties.  
We also write
\[
T_{p,0}:=t_{p^{(1)},0}+t_{p^{(2)},0},\qquad p\in\mathcal R.
\]
These trace-monodromy parameters are fixed and are not additional
deformation times.

\subsubsection{The case $r_\infty\geq 3$}
As explained above, when $r_\infty\geq 3$, one may fix the two leading coefficients at infinity. This corresponds to choosing the following set of times.

\begin{definition}[Trivial and isomonodromic times]\label{DefTrivialrgeq3} For $r_\infty\geq 3$, we define the following set of trivial coordinates denoted $\mathcal{T}_{\text{trivial}}$.
\bea T_{\infty,k}&=&t_{\infty^{(1)},k}+t_{\infty^{(2)},k} \,\,\,,\,\, \forall \, k\in \llbracket 1,r_\infty-1\rrbracket, \cr
T_{X_s,k}&=&t_{X_s^{(1)},k}+t_{X_s^{(2)},k} \,\,\,,\,\, \forall \, (s,k)\in \llbracket 1,n\rrbracket\times \llbracket 1,r_s-1\rrbracket, \cr
T_{1}&=&\frac{t_{\infty^{(1)},r_\infty-2}-t_{\infty^{(2)},r_\infty-2}}{2^{\frac{1}{r_\infty-1}}(r_\infty-2)(t_{\infty^{(1)},r_\infty-1}-t_{\infty^{(2)},r_\infty-1})^{\frac{r_\infty-2}{r_\infty-1}}} ,\cr
T_{2}&=&\left(\frac{t_{\infty^{(1)},r_\infty-1}-t_{\infty^{(2)},r_\infty-1}}{2}\right)^{\frac{1}{r_\infty-1}}.
\eea
Moreover, the set of isomonodromic times, denoted $\mathcal{T}_{\text{iso}}$, contains
\begin{itemize}\item For all $j\in \llbracket 1, r_\infty-3\rrbracket$:
\small{\bea \tau_{\infty,j}&=&2^{\frac{j}{r_\infty-1}}\Big[ \sum_{i=0}^{r_\infty-j-3} \frac{(-1)^i (j+i-1)!}{i! (j-1)! (r_\infty-2)^i}\frac{(t_{\infty^{(1)},r_\infty-2}-t_{\infty^{(2)},r_\infty-2})^i(t_{\infty^{(1)},j+i}-t_{\infty^{(2)},j+i})}{(t_{\infty^{(1)},r_\infty-1}-t_{\infty^{(2)},r_\infty-1})^{\frac{i(r_\infty-1)+j}{r_\infty-1}}}\cr
&&+\frac{(-1)^{r_\infty-j-2} (r_\infty-3)!}{(r_\infty-1-j)(r_\infty-j-3)!(j-1)!(r_\infty-2)^{r_\infty-j-2}}\frac{(t_{\infty^{(1)},r_\infty-2}-t_{\infty^{(2)},r_\infty-2})^{r_\infty-1-j}}{(t_{\infty^{(1)},r_\infty-1}-t_{\infty^{(2)},r_\infty-1})^{\frac{(r_\infty-2)(r_\infty-1-j)}{r_\infty-1}}}\Big].\cr&&
\eea}\normalsize{}
\item For all $s\in \llbracket 1,n\rrbracket$: $(\tau_{X_s,k})_{1\leq k\leq r_s-1}$ are defined by:
\bea \tau_{X_s,k}&=&(t_{X_s^{(1)},k}-t_{X_s^{(2)},k}) \left(\frac{t_{\infty^{(1)},r_\infty-1}-t_{\infty^{(2)},r_\infty-1}}{2}\right)^{\frac{k}{r_\infty-1}}\,\,,\,\, \forall\, k\in \llbracket 1,r_s-1\rrbracket.\cr
&&
\eea
\item For all $s\in \llbracket 1,n\rrbracket$, the times $(\td{X}_s)_{1\leq s\leq n}$ are defined by 
\beq \td{X}_s=X_s\left(\frac{t_{\infty^{(1)},r_\infty-1}-t_{\infty^{(2)},r_\infty-1}}{2}\right)^{\frac{1}{r_\infty-1}}+\frac{\left(t_{\infty^{(1)},r_\infty-2}-t_{\infty^{(2)},r_\infty-2}\right)}{2^{\frac{1}{r_\infty-1}}(r_\infty-2)\left(t_{\infty^{(1)},r_\infty-1}-t_{\infty^{(2)},r_\infty-1}\right)^{\frac{r_\infty-2}{r_\infty-1}}}.\eeq
\end{itemize}
\end{definition}

We have the inverse relations

\begin{proposition}\label{InverseRelationsrgeq3} Irregular times and the locations of the poles are related to the set of trivial and isomonodromic times by the relations
\begin{itemize} 
\item For all $k\in \llbracket 1,r_\infty-3\rrbracket$, $i\in\llbracket 1,2\rrbracket$:
\bea t_{\infty^{(i)},r_\infty-1}&=&\frac{1}{2}T_{\infty, r_\infty-1}+  (-1)^{i+1}T_{2}^{r_\infty-1},\cr
t_{\infty^{(i)},r_\infty-2}&=&\frac{1}{2}T_{\infty, r_\infty-2}+(-1)^{i+1}(r_\infty-2) T_{1} T_{2}^{r_\infty-2},\cr
t_{\infty^{(i)},k}&=&\frac{1}{2}T_{\infty, k}+ \frac{(-1)^{i+1}}{2}T_2^{k}\Big(\frac{2(r_\infty-2)!}{(k-1)! (r_\infty-1-k)!} T_1^{r_\infty-1-k}\cr
&&+\sum_{j=2}^{r_\infty-1-k}\frac{(r_\infty-2-j)!}{(k-1)! (r_\infty-1-k-j)!}T_1^{r_\infty-1-j-k}\tau_{\infty,r_\infty-1-j}\Big).
\eea
\item For all $i\in \llbracket 1,2\rrbracket$, $s\in \llbracket 1,n\rrbracket$, for all $k\in \llbracket 1,r_s-1\rrbracket$:
\beq 
t_{X_s^{(i)},k}=\frac{1}{2}T_{X_s,k}+\frac{(-1)^{i+1}}{2}T_{2}^{-k}\tau_{X_s,k}.
\eeq
\item The positions of the poles are given by
\beq X_s=\frac{\td{X}_s-T_{1}}{T_{2}} \,\,\,,\,\,\forall\, s\in \llbracket 1,n\rrbracket.\eeq
\end{itemize}
\end{proposition}

\begin{proof}The proof is presented in Appendix \ref{AppendixH}.
\end{proof}

The inverse relations provide the following proposition.

\begin{proposition}[Deformations dual to the isomonodromic times for $r_\infty\geq 3$]\label{TheoDualIsorgeq3}
For $r_\infty\geq 3$, we have, for all $r\in \llbracket 1, r_\infty-3\rrbracket$,
\beq 
\partial_{\tau_{\infty,r}}=\frac{1}{2}\sum_{k=1}^r T_2^{k}\frac{(r-1)!}{(k-1)! (r-k)!}T_1^{r-k} \left(\partial_{t_{\infty^{(1)},k}} -\partial_{t_{\infty^{(2)},k}} \right) \,\,,\,\, \forall \, r\in \llbracket 1, r_\infty-3\rrbracket
\eeq
and for all $s\in \llbracket 1,n\rrbracket$:
\beq
\partial_{\tau_{X_s,k}}=\frac{1}{2}T_{2}^{-k}\left(\partial_{t_{X_s^{(1)},k}}-\partial_{t_{X_s^{(2)},k}}\right)\,\,,\,\, \forall\, k\in   \llbracket 1,r_s-1\rrbracket
\eeq
and
\beq 
\partial_{\td{X}_s}=\frac{1}{T_{2}}\partial_{X_s} \,\,,\,\,\forall\, s\in \llbracket 1,n\rrbracket.
\eeq
\end{proposition}

\subsubsection{The case $r_\infty=2$}
In the case $r_\infty=2$ we necessarily have $n\geq 1$ in order to have $g>0$. We may thus fix the coefficient $t_{\infty^{(1)},1}$ and the location of the pole $X_1$.
\begin{definition}[Trivial and isomonodromic times]\label{DefTrivialrequal2} For $r_\infty=2$, we define the following set of trivial coordinates denoted $\mathcal{T}_{\text{trivial}}$.
\bea T_{\infty,1}&=&t_{\infty^{(1)},1}+t_{\infty^{(2)},1}, \cr
T_{X_s,k}&=&t_{X_s^{(1)},k}+t_{X_s^{(2)},k} \,\,\,,\,\, \forall \, (s,k)\in \llbracket 1,n\rrbracket\times \llbracket 1,r_s-1\rrbracket, \cr
T_{1}&=&-X_1\left(\frac{t_{\infty^{(1)},1}-t_{\infty^{(2)},1}}{2}\right),\cr
T_{2}&=&\left(\frac{t_{\infty^{(1)},1}-t_{\infty^{(2)},1}}{2}\right).
\eea
Moreover, the set of isomonodromic times, denoted $\mathcal{T}_{\text{iso}}$, contains
\begin{itemize}
\item For all $s\in \llbracket 1,n\rrbracket$: $(\tau_{X_s,k})_{1\leq k\leq r_s-1}$ are defined by:
\bea \tau_{X_s,k}&=&(t_{X_s^{(1)},k}-t_{X_s^{(2)},k}) \left(\frac{t_{\infty^{(1)},1}-t_{\infty^{(2)},1}}{2}\right)^{k}\,\,,\,\, \forall\, k\in \llbracket 1,r_s-1\rrbracket.\cr
&&
\eea
\item For all $s\in \llbracket 2,n\rrbracket$, the times $(\td{X}_s)_{2\leq s\leq n}$ are defined by 
\beq \td{X}_s=(X_s-X_1)\left(\frac{t_{\infty^{(1)},1}-t_{\infty^{(2)},1}}{2}\right).\eeq
\end{itemize}
\end{definition} 

We have the inverse relations

\begin{proposition}\label{InverseRelationsrequal2} Irregular times and the locations of the poles are related to the set of trivial and isomonodromic times by the relations
\begin{itemize} 
\item For all $i\in \llbracket 1,2\rrbracket$:
\beq t_{\infty^{(i)},1}=\frac{1}{2}T_{\infty, 1}+  (-1)^{i+1}T_2.\eeq
\item For all $i\in \llbracket 1,2\rrbracket$, $s\in \llbracket 1,n\rrbracket$, for all $k\in \llbracket 1,r_s-1\rrbracket$:
\beq 
t_{X_s^{(i)},k}=\frac{1}{2}T_{X_s,k}+\frac{(-1)^{i+1}}{2}T_{2}^{-k}\tau_{X_s,k}.
\eeq
\item The positions of the poles are given by
\bea X_1&=&-\frac{T_1}{T_2},\cr
X_s&=&\frac{\td{X}_s-T_1}{T_2}\,\,\,,\,\,\forall\, s\in \llbracket 2,n\rrbracket.\eea
\end{itemize}
\end{proposition}

\begin{proof}The proof is obvious.
\end{proof}

The inverse relations provide the following proposition.

\begin{proposition}[Deformations dual to the isomonodromic times for $r_\infty=2$]\label{TheoDualIsorequal2}
For $r_\infty=2$, for all $s\in \llbracket 1,n\rrbracket$:
\beq
\partial_{\tau_{X_s,k}}=\frac{1}{2}T_{2}^{-k}\left(\partial_{t_{X_s^{(1)},k}}-\partial_{t_{X_s^{(2)},k}}\right)\,\,,\,\, \forall\, k\in   \llbracket 1,r_s-1\rrbracket
\eeq
and for all $s\in \llbracket 2,n\rrbracket$:
\beq 
\partial_{\td{X}_s}=\frac{1}{T_{2}}\partial_{X_s} \,\,,\,\,\forall\, s\in \llbracket 2,n\rrbracket.
\eeq
\end{proposition}

\subsubsection{The case $r_\infty=1$ and $n\geq 2$}
In this section we consider $r_\infty=1$ and $n\geq 2$ and we fix the position of two finite poles. This corresponds to the following set of trivial and isomonodromic times.
\begin{definition}[Trivial and isomonodromic times]\label{DefTrivialrequal1} For $r_\infty=1$ and $n\geq 2$, we define the following set of trivial coordinates denoted $\mathcal{T}_{\text{trivial}}$.
\bea 
T_{X_s,k}&=&t_{X_s^{(1)},k}+t_{X_s^{(2)},k} \,\,\,,\,\, \forall \, (s,k)\in \llbracket 1,n\rrbracket\times \llbracket 1,r_s-1\rrbracket, \cr
T_{1}&=&-\frac{X_1}{X_2-X_1},\cr
T_{2}&=&\left(X_2-X_1\right)^{-1}.
\eea
Moreover, the set of isomonodromic times, denoted $\mathcal{T}_{\text{iso}}$, contains
\begin{itemize}
\item For all $s\in \llbracket 1,n\rrbracket$: $(\tau_{X_s,k})_{1\leq k\leq r_s-1}$ are defined by:
\beq \tau_{X_s,k}=(t_{X_s^{(1)},k}-t_{X_s^{(2)},k}) \left(X_2-X_1\right)^{-k}\,\,,\,\, \forall\, k\in \llbracket 1,r_s-1\rrbracket.
\eeq
\item For all $s\in \llbracket 3,n\rrbracket$, the times $(\td{X}_s)_{3\leq s\leq n}$ are defined by 
\beq \td{X}_s=\frac{X_s-X_1}{X_2-X_1}.\eeq
\end{itemize}
\end{definition} 

The inverse relations provide the following proposition.

\begin{proposition}\label{InverseRelationsrequal1} Irregular times and the locations of the poles are related to the set of trivial and isomonodromic times by the relations
\begin{itemize} 
\item For all $i\in \llbracket 1,2\rrbracket$, $s\in \llbracket 1,n\rrbracket$, for all $k\in \llbracket 1,r_s-1\rrbracket$:
\beq 
t_{X_s^{(i)},k}=\frac{1}{2}T_{X_s,k}+\frac{(-1)^{i+1}}{2}T_{2}^{-k}\tau_{X_s,k}.
\eeq
\item The positions of the poles are given by
\bea X_1&=&-\frac{T_1}{T_2},\cr
X_2&=&\frac{1-T_1}{T_2},\cr
X_s&=&\frac{\td{X}_s-T_1}{T_2}\,\,\,,\,\,\forall\, s\in \llbracket 3,n\rrbracket.\eea
\end{itemize}
\end{proposition}

\begin{proof}The proof is obvious.
\end{proof}

The inverse relations provide the following proposition.

\begin{proposition}[Deformations dual to the isomonodromic times for $r_\infty=1$ and $n\geq 2$]\label{TheoDualIsorequal1}
For $r_\infty=1$ and $n\geq 2$, for all $s\in \llbracket 1,n\rrbracket$:
\beq
\partial_{\tau_{X_s,k}}=\frac{1}{2}T_{2}^{-k}\left(\partial_{t_{X_s^{(1)},k}}-\partial_{t_{X_s^{(2)},k}}\right)\,\,,\,\, \forall\, k\in   \llbracket 1,r_s-1\rrbracket
\eeq
and for all $s\in \llbracket 3,n\rrbracket$:
\beq 
\partial_{\td{X}_s}=\frac{1}{T_{2}}\partial_{X_s} \,\,,\,\,\forall\, s\in \llbracket 3,n\rrbracket.
\eeq
\end{proposition}

\subsubsection{The case $r_\infty=1$ and $n=1$}
In this section we consider $r_\infty=1$ and $n=1$ (hence $r_1\geq 3$ to have $g>0$). In this case, we may fix the position of the finite pole and its leading coefficient. It corresponds to the following definition of trivial and isomonodromic times. 
\begin{definition}[Trivial and isomonodromic times]\label{DefTrivialrequal1n1} For $r_\infty=1$ and $n=1$, we define the following set of trivial coordinates denoted $\mathcal{T}_{\text{trivial}}$.
\bea 
T_{X_1,k}&=&t_{X_1^{(1)},k}+t_{X_1^{(2)},k} \,\,\,,\,\, \forall \, k\in \llbracket 1,r_1-1\rrbracket, \cr
T_{1}&=&-X_1\left(\frac{t_{X_1^{(1)},r_1-1} -t_{X_1^{(2)},r_1-1}}{2}\right)^{-\frac{1}{r_1-1}},\cr
T_{2}&=&\left(\frac{t_{X_1^{(1)},r_1-1} -t_{X_1^{(2)},r_1-1}}{2}\right)^{-\frac{1}{r_1-1}}.
\eea
Moreover, the set of isomonodromic times, denoted $\mathcal{T}_{\text{iso}}$, contains
\begin{itemize}
\item $(\tau_{X_1,k})_{1\leq k\leq r_1-2}$ are defined by:
\beq \tau_{X_1,k}=(t_{X_1^{(1)},k}-t_{X_1^{(2)},k}) \left(\frac{t_{X_1^{(1)},r_1-1} -t_{X_1^{(2)},r_1-1}}{2}\right)^{-\frac{k}{r_1-1}}\,\,,\,\, \forall\, k\in \llbracket 1,r_1-2\rrbracket.
\eeq
\end{itemize}
\end{definition} 

The inverse relations are given by the following proposition.

\begin{proposition}\label{InverseRelationsrequal1n1} Irregular times and the locations of the poles are related to the set of trivial and isomonodromic times by the relations
\begin{itemize} 
\item For all $i\in \llbracket 1,2\rrbracket$:
\bea 
t_{X_1^{(i)},r_1-1}&=&\frac{1}{2}T_{X_1,r_1-1}+(-1)^{i+1}T_2^{-(r_1-1)},\cr
t_{X_1^{(i)},k}&=&\frac{1}{2}T_{X_1,k}+\frac{(-1)^{i+1}}{2}T_{2}^{-k}\tau_{X_1,k} \,\,\,,\,\,  \forall \, k\in \llbracket 1,r_1-2\rrbracket.
\eea
\item The position of the finite pole is given by
\beq X_1=-\frac{T_1}{T_2}.\eeq
\end{itemize}
\end{proposition}

\begin{proof}The proof is obvious.
\end{proof}

The inverse relations provide the following proposition.

\begin{proposition}[Deformations dual to the isomonodromic times for $r_\infty=1$ and $n=1$]\label{TheoDualIsorequal1n1}
For $r_\infty=1$ and $n=1$, we have:
\beq
\partial_{\tau_{X_1,k}}=\frac{1}{2}T_{2}^{-k}\left(\partial_{t_{X_1^{(1)},k}}-\partial_{t_{X_1^{(2)},k}}\right)\,\,,\,\, \forall\, k\in   \llbracket 1,r_1-2\rrbracket.
\eeq
\end{proposition}

\subsection{Properties of trivial and isomonodromic times}
We first note that in each case, we always have exactly $g$ isomonodromic times that are complemented by $r_\infty+1+\underset{s=1}{\overset{n}{\sum}}r_s -n$ trivial times so that we have in total $2r_\infty+2\underset{s=1}{\overset{n}{\sum}}r_s-n-2$ new coordinates corresponding to the dimension of the initial deformation space. Choose branches of all fractional powers on a simply connected
domain where the relevant leading differences and pole separations
are nonzero. On that domain the set of new coordinates $\mathcal{T}:=\mathcal{T}_{\text{trivial}}\cup \mathcal{T}_{\text{iso}}$ is locally in one-to-one correspondence with the set of initial coordinates $\{(t_{\infty^{(i)},k})_{1\leq i\leq 2, 1\leq k\leq r_\infty-1}\} \cup\{ (t_{X_s^{(i)},k})_{1\leq i\leq 2, 1\leq s\leq n, 1\leq k\leq r_s-1}\}\cup\{X_s\}_{1\leq s\leq n}$ since we could exhibit the inverse relations in Propositions \ref{InverseRelationsrgeq3}, \ref{InverseRelationsrequal2}, \ref{InverseRelationsrequal1} and \ref{InverseRelationsrequal1n1}. 

Let us first observe that the trivial times $T_1$ and $T_2$ are chosen to satisfy the following proposition.

\begin{proposition}\label{PropositionT1T2} The trivial times $T_1$ and $T_2$ satisfy
\bea 0&=&\mathcal{L}_{\mathbf{v}_{\infty,k}}[T_2]\,\,,\,\,\forall\, k\in \llbracket 1, r_\infty-1\rrbracket,\cr
0&=&\mathcal{L}_{\mathbf{v}_{X_s,k}}[T_2]\,\,,\,\,\forall\, (s,k)\in \llbracket 1,n\rrbracket\times \llbracket 1, r_s-1\rrbracket,\cr
\hbar T_2&=&\mathcal{L}_\mathbf{a}[T_2],\cr
0&=&\mathcal{L}_\mathbf{b}[T_2],\cr
0&=&\mathcal{L}_{\mathbf{v}_{\infty,k}}[T_1]\,\,,\,\,\forall\, k\in \llbracket 1, r_\infty-1\rrbracket,\cr
0&=&\mathcal{L}_{\mathbf{v}_{X_s,k}}[T_1]\,\,,\,\,\forall\, (s,k)\in \llbracket 1,n\rrbracket\times \llbracket 1, r_s-1\rrbracket,\cr
0&=&\mathcal{L}_\mathbf{a}[T_1],\cr
\hbar T_2&=&\mathcal{L}_\mathbf{b}[T_1].
\eea
\end{proposition}

\begin{proof}The proof is given in Appendix \ref{AppendixT1T2}.\end{proof}

Moreover, the isomonodromic times are defined so that the following theorem holds:
\begin{theorem}\label{MainTheotau} The general solutions $f(X_s,t_{\infty^{(i)},k}, t_{X_s^{(i)},k})$ of the partial differential equations
\bea\label{ConditionsToImpose2}
0&=&\mathcal{L}_{v_{\infty,k}}[f] \,\,,\,\, \forall \, k\in \llbracket 1, r_\infty-1\rrbracket ,\cr
0&=&\mathcal{L}_{v_{X_s,k}}[f]\,\,,\,\, \forall \, (s,k)\in \llbracket 1,n\rrbracket\times \llbracket 1, r_s-1\rrbracket,\cr
0&=&\mathcal{L}_{\mathbf{a}}[f],\cr
0&=&\mathcal{L}_{\mathbf{b}}[f]
\eea
are arbitrary functions of the isomonodromic times. Note in particular that the isomonodromic times are themselves solutions of \eqref{ConditionsToImpose2}.
\end{theorem}

\begin{proof} 
The proof is given in Appendix \ref{AppendixF}.
\end{proof}

\subsection{Shifted Darboux coordinates and symplectic reduction}
Theorem \ref{TheoremTrivialSubspace1} indicates that deformations $\mathcal{L}_{\mathbf{a}}$ and $\mathcal{L}_{\mathbf{b}}$ do not act trivially on the Darboux coordinates $(q_j,p_j)_{1\leq j\leq g}$. However, since the action is very simple, we may easily perform a symplectic transformation on the Darboux coordinates to obtain ``shifted Darboux coordinates'' for which the action of $\mathcal{L}_{\mathbf{a}}$ and $\mathcal{L}_{\mathbf{b}}$ becomes trivial.

\begin{definition}\label{DefSymplecticChange} For any $j\in \llbracket 1,g\rrbracket$, we define
\bea \label{SymplecticChange}
\check{q}_j&=&T_2 q_j+T_1 ,  \cr
\check{p}_j&=&T_2^{-1}\left(p_j-\frac{1}{2}P_1(q_j)\right).
\eea
that we call ``shifted Darboux coordinates''.
\end{definition}

Work locally with $\hbar\ne0$, fixed formal monodromies and $T_2\ne0$. In the following statement, the Hamiltonians defining $\Omega$ use the
normalization of Definition \ref{DefinitionSymplecticForm}. The
Hamiltonians in the shifted coordinates include the time-dependent
canonical transformation terms: their normalization is fixed explicitly
by \eqref{ReducedHamiltonianNormalizationP} in Appendix
\ref{AppendixSymplecticForm}.

\begin{theorem}[Symplectic reduction]\label{MainTheotau0}
The fundamental two-form $\Omega$ defined in Definition
\ref{DefinitionSymplecticForm} reduces to
\bea
\Omega&=&\hbar\sum_{j=1}^g dq_j\wedge dp_j
-\sum_{s=1}^n\sum_{i=1}^2\sum_{k=1}^{r_s-1}
dt_{X_s^{(i)},k}\wedge
d\text{Ham}^{(\mathbf e_{X_s^{(i)},k})}\cr
&&-\sum_{i=1}^2\sum_{k=1}^{r_\infty-1}
dt_{\infty^{(i)},k}\wedge
d\text{Ham}^{(\mathbf e_{\infty^{(i)},k})}
-\sum_{s=1}^n dX_s\wedge d\text{Ham}^{(\mathbf e_{X_s})}\cr
&=&\hbar\sum_{j=1}^g d\check q_j\wedge d\check p_j
-\sum_{\tau\in\mathcal T_{\mathrm{iso}}}
d\tau\wedge d\text{Ham}^{(\boldsymbol\alpha_\tau)}.
\eea
\end{theorem}

\begin{proof}
The detailed calculation is given in Appendix
\ref{AppendixSymplecticForm} for each of the four choices of trivial
and isomonodromic times. The normalization of terms depending only
on the times is explained in Remark \ref{RemarkPureTimeAppendixP}.
\end{proof}

Theorem \ref{MainTheotau0} (or Theorem \ref{MainTheotau}) implies the following corollary:

\begin{corollary}\label{MainTheoTrivial} The shifted Darboux coordinates $\left(\check{q}_j,\check{p}_j\right)_{1\leq j\leq g}$ are independent of the trivial times and thus may only depend on the isomonodromic times.
\end{corollary}

Moreover, we may complement the previous result with the following theorem regarding the dependence on the monodromy parameters:

\begin{theorem}\label{TheoMonodromies}
Identify connections that differ by a scalar twist, keeping the
trace-free connection and the reduced initial data fixed. Under this
identification the shifted Darboux coordinates depend on the differences
of the two formal-monodromy parameters at each pole, and are independent
of their sums. The sums obey the global residue relation
\eqref{SumResidues}; they are not independent unconstrained parameters.
Their variations are scalar-twist parameters, not additional
isomonodromic times of the previously fixed-monodromy problem.
\end{theorem}

\begin{proof}The proof is given in Appendix \ref{AppendixMonodromies}.
\end{proof}

\begin{remark}Corollary \ref{MainTheoTrivial} implies that the shifted Darboux coordinates $\left(\check{q}_j,\check{p}_j\right)_{1\leq j\leq g}$ are solutions of the partial differential equations \eqref{ConditionsToImpose2}.
\end{remark}

Finally let us mention the following observation.

\begin{proposition}\label{PropositionTrace}
Assume $\hbar\ne0$ and that the deformations $\tau_1,\ldots,\tau_g$
preserve $P_1(\lambda)$ as a function of $\lambda$. This holds, in
particular, on the trace-free slice $P_1=0$.
For either $L^*=\check L$ or $L^*=\td{L}$, write $A^*_j$ for
the corresponding auxiliary matrices. Their traces are independent of
$\lambda$. Locally on a simply connected time domain, let
\[
dF^*=\sum_j\operatorname{Tr}(A^*_j)\,d\tau_j,
\qquad G^*=\exp\left(-\frac{F^*}{2\hbar}\right)I_2.
\]
Then $\Psi^*\mapsto G^*\Psi^*$ makes all $A^*_j$ traceless
simultaneously and leaves $L^*$ unchanged.
\end{proposition}
\begin{proof}
Taking the trace of zero curvature gives
$\partial_\lambda\operatorname{Tr}(A^*_j)=
\partial_{\tau_j}P_1=0$. Multi-time compatibility reads
\begin{equation}\label{EqGauge}
\hbar\partial_{\tau_j}A^*_i
=\hbar\partial_{\tau_i}A^*_j+[A^*_j,A^*_i].
\end{equation}
Its trace shows that $\underset{j}{\sum}\operatorname{Tr}(A^*_j)d\tau_j$ is
closed, so the local primitive $F^*$ exists. The scalar gauge adds
$-\tfrac12\partial_{\tau_j}F^*I_2$ to $A^*_j$, cancelling its trace.
Since $G^*$ is independent of $\lambda$, it leaves $L^*$ unchanged.
\end{proof}

The last proposition shall be useful when $\td{L}$ and $\check{L}$ are traceless. In this case, it is interesting to perform this additional gauge transformation in order to obtain a Lax pair that belongs to $\mathfrak{sl}_2(\mathbb{C})$ rather than $\mathfrak{gl}_2(\mathbb{C})$. In particular, this is always possible for the canonical choice of trivial times proposed in Section \ref{SectionMainResult}.

\section{Canonical choice of trivial times and simplification of the Hamiltonian systems}\label{SectionMainResult}

The aim of this section is to combine the general expressions of the Hamiltonian systems given by Theorem \ref{HamTheorem} and the reduction of the deformation space developed in Section \ref{SectionSympRed}. The purpose is thus to obtain the non-trivial isomonodromic evolutions for the shifted Darboux coordinates $\left(\check{q}_j,\check{p}_j\right)_{1\leq j\leq g}$ for any isomonodromic time $\tau$ in a simpler form, using the fact that these quantities are independent of the trivial times. Indeed, since the shifted Darboux coordinates $\left(\check{q}_j,\check{p}_j\right)_{1\leq j\leq g}$ are independent of the trivial times from Corollary \ref{MainTheoTrivial} and Theorem \ref{TheoMonodromies}, the evolutions $\left(\partial_\tau \check{q}_j,\partial_\tau \check{p}_j\right)_{1\leq j\leq g}$ do not depend on the trivial times for any isomonodromic time $\tau$. Thus, one may obtain these evolutions by taking any value of the trivial times. For computational purposes, there exists a canonical choice of the trivial times for which formulas get simpler. But we stress again that these evolution equations would be exactly the same for any other choice of the trivial times.

\subsection{Canonical choice of the trivial times and main theorem}

In the rest of Section \ref{SectionMainResult}, we set the trivial times to the following particular values.
\begin{definition}[Canonical choice of trivial times]We shall call ``canonical choice of the trivial times'' the choice of the following values:
\bea \label{TrivialTimesChoice}T_{\infty,k}&=&0 \,\,\,,\,\, \forall \, k\in \llbracket 0,r_\infty-1\rrbracket, \cr
T_{X_s,k}&=&0\,\,\,,\,\, \forall \, (s,k)\in \llbracket 1,n\rrbracket\times \llbracket 0,r_s-1\rrbracket ,\cr
T_{1}&=&0 ,\cr
T_{2}&=&1.
\eea 
\end{definition}

The conditions $T_{p,0}=0$ use the scalar-twist identification of
Theorem \ref{TheoMonodromies}; they are not additional isomonodromic
deformations of the fixed-monodromy problem. If the original trace
formal monodromies are nonzero, the parameters used in the trace-free
formulas are
$\bar t_{p^{(1)},0}=\frac{t_{p^{(1)},0}-t_{p^{(2)},0}}{2}$ and
$\bar t_{p^{(2)},0}=-\bar t_{p^{(1)},0}$ for $p\in\mathcal R$.
We suppress the bars after this identification.

In particular, the canonical choice of trivial times implies that
\bea 
t_{\infty^{(2)},k}&=&-t_{\infty^{(1)},k}\,\,,\,\, \forall\, k\in \llbracket 0, r_\infty-1\rrbracket ,\cr
\alpha_{\infty^{(2)},k}&=&-\alpha_{\infty^{(1)},k}\,\,,\,\, \forall\, k\in \llbracket 1, r_\infty-1\rrbracket,\cr
t_{X_s^{(2)},k}&=&-t_{X_s^{(1)},k}\,\,,\,\, \forall\, (s,k)\in \llbracket 1,n\rrbracket\times \llbracket 0, r_s-1\rrbracket,\cr
\alpha_{X_s^{(2)},k}&=&-\alpha_{X_s^{(1)},k}\,\,,\,\, \forall\, (s,k)\in \llbracket 1,n\rrbracket\times \llbracket 1, r_s-1\rrbracket.
\eea
Moreover, we always have under this canonical choice of trivial times: 
\beq P_1(\lambda)=0 \,\, \text{ and }\, (q_j,p_j)=(\check{q}_j,\check{p}_j) \,\, ,\forall \, j\in \llbracket 1,g\rrbracket.\eeq
In particular, \textbf{under this canonical choice of trivial times \eqref{TrivialTimesChoice}, the initial Darboux coordinates $\left(q_j,p_j\right)_{1\leq j\leq g}$ identify with the shifted Darboux coordinates $\left(\check{q}_j,\check{p}_j\right)_{1\leq j\leq g}$}.
Note that we also get from \eqref{Relationckalphainfty} and \eqref{Relationckalphas} that for this specific choice of trivial times and for any isomonodromic time $\tau$,
\bea\label{cobservation} c_{\infty,k}^{(\boldsymbol{\alpha}_{\tau})}&=&0 \,\,,\,\, \forall \, k\in \llbracket 1, r_\infty-1\rrbracket,\cr
c_{X_s,k}^{(\boldsymbol{\alpha}_{\tau})}&=&0 \,\,,\,\, \forall \, (s,k)\in \llbracket 1,n\rrbracket\times \llbracket 1, r_s-1\rrbracket.
\eea
Indeed, the right-hand sides of \eqref{Relationckalphainfty} and \eqref{Relationckalphas} are always vanishing by obvious symmetry.
\medskip 

Finally, note that since $P_1=0$, $\td{L}$ and $\check{L}$ are traceless. Hence, Proposition \ref{PropositionTrace} implies that under a potential additional trivial gauge transformation, we may choose a gauge in which $\td{L}$, $\check{L}$, $\check{A}_{\boldsymbol{\alpha}_{\tau}}$ and $\td{A}_{\boldsymbol{\alpha}_{\tau}}$ are traceless for any isomonodromic time $\tau$. 

\medskip

We shall now apply Theorem \ref{HamTheorem} for the canonical values of the trivial times given by \eqref{TrivialTimesChoice}. The following formulas use the explicit Hamiltonian representatives of that theorem. When they are used in the fundamental two-form, the compatible time-only normalization of Definition \ref{DefinitionSymplecticForm} and Appendix \ref{AppendixSymplecticForm} is understood. In particular, we obtain the following simplification:

\begin{theorem}[Hamiltonian representation for the canonical choice of trivial times]\label{HamTheoremReduced} \sloppy{The canonical choice of the trivial times given by \eqref{TrivialTimesChoice} and the definitions of trivial times made in Definitions \ref{DefTrivialrgeq3}, \ref{DefTrivialrequal2}, \ref{DefTrivialrequal1} and \ref{DefTrivialrequal1n1} imply that for any non-trivial isomonodromic time $\tau\in \mathcal{T}_{\text{iso}}$:
\beq \label{DefHamReduced} \text{Ham}^{(\boldsymbol{\alpha}_\tau)}(\check{\mathbf{q}},\check{\mathbf{p}})=\sum_{k=0}^{r_\infty-4} \nu_{\infty,k+1}^{(\boldsymbol{\alpha}_\tau)}H_{\infty,k}-\sum_{s=1}^n\sum_{k=2}^{r_s}\nu_{X_s,k-1}^{(\boldsymbol{\alpha}_\tau)}H_{X_s,k}+\sum_{s=1}^n \alpha_{X_s}H_{X_s,1}.\eeq
Thus the Hamiltonians are time-dependent linear combinations of the
accessory coefficients. Explicitly, for the independent times,}
\begin{equation}\label{NewHamReduced}
\begin{aligned}
\begin{pmatrix}
\text{Ham}^{(\boldsymbol\alpha_{\tau_{\infty,1}})}(\check{\mathbf q},\check{\mathbf p})\\
\vdots\\
(r_\infty-3)\text{Ham}^{(\boldsymbol\alpha_{\tau_{\infty,r_\infty-3}})}(\check{\mathbf q},\check{\mathbf p})
\end{pmatrix}
&=\td{M}_\infty^{-1}
\begin{pmatrix}H_{\infty,r_\infty-4}\\\vdots\\H_{\infty,0}\end{pmatrix},\\[2mm]
\begin{pmatrix}
\text{Ham}^{(\boldsymbol\alpha_{\tau_{X_s,1}})}(\check{\mathbf q},\check{\mathbf p})\\
\vdots\\
(r_s-1)\text{Ham}^{(\boldsymbol\alpha_{\tau_{X_s,r_s-1}})}(\check{\mathbf q},\check{\mathbf p})
\end{pmatrix}
&=M_s^{-1}\begin{pmatrix}H_{X_s,r_s}\\\vdots\\H_{X_s,2}\end{pmatrix},\\[2mm]
\begin{pmatrix}
\text{Ham}^{(\boldsymbol\alpha_{\td{X}_{s_0}})}(\check{\mathbf q},\check{\mathbf p})\\
\vdots\\
\text{Ham}^{(\boldsymbol\alpha_{\td{X}_n})}(\check{\mathbf q},\check{\mathbf p})
\end{pmatrix}
&=\begin{pmatrix}H_{X_{s_0},1}\\\vdots\\H_{X_n,1}\end{pmatrix}.
\end{aligned}
\end{equation}
The first line applies when $r_\infty\ge4$. The second applies to
each $s$ with $r_s\ge2$, except in the case $r_\infty=1,n=1$,
where the leading finite time is fixed and it is replaced by
\[
\begin{pmatrix}
\text{Ham}^{(\boldsymbol\alpha_{\tau_{X_1,1}})}(\check{\mathbf q},\check{\mathbf p})\\
\vdots\\
(r_1-2)\text{Ham}^{(\boldsymbol\alpha_{\tau_{X_1,r_1-2}})}(\check{\mathbf q},\check{\mathbf p})
\end{pmatrix}
=\begin{pmatrix}I_{r_1-2}&0\end{pmatrix}M_1^{-1}
\begin{pmatrix}H_{X_1,r_1}\\\vdots\\H_{X_1,2}\end{pmatrix}.
\]
In the position line, $s_0=1$ if $r_\infty\ge3$, $s_0=2$ if
$r_\infty=2$, and $s_0=3$ if $r_\infty=1,n\ge2$.
That line is absent if $n<s_0$ or if $r_\infty=1,n=1$.
The matrices are

\beq \td{M}_\infty=\begin{pmatrix}2&0&\dots&& &\dots &0\\
0&2& 0& && &\vdots\\
\tau_{\infty,r_\infty-3}&0& 2&\ddots && &\vdots\\
\vdots & \ddots&\ddots &\ddots  && &\vdots\\
\vdots &\ddots&\ddots&\ddots&2&\ddots&\vdots\\
\tau_{\infty,4} &\ddots &\ddots&&\ddots& 2 &0\\
\tau_{\infty,3}&\tau_{\infty,4}& \dots && \tau_{\infty,r_\infty-3}& 0& 2
 \end{pmatrix}\in \mathcal{M}_{r_\infty-3}(\mathbb{C})
\eeq
and 
\beq M_s=\begin{pmatrix}\tau_{X_s,r_s-1}&0&\dots& &\dots &0\\
\tau_{X_s,r_s-2}&\tau_{X_s,r_s-1}& 0& & &\vdots\\
\vdots & \ddots&\ddots &\ddots  & &\vdots\\
\vdots &\ddots&\ddots&\ddots&0&\vdots\\
\tau_{X_s,2}&\ddots &\ddots&\ddots& \tau_{X_s,r_s-1}&0\\
\tau_{X_s,1}&\tau_{X_s,2}& \dots & & \tau_{X_s,r_s-2}& \tau_{X_s,r_s-1}
 \end{pmatrix} \,,\,\, \forall \, s\in \llbracket 1,n\rrbracket.
\eeq
Only the coefficients of these linear combinations depend on the choice of deformation direction; the accessory coefficients are determined by Proposition \ref{PropDefCi2}. 
\end{theorem}

\begin{proof}The proof is presented in Appendix \ref{ProofTheoremHamTheoremReduced}.
\end{proof}

\begin{remark}Note that the previous proof is only valid for the choice of trivial times made in Definitions \ref{DefTrivialrgeq3}, \ref{DefTrivialrequal2}, \ref{DefTrivialrequal1} and \ref{DefTrivialrequal1n1}. In particular, in order to obtain this simplification, one needs to use the additional degrees of freedom to first fix coefficients at infinity rather than location of the poles. Fixing the location of the poles first would imply that for $r_\infty\geq 2$, terms proportional to $\nu_{\infty,-1}^{(\boldsymbol{\alpha}_\tau)}$ or $\nu_{\infty,0}^{(\boldsymbol{\alpha}_\tau)}$ do not vanish and thus, one should keep the full formula of Theorem \ref{HamTheorem} creating unnecessary complications for interpretation.
\end{remark}

\begin{remark}\label{RemarkHInTermsOfTraceSquared}
One of the main advantages of the reduced form of Theorem \ref{HamTheoremReduced} is that we may combine it with Proposition \ref{HInTermsOfTraceSquared} to make some connections with results of \cite{MartaPaper2022}. Indeed, under the canonical choice of trivial times \eqref{TrivialTimesChoice}, Proposition \ref{HInTermsOfTraceSquared} reduces to
\bea \label{ReconstructionHamiltonian} H_{\infty,j}&=&-\frac{1}{2}\Res_{\lambda\to \infty} \Tr(L_{\text{GMR}}(\lambda,\hbar)^2)\lambda^{-j-1} \,,\,\forall \, j\in \llbracket 0,r_\infty-4\rrbracket,\cr
H_{X_s,j}&=&\frac{1}{2}\Res_{\lambda\to X_s} \Tr(L_{\text{GMR}}(\lambda,\hbar)^2)(\lambda-X_s)^{j-1}\,,\,\forall \, (s,j)\in \llbracket 1,n\rrbracket\times \llbracket 1,r_s\rrbracket.\cr&&
\eea
The trace-residue representation \eqref{ReconstructionHamiltonian},
combined with the matrix relations \eqref{NewHamReduced}, is
analogous in form to the Hamiltonian formulas of
\cite{MartaPaper2022}.
However, our matrix $L_{\mathrm{GMR}}$ is obtained by the meromorphic
gauge transformation \eqref{DefPsic} and generically has additional
apparent poles at the points $q_j$, whereas the matrix $A(\lambda)$
of \cite{MartaPaper2022} has only the prescribed poles.
For $\hbar\ne0$, the accessory coefficients reconstructed by
\eqref{ReconstructionHamiltonian} must therefore be distinguished
from the commuting isospectral Hamiltonians of the prescribed-pole
connection.
A detailed identification of the two Hamiltonian formulations
requires matching their deformation parameters and Darboux
coordinates, including the Hamiltonian corrections arising from
time-dependent changes of coordinates.
\end{remark}

Theorem \ref{HamTheoremReduced} implies that the Hamiltonians are explicit time-dependent linear combinations of the accessory coefficients determined by Proposition \ref{PropDefCi2}. Under our canonical choice of trivial times, quantities involved in Proposition \ref{PropDefCi2} simplify and we propose the corresponding expressions depending on the value of $r_\infty$. In particular, we believe that this shall be useful for readers focused on immediate use for applications.

\subsection{The case $r_\infty\geq 3$}\label{Sectionrinfty3}
\subsubsection{General case $n\geq 0$}\label{Sectionrgeq3ngeq0}
For $r_\infty\geq 3$, the canonical choice of trivial times \eqref{TrivialTimesChoice} implies that
\beq t_{\infty^{(1)},r_\infty-1}=1 \,\, \text{ and } \,\,  t_{\infty^{(1)},r_\infty-2}=0.\eeq
In particular, isomonodromic times simplify to
\bea \tau_{\infty,j}&=&2 t_{\infty^{(1)},j} \,,\,\, \forall \, j\in \llbracket 1,r_\infty-3\rrbracket,\cr
\tau_{X_s,k}&=&2t_{X_s^{(1)},k}   \,,\,\, \forall \, (s,k)\in \llbracket 1,n\rrbracket\times \llbracket 1, r_s-1\rrbracket,\cr
\td{X}_s&=&X_s \,,\, \forall \, s\in \llbracket 1,n\rrbracket.
\eea
The expression for $\td{P}_2$ (eq. \eqref{DeftdP2}) simplifies to
\bea\label{tdP2reducedrinftygeq3} \td{P}_2(\lambda)
&=&-\lambda^{2r_\infty-4} -\sum_{r=r_\infty-2}^{2r_\infty-6} \left(\tau_{\infty,r-r_\infty+3}+\frac{1}{4}\sum_{j=2}^{2r_\infty-6-r}\tau_{\infty,r_\infty-1-j}\tau_{\infty,j+r-r_\infty+3}\right)\lambda^r \cr
&&-\left(2t_{\infty^{(1)},0}+\frac{1}{4}\sum_{j=2}^{r_\infty-3} \tau_{\infty,r_\infty-1-j}\tau_{\infty,j}\right)\lambda^{r_\infty-3}\cr
&&-\frac{1}{4}\sum_{s=1}^n\sum_{j=r_s+2}^{2r_s}\sum_{i=0}^{2r_s-j}\tau_{X_s,r_s-1-i}\tau_{X_s,i+j-r_s-1}(\lambda-\td{X}_s)^{-j}\cr
&&-\frac{1}{4}\sum_{s=1}^n \sum_{i=1}^{r_s-2}\tau_{X_s, r_s-1-i}\tau_{X_s,i}(\lambda-\td{X}_s)^{-r_s-1}\cr
&&-\sum_{s=1}^n \left(t_{X_s^{(1)},0}\tau_{X_s,r_s-1}\delta_{r_s\geq 2}+(t_{X_s^{(1)},0})^2\delta_{r_s=1} \right)(\lambda-\td{X}_s)^{-r_s-1}.
\eea
For $r_\infty=3$, the blocks $V_\infty$ and $\mathbf H_\infty$
below are empty and are omitted, as in \eqref{RelationNuMuMatrixForm}.
Matrices $(V_s)_{1\leq s\leq n}$ and $V_\infty$, defined by \eqref{DefVinfty}, reduce to
\beq\label{DefVinftyReducesrinftygeq3}V_\infty=\begin{pmatrix}1&1 &\dots &\dots &1\\
\check{q}_1& \check{q}_2&\dots &\dots& \check{q}_{g}\\
\vdots & & & & \vdots\\
\vdots & & & & \vdots\\
\check{q}_1^{r_\infty-4}& \check{q}_2^{r_\infty-4} &\dots & \dots& \check{q}_{g}^{r_\infty-4}\end{pmatrix}\,,\, 
V_s=\begin{pmatrix}\frac{1}{\check{q}_1-\td{X}_s}& \dots &\dots& \frac{1}{\check{q}_g-\td{X}_s}\\
\frac{1}{(\check{q}_1-\td{X}_s)^2}& \dots &\dots& \frac{1}{(\check{q}_g-\td{X}_s)^2}\\
\vdots & & & \vdots\\
\vdots & & & \vdots\\
\frac{1}{(\check{q}_1-\td{X}_s)^{r_s}}& \dots &\dots& \frac{1}{(\check{q}_g-\td{X}_s)^{r_s}}\\
\end{pmatrix}
\eeq
and we have 
\beq \label{Hreducedrinftygeq3}\begin{pmatrix} V_\infty^{t}& V_1^t&\dots &V_n^{t}\end{pmatrix}\begin{pmatrix}\mathbf{H}_{\infty}\\\mathbf{H}_{X_1}\\ \vdots \\ \mathbf{H}_{X_n} \end{pmatrix}=\begin{pmatrix} \check{p}_1^2 + \check{p}_1\underset{s=1}{\overset{n}{\sum}} \frac{\hbar r_s}{\check{q}_1-\td{X}_s}+\td{P}_2(\check{q}_1)+\hbar \underset{i\neq 1}{\sum}\frac{\check{p}_i-\check{p}_1}{\check{q}_1-\check{q}_i}+\hbar \check{q}_1^{r_\infty-3}\\
\vdots\\
\vdots\\
\check{p}_g^2 + \check{p}_g\underset{s=1}{\overset{n}{\sum}} \frac{\hbar r_s}{\check{q}_g-X_s}+\td{P}_2(\check{q}_g)+\hbar \underset{i\neq g}{\sum}\frac{\check{p}_i-\check{p}_g}{\check{q}_g-\check{q}_i}+\hbar \check{q}_g^{r_\infty-3}
\end{pmatrix}.
\eeq
 
\begin{remark}\label{Remarkmucoeffrinftygeq3}Note also that coefficients $\left(\mu^{(\boldsymbol{\alpha}_\tau)}_j\right)_{1\leq j\leq g}$ are given by Proposition \ref{PropA12Form} using the simplified expressions \eqref{DefVinftyReducesrinftygeq3} for the matrices $V_\infty$ and $(V_s)_{1\leq s\leq n}$ and $\boldsymbol{\nu}_{{X_s}}^{(\boldsymbol{\alpha}_\tau)}=\left(\alpha_{X_s},-\nu^{(\boldsymbol{\alpha}_{\tau})}_{{X_s},1},\dots,-\nu^{(\boldsymbol{\alpha}_{\tau})}_{{X_s},r_s-1}\right)^t$.
\end{remark}

\subsubsection{The sub-case $n=0$: single-pole polynomial connections}\label{P2Hierarchy}
Assume $r_\infty\ge4$ and $n=0$. This family contains Painlev\'e II
at $r_\infty=4$ and the genus-two system identified in
Section \ref{SectionSecondElementP2} with Chiba's
$(P_{\mathrm{II}-2})_2$ at $r_\infty=5$.
The two Painlev\'e-II hierarchies distinguished in \cite{Chiba}
coincide for the ordinary Painlev\'e II equation but differ in
higher dimension. A general identification of the present family
with the specific Lenard hierarchy studied in
\cite{mazzocco2007hamiltonian} is not established here.
Theorem \ref{HamTheoremReduced} takes the following simple form:
\bea\label{tdP2reducedPainleve2} \td{P}_2(\lambda)
&=&-\lambda^{2r_\infty-4} -\sum_{r=r_\infty-2}^{2r_\infty-6} \left(\tau_{\infty,r-r_\infty+3}+\frac{1}{4}\sum_{j=2}^{2r_\infty-6-r}\tau_{\infty,r_\infty-1-j}\tau_{\infty,j+r-r_\infty+3}\right)\lambda^r \cr
&&-\left(2t_{\infty^{(1)},0}+\frac{1}{4}\sum_{j=2}^{r_\infty-3} \tau_{\infty,r_\infty-1-j}\tau_{\infty,j}\right)\lambda^{r_\infty-3}.
\eea
Coefficients $\left(H_{\infty,k}\right)_{0\leq k\leq r_\infty-4}$ are determined by 
\beq\label{HCoeffPainleve2Hierarchy} \begin{pmatrix}1 &\dots &\dots &1\\
\check{q}_1&\dots &\dots& \check{q}_{g}\\
\vdots  & & & \vdots\\
\vdots  & & & \vdots\\
\check{q}_1^{r_\infty-4}& \dots & \dots& \check{q}_{g}^{r_\infty-4}\end{pmatrix}^t \begin{pmatrix}H_{\infty,0}\\ \vdots\\ H_{\infty,r_\infty-4}\end{pmatrix}=\begin{pmatrix} \check{p}_1 ^2+\td{P}_2(\check{q}_1)+\hbar \underset{i\neq 1}{\sum}\frac{\check{p}_i-\check{p}_1}{\check{q}_1-\check{q}_i}+\hbar \check{q}_1^{r_\infty-3}\\
\vdots\\
\check{p}_{g} ^2+\td{P}_2(\check{q}_{g})+\hbar \underset{i\neq g}{\sum}\frac{\check{p}_i-\check{p}_{g}}{\check{q}_{g}-\check{q}_i}+\hbar \check{q}_g^{r_\infty-3}
\end{pmatrix}.
\eeq
With the original two-sheet irregular times treated as independent
coordinates, the reduced derivatives on the canonical slice are
\[
\partial_{\tau_{\infty,j}}
=\frac12\left(
\partial_{t_{\infty^{(1)},j}}
-\partial_{t_{\infty^{(2)},j}}\right),
\qquad 1\le j\le r_\infty-3.
\]
Equivalently, after substituting
$t_{\infty^{(2)},j}=-t_{\infty^{(1)},j}$ into a function,
this is one half of its derivative with respect to the remaining
one-sheet coordinate. The Hamiltonians are given by
\beq \label{HamPainleve2Hierarchy}\begin{pmatrix}\text{Ham}^{(\alpha_{\tau_{\infty,1}})}(\check{\mathbf{q}},\check{\mathbf{p}})\\ \vdots \\ (r_\infty-3)\text{Ham}^{(\alpha_{\tau_{\infty,r_\infty -3}})}(\check{\mathbf{q}},\check{\mathbf{p}})\end{pmatrix}=\td{M}_{\infty}^{-1}\begin{pmatrix}H_{\infty,r_\infty-4}\\ \vdots\\ H_{\infty,0} \end{pmatrix}.
\eeq

\subsection{The case $r_\infty=2$}\label{Sectionrinftyequal2}
For $r_\infty=2$ (thus $n\geq 1$ in order to have $g>0$), the canonical choice of trivial times \eqref{TrivialTimesChoice} implies that
\beq t_{\infty^{(1)},1}=1 \,\, \text{ and } \,\,  X_1=0.\eeq
In particular, isomonodromic times simplify to
\bea
\tau_{X_s,k}&=&2t_{X_s^{(1)},k}   \,,\,\, \forall \, (s,k)\in \llbracket 1,n\rrbracket\times \llbracket 1, r_s-1\rrbracket,\cr
\td{X}_s&=&X_s \,,\, \forall \, s\in \llbracket 2,n\rrbracket.
\eea
The expression for $\td{P}_2$ reduces to
\bea\label{tdP2reducedrinftyequal2} \td{P}_2(\lambda)
&=&-1-\frac{1}{4}\sum_{j=r_1+2}^{2r_1}\sum_{i=0}^{2r_1-j}\tau_{X_1,r_1-1-i}\tau_{X_1,i+j-r_1-1}\lambda^{-j}\cr
&&-\frac{1}{4}\sum_{i=1}^{r_1-2}\tau_{X_1, r_1-1-i}\tau_{X_1,i}\lambda^{-r_1-1}\cr
&&-\left(t_{X_1^{(1)},0}\tau_{X_1,r_1-1}\delta_{r_1\geq 2}+(t_{X_1^{(1)},0})^2\delta_{r_1=1} \right)\lambda^{-r_1-1}\cr
&&-\frac{1}{4}\sum_{s=2}^n\sum_{j=r_s+2}^{2r_s}\sum_{i=0}^{2r_s-j}\tau_{X_s,r_s-1-i}\tau_{X_s,i+j-r_s-1}(\lambda-\td{X}_s)^{-j}\cr
&&-\frac{1}{4}\sum_{s=2}^n \sum_{i=1}^{r_s-2}\tau_{X_s, r_s-1-i}\tau_{X_s,i}(\lambda-\td{X}_s)^{-r_s-1}\cr
&&-\sum_{s=2}^n \left(t_{X_s^{(1)},0}\tau_{X_s,r_s-1}\delta_{r_s\geq 2}+(t_{X_s^{(1)},0})^2\delta_{r_s=1} \right)(\lambda-\td{X}_s)^{-r_s-1}.
\eea
Matrices $(V_s)_{1\leq s\leq n}$, defined by \eqref{DefVinfty}, reduce to
\beq\label{DefVinftyReducesrinftyequal2}V_1=\begin{pmatrix}\frac{1}{\check{q}_1}& \dots &\dots& \frac{1}{\check{q}_g}\\
\frac{1}{\check{q}_1^2}& \dots &\dots& \frac{1}{\check{q}_g^2}\\
\vdots & & & \vdots\\
\vdots & & & \vdots\\
\frac{1}{\check{q}_1^{r_1}}& \dots &\dots& \frac{1}{\check{q}_g^{r_1}}\\
\end{pmatrix},\,\, 
V_s=\begin{pmatrix}\frac{1}{\check{q}_1-\td{X}_s}& \dots &\dots& \frac{1}{\check{q}_g-\td{X}_s}\\
\frac{1}{(\check{q}_1-\td{X}_s)^2}& \dots &\dots& \frac{1}{(\check{q}_g-\td{X}_s)^2}\\
\vdots & & & \vdots\\
\vdots & & & \vdots\\
\frac{1}{(\check{q}_1-\td{X}_s)^{r_s}}& \dots &\dots& \frac{1}{(\check{q}_g-\td{X}_s)^{r_s}}\\
\end{pmatrix} \,,\forall \, s\in \llbracket 2,n\rrbracket.
\eeq
We have
\beq \label{Hreducedrinftyequal2}\begin{pmatrix} \mathbf{e}_1^t& \mathbf{e}_1^t&\dots&\mathbf{e}_1^t\\
V_1^t&V_2^t&\dots &V_n^{t}\end{pmatrix}\begin{pmatrix}\mathbf{H}_{X_1}\\ \vdots \\ \mathbf{H}_{X_n} \end{pmatrix}=\begin{pmatrix} 
\hbar \underset{j=1}{\overset{g}{\sum}} \check{p}_j+2t_{\infty^{(1)},0}-\hbar\\
\check{p}_1^2 + \check{p}_1\left(\frac{\hbar r_1}{\check{q}_1}+\underset{s=2}{\overset{n}{\sum}} \frac{\hbar r_s}{\check{q}_1-\td{X}_s}\right)+\td{P}_2(\check{q}_1)+\hbar \underset{i\neq 1}{\sum}\frac{\check{p}_i-\check{p}_1}{\check{q}_1-\check{q}_i}\\
\vdots\\
\vdots\\
\check{p}_g^2 + \check{p}_g\left(\frac{\hbar r_1}{\check{q}_g}+\underset{s=2}{\overset{n}{\sum}} \frac{\hbar r_s}{\check{q}_g-\td{X}_s}\right)+\td{P}_2(\check{q}_g)+\hbar \underset{i\neq g}{\sum}\frac{\check{p}_i-\check{p}_g}{\check{q}_g-\check{q}_i}.
\end{pmatrix}
\eeq

Note that one may obtain the coefficients $\left(\mu_j^{(\boldsymbol{\alpha}_\tau)}\right)_{1\leq j\leq g}\cup\left\{\nu^{(\boldsymbol{\alpha}_\tau)}_{\infty,0}\right\}$ for any $\tau \in \mathcal{T}_{\text{iso}}$ by solving
\beq \label{rinftyequal2Coeffmu}\begin{pmatrix}  V_1 \\ \vdots \\\vdots \\V_n\end{pmatrix}\begin{pmatrix} \mu^{(\boldsymbol{\alpha}_\tau)}_1\\ \vdots\\ \vdots\\\vdots\\ \mu^{(\boldsymbol{\alpha}_\tau)}_g\end{pmatrix}= \begin{pmatrix} \boldsymbol{\nu}^{(\boldsymbol{\alpha}_\tau)}_{X_1} \\ \vdots\\\vdots\\ \ \boldsymbol{\nu}^{(\boldsymbol{\alpha}_\tau)}_{X_n}\end{pmatrix}\, \text{ with }\, \boldsymbol{\nu}^{(\boldsymbol{\alpha}_\tau)}_{{X_s}}=\begin{pmatrix} \alpha_{X_s}+\nu^{(\boldsymbol{\alpha}_\tau)}_{\infty,0}\\ -\nu^{(\boldsymbol{\alpha}_\tau)}_{{X_s},1}\\-\nu^{(\boldsymbol{\alpha}_\tau)}_{{X_s},2}\\   \vdots \\ -\nu^{(\boldsymbol{\alpha}_\tau)}_{{X_s},r_s-1}\end{pmatrix}\,,\,\,\forall\, s\in\llbracket 1,n\rrbracket.\eeq

\subsection{The case $r_\infty=1$ and $n\geq 2$}\label{Sectionrinftyequal1ngeq2}
For $r_\infty=1$ with $n\geq 2$, the canonical choice of trivial times \eqref{TrivialTimesChoice} implies that
\beq X_1=0 \,\, \text{ and } \,\,  X_2=1.\eeq
In particular, isomonodromic times simplify to
\bea
\tau_{X_s,k}&=&2t_{X_s^{(1)},k}   \,,\,\, \forall \, (s,k)\in \llbracket 1,n\rrbracket\times \llbracket 1, r_s-1\rrbracket,\cr
\td{X}_s&=&X_s \,,\, \forall \, s\in \llbracket 3,n\rrbracket.
\eea
The expression for $\td{P}_2$ reduces to
\bea\label{tdP2reducedrinftyequal1ngeq2} \td{P}_2(\lambda)
&=&-\frac{1}{4}\sum_{j=r_1+2}^{2r_1}\sum_{i=0}^{2r_1-j}\tau_{X_1,r_1-1-i}\tau_{X_1,i+j-r_1-1}\lambda^{-j}\cr
&&-\frac{1}{4}\sum_{i=1}^{r_1-2}\tau_{X_1, r_1-1-i}\tau_{X_1,i}\lambda^{-r_1-1}\cr
&&-\left(t_{X_1^{(1)},0}\tau_{X_1,r_1-1}\delta_{r_1\geq 2}+(t_{X_1^{(1)},0})^2\delta_{r_1=1} \right)\lambda^{-r_1-1}\cr
&&-\frac{1}{4}\sum_{j=r_2+2}^{2r_2}\sum_{i=0}^{2r_2-j}\tau_{X_2,r_2-1-i}\tau_{X_2,i+j-r_2-1}(\lambda-1)^{-j}\cr
&&-\frac{1}{4}\sum_{i=1}^{r_2-2}\tau_{X_2, r_2-1-i}\tau_{X_2,i}(\lambda-1)^{-r_2-1}\cr
&&-\left(t_{X_2^{(1)},0}\tau_{X_2,r_2-1}\delta_{r_2\geq 2}+(t_{X_2^{(1)},0})^2\delta_{r_2=1} \right)(\lambda-1)^{-r_2-1}\cr
&&-\frac{1}{4}\sum_{s=3}^n\sum_{j=r_s+2}^{2r_s}\sum_{i=0}^{2r_s-j}\tau_{X_s,r_s-1-i}\tau_{X_s,i+j-r_s-1}(\lambda-\td{X}_s)^{-j}\cr
&&-\frac{1}{4}\sum_{s=3}^n \sum_{i=1}^{r_s-2}\tau_{X_s, r_s-1-i}\tau_{X_s,i}(\lambda-\td{X}_s)^{-r_s-1}\cr
&&-\sum_{s=3}^n \left(t_{X_s^{(1)},0}\tau_{X_s,r_s-1}\delta_{r_s\geq 2}+(t_{X_s^{(1)},0})^2\delta_{r_s=1} \right)(\lambda-\td{X}_s)^{-r_s-1}.
\eea
Matrices $(V_s)_{1\leq s\leq n}$, defined by \eqref{DefVinfty}, reduce to
\beq\label{DefVinftyReducesrinftyequal1ngeq2}V_1=\begin{pmatrix}\frac{1}{\check{q}_1}& \dots &\dots& \frac{1}{\check{q}_g}\\
\frac{1}{\check{q}_1^2}& \dots &\dots& \frac{1}{\check{q}_g^2}\\
\vdots & & & \vdots\\
\vdots & & & \vdots\\
\frac{1}{\check{q}_1^{r_1}}& \dots &\dots& \frac{1}{\check{q}_g^{r_1}}\\
\end{pmatrix},\,\, V_2=\begin{pmatrix}\frac{1}{(\check{q}_1-1)}& \dots &\dots& \frac{1}{(\check{q}_g-1)}\\
\frac{1}{(\check{q}_1-1)^2}& \dots &\dots& \frac{1}{(\check{q}_g-1)^2}\\
\vdots & & & \vdots\\
\vdots & & & \vdots\\
\frac{1}{(\check{q}_1-1)^{r_2}}& \dots &\dots& \frac{1}{(\check{q}_g-1)^{r_2}}\\
\end{pmatrix}
\eeq
and 
\beq \forall\, s\in \llbracket 3,n\rrbracket \,:\, V_s=\begin{pmatrix}\frac{1}{\check{q}_1-\td{X}_s}& \dots &\dots& \frac{1}{\check{q}_g-\td{X}_s}\\
\frac{1}{(\check{q}_1-\td{X}_s)^2}& \dots &\dots& \frac{1}{(\check{q}_g-\td{X}_s)^2}\\
\vdots & & & \vdots\\
\vdots & & & \vdots\\
\frac{1}{(\check{q}_1-\td{X}_s)^{r_s}}& \dots &\dots& \frac{1}{(\check{q}_g-\td{X}_s)^{r_s}}\\
\end{pmatrix} \,,\forall \, s\in \llbracket 3,n\rrbracket.
\eeq
We have
\bea\label{Hreducedrinftyequal1ngeq2}&&\begin{pmatrix} 
\mathbf{e}_1^t&\mathbf{e}_1^t&\dots&\dots& \mathbf{e}_1^t\\
\mathbf{e}_2^t\delta_{r_1\geq 2}&\mathbf{e}_2^t\delta_{r_2\geq 2}+\mathbf{e}_1^t&\mathbf{e}_2^t\delta_{r_3\geq 2}+\td{X}_3\mathbf{e}_1^t &\dots& \mathbf{e}_2^t\delta_{r_n\geq 2}+\td{X}_n\mathbf{e}_1^t\\ 
V_1^t&V_2^t&V_3^t&\dots &V_n^{t}\end{pmatrix}\begin{pmatrix}\mathbf{H}_{X_1}\\ \vdots \\ \mathbf{H}_{X_n} \end{pmatrix}\cr
&&=\begin{pmatrix} 
\hbar \underset{j=1}{\overset{g}{\sum}}  \check{p}_j\\
\hbar \underset{j=1}{\overset{g}{\sum}} \check{q}_j\check{p}_j -\underset{s=1}{\overset{n}{\sum}} (t_{X_s^{(1)},0})^2\delta_{r_s=1}+t_{\infty^{(1)},0}(t_{\infty^{(1)},0} -\hbar)\\
\check{p}_1^2 + \check{p}_1\left(\frac{\hbar r_1}{\check{q}_1}+\frac{\hbar r_2}{\check{q}_1-1}+\underset{s=3}{\overset{n}{\sum}} \frac{\hbar r_s}{\check{q}_1-\td{X}_s}\right)+\td{P}_2(\check{q}_1)+\hbar \underset{i\neq 1}{\sum}\frac{\check{p}_i-\check{p}_1}{\check{q}_1-\check{q}_i}\\
\vdots\\
\vdots\\
\check{p}_g^2 + \check{p}_g\left(\frac{\hbar r_1}{\check{q}_g}+\frac{\hbar r_2}{\check{q}_g-1}+\underset{s=3}{\overset{n}{\sum}} \frac{\hbar r_s}{\check{q}_g-\td{X}_s}\right)+\td{P}_2(\check{q}_g)+\hbar \underset{i\neq g}{\sum}\frac{\check{p}_i-\check{p}_g}{\check{q}_g-\check{q}_i}
\end{pmatrix}.\cr&&
\eea
Note that one may obtain the coefficients $\left(\mu_j^{(\boldsymbol{\alpha}_\tau)}\right)_{1\leq j\leq g}\cup\left\{\nu^{(\boldsymbol{\alpha}_\tau)}_{\infty,-1},\nu^{(\boldsymbol{\alpha}_\tau)}_{\infty,0}\right\}$ for any $\tau \in \mathcal{T}_{\text{iso}}$ by solving
\beq \label{muCoeffsrinftyequal1ngeq2} \begin{pmatrix}  V_1 \\ V_2 \\\vdots \\V_n\end{pmatrix}\begin{pmatrix} \mu^{(\boldsymbol{\alpha}_\tau)}_1\\ \vdots\\ \vdots\\\vdots\\ \mu^{(\boldsymbol{\alpha}_\tau)}_g\end{pmatrix}= \begin{pmatrix} \boldsymbol{\nu}^{(\boldsymbol{\alpha}_\tau)}_{X_1} \\ \vdots\\\vdots\\ \ \boldsymbol{\nu}^{(\boldsymbol{\alpha}_\tau)}_{X_n}\end{pmatrix}\, \text{ with }\, \boldsymbol{\nu}^{(\boldsymbol{\alpha}_\tau)}_{X_s}=\begin{pmatrix} \alpha_{X_s}+\nu^{(\boldsymbol{\alpha}_\tau)}_{\infty,0}+\nu^{(\boldsymbol{\alpha}_\tau)}_{\infty,-1}\td{X}_s\\ -\nu^{(\boldsymbol{\alpha}_\tau)}_{{X_s},1}+\nu^{(\boldsymbol{\alpha}_\tau)}_{\infty,-1}\\-\nu^{(\boldsymbol{\alpha}_\tau)}_{{X_s},2}\\   \vdots \\ -\nu^{(\boldsymbol{\alpha}_\tau)}_{{X_s},r_s-1}\end{pmatrix}.\eeq

\subsection{The case $r_\infty=1$ and $n=1$}\label{Sectionrinftyequal1nequal1}
For $r_\infty=1$ and $n=1$, ($r_1\geq 3$ in order to have $g>0$) the canonical choice of trivial times \eqref{TrivialTimesChoice} implies that
\beq X_1=0 \,\, \text{ and } \,\, t_{X_1^{(1)},r_1-1}=1.\eeq
In particular, isomonodromic times simplify to
\beq
\tau_{X_1,k}=2t_{X_1^{(1)},k}   \,,\,\, \forall \, k\in \llbracket 1, r_1-2\rrbracket.
\eeq
The expression for $\td{P}_2$ reduces to
\bea\label{tdP2reducedrinftyequal1nequal1} \td{P}_2(\lambda)
&=&-\lambda^{-2r_1}
-\sum_{k=1}^{r_1-2}\tau_{X_1,k}\lambda^{-r_1-k-1}
-2t_{X_1^{(1)},0}\lambda^{-r_1-1}\cr
&&-\frac14\sum_{\substack{1\le a,b\le r_1-2\\a+b\ge r_1-1}}
\tau_{X_1,a}\tau_{X_1,b}\lambda^{-a-b-2}.
\eea
Matrix $V_1$ is defined by
\beq\label{DefVinftyReducesrinftyequal1nequal1}V_1=\begin{pmatrix}\frac{1}{\check{q}_1}& \dots &\dots& \frac{1}{\check{q}_g}\\
\frac{1}{\check{q}_1^2}& \dots &\dots& \frac{1}{\check{q}_g^2}\\
\vdots & & & \vdots\\
\vdots & & & \vdots\\
\frac{1}{\check{q}_1^{r_1}}& \dots &\dots& \frac{1}{\check{q}_g^{r_1}}\\
\end{pmatrix}.
\eeq
We have
\beq \label{Hreducedrinftyequal1nequal1}\begin{pmatrix}1&0&\dots&\dots& 0\\ 
0&1&0&\dots& 0\\\\
&& (V_1)^t&&\end{pmatrix}\begin{pmatrix}H_{X_1,1}\\ \vdots \\ H_{X_1,r_1}\end{pmatrix}=\begin{pmatrix}\hbar \underset{j=1}{\overset{g}{\sum}}  \check{p}_j\\
\hbar \underset{j=1}{\overset{g}{\sum}} \check{q}_j\check{p}_j +t_{\infty^{(1)},0}(t_{\infty^{(1)},0} -\hbar)\\
 \check{p}_1^2 + \check{p}_1\left(\frac{\hbar r_1}{\check{q}_1}\right)+\td{P}_2(\check{q}_1)+\hbar \underset{i\neq 1}{\sum}\frac{\check{p}_i-\check{p}_1}{\check{q}_1-\check{q}_i}\\
\vdots\\
\vdots\\
\check{p}_g^2 + \check{p}_g\left(\frac{\hbar r_1}{\check{q}_g}\right)+\td{P}_2(\check{q}_g)+\hbar \underset{i\neq g}{\sum}\frac{\check{p}_i-\check{p}_g}{\check{q}_g-\check{q}_i}
\end{pmatrix}.
\eeq
Note that in this case ($t_{X_1^{(1)},r_1-1}=1=\frac{1}{2}\tau_{X_1,r_1-1}$), the matrix $M_1$ reduces to:
\beq \label{MatrixMReducedrinftyequal2ngeq1}
M_1=\begin{pmatrix}2&0&\dots& &\dots &0\\
\tau_{X_1,r_1-2}&2& 0& & &\vdots\\
\vdots & \ddots&\ddots &\ddots  & &\vdots\\
\vdots &\ddots&\ddots&\ddots&0&\vdots\\
\tau_{X_1,2}&\ddots &\ddots&\ddots& 2&0\\
\tau_{X_1,1}&\tau_{X_1,2}& \dots & & \tau_{X_1,r_1-2}& 2
 \end{pmatrix}.
\eeq

Note that one could obtain the coefficients $\left(\mu_j^{(\boldsymbol{\alpha}_\tau)}\right)_{1\leq j\leq g}\cup\left\{\nu^{(\boldsymbol{\alpha}_\tau)}_{\infty,-1},\nu^{(\boldsymbol{\alpha}_\tau)}_{\infty,0}\right\}$ for any $\tau \in \mathcal{T}_{\text{iso}}$ by solving
\beq\label{muCoeffsrinftyequal1nequal2} V_1 \begin{pmatrix} \mu^{(\boldsymbol{\alpha}_\tau)}_1\\ \vdots\\\vdots\\ \mu^{(\boldsymbol{\alpha}_\tau)}_g\end{pmatrix}= \boldsymbol{\nu}^{(\boldsymbol{\alpha}_\tau)}_{X_1}\, \text{ with }\, \boldsymbol{\nu}^{(\boldsymbol{\alpha}_\tau)}_{X_1}=\begin{pmatrix} \nu^{(\boldsymbol{\alpha}_\tau)}_{\infty,0}\\ \nu^{(\boldsymbol{\alpha}_\tau)}_{\infty,-1}\\-\nu^{(\boldsymbol{\alpha}_\tau)}_{{X_1},2}\\   \vdots \\ -\nu^{(\boldsymbol{\alpha}_\tau)}_{{X_1},r_1-1}\end{pmatrix}.\eeq

\section{Examples: Painlev\'{e} equations and Fuchsian systems}\label{SectionExemples}
We shall apply the general theory developed in this article to the examples of the Painlev\'{e} equations (giving all possibilities covered by the present work  to have $g=1$), Chiba's genus-two system $(P_{\mathrm{II}-2})_2$ and Fuchsian systems. In each case, we shall write explicitly the evolution equations for $(\check{\mathbf{q}},\check{\mathbf{p}})$, the corresponding Hamiltonians and the associated Lax pairs $\left(\td{L}(\lambda,\hbar),\td{A}_{\boldsymbol{\alpha}_\tau}(\lambda,\hbar)\right)$. Note that the Painlev\'{e} $1$ equation is not covered by the present work since it corresponds to the case where $\infty$ is ramified.\footnote{The Painlev\'{e} $1$ hierarchy is treated in detail in \cite{MarchalAlameddineP1Hierarchy2023} using similar techniques.} These examples have been studied in the literature so this section is meant as a test to show that the present work correctly recovers these known cases with almost no computation. It may also help the reader to understand how to apply concretely the various formulas obtained in the general setting.

\subsection{Painlev\'{e} $2$ case}\label{SectionP2}
The Painlev\'{e} $2$ case corresponds to $n=0$ and $r_\infty=4$ and is a direct application of Section \ref{P2Hierarchy}. It provides a genus $1$ curve. The only isomonodromic time is $t=\tau_{\infty,1}=2t_{\infty^{(1)},1}$. Proposition \ref{PropAsymptoticExpansionA12} gives $\nu_{\infty,-1}^{(\boldsymbol{\alpha}_{t})}=0$, $\nu_{\infty,0}^{(\boldsymbol{\alpha}_{t})}=0$ and equation \eqref{nuinftyReduced} gives $\nu_{\infty,1}^{(\boldsymbol{\alpha}_{t})}=\frac{1}{2}$ so that from Remark \ref{Remarkmucoeffrinftygeq3} $\mu^{(\boldsymbol{\alpha}_{t})}_1=\nu_{\infty,1}^{(\boldsymbol{\alpha}_{t})}=\frac{1}{2}$.
Using \eqref{tdP2reducedPainleve2} we have,
\beq \td{P}_2(\lambda)=
-\lambda^{4} -t \lambda^2 -2t_{\infty^{(1)},0}\lambda.
\eeq
Thus, we get from \eqref{HCoeffPainleve2Hierarchy} that the Hamiltonian is given by 
\beq \text{Ham}^{(\boldsymbol{\alpha}_t)}(\check{q},\check{p})=\frac{1}{2} H_{\infty,0}= \frac{1}{2}\check{p}^2-\frac{1}{2}\check{q}^4-\frac{t}{2}\check{q}^2-\frac{\check{q}}{2}\left(2t_{\infty^{(1)},0}-\hbar\right).
\eeq
Moreover, Theorem \ref{HamTheorem} provides
\bea \hbar \partial_{t}[\check{q}]&=&\check{p},\cr
\hbar \partial_{t}[\check{p}]&=&\frac{1}{2}\left(4\check{q}^{3} +2t\check{q} + 2t_{\infty^{(1)},0}-\hbar\right)=2\check{q}^{3} +t\check{q} + t_{\infty^{(1)},0}-\frac{\hbar}{2}.
\eea
In particular, $\check{q}$ satisfies the ODE
\beq \hbar^2 \frac{\partial^2 \check{q}}{\partial t^2}=2\check{q}^{3} +t\check{q} + t_{\infty^{(1)},0}-\frac{\hbar}{2},\eeq
\sloppy{i.e. the standard Painlev\'{e} $2$ equation with parameter $\alpha=t_{\infty^{(1)},0}-\frac{\hbar}{2}=\frac{1}{2}\left(t_{\infty^{(1)},0}-t_{\infty^{(2)},0}-\hbar\right)$.}
The associated Lax pair in the gauge without apparent singularities, normalized at infinity according to \eqref{NormalizationInfty} and under the choice of trivial times \eqref{TrivialTimesChoice} (i.e. $t_{\infty^{(2)},k}=-t_{\infty^{(1)},k}$ for all $k\in \llbracket 0,3\rrbracket$ and $t_{\infty^{(1)},3}=1$, $t_{\infty^{(1)},2}=0$ and $t_{\infty^{(1)},1}=\frac{t}{2}$) is

\bea \td{L}(\lambda,\hbar)&=&\begin{pmatrix}-\lambda^2+\check{q}^2+\check{p} & \lambda-\check{q}\\ (2\check{q}^2+2\check{p}+t)\lambda+2\check{q}^3+(2\check{p}+t)\check{q}+2t_{\infty^{(1)},0}& \lambda^2-\check{q}^2-\check{p} \end{pmatrix}\cr
\td{A}_{\boldsymbol{\alpha}_t}(\lambda,\hbar)&=&\begin{pmatrix}-\frac{\lambda+\check{q}}{2} & \frac{1}{2}\\\check{q}^2+\check{p}+\frac{t}{2} & \frac{\lambda+\check{q}}{2} \end{pmatrix}.\cr&&
\eea

One may also recover the standard Jimbo-Miwa Lax pair \cite{JimboMiwa} by defining $(Q,P)=(-\check{q},-\check{p}-\check{q}^2-\frac{t}{2})$ which is time-dependent. The transformed Hamiltonian is
\beq\text{Ham}^{(\boldsymbol{\alpha}_t)}(Q,P)=\frac{1}{2}P^2+\left(Q^2+\frac{t}{2}\right)P+\frac{t^2}{8}+t_{\infty^{(1)},0}Q.\eeq

\subsection{Painlev\'{e} $3$ case}\label{SectionP3}
The type $D_6$  Painlev\'{e} $3$ case\footnote{Degenerate Painlev\'{e} $3$ equations of types $D_7$ and $D_8$ correspond to
isomonodromic deformations of $2\times 2$ linear systems with ramified irregular singularities that do not fit into the present setup.} corresponds to $n=1$ with $r_1=2$ and $r_\infty=2$ and thus provides a spectral curve of genus $g=1$. It is a special case of Section \ref{Sectionrinftyequal2}. The canonical choice of trivial times provides $X_1=0$, $t_{\infty^{(1)},1}=1=-t_{\infty^{(2)},1}$ and the isomonodromic time is $t=\tau_{X_1,1}=2t_{X_1^{(1)},1}$. In particular from \eqref{nuXsReduced}, we get that
\beq \,\nu^{(\boldsymbol{\alpha}_t)}_{\infty,-1}=0 \, \text{ and }\,  \nu^{(\boldsymbol{\alpha}_t)}_{{X_1},0}=0  \, \text{ and }\, \nu^{(\boldsymbol{\alpha}_t)}_{{X_1},1}=-\frac{1}{t}. \eeq
In this case, we have
\beq 
\td{P}_2(\lambda)
=-1- \frac{t_{X_1^{(1)},0}t}{\lambda^3}- \frac{t^2}{4\lambda^4}.
\eeq
Since $r_\infty=2$, the coefficients $\left(\nu^{(\boldsymbol{\alpha}_t)}_{\infty,0},\mu_1^{(\boldsymbol{\alpha}_t)}\right)$ are determined by \eqref{rinftyequal2Coeffmu}:
\beq \begin{pmatrix} \frac{1}{\check{q}}\\\frac{1}{\check{q}^2} \end{pmatrix} \mu_1^{(\boldsymbol{\alpha}_t)}=\begin{pmatrix} -\nu^{(\boldsymbol{\alpha}_t)}_{{X_1},0}+\nu^{(\boldsymbol{\alpha}_t)}_{\infty,0}\\ -\nu^{(\boldsymbol{\alpha}_t)}_{{X_1},1}+\nu^{(\boldsymbol{\alpha}_t)}_{\infty,-1}\end{pmatrix}=\begin{pmatrix}  \nu^{(\boldsymbol{\alpha}_t)}_{\infty,0}\\ \frac{1}{t}\end{pmatrix}
\eeq
which is equivalent to
\beq \nu^{(\boldsymbol{\alpha}_t)}_{\infty,0}=\frac{\check{q}}{t}\,,\, \mu_1^{(\boldsymbol{\alpha}_t)}=\frac{\check{q}^2}{t}.\eeq
The constants $H_{X_1,1}$ and $H_{X_1,2}$ are determined by equation \eqref{Hreducedrinftyequal2},
\beq  \begin{pmatrix} 1& 0\\ \frac{1}{\check{q}} & \frac{1}{\check{q}^2}\end{pmatrix} \begin{pmatrix} H_{X_1,1}\\ H_{X_1,2}\end{pmatrix}=\begin{pmatrix}\hbar \check{p}+ (2t_{\infty^{(1)},0}-\hbar) \\ \check{p}^2+\frac{2\hbar}{\check{q}}\check{p} -1- \frac{t_{X_1^{(1)},0}t}{\check{q}^3}- \frac{t^2}{4\check{q}^4}\end{pmatrix},
\eeq
i.e.
\bea H_{X_1,1}&=&\hbar \check{p}+ 2t_{\infty^{(1)},0}-\hbar,\cr
H_{X_1,2}&=&\check{q}^2\check{p}^2+\hbar\check{q} \check{p}-\frac{t^2}{4\check{q}^2}-\frac{t_{X_1^{(1)},0}t}{\check{q}}-\check{q}^2-
\left(2t_{\infty^{(1)},0}-\hbar\right)\check{q}.
\eea
The corresponding Hamiltonian is
\bea \text{Ham}^{(\boldsymbol{\alpha}_t)}(\check{q},\check{p})&=&-\nu^{(\boldsymbol{\alpha}_t)}_{{X_1},0}H_{X_1,1}- \nu^{(\boldsymbol{\alpha}_t)}_{{X_1},1}H_{X_1,2}=\frac{1}{t}H_{X_1,2} \cr
&=& \frac{1}{t}\left(\check{q}^2\check{p}^2+ \hbar \check{q} \check{p} -\check{q}^2 -(2t_{\infty^{(1)},0}-\hbar)\check{q} -\frac{t_{X_1^{(1)},0} t}{\check{q}} -\frac{t^2}{4\check{q}^2} \right).\eea
We also get from Theorem \ref{HamTheorem}:
\bea \hbar \partial_t \check{q}&=&
\frac{2\check{q}^2}{t} \check{p}+\hbar\frac{\check{q}}{t}, \cr
\hbar \partial_t \check{p}&=&-\frac{2\check{q}}{t}\check{p}^2-\hbar\frac{\check{p}}{t}-\frac{t}{2\check{q}^3}-\frac{t_{X_1^{(1)},0}}{\check{q}^2}+\frac{2\check{q}}{t}+ \frac{2t_{\infty^{(1)},0}-\hbar}{t}.
\eea
Finally $\check{q}$ satisfies the Painlev\'{e} $3$ equation:
\beq \hbar^2 \ddot{\check{q}}=\frac{(\hbar \dot{\check{q}})^2}{\check{q}}-\frac{\hbar^2 \dot{\check{q}}}{t}+\frac{\alpha \check{q}^2+\gamma \check{q}^3}{t^2}+\frac{\beta}{t}+\frac{\delta}{\check{q}},\eeq
with
\beq \alpha=2\left(2t_{\infty^{(1)},0}-\hbar\right) \,,\, \beta=-2t_{X_1^{(1)},0}\,,\, \gamma=4 \,,\, \delta=-1.\eeq

The associated Lax pair in the gauge without apparent singularities, normalized at infinity according to \eqref{NormalizationInfty} and under the canonical choice of trivial times \eqref{TrivialTimesChoice} (i.e. $t_{\infty^{(2)},k}=-t_{\infty^{(1)},k}$ and $t_{X_1^{(2)},k}=-t_{X_1^{(1)},k}$ for all $k\in \llbracket 0,1\rrbracket$, $X_1=0$, $t_{\infty^{(1)},1}=1$ and $t_{X_1^{(1)},1}=\frac{t}{2}$) is
\bea \td{L}(\lambda,\hbar)&=&\text{diag}(-1,1)+\frac{\td{L}_1}{\lambda}+\frac{\td{L}_2}{\lambda^2} \,\,\text{ with }\cr
\td{L}_1&=& \begin{pmatrix}-t_{\infty^{(1)},0}& 1\\ \check{q}^2\check{p}^2+2\check{q}^2\check{p}+ \frac{4\check{q}^4-4t_{\infty^{(1)},0}^2\check{q}^2 -4t_{X_1^{(1)},0} t\check{q}-t^2}{4\check{q}^2}&t_{\infty^{(1)},0}
\end{pmatrix},\cr
\td{L}_2&=&\begin{pmatrix} \check{q}^2\check{p}+\check{q}^2+t_{\infty^{(1)},0}\check{q}& -\check{q}\\ 
\check{q}^3\check{p}^2+2\check{q}^2(\check{q}+t_{\infty^{(1)},0} )\check{p}+\frac{(2\check{q}^2+2\check{q}t_{\infty^{(1)},0})^2-t^2}{4\check{q}}&-(\check{q}^2\check{p}+\check{q}^2+t_{\infty^{(1)},0}\check{q})\end{pmatrix} .\cr
&&
\eea
The associated auxiliary deformation matrix is
\bea \td{A}_{\boldsymbol{\alpha}_t}(\lambda,\hbar)&=&\text{diag}\left(-\frac{\check{q}}{t},\frac{\check{q}}{t}\right)\cr
&&+\frac{1}{\lambda} \begin{pmatrix}-\frac{\check{q}^2\check{p}}{t}-\frac{\check{q}(\check{q}+t_{\infty^{(1)},0})}{t}& \frac{\check{q}}{t}\\
-\frac{\left((\check{p}+1)\check{q}^2+t_{\infty^{(1)},0}\check{q}+\frac{t}{2}\right)\left((\check{p}+1)\check{q}^2+t_{\infty^{(1)},0}\check{q}-\frac{t}{2}\right)}{t\check{q}}
 & -\left(-\frac{\check{q}^2\check{p}}{t}-\frac{\check{q}(\check{q}+t_{\infty^{(1)},0})}{t}\right)\end{pmatrix}.\cr
&&
\eea

\subsection{Painlev\'{e} $4$ case}
\subsubsection{Painlev\'{e} $4$ case in the canonical choice of trivial times}\label{SectionP4new}
The Painlev\'{e} $4$ case corresponds to $n=1$ with $r_1=1$ and $r_\infty=3$. It is a special case of Section \ref{Sectionrgeq3ngeq0}. The canonical choice of trivial times \eqref{TrivialTimesChoice} provides $t_{\infty^{(1)},2}=1$  and $t_{\infty^{(1)},1}=0$ and the only non-trivial isomonodromic time is $t=\td{X}_1=X_1$. First, we observe from \eqref{tdP2reducedrinftygeq3} that
\beq 
\td{P}_2(\lambda)=-\frac{t_{X_1^{(1)},0}^2}{(\lambda-t)^2}- \lambda^2    -2t_{\infty^{(1)},0}.
\eeq
Since the only isomonodromic time is a position of a pole we get
\[
\nu^{(\boldsymbol\alpha_t)}_{\infty,-1}
=\nu^{(\boldsymbol\alpha_t)}_{\infty,0}=0,
\qquad \nu^{(\boldsymbol\alpha_t)}_{X_1,0}=-1.
\]
The constant $H_{X_1,1}$ is determined by Proposition \ref{PropDefCi2},
\beq V_1^t H_{X_1,1}=\check{p}^2+ \frac{\hbar \check{p}}{\check{q}-t}-\frac{t_{X_1^{(1)},0}^2}{(\check{q}-t)^2}- \check{q}^2-2t_{\infty^{(1)},0}  +\hbar,\eeq
i.e.
\beq \label{HX11P40}H_{X_1,1}=\check{p}^2(\check{q}-t)+ \hbar \check{p}-\frac{t_{X_1^{(1)},0}^2}{\check{q}-t}- (\check{q}^2+2t_{\infty^{(1)},0}  -\hbar)(\check{q}-t).\eeq
The Hamiltonian is given by \eqref{NewHamReduced}
\bea \label{HamP40}\text{Ham}^{(\boldsymbol{\alpha}_t)}(\check{q},\check{p})&=&H_{X_1,1}\cr
&=&\check{p}^2(\check{q}-t)+ \hbar \check{p}-\frac{t_{X_1^{(1)},0}^2}{\check{q}-t}- (\check{q}^2+2t_{\infty^{(1)},0}  -\hbar)(\check{q}-t)\cr
&=& (\check{q}-t)\check{p}^2+ \hbar \check{p}-\frac{t_{X_1^{(1)},0}^2}{\check{q}-t}-\check{q}^3+t\check{q}^2-(2t_{\infty^{(1)},0}  -\hbar)\check{q}+(2t_{\infty^{(1)},0}  -\hbar)t\cr
&=& (\check{q}-t)\check{p}^2+ \hbar \check{p}-\frac{t_{X_1^{(1)},0}^2}{\check{q}-t}-(\check{q}-t)^3-2t(\check{q}-t)^2-t^2(\check{q}-t) \cr
&&-(2t_{\infty^{(1)},0}  -\hbar)(\check{q}-t).
\eea
Coefficient $\mu_1^{(\boldsymbol{\alpha}_t)}$ is determined by Remark \ref{Remarkmucoeffrinftygeq3}: $\mu_1^{(\boldsymbol{\alpha}_t)}=\check{q}-t$ so that the evolution equations of Theorem \ref{HamTheorem} read:
\bea \label{EvoP40}\hbar \partial_t \check{q}&=&2(\check{q}-t)\left(\check{p} +\frac{\hbar}{2(\check{q}-t)}\right)=2\check{p}(\check{q}-t) +\hbar,\cr
\hbar \partial_t \check{p}&=&(\check{q}-t)\left(\frac{\hbar \check{p}}{(\check{q}-t)^2}-\frac{2t_{X_1^{(1)},0}^2}{(\check{q}-t)^3}+ 2\check{q}-\frac{H_{X_1,1}}{(\check{q}-t)^2} \right)\cr
&=& -\check{p}^2-\frac{t_{X_1^{(1)},0}^2}{(\check{q}-t)^2}+ 3\check{q}^2- 2t\check{q}+ 2t_{\infty^{(1)},0}  -\hbar\cr
&=& -\check{p}^2-\frac{t_{X_1^{(1)},0}^2}{(\check{q}-t)^2}+3(\check{q}-t)^2+4t(\check{q}-t)+t^2 +2t_{\infty^{(1)},0}  -\hbar.
\eea
Thus, $\check{q}(t)$ satisfies the differential equation:
\beq \hbar^2(\check{q}-t)\ddot{\check{q}}=\frac{1}{2}(\hbar \dot{\check{q}})^2-\hbar^2\dot{\check{q}}+6(\check{q}-t)^4+8t(\check{q}-t)^3+2(t^2+2t_{\infty^{(1)},0}-\hbar)(\check{q}-t)^2-2t_{X_1^{(1)},0}^2+\frac{\hbar^2}{2}.
\eeq
Set $u=\check q-t$. The transformed Hamiltonian is
$K(u,\check p,t)=\operatorname{Ham}^{(\boldsymbol\alpha_t)}
(u+t,\check p,t)-\hbar\check p$.
It generates the fixed-pole Hamilton equations \eqref{EvoP4}, with
$u$ in place of $q$, and agrees with \eqref{HamP4}.
Consequently $u(t)$ satisfies \eqref{P4eq}.

The associated Lax pair is
\bea \td{L}(\lambda,\hbar)&=&\begin{pmatrix}-\lambda+\frac{(\check{q}-t)(\check{p}+\check{q})}{\lambda-t}&\frac{\lambda-\check{q}}{\lambda-t}\\
-\frac{(t_{X_1^{(1)},0})^2}{(\check{q}-t)(\lambda-t)}+ \frac{(\check{q}-t)(\check{p}+\check{q})^2}{\lambda-t}+2(\check{q}-t)(\check{p}+\check{q})+2t_{\infty^{(1)},0}& \lambda-\frac{(\check{q}-t)(\check{p}+\check{q})}{\lambda-t}
\end{pmatrix},\cr
\td{A}_{\boldsymbol{\alpha}_t}(\lambda,\hbar)&=&\begin{pmatrix}-(\check{q}-t)\left(\frac{\check{p}+\check{q}}{\lambda-t}+1\right)& \frac{\check{q}-t}{\lambda-t}\\
\frac{(t_{X_1^{(1)},0})^2}{(\check{q}-t)(\lambda-t)}-\frac{(\check{q}-t)(\check{p}+\check{q})^2}{\lambda-t}&(\check{q}-t)\left(\frac{\check{p}+\check{q}}{\lambda-t}+1\right)
\end{pmatrix}.\cr
&& 
\eea

\subsubsection{Painlev\'{e} $4$ case in the Jimbo-Miwa setup}\label{SectionP4}
The Painlev\'{e} $4$ case may also be recovered by another choice of trivial times. Indeed, since it corresponds to $n=1$ with $r_1=1$ and $r_\infty=3$ one may decide to fix $t_{\infty^{(1)},2}=1=-t_{\infty^{(2)},2}$ and $X_1=0$ and take $t=t_{\infty^{(1)},1}$ as the only non-trivial isomonodromic time. This is the standard setup used by the Japanese school and in the works of Jimbo-Miwa \cite{JimboMiwa} rather than fixing the two leading coefficients at infinity. In this setup, Theorem \ref{HamTheorem} remains valid and we propose for completeness to apply it with this choice of trivial times. Throughout this subsection, $(q,p)$ denote the unshifted
apparent-singularity coordinates in the traceless normalization.
Here $T_2=1$ and $T_1=t$, so the invariant coordinates of
\eqref{SymplecticChange} are $\check q=q+t$ and $\check p=p$.
First, we observe that
\beq 
\td{P}_2(\lambda)=-\frac{t_{X_1^{(1)},0}^2}{\lambda^2}- \lambda^2-2t \lambda-t^2-2t_{\infty^{(1)},0}.
\eeq
Since $r_\infty=3$, coefficients $(\nu^{(\boldsymbol{\alpha}_t)}_{\infty,k})_{-1\leq k\leq 0}$ are determined by Proposition \ref{PropAsymptoticExpansionA12},
\beq \label{Resspecial}\begin{pmatrix} 2 &0 \\
2t & 2 
\end{pmatrix} \begin{pmatrix}\nu^{(\boldsymbol{\alpha}_t)}_{\infty,-1}\\ \nu^{(\boldsymbol{\alpha}_t)}_{\infty,0} \end{pmatrix} = \begin{pmatrix}0 \\2 \end{pmatrix} \,\, \Leftrightarrow\,\, \nu^{(\boldsymbol{\alpha}_t)}_{\infty,-1}=0 \,\text{ and }\, \nu^{(\boldsymbol{\alpha}_t)}_{\infty,0}=1.\eeq
Proposition \ref{PropA12Form} implies that $\mu_1^{(\boldsymbol{\alpha}_t)}$ is determined by
\beq V_1 \mu_1^{(\boldsymbol{\alpha}_t)}=\frac{1}{q} \mu_1^{(\boldsymbol{\alpha}_t)}= -\nu^{(\boldsymbol{\alpha}_t)}_{{X_1},0}+\nu^{(\boldsymbol{\alpha}_t)}_{\infty,0}
+\nu^{(\boldsymbol{\alpha}_t)}_{\infty,-1}X_1= \nu^{(\boldsymbol{\alpha}_t)}_{\infty,0} \,\, \Leftrightarrow \,\, \mu_1^{(\boldsymbol{\alpha}_t)}=q.\eeq
The constant $H_{X_1,1}$ is determined by Proposition \ref{PropDefCi2},
\beq V_1^t H_{X_1,1}=p^2+ \frac{\hbar p}{q}-\frac{t_{X_1^{(1)},0}^2}{q^2}- q^2-2t q-t^2-2t_{\infty^{(1)},0}  +\hbar,\eeq
i.e.
\beq \label{HX11P4}H_{X_1,1}=p^2q+ \hbar p-\frac{t_{X_1^{(1)},0}^2}{q}- q^3-2t q^2-(t^2+2t_{\infty^{(1)},0}-\hbar)q.\eeq
Thus, we get from Theorem \ref{HamTheorem}:
\bea \label{EvoP4}\hbar \partial_t q&=&2pq,\cr 
\hbar \partial_t p&=&-p^2-\frac{t_{X_1^{(1)},0}^2}{q^2}+(t^2-\hbar+2t_{\infty^{(1)},0})+4t q+3q^2.
\eea
The corresponding Hamiltonian is
\bea \label{HamP4}\text{Ham}^{(\boldsymbol{\alpha}_t)}(q,p)&=& \nu^{(\boldsymbol{\alpha}_t)}_{\infty,0}H_{X_1,1}-\hbar  \nu^{(\boldsymbol{\alpha}_t)}_{\infty,0} p-\hbar  \nu^{(\boldsymbol{\alpha}_t)}_{\infty,-1}qp\overset{\eqref{Resspecial}}{=}H_{X_1,1}-\hbar p\cr
&\overset{\eqref{HX11P4}}{=}&qp^2-q^3-2tq^2-(t^2+2t_{\infty^{(1)},0}-\hbar)q-\frac{t_{X_1^{(1)},0}^2}{q}.
\eea
In particular, $q(t)$ satisfies the Painlev\'{e} $4$ equation,
\beq \label{P4eq} \hbar^2q\ddot{q}= \frac{1}{2}(\hbar \dot{q})^2 +6q^4+8tq^3+2(t^2+2t_{\infty^{(1)},0}-\hbar)q^2-2t_{X_1^{(1)},0}^2.\eeq
This is equivalent to saying that $u(t):=2q(t)$ satisfies the normalized Gambier Painlev\'{e} $4$ equation 
\beq \label{NormalizedP4}\hbar^2u\ddot{u}= \frac{1}{2}(\hbar \dot{u})^2 +\frac{3}{2}u^4+4tu^3+2(t^2-\alpha)u^2+\beta\, \text{ with }\,\alpha=-(2t_{\infty^{(1)},0}-\hbar) \,\,,\,\, \beta=-8t_{X_1^{(1)},0}^{\, 2}.\eeq 

\begin{remark}From our work, it appears that the natural form of the Painlev\'{e} $4$ equation is \eqref{P4eq} and not the normalized Gambier Painlev\'{e} $4$ equation \eqref{NormalizedP4}. Indeed, the natural coordinates are the zeros of the Wronskian given by $q$ and not $2q$. This observation was already made in \cite{IwakiMarchalSaenz,MOsl2}.
\end{remark}

The associated Lax pair in the gauge without apparent singularities, normalized at infinity according to \eqref{NormalizationInfty} and under the special choice of trivial times $X_1=0$, $t_{\infty^{(1)},2}=1=-t_{\infty^{(2)},2}$, $t=t_{\infty^{(1)},1}=-t_{\infty^{(2)},1}$, $t_{\infty^{(1)},0}=-t_{\infty^{(2)},0}$ and $t_{X_1^{(1)},0}=-t_{X_1^{(2)},0}$ are given below. With $K=q(p+q+t)$, the matrices are
\begin{align*}
\td{L}(\lambda,\hbar)&=
\begin{pmatrix}
-\lambda-t+\frac{K}{\lambda}&1-\frac{q}{\lambda}\\
2K+2t_{\infty^{(1)},0}
+\dfrac{K^2-t_{X_1^{(1)},0}^2}{q\lambda}
&\lambda+t-\frac{K}{\lambda}
\end{pmatrix},\\
\td{A}_{\boldsymbol\alpha_t}(\lambda,\hbar)&=
\begin{pmatrix}
-\lambda-q-t&1\\
2K+2t_{\infty^{(1)},0}&\lambda+q+t
\end{pmatrix}.
\end{align*}
\subsection{Painlev\'{e} $5$ case}\label{SectionP5}
The Painlev\'{e} $5$ case corresponds to $r_\infty=1$ with $n=2$, $r_1=1$ and $r_2=2$. It is a special case of Section \ref{Sectionrinftyequal1ngeq2}. Under the canonical choice of trivial times \eqref{TrivialTimesChoice}, the isomonodromic time is $t=\tau_{X_2,1}=2t_{X_2^{(1)},1}$. Equation \eqref{nuXsReduced} provides $\nu^{(\boldsymbol{\alpha}_t)}_{{X_2},1}=-\frac{1}{t}$. We also have
\beq
\td{P}_2(\lambda)=
-\frac{t^2}{4(\lambda-1)^4}-\frac{t_{X_2^{(1)},0}t}{(\lambda-1)^3}-\frac{t_{X_1^{(1)},0}^2}{\lambda^2}.
\eeq
Since $r_\infty=1$, the coefficients  $\nu^{(\boldsymbol{\alpha}_t)}_{\infty,-1}$, $\nu^{(\boldsymbol{\alpha}_t)}_{\infty,0}$ and $\mu_1^{(\boldsymbol{\alpha}_t)}$ are determined by \eqref{muCoeffsrinftyequal1ngeq2}:
\beq \begin{pmatrix} \frac{1}{\check{q}} \\ \frac{1}{\check{q}-1}\\ \frac{1}{(\check{q}-1)^2}\end{pmatrix} \mu_1^{(\boldsymbol{\alpha}_t)}=\begin{pmatrix}
\nu^{(\boldsymbol{\alpha}_t)}_{\infty,0}\\\nu^{(\boldsymbol{\alpha}_t)}_{\infty,0}+\nu^{(\boldsymbol{\alpha}_t)}_{\infty,-1}\\
-\nu^{(\boldsymbol{\alpha}_t)}_{{X_2},1}+\nu^{(\boldsymbol{\alpha}_t)}_{\infty,-1}\end{pmatrix}=\begin{pmatrix}\nu^{(\boldsymbol{\alpha}_t)}_{\infty,0}\\\nu^{(\boldsymbol{\alpha}_t)}_{\infty,0}+\nu^{(\boldsymbol{\alpha}_t)}_{\infty,-1}\\ \frac{1}{t}+\nu^{(\boldsymbol{\alpha}_t)}_{\infty,-1} \end{pmatrix}\eeq
which is equivalent to
\beq \nu^{(\boldsymbol{\alpha}_t)}_{\infty,-1}=\frac{\check{q}-1}{t} \,,\,\nu^{(\boldsymbol{\alpha}_t)}_{\infty,0}=\frac{(\check{q}-1)^2}{t}\,,\,  \mu_1^{(\boldsymbol{\alpha}_t)}=\frac{\check{q}(\check{q}-1)^2}{t}.\eeq
Coefficients $H_{X_1,1}$, $H_{X_2,1}$ and $H_{X_2,2}$ are determined by \eqref{Hreducedrinftyequal1ngeq2}:
\beq  \begin{pmatrix} 1& 1& 0\\ 
0&1&1\\
\frac{1}{\check{q}} & \frac{1}{\check{q}-1} & \frac{1}{(\check{q}-1)^2}\end{pmatrix} \begin{pmatrix} H_{X_1,1}\\ H_{X_2,1}\\ H_{X_2,2}\end{pmatrix}=\begin{pmatrix}\hbar \check{p}\\
\hbar \check{p}\check{q}- t_{X_1^{(1)},0}^2+t_{\infty^{(1)},0}(t_{\infty^{(1)},0}-\hbar) \\
\check{p}^2+\left(\frac{\hbar}{\check{q}}+\frac{2\hbar}{\check{q}-1}\right)\check{p}+\td{P}_2(\check{q})
\end{pmatrix}
\eeq
i.e.
\bea H_{X_1,1}&=&\check{q}(\check{q}-1)^2\check{p}^2+\hbar(\check{q}-1)^2\check{p}- \frac{t_{X_1^{(1)},0}^2}{\check{q}}\cr
&&-\frac{t^2}{4(\check{q}-1)^2}-\frac{(4t_{X_2^{(1)},0}+t)t}{4(\check{q}-1)}+t_{\infty^{(1)},0}(-t_{\infty^{(1)},0}+\hbar)\check{q}-t_{X_2^{(1)},0}t+2t_{X_1^{(1)},0}^2,\cr
H_{X_2,1}&=&-\check{q}(\check{q}-1)^2\check{p}^2-\hbar \check{q}(\check{q}-2)\check{p}+ \frac{t_{X_1^{(1)},0}^2}{\check{q}}\cr
&&+\frac{t^2}{4(\check{q}-1)^2}+\frac{(4t_{X_2^{(1)},0}+t)t}{4(\check{q}-1)}-t_{\infty^{(1)},0}(-t_{\infty^{(1)},0}+\hbar)\check{q}+t_{X_2^{(1)},0}t-2t_{X_1^{(1)},0}^2,\cr
 H_{X_2,2}&=&\check{q}(\check{q}-1)^2\check{p}^2+\hbar\check{q}(\check{q}-1)\check{p}- \frac{t_{X_1^{(1)},0}^2}{\check{q}}\cr
&&-\frac{t^2}{4(\check{q}-1)^2}-\frac{(4t_{X_2^{(1)},0}+t)t}{4(\check{q}-1)}+t_{\infty^{(1)},0}(-t_{\infty^{(1)},0}+\hbar)(\check{q}-1)-t_{X_2^{(1)},0}t+t_{X_1^{(1)},0}^2.\cr
&&
\eea
Thus, we get from Theorem \ref{HamTheorem}
\bea \hbar \partial_t \check{q}&=&\frac{2\check{q}(\check{q}-1)^2\check{p}}{t}+\hbar\frac{\check{q}(\check{q}-1)}{t},\cr
\hbar \partial_t \check{p}&=&-\frac{(3\check{q}-1)(\check{q}-1)}{t}\check{p}^2-\frac{\hbar(2\check{q}-1)}{t}\check{p}\cr 
&& -\frac{t_{X_1^{(1)},0}^2}{t \check{q}^2}-\frac{t}{2(\check{q}-1)^3}-\frac{4t_{X_2^{(1)},0}+t}{4(\check{q}-1)^2}+\frac{t_{\infty^{(1)},0}(t_{\infty^{(1)},0}-\hbar)}{t}.
\eea
In particular, $\check{q}$ satisfies the Painlev\'{e} $5$ equation
\beq  \label{P5Gambier}\hbar^2\ddot{\check{q}}=\left(\frac{1}{2\check{q}}+\frac{1}{\check{q}-1}\right)(\hbar \dot{\check{q}})^2- \frac{\hbar^2}{t}\dot{\check{q}}+\frac{(\check{q}-1)^2}{t^2}\left(\alpha \check{q}+\frac{\beta}{\check{q}}\right) +\gamma\frac{\check{q}}{t}+\delta\frac{\check{q}(\check{q}+1)}{\check{q}-1}
\eeq
with 
\beq \alpha=\frac{(2t_{\infty^{(1)},0}-\hbar)^2}{2}\,,\, \beta=-2t_{X_1^{(1)},0}^2\,,\, \gamma=-2t_{X_2^{(1)},0}\,,\, \delta=-\frac{1}{2}.\eeq

The corresponding Hamiltonian is given by Theorem \ref{HamTheoremReduced}:
\bea \text{Ham}^{(\boldsymbol{\alpha}_t)}(\check{q},\check{p})&=&- \nu^{(\boldsymbol{\alpha}_t)}_{{X_2},1}H_{X_2,2}\cr
&=&\frac{1}{t}\Big[\check{q}(\check{q}-1)^2\check{p}^2+\hbar\check{q}(\check{q}-1)\check{p}- \frac{t_{X_1^{(1)},0}^2}{\check{q}}\cr
&&-\frac{t^2}{4(\check{q}-1)^2}-\frac{(4t_{X_2^{(1)},0}+t)t}{4(\check{q}-1)}+t_{\infty^{(1)},0}(-t_{\infty^{(1)},0}+\hbar)(\check{q}-1)\cr
&&-t_{X_2^{(1)},0}t+t_{X_1^{(1)},0}^2\Big].
\eea

The associated Lax pair without apparent singularities is normalized
at infinity according to \eqref{NormalizationInfty}. Here $X_1=0$,
$X_2=1$, $t_{X_2^{(1)},1}=t/2=-t_{X_2^{(2)},1}$,
$t_{\infty^{(2)},0}=-t_{\infty^{(1)},0}$, and
$t_{X_s^{(2)},0}=-t_{X_s^{(1)},0}$ for $s=1,2$.
With these conventions the pair is

\bea \td{L}_{1,1}(\lambda,\hbar)&=&\frac{(\check{q}-1)\left(\eta_0+ \check{p}\check{q}(\check{q}-1)+t_{\infty^{(1)},0}\right)}{(\lambda-1)^2}-\frac{\check{q}\eta_0+\check{p}\check{q}(\check{q}-1)^2+t_{\infty^{(1)},0}}{\lambda-1}\cr
&&+\frac{\check{q}\eta_0+\check{q}(\check{q}-1)^2\check{p}}{\lambda},\cr
\td{L}_{2,2}(\lambda,\hbar)&=&-\td{L}_{1,1}(\lambda,\hbar),\cr 
\td{L}_{1,2}(\lambda,\hbar)&=&\frac{\lambda-\check{q}}{\lambda(\lambda-1)^2},\cr
\td{L}_{2,1}(\lambda,\hbar)&=&\frac{(\check{q}-1)\left((\eta_0+\check{p}\check{q}(\check{q}-1)+t_{\infty^{(1)},0})^2-\frac{t^2}{4(\check{q}-1)^2}\right)}{(\lambda-1)^2}\cr
&&- \frac{\check{q}\left((\eta_0+\check{p}(\check{q}-1)^2)^2-\frac{t_{X_1^{(1)},0}^2}{\check{q}^2}\right)}{\lambda-1}\cr
&&+\frac{\check{q}\left((\eta_0+\check{p}(\check{q}-1)^2)^2-\frac{t_{X_1^{(1)},0}^2}{\check{q}^2}\right)}{\lambda}
\eea
and

\bea\td{A}_{1,1}(\lambda,\hbar)&=&-\frac{(\check{q}-1)(\check{p}\check{q}(\check{q}-1)+\eta_0+t_{\infty^{(1)},0})}{t(\lambda-1)} -\frac{t_{\infty^{(1)},0}(\check{q}-1)}{t},\cr
\td{A}_{2,2}(\lambda,\hbar)&=&\frac{(\check{q}-1)(\check{p}\check{q}(\check{q}-1)+\eta_0+t_{\infty^{(1)},0})}{t(\lambda-1)} +\frac{(t_{\infty^{(1)},0}-\hbar)(\check{q}-1)}{t},\cr
\td{A}_{1,2}(\lambda,\hbar)&=&\frac{\check{q}-1}{t(\lambda-1)},\cr
\td{A}_{2,1}(\lambda,\hbar)&=&-\frac{ \frac{(\check{q}-1)}{t}\left( (\eta_0+\check{p}\check{q}(\check{q}-1)+t_{\infty^{(1)},0})^2-\frac{t^2}{4} \right) +\frac{t\check{q}(\check{q}-2)}{4(\check{q}-1)} }{\lambda-1}\cr
&&
\eea

where we have defined following \eqref{gspecial}
\beq \eta_0=\frac{1}{2t_{\infty^{(1)},0}}\left[\check{q}(\check{q}-1)^2\check{p}^2-\frac{t_{X_1^{(1)},0}^2}{\check{q}}-\frac{t^2}{4(\check{q}-1)^2}-\frac{t(t+4t_{X_2^{(1)},0}) }{4(\check{q}-1)}+t_{\infty^{(1)},0}^2(\check{q}-2) \right].
\eeq

Note that both matrices $\td{L}(\lambda,\hbar)$ and $\td{A}(\lambda,\hbar)$ may be set traceless by the additional trivial gauge transformation defined by $G=u(t) I_2$ with $u(t)=\exp\left(\int^t\frac{\check{q}(s)-1}{2s}ds \right)$ following Proposition \ref{PropositionTrace}.

\begin{remark} The Painlev\'{e} $5$ case may also be recovered by setting $n=2$, $r_1=1$, $r_2=1$ and $r_\infty=2$ that corresponds to a M\"{o}bius transformation exchanging only $\lambda=X_2$ with $\lambda=\infty$. In this case, the standard approach is to use a choice of trivial times so that it fixes the position of the two finite poles at $X_1=0$ and $X_2=1$ so that the only non-trivial isomonodromic time is $t=\td{X}_2=t_{\infty^{(1)},1}$. We use unshifted apparent-singularity coordinates $(q,p)$ on the
traceless slice, so $t_{p^{(2)},k}=-t_{p^{(1)},k}$ at every pole.
Here $T_2=t$ and $T_1=0$; the invariant coordinates of
\eqref{SymplecticChange} are $\check q=tq$ and $\check p=\frac{p}{t}$.
The raw apparent root $q$ satisfies the ODE
\bea\label{P5q}
\hbar^2\ddot{q}&=&\frac{(2q-1)}{2q(q-1)}(\hbar \dot{q})^2- \frac{\hbar^2}{t}\dot{q}+\frac{2t_{X_1^{(1)},0}^2(q-1)}{t^2q}
-\frac{(t_{X_2^{(1)},0}-t_{X_2^{(2)},0})^2q}{2t^2(q-1)}\cr
&&+2q(q-1)(2q-1)+\frac{2q(q-1)}{t}(2t_{\infty^{(1)},0}-\hbar)
\eea
Put $a=t_{\infty^{(1)},0}$, $b=t_{X_1^{(1)},0}$ and
$c=t_{X_2^{(1)},0}=-t_{X_2^{(2)},0}$.
Then $y=\frac{q}{q-1}$ satisfies the form \eqref{P5Gambier}, with the
same independent variable $t$, for
\begin{equation}\label{AlternativePVParameters}
\alpha=2c^2,\qquad\beta=-2b^2,\qquad
\gamma=-2(2a-\hbar),\qquad\delta=-2.
\end{equation}
To obtain the convention $\delta=-\tfrac12$, set $T=2t$ and
$Y(T)=\frac{q(T/2)}{q(T/2)-1}$. The parameters in the $T$ equation are
\[
\alpha=2c^2,\qquad\beta=-2b^2,\qquad
\gamma=-(2a-\hbar),\qquad\delta=-\tfrac12.
\]
 This indirect approach was used for example in \cite{JimboMiwa,MOsl2} or in articles in which infinity is always assumed to be one of the poles of highest order.
\end{remark}
\subsection{Painlev\'{e} $6$ case}\label{SectionP6}
The Painlev\'{e} $6$ case corresponds to $n=3$ with $r_1=1$, $r_2=1$, $r_3=1$ and $r_\infty=1$ and thus is a direct application of Section \ref{Sectionrinftyequal1ngeq2}. The isomonodromic time is $t=\td{X}_3$. The corresponding derivative under the canonical choice of trivial times is $\partial_{\td{X}_3}=\partial_{X_3}$ corresponding to $\alpha_{X_3}=1$ while all other coefficients vanish. Coefficients $\left(\nu^{(\boldsymbol{\alpha}_{\td{X}_3})}_{\infty,-1},\nu^{(\boldsymbol{\alpha}_{\td{X}_3})}_{\infty,0},\mu_1^{(\boldsymbol{\alpha}_{\td{X}_3})}\right)$ are determined by Proposition \ref{PropA12Form}
\beq \begin{pmatrix} \frac{1}{\check{q}}\\ \frac{1}{\check{q}-1}\\\frac{1}{\check{q}-t}\end{pmatrix}\mu_1^{(\boldsymbol{\alpha}_{\td{X}_3})}=\begin{pmatrix}\nu^{(\boldsymbol{\alpha}_{\td{X}_3})}_{\infty,0}\\ \nu^{(\boldsymbol{\alpha}_{\td{X}_3})}_{\infty,0}+ \nu^{(\boldsymbol{\alpha}_{\td{X}_3})}_{\infty,-1}\\ 1+\nu^{(\boldsymbol{\alpha}_{\td{X}_3})}_{\infty,0}+\nu^{(\boldsymbol{\alpha}_{\td{X}_3})}_{\infty,-1}t\end{pmatrix}\eeq
i.e.
\beq \nu^{(\boldsymbol{\alpha}_{\td{X}_3})}_{\infty,0}=\frac{\mu_1^{(\boldsymbol{\alpha}_{\td{X}_3})}}{\check{q}} \,,\, \nu^{(\boldsymbol{\alpha}_{\td{X}_3})}_{\infty,-1}=\frac{\mu_1^{(\boldsymbol{\alpha}_{\td{X}_3})}}{\check{q}(\check{q}-1)}\,\,,\, \left(\frac{1}{\check{q}-t}-\frac{1}{\check{q}}-\frac{t}{\check{q}(\check{q}-1)}\right) \mu_1^{(\boldsymbol{\alpha}_{\td{X}_3})}=1\eeq
so that
\beq \mu_1^{(\boldsymbol{\alpha}_{\td{X}_3})}=\frac{\check{q}(\check{q}-1)(\check{q}-t)}{t(t-1)}.\eeq

Coefficients $\left(H_{X_1,1},H_{X_2,1},H_{X_3,1}\right)$ are determined by \eqref{Hreducedrinftyequal1ngeq2}:
\beq 
\begin{pmatrix} 1&1&1\\
0&1&t\\
\frac{1}{\check{q}}& \frac{1}{\check{q}-1}& \frac{1}{\check{q}-t}
\end{pmatrix}\begin{pmatrix}H_{X_1,1}\\ H_{X_2,1}\\H_{X_3,1}\end{pmatrix}= \begin{pmatrix} \hbar \check{p}\\
 \hbar \check{p}\check{q}-t_{X_1^{(1)},0}^2-t_{X_2^{(1)},0}^2-t_{X_3^{(1)},0}^2+t_{\infty^{(1)},0}(t_{\infty^{(1)},0} -\hbar) \\
 \check{p}^2+\hbar \check{p}\left(\frac{1}{\check{q}}+\frac{1}{\check{q}-1}+\frac{1}{\check{q}-t}\right)-\frac{t_{X_1^{(1)},0}^2}{\check{q}^2}-\frac{t_{X_2^{(1)},0}^2}{(\check{q}-1)^2}-\frac{t_{X_3^{(1)},0}^2}{(\check{q}-t)^2}
\end{pmatrix}
\eeq
so that we get from Theorem \ref{HamTheorem}
\footnotesize{\bea
\hbar \partial_t \check{q}&=& \frac{2\check{q}(\check{q}-1)(\check{q}-t)}{t(t-1)}\check{p} +\hbar \frac{(\check{q}-1)\check{q}}{t(t-1)},\cr
\hbar \partial_t \check{p}&=&\mu_1^{(\boldsymbol{\alpha}_{\td{X}_3})}\Big[-\td{P}_2'(\check{q})+\check{p}\left(\frac{\hbar}{\check{q}^2}+\frac{\hbar}{(\check{q}-1)^2}+\frac{\hbar}{(\check{q}-t)^2}\right)-\frac{H_{X_1,1}}{\check{q}^2}-\frac{H_{X_2,1}}{(\check{q}-1)^2}-\frac{H_{X_3,1}}{(\check{q}-t)^2} \Big]+\hbar\nu^{(\boldsymbol{\alpha}_{\td{X}_3})}_{\infty,-1} \check{p} \cr
&=&-\frac{3\check{q}^2-2t\check{q}-2\check{q}+t}{t(t-1)}\check{p}^2-\hbar\frac{(2\check{q}-1)\check{p}}{t(t-1)}-\frac{t_{X_1^{(1)},0}^2}{(t-1)\check{q}^2}+\frac{t_{X_2^{(1)},0}^2}{t(\check{q}-1)^2}-\frac{t_{X_3^{(1)},0}^2}{(\check{q}-t)^2}+\frac{t_{\infty^{(1)},0}(t_{\infty^{(1)},0}-\hbar) }{t(t-1)}.\cr
&&
\eea} 
\normalsize{Finally}, we obtain that $\check{q}$ satisfies the Painlev\'{e} $6$ equation 
\bea\hbar^2\ddot{\check{q}}&=&\frac{1}{2}\left(\frac{1}{\check{q}}+\frac{1}{\check{q}-1}+\frac{1}{\check{q}-t}\right) (\hbar\dot{\check{q}})^2-\hbar^2 \dot{\check{q}}\left(\frac{1}{t}+\frac{1}{t-1}+\frac{1}{\check{q}-t}\right)\cr
&&+\frac{\check{q}(\check{q}-1)(\check{q}-t)}{t^2(t-1)^2}\left(\alpha+\beta\frac{t}{\check{q}^2}+\gamma \frac{t-1}{(\check{q}-1)^2}+\delta \frac{t(t-1)}{(\check{q}-t)^2}\right)\cr
&& 
\eea
with $\alpha=\frac{(2t_{\infty^{(1)},0}-\hbar)^2}{2}$, $\beta=-2t_{X_1^{(1)},0}^2$, $\gamma=2t_{X_2^{(1)},0}^2 $ and $\delta=-\left(2t_{X_3^{(1)},0}^2-\frac{\hbar^2}{2}\right)$.

A Hamiltonian generating this system is
\begin{align*}
\operatorname{Ham}^{(\boldsymbol\alpha_t)}
&=\frac{\check q(\check q-1)(\check q-t)}{t(t-1)}\check p^2
+\frac{\hbar\check q(\check q-1)}{t(t-1)}\check p
-\frac{t_{X_1^{(1)},0}^2}{(t-1)\check q}
+\frac{t_{X_2^{(1)},0}^2}{t(\check q-1)}\\
&\quad-\frac{t_{X_3^{(1)},0}^2}{\check q-t}
-\frac{t_{\infty^{(1)},0}(t_{\infty^{(1)},0}-\hbar)\check q}{t(t-1)}.
\end{align*}
The accessory coefficient is
$H_{X_3,1}=\operatorname{Ham}^{(\boldsymbol\alpha_t)}+C(t)$, where
\[
C(t)=\frac{
t_{\infty^{(1)},0}(t_{\infty^{(1)},0}-\hbar)t
+t_{X_1^{(1)},0}^2-t_{X_2^{(1)},0}^2
+t_{X_3^{(1)},0}^2(1-2t)}{t(t-1)}.
\]
This time-only difference does not affect the Hamilton equations.

The associated Lax pair in the gauge without apparent singularities, normalized at infinity according to \eqref{NormalizationInfty} and under the choice of trivial times \eqref{TrivialTimesChoice} (i.e. $X_1=0$, $X_2=1$, $X_3=t$ and $t_{\infty^{(2)},0}=-t_{\infty^{(1)},0}$, $t_{X_j^{(2)},0}=-t_{X_j^{(1)},0}$ for all $1\leq j\leq 3$ ) is

\beq \td{L}(\lambda,\hbar)= \frac{\td{L}_0}{\lambda}+ \frac{\td{L}_1}{\lambda-1}+\frac{\td{L}_t}{\lambda-t}\,\,,\,\,\td{A}_{\boldsymbol{\alpha}_t}(\lambda,\hbar)= \frac{\td{A}_t}{\lambda-t} +\td{A}_\infty
\eeq
with
\bea \td{L}_0&=&\begin{pmatrix}\frac{\check{q}}{t}(\eta_0 +(\check{q}-1)(\check{q}-t)\check{p})& -\frac{\check{q}}{t}\\
\frac{\check{q}}{t}\left( (\eta_0+\check{p}(\check{q}-1)(\check{q}-t))^2-\frac{t_{X_1^{(1)},0}^2 t^2}{\check{q}^2}\right)& -\frac{\check{q}}{t}(\eta_0 +(\check{q}-1)(\check{q}-t)\check{p})  \end{pmatrix},\cr
\td{L}_1&=&\begin{pmatrix}-\frac{(\check{q}-1)}{(t-1)}(\eta_0+\check{q}(\check{q}-t)\check{p}+t_{\infty^{(1)},0})&\frac{\check{q}-1}{t-1}\\
-\frac{\check{q}-1}{t-1}\left((\eta_0+\check{q}(\check{q}-t)\check{p}+t_{\infty^{(1)},0})^2-\frac{t_{X_2^{(1)},0}^2(t-1)^2}{(\check{q}-1)^2}\right)& \frac{(\check{q}-1)}{(t-1)}(\eta_0+\check{q}(\check{q}-t)\check{p}+t_{\infty^{(1)},0})\end{pmatrix},\cr
\td{L}_t&=&\begin{pmatrix}\frac{(\check{q}-t)}{t(t-1)}(\eta_0+\check{q}(\check{q}-1)\check{p}+t_{\infty^{(1)},0}t)& -\frac{\check{q}-t}{t(t-1)}\\
\frac{\check{q}-t}{t(t-1)}\left((\eta_0+\check{q}(\check{q}-1)\check{p}+t_{\infty^{(1)},0}t )^2-\frac{ t_{X_3^{(1)},0}^2 t^2(t-1)^2 }{(\check{q}-t)^2}\right)&-\frac{(\check{q}-t)}{t(t-1)}(\eta_0+\check{q}(\check{q}-1)\check{p}+t_{\infty^{(1)},0}t)\end{pmatrix}\cr
&&
\eea
satisfying
\beq \td{L}_\infty:=\td{L}_0+\td{L}_1+\td{L}_t=\begin{pmatrix} -t_{\infty^{(1)},0}&0\\ 0&t_{\infty^{(1)},0}\end{pmatrix}\eeq
and
\bea \td{A}_t&=&\begin{pmatrix}-\frac{(\eta_0+\check{p}\check{q}(\check{q}-1)+t_{\infty^{(1)},0}t)(\check{q}-t)}{t(t-1)}& \frac{\check{q}-t}{t(t-1)}\\
-\frac{(\check{q}-t)}{t(t-1)}\left((\eta_0+\check{q}(\check{q}-1)\check{p}+t_{\infty^{(1)},0}t)^2-\frac{t_{X_3^{(1)},0}^2t^2(t-1)^2}{(\check{q}-t)^2}\right)&\frac{(\eta_0+\check{p}\check{q}(\check{q}-1)+t_{\infty^{(1)},0}t)(\check{q}-t)}{t(t-1)} 
\end{pmatrix},\cr
\td{A}_\infty&=&\begin{pmatrix}-\frac{t_{\infty^{(1)},0}(\check{q}-t)}{t(t-1)}&0\\0&\frac{(t_{\infty^{(1)},0}-\hbar)(\check{q}-t)}{t(t-1)}  \end{pmatrix}\cr
&&
\eea
where we have defined following \eqref{gspecial}:
\small{\beq \eta_0=\frac{1}{2t_{\infty^{(1)},0}}\left[\check{q}(\check{q}-1)(\check{q}-t)\check{p}^2-\frac{t_{X_1^{(1)},0}^2t}{\check{q}}+ \frac{t_{X_2^{(1)},0}^2(t-1)}{\check{q}-1}-\frac{t_{X_3^{(1)},0}^2t(t-1)}{\check{q}-t}+t_{\infty^{(1)},0}^2(\check{q}-1-t)\right].
\eeq}

\normalsize{Note} that both matrices $\td{L}(\lambda,\hbar)$ and $\td{A}_{\boldsymbol{\alpha}_t}(\lambda,\hbar)$ may be set traceless by the additional trivial gauge transformation defined by $G=u(t) I_2$ with $u(t)=\exp\left(\int^t\frac{(\check{q}(s)-s)}{2s(s-1)}ds \right)$ following Proposition \ref{PropositionTrace}.

\subsection{Fuchsian systems}\label{SectionFuchsian}
Fuchsian systems correspond to the case $r_\infty=1$ and $n\geq 3$ with $r_i=1$ for all $i\in \llbracket 1,n\rrbracket$. They provide a generic spectral curve of genus $g=n-2$. It is a special case of Section \ref{Sectionrinftyequal1ngeq2}. The canonical choice of trivial times \eqref{TrivialTimesChoice} implies that $X_1=0$, $X_2=1$ and that the isomonodromic times are $\td{X}_s=X_s$ for $s\geq 3$. 
We get from \eqref{tdP2reducedrinftyequal1ngeq2}:
\beq \td{P}_2(\lambda)
=-\frac{t_{X_1^{(1)},0}^2}{\lambda^2}-\frac{t_{X_2^{(1)},0}^2}{(\lambda-1)^2}-\sum_{s=3}^n \frac{t_{X_s^{(1)},0}^2}{(\lambda-\td{X}_s)^2}.
 \eeq
Coefficients $\left(H_{X_s,1}\right)_{1\leq s\leq n}$ are given by \eqref{Hreducedrinftyequal1ngeq2}

\tiny{\bea\label{FuchsianHam} &&\begin{pmatrix}1&1&1&\dots &1 \\
0&1&\td{X}_3&\dots&\td{X}_n\\
\frac{1}{\check{q_1}}&\frac{1}{\check{q_1}-1}&\frac{1}{\check{q_1}-\td{X}_3}& \dots & \frac{1}{\check{q_1}-\td{X}_n}\\
\frac{1}{\check{q_2}}& \frac{1}{\check{q_2}-1}&\frac{1}{\check{q_2}-\td{X}_3} &\dots & \frac{1}{\check{q_2}-\td{X}_n}\\
\vdots& \vdots& & \vdots& \vdots\\
\frac{1}{\check{q_g}}&\frac{1}{\check{q_g}-1}&\frac{1}{\check{q_g}-\td{X}_3}&\dots&\frac{1}{\check{q_g}-\td{X}_n} 
\end{pmatrix}\begin{pmatrix}H_{X_1,1}\\ \vdots \\ \\ \vdots\\ H_{X_n,1}\end{pmatrix}
=\cr
&&\begin{pmatrix}\hbar \underset{j=1}{\overset{g}{\sum}} \check{p}_j\\
\hbar \underset{j=1}{\overset{g}{\sum}} \check{q}_j \check{p}_j -\underset{s=1}{\overset{n}{\sum}} t_{X_s^{(1)},0}^2 +t_{\infty^{(1)},0}(t_{\infty^{(1)},0}-\hbar)\\
 \check{p}_1^2 + \hbar\check{p}_1\left(\frac{1 }{\check{q}_1}+\frac{1 }{\check{q}_1-1}+ \underset{s=3}{\overset{n}{\sum}} \frac{1 }{\check{q}_1-\td{X}_s}\right)-\frac{t_{X_1^{(1)},0}^2}{\check{q}_1^2}-\frac{t_{X_2^{(1)},0}^2}{(\check{q}_1-1)^2}-\underset{s=3}{\overset{n}{\sum}} \frac{t_{X_s^{(1)},0}^2}{(\check{q}_1-\td{X}_s)^2}-\hbar \underset{i\neq 1}{\sum}\frac{\check{p}_1-\check{p}_i}{\check{q}_1-\check{q}_i}\\
\vdots\\
 \check{p}_g^2 + \hbar\check{p}_g\left(\frac{1 }{\check{q}_g}+\frac{1 }{\check{q}_g-1}+ \underset{s=3}{\overset{n}{\sum}} \frac{1 }{\check{q}_g-\td{X}_s}\right)-\frac{t_{X_1^{(1)},0}^2}{\check{q}_g^2}-\frac{t_{X_2^{(1)},0}^2}{(\check{q}_g-1)^2}-\underset{s=3}{\overset{n}{\sum}} \frac{t_{X_s^{(1)},0}^2}{(\check{q}_g-\td{X}_s)^2}-\hbar \underset{i\neq g}{\sum}\frac{\check{p}_g-\check{p}_i}{\check{q}_g-\check{q}_i}
\end{pmatrix}.
\eea}

\normalsize{The} Hamiltonian corresponding to a deformation with respect to $\td{X}_s$ is 
\beq \label{Simplification}\text{Ham}^{(\boldsymbol{\alpha}_{\td{X}_s})}(\mathbf{\check{q}},\mathbf{\check{p}})=H_{X_s,1} 
\eeq
that is obtained from inverting \eqref{FuchsianHam}.

\normalsize{For} any $s\in \llbracket 3,n\rrbracket$, the deformation with respect to $\td{X}_s=X_s$ provides the relation
\beq \begin{pmatrix}\frac{1}{\check{q}_1}&\dots  & \dots& \frac{1}{\check{q}_g} \\
\frac{1}{\check{q}_1-1}&\dots& \dots& \frac{1}{\check{q}_g-1} \\
\frac{1}{\check{q}_1-\td{X}_3}&\dots &\dots& \frac{1}{\check{q}_g-\td{X}_3} \\
\vdots& &&\vdots\\
\frac{1}{\check{q}_1-\td{X}_n}&\dots &\dots& \frac{1}{\check{q}_g-\td{X}_n}
\end{pmatrix}
\begin{pmatrix} \mu_{1}^{(\boldsymbol{\alpha}_{\td{X}_s})}\\ \vdots\\\vdots\\ \mu_{g}^{(\boldsymbol{\alpha}_{\td{X}_s})} \end{pmatrix}=\begin{pmatrix} \nu_{\infty,0}^{(\boldsymbol{\alpha}_{\td{X}_s})}\\
\nu_{\infty,0}^{(\boldsymbol{\alpha}_{\td{X}_s})}+  \nu_{\infty,-1}^{(\boldsymbol{\alpha}_{\td{X}_s})}\\
\delta_{s,3}+\nu_{\infty,0}^{(\boldsymbol{\alpha}_{\td{X}_s})}+ \td{X}_3 \nu_{\infty,-1}^{(\boldsymbol{\alpha}_{\td{X}_s})}\\
 \vdots\\ \delta_{s,n}+\nu_{\infty,0}^{(\boldsymbol{\alpha}_{\td{X}_s})}+ \td{X}_{n} \nu_{\infty,-1}^{(\boldsymbol{\alpha}_{\td{X}_s})} \end{pmatrix}
\eeq
that determines $\left(\nu^{(\boldsymbol{\alpha}_{\td{X}_s})}_{\infty,-1}\right)$, $\nu^{(\boldsymbol{\alpha}_{\td{X}_s})}_{\infty,0}$ and all $\left(\mu_{j}^{(\boldsymbol{\alpha}_{\td{X}_s})}\right)_{1\leq j\leq g}$. The evolution equations are given by
\bea 
\hbar \partial_{\td{X}_s}\check{q}_j&=&2\mu^{(\boldsymbol{\alpha}_{\td{X}_s})}_j\left(\check{p}_j +\frac{1}{2}\left(\frac{ \hbar }{\check{q}_j}+\frac{ \hbar }{\check{q}_j-1}+  \sum_{s=3}^n \frac{ \hbar}{\check{q}_j-\td{X}_s}\right)\right)-\hbar \nu^{(\boldsymbol{\alpha}_{\td{X}_s})}_{\infty,0} -\hbar \nu^{(\boldsymbol{\alpha}_{\td{X}_s})}_{\infty,-1} \check{q}_j\cr
&&-\hbar \sum_{i\neq j}\frac{\mu^{(\boldsymbol{\alpha}_{\td{X}_s})}_j+\mu^{(\boldsymbol{\alpha}_{\td{X}_s})}_i}{\check{q}_j-\check{q}_i},\cr
\hbar \partial_{\td{X}_s}\check{p}_j&=&\hbar \sum_{i\neq j}\frac{(\mu^{(\boldsymbol{\alpha}_{\td{X}_s})}_i+\mu^{(\boldsymbol{\alpha}_{\td{X}_s})}_j)(\check{p}_i-\check{p}_j)}{(\check{q}_j-\check{q}_i)^2} +\mu^{(\boldsymbol{\alpha}_{\td{X}_s})}_j\Big[ \check{p}_j\left(\frac{\hbar }{\check{q}_j^2}+\frac{\hbar}{(\check{q}_j-1)^2}+ \sum_{s=3}^n \frac{\hbar }{(\check{q}_j-\td{X}_s)^2}\right) \cr
&&-\frac{2t_{X_1^{(1)},0}^2}{\check{q}_j^3}-\frac{2t_{X_2^{(1)},0}^2}{(\check{q}_j-1)^3}-2\sum_{s=3}^n\frac{t_{X_s^{(1)},0}^2}{(\check{q}_j-\td{X}_s)^3}-\sum_{s=1}^nH_{X_s,1}(\check{q}_j-\td{X}_s)^{-2} \Big]+\hbar \nu^{(\boldsymbol{\alpha}_{\td{X}_s})}_{\infty,-1}\check{p}_j.\cr
&&
\eea

The Lax pair in the companion-like gauge is given by
\footnotesize{\beq \label{OkamotoSp}L(\lambda,\hbar)=\begin{pmatrix}0&1\\ 
\frac{t_{X_1^{(1)},0}^2}{\lambda^2}+\frac{t_{X_2^{(1)},0}^2}{(\lambda-1)^2}+\underset{s=3}{\overset{n}{\sum}} \frac{t_{X_s^{(1)},0}^2}{(\lambda-\td{X}_s)^2}+\frac{H_{X_1,1}}{\lambda}+\frac{H_{X_2,1}}{\lambda-1}+\underset{s=3}{\overset{n}{\sum}}\frac{H_{X_s,1}}{\lambda-\td{X}_s}-\underset{j=1}{\overset{g}{\sum}}\frac{\hbar \check{p}_j}{\lambda-\check{q}_j}
&\underset{j=1}{\overset{g}{\sum}}\frac{\hbar}{\lambda-\check{q}_j}-\frac{\hbar}{\lambda}-\frac{\hbar}{\lambda-1}-\underset{s=3}{\overset{n}{\sum}}\frac{\hbar}{\lambda-\td{X}_s}  \end{pmatrix}
\eeq}
\normalsize{and}
\bea \left[A_{\boldsymbol{\alpha}_{\td{X}_s}}\right]_{1,1}(\lambda,\hbar)&=&-\sum_{j=1}^g \frac{\mu_j^{(\boldsymbol{\alpha}_{\td{X}_s})} \check{p}_j}{\lambda-\check{q}_j},\cr
\left[A_{\boldsymbol{\alpha}_{\td{X}_s}}\right]_{1,2}(\lambda,\hbar)&=& \nu_{\infty,-1}^{(\boldsymbol{\alpha}_{\td{X}_s})}\lambda+\nu_{\infty,0}^{(\boldsymbol{\alpha}_{\td{X}_s})}+\underset{j=1}{\overset{g}{\sum}}\frac{\mu_j^{(\boldsymbol{\alpha}_{\td{X}_s})}}{\lambda-\check{q}_j}.
\eea
Using \eqref{TrivialEntriesA}, one may obtain the other two entries of $A_{\boldsymbol{\alpha}_{\td{X}_s}}(\lambda,\hbar)$ and use \eqref{CheckLEquations} and \eqref{TdLEquations} in order to obtain $\td{L}(\lambda,\hbar)$ and $\td{A}_{\boldsymbol{\alpha}_{\td{X}_s}}(\lambda,\hbar)$.

Note that our expression for the companion-like Lax matrix \eqref{OkamotoSp} recovers the Garnier formulation of Fuchsian systems described by Okamoto \cite{Okamoto1986Iso}. Moreover, the simple formula for the Hamiltonian \eqref{Simplification} implies that 
\bea \text{Ham}^{(\boldsymbol{\alpha}_{\td{X}_s})}(\mathbf{\check{q}},\mathbf{\check{p}})&=&H_{X_s,1} =\Res_{\lambda\to \td{X}_s} \left[L(\lambda,\hbar)\right]_{2,1} \,\,,\,\, \forall\, s\in \llbracket3, n\rrbracket\cr
p_j&=&-\frac{1}{\hbar}\Res_{\lambda\to q_j} \left[L(\lambda,\hbar)\right]_{2,1} \,,\, \forall \, j\in \llbracket1,g\rrbracket.
\eea
recovering the main theorem of Okamoto \cite{Okamoto1986Iso}. In order to obtain the Schlesinger formulation of Fuchsian systems, one would need the expression of the Lax matrix in the geometric gauge rather than the oper gauge. This goes beyond the scope of the present paper but has been done in \cite{MarchalAlameddineIsospectralIsomono2023}.

\subsection{The genus-two system $(P_{\mathrm{II}-2})_2$}\label{SectionSecondElementP2}
This genus-two example has $n=0$ and $r_\infty=5$.
In Chiba's terminology it is the second member of the second
Painlev\'e-II hierarchy, $(P_{\mathrm{II}-2})_2$ (Section 4.1(i) of
\cite{Chiba}). This designation distinguishes it from
$(P_{\mathrm{II}-1})_2$; the phrase ``Painlev\'e-II hierarchy'' below
is used in this specified sense for the present example.
The canonical choice \eqref{TrivialTimesChoice} is
\bea t_{\infty^{(1)},4}&=&1=-t_{\infty^{(2)},4}\,,\,t_{\infty^{(1)},3}=0=-t_{\infty^{(2)},3},\cr
t_{\infty^{(1)},2}&=&\frac{1}{2}\tau_{\infty,2}=-t_{\infty^{(2)},2}  \,,\,t_{\infty^{(1)},1}=\frac{1}{2}\tau_{\infty,1}=-t_{\infty^{(2)},1},\cr
t_{\infty^{(2)},0}&=&-t_{\infty^{(1)},0}.
\eea
We get from \eqref{tdP2reducedPainleve2}:
\beq \td{P}_2(\lambda)=-\lambda^6 -\tau_{\infty,2}\lambda^4-\tau_{\infty,1}\lambda^3-\left(2t_{\infty^{(1)},0}+\frac{1}{4}\tau_{\infty,2}^2\right)\lambda^2.\eeq
Coefficients $\left(H_{\infty,0},H_{\infty,1}\right)$ are determined by \eqref{HCoeffPainleve2Hierarchy}:
\beq \begin{pmatrix} 1& \check{q}_1\\ 1&\check{q}_2\end{pmatrix}\begin{pmatrix} H_{\infty,0}\\ H_{\infty,1}\end{pmatrix}=\begin{pmatrix}\check{p}_1^2+\td{P_2}(\check{q}_1)-\hbar\frac{\check{p}_2-\check{p}_1}{\check{q}_2-\check{q}_1}+\hbar \check{q}_1^2\\
\check{p}_2^2+\td{P_2}(\check{q}_2)-\hbar\frac{\check{p}_1-\check{p}_2}{\check{q}_1-\check{q}_2}+\hbar \check{q}_2^2
\end{pmatrix}
\eeq
i.e.
\small{
\bea H_{\infty,0}&=&\frac{\check{q}_1\check{p}_2^2-\check{q}_2\check{p}_1^2}{\check{q}_1-\check{q}_2}-\hbar \frac{\check{p}_1-\check{p}_2}{\check{q}_1-\check{q}_2} \cr
&&+\check{q}_1\check{q}_2\left(\check{q}_1^4+\check{q}_2^4+\check{q}_1^3\check{q}_2+\check{q}_2^3\check{q}_1+\check{q}_1^2\check{q}_2^2+(\check{q}_1^2+\check{q}_1\check{q}_2+\check{q}_2^2)\tau_{\infty,2}+(\check{q}_1+\check{q}_2)\tau_{\infty,1}+\frac{\tau_{\infty,2}^2}{4}+2t_{\infty^{(1)},0}-\hbar\right),\cr
H_{\infty,1}&=&\frac{\check{p}_1^2-\check{p}_2^2}{\check{q}_1-\check{q}_2}-(\check{q}_1+\check{q}_2)\left((\check{q}_1^2+\check{q}_2^2+\frac{1}{2}\tau_{\infty,2})^2-\check{q}_1^2\check{q}_2^2+2t_{\infty^{(1)},0}-\hbar\right)-(\check{q}_1^2+\check{q}_1\check{q}_2+\check{q}_2^2)\tau_{\infty,1}.\cr
&&
\eea}
\normalsize{Coefficients} $\left(\nu^{(\boldsymbol{\alpha}_{{\infty,1}})}_{\infty,1},\nu^{(\boldsymbol{\alpha}_{{\infty,1}})}_{\infty,2},\nu^{(\boldsymbol{\alpha}_{{\infty,2}})}_{\infty,1},\nu^{(\boldsymbol{\alpha}_{{\infty,2}})}_{\infty,2}\right)$ are determined by Proposition \ref{PropA12Form} that trivially gives
\beq \nu^{(\boldsymbol{\alpha}_{{\infty,1}})}_{\infty,1}=0\,\,\,,\,\,\,\nu^{(\boldsymbol{\alpha}_{{\infty,1}})}_{\infty,2}=\frac{1}{2}
\,\,\,,\,\,\,\nu^{(\boldsymbol{\alpha}_{{\infty,2}})}_{\infty,1}=\frac{1}{4}\,\,\,,\,\,\,\nu^{(\boldsymbol{\alpha}_{{\infty,2}})}_{\infty,2}=0.
\eeq
According to \eqref{HamPainleve2Hierarchy}, it provides the Hamiltonians:
\small{\bea\text{Ham}^{(\boldsymbol{\alpha}_{{\infty,1}})}(\check{q}_1,\check{q}_2,\check{p}_1,\check{p}_2)&=&\frac{1}{2}H_{\infty,1}  \cr
&=&\frac{\check{p}_1^2-\check{p}_2^2}{2(\check{q}_1-\check{q}_2)}-\frac{(\check{q}_1+\check{q}_2)( (\check{q}_1^2+\check{q}_2^2)^2-\check{q}_1^2\check{q}_2^2)}{2}-\frac{\check{q}_1^2+\check{q}_1\check{q}_2+\check{q}_2^2}{2}\tau_{\infty,1} \cr
&&-\frac{(\check{q}_1+\check{q}_2)(\check{q}_1^2+\check{q}_2^2)}{2}\tau_{\infty,2}-\frac{\check{q}_1+\check{q}_2}{8}\tau_{\infty,2}^2 -\frac{(\check{q}_1+\check{q}_2)(2t_{\infty^{(1)},0}-\hbar)}{2}\cr
\text{Ham}^{(\boldsymbol{\alpha}_{{\infty,2}})}(\check{q}_1,\check{q}_2,\check{p}_1,\check{p}_2)&=&\frac{1}{4}H_{\infty,0}\cr
&=&\frac{\check{q}_1\check{p}_2^2-\check{q}_2\check{p}_1^2-\hbar(\check{p}_1-\check{p}_2)}{4(\check{q}_1-\check{q}_2)}+\frac{(\check{q}_1^4+\check{q}_1^3\check{q}_2+\check{q}_1^2\check{q}_2^2+\check{q}_1\check{q}_2^3+\check{q}_2^4)\check{q}_1\check{q}_2}{4}\cr
&&+\frac{(\check{q}_1+\check{q}_2)\check{q}_1\check{q}_2}{4}\tau_{\infty,1}+ \frac{(\check{q}_1^2+\check{q}_1\check{q}_2+\check{q}_2^2)\check{q}_1\check{q}_2}{4}\tau_{\infty,2}+\frac{\check{q}_1\check{q}_2}{16}\tau_{\infty,2}^2+\frac{(2t_{\infty^{(1)},0}-\hbar)\check{q}_1\check{q}_2}{4}.\cr
&&
\eea}
\normalsize{}
\begin{remark}[Flat representatives in genus two]
Put $s=\tau_{\infty,1}$ and $u=\tau_{\infty,2}$, and denote the two
displayed Hamiltonians by
$H_s=\operatorname{Ham}^{(\boldsymbol\alpha_{\infty,1})}$ and
$H_u=\operatorname{Ham}^{(\boldsymbol\alpha_{\infty,2})}$.
With the bracket of Remark \ref{RemarkFlatHamiltonianNormalization},
direct differentiation gives
\begin{equation}\label{GenusTwoHamiltonianCurvature}
\hbar(\partial_sH_u-\partial_uH_s)-\{H_s,H_u\}
=-\frac{\hbar s}{8}.
\end{equation}
Thus the Hamiltonian vector fields are compatible, but the
fundamental two-form requires time-only normalization. One flat pair is
\begin{equation}\label{GenusTwoFlatHamiltonians}
H_s^{\mathrm{flat}}=H_s,\qquad H_u^{\mathrm{flat}}=H_u+\frac{s^2}{16}.
\end{equation}
This changes neither flow nor either Lax equation and leaves the
accessory coefficients $H_{\infty,0},H_{\infty,1}$ unchanged.
The polynomial Hamiltonians obtained below by a time-dependent
canonical transformation are already flat in their displayed
normalization; the term $s^2/16$ is not to be added to that pair.
\end{remark}
\normalsize{Finally} coefficients $\left(\mu^{(\boldsymbol{\alpha}_{{\infty,1}})}_1,\mu^{(\boldsymbol{\alpha}_{{\infty,1}})}_2,\mu^{(\boldsymbol{\alpha}_{{\infty,2}})}_1,\mu^{(\boldsymbol{\alpha}_{{\infty,2}})}_2\right)$ are determined by
\beq \begin{pmatrix} 1&1\\ \check{q}_1& \check{q}_2\end{pmatrix} \begin{pmatrix} \mu^{(\boldsymbol{\alpha}_{{\infty,1}})}_1\\\mu^{(\boldsymbol{\alpha}_{{\infty,1}})}_2\end{pmatrix}=\begin{pmatrix}\nu^{(\boldsymbol{\alpha}_{{\infty,1}})}_{\infty,1}\\\nu^{(\boldsymbol{\alpha}_{{\infty,1}})}_{\infty,2} \end{pmatrix}\,\,,\,\,  \begin{pmatrix} 1&1\\ \check{q}_1& \check{q}_2\end{pmatrix} \begin{pmatrix} \mu^{(\boldsymbol{\alpha}_{{\infty,2}})}_1\\\mu^{(\boldsymbol{\alpha}_{{\infty,2}})}_2\end{pmatrix}=\begin{pmatrix}\nu^{(\boldsymbol{\alpha}_{{\infty,2}})}_{\infty,1}\\\nu^{(\boldsymbol{\alpha}_{{\infty,2}})}_{\infty,2} \end{pmatrix}
\eeq
i.e.
\beq \mu^{(\boldsymbol{\alpha}_{{\infty,1}})}_1=\frac{1}{2(\check{q}_1-\check{q}_2)}\,\,\,,\,\,\, \mu^{(\boldsymbol{\alpha}_{{\infty,1}})}_2=-\frac{1}{2(\check{q}_1-\check{q}_2)}\,\,\,,\,\,\,\mu^{(\boldsymbol{\alpha}_{{\infty,2}})}_1=-\frac{\check{q}_2}{4(\check{q}_1-\check{q}_2)}\,\,\,,\,\,\,\mu^{(\boldsymbol{\alpha}_{{\infty,2}})}_2=\frac{\check{q}_1}{4(\check{q}_1-\check{q}_2)}.
\eeq

The corresponding evolution equations are
\bea \hbar \partial_{\tau_{\infty,1}} \check{q}_1&=&
\frac{\check{p}_1}{\check{q}_1-\check{q}_2},\cr
\hbar \partial_{\tau_{\infty,1}} \check{q}_2&=&
-\frac{\check{p}_2}{\check{q}_1-\check{q}_2},\cr
\hbar \partial_{\tau_{\infty,1}} \check{p}_1
&=&\frac{\check{p}_1^2-\check{p}_2^2}{2(\check{q}_1-\check{q}_2)^2}+\frac{1}{2}\left(5\check{q}_1^4+4\check{q}_1^3\check{q}_2+3\check{q}_1^2\check{q}_2^2+2\check{q}_1\check{q}_2^3+\check{q}_2^4\right)+\left(\check{q}_1+\frac{\check{q}_2}{2}\right)\tau_{\infty,1}\cr
&&+\frac{1}{2}\left(3\check{q}_1^2+2\check{q}_1\check{q}_2+\check{q}_2^2\right)\tau_{\infty,2}+\frac{1}{8}\tau_{\infty,2}^2+t_{\infty^{(1)},0}-\frac{\hbar}{2},\cr 
\hbar \partial_{\tau_{\infty,1}} \check{p}_2
&=&-\frac{\check{p}_1^2-\check{p}_2^2}{2(\check{q}_1-\check{q}_2)^2}+\frac{1}{2}\left(5\check{q}_2^4+4\check{q}_2^3\check{q}_1+3\check{q}_2^2\check{q}_1^2+2\check{q}_2\check{q}_1^3+\check{q}_1^4\right)+\left(\check{q}_2+\frac{\check{q}_1}{2}\right)\tau_{\infty,1}\cr
&&+\frac{1}{2}\left(3\check{q}_2^2+2\check{q}_2\check{q}_1+\check{q}_1^2\right)\tau_{\infty,2}+\frac{1}{8}\tau_{\infty,2}^2+t_{\infty^{(1)},0}-\frac{\hbar}{2}
\eea
and
\bea \hbar \partial_{\tau_{\infty,2}} \check{q}_1&=&
-\frac{\check{p}_1\check{q}_2}{2(\check{q}_1-\check{q}_2)}-\frac{\hbar}{4(\check{q}_1-\check{q}_2)},\cr
\hbar \partial_{\tau_{\infty,2}} \check{q}_2&=&
\frac{\check{p}_2\check{q}_1}{2(\check{q}_1-\check{q}_2)}+\frac{\hbar}{4(\check{q}_1-\check{q}_2)},\cr
\hbar \partial_{\tau_{\infty,2}} \check{p}_1
&=&-\frac{\check{q}_2(\check{p}_1^2-\check{p}_2^2)}{4(\check{q}_1-\check{q}_2)^2}-\frac{\hbar(\check{p}_1-\check{p}_2)}{4(\check{q}_1-\check{q}_2)^2}-\frac{\check{q}_2(5\check{q}_1^4+4\check{q}_1^3\check{q}_2+3\check{q}_1^2\check{q}_2^2+2\check{q}_1\check{q}_2^3+\check{q}_2^4)}{4}\cr
&&-\frac{\check{q}_2(2\check{q}_1+\check{q}_2)}{4}\tau_{\infty,1}-\frac{\check{q}_2(3\check{q}_1^2+2\check{q}_1\check{q}_2+\check{q}_2^2)}{4}\tau_{\infty,2}-\frac{\check{q}_2}{16}\tau_{\infty,2}^2-\frac{\check{q}_2(2t_{\infty^{(1)},0}-\hbar)}{4},\cr
\hbar \partial_{\tau_{\infty,2}} \check{p}_2
&=&\frac{\check{q}_1(\check{p}_1^2-\check{p}_2^2)}{4(\check{q}_1-\check{q}_2)^2}+\frac{\hbar(\check{p}_1-\check{p}_2)}{4(\check{q}_1-\check{q}_2)^2}-\frac{\check{q}_1(5\check{q}_2^4+4\check{q}_2^3\check{q}_1+3\check{q}_2^2\check{q}_1^2+2\check{q}_2\check{q}_1^3+\check{q}_1^4)}{4}\cr
&&-\frac{\check{q}_1(2\check{q}_2+\check{q}_1)}{4}\tau_{\infty,1}-\frac{\check{q}_1(3\check{q}_2^2+2\check{q}_2\check{q}_1+\check{q}_1^2)}{4}\tau_{\infty,2}-\frac{\check{q}_1}{16}\tau_{\infty,2}^2-\frac{\check{q}_1(2t_{\infty^{(1)},0}-\hbar)}{4}.\cr
&&
\eea
Moreover, the Lax matrices normalized at infinity according to \eqref{NormalizationInfty} are
\small{\bea \td{L}_{1,1}&=&-\lambda^3+ \left(\frac{\check{p}_1-\check{p}_2}{\check{q}_1-\check{q}_2}+\check{q}_1^2+\check{q}_2^2+\check{q}_1\check{q}_2\right)\lambda+ \frac{\check{p}_2\check{q}_1-\check{p}_1\check{q}_2}{\check{q}_1-\check{q}_2}-\check{q}_1\check{q}_2(\check{q}_1+\check{q}_2),\cr
\td{L}_{1,2}&=&  (\lambda-\check{q}_1)(\lambda-\check{q}_2),\cr
\td{L}_{2,2}&=&-\td{L}_{1,1},\cr
\td{L}_{2,1}&=&2\left(\frac{\check{p}_1-\check{p}_2}{\check{q}_1-\check{q}_2}+\check{q}_1^2+\check{q}_1\check{q}_2+\check{q}_2^2+\frac{\tau_{\infty,2}}{2}\right)\lambda^2+2\left(\frac{\check{p}_1\check{q}_1-\check{p}_2\check{q}_2}{\check{q}_1-\check{q}_2}+(\check{q}_1+\check{q}_2)(\check{q}_1^2+\check{q}_2^2+\frac{\tau_{\infty,2}}{2})+\frac{\tau_{\infty,1}}{2}\right)\lambda\cr
&&-\frac{(\check{p}_1-\check{p}_2)^2}{(\check{q}_1-\check{q}_2)^2}-\frac{2(\check{p}_1\check{q}_2-\check{p}_2\check{q}_1)(\check{q}_1+\check{q}_2)}{\check{q}_1-\check{q}_2}+\check{q}_1^4+\check{q}_2^4-\check{q}_1^2\check{q}_2^2+(\check{q}_1^2+\check{q}_1\check{q}_2+\check{q}_2^2)\tau_{\infty,2}\cr
&&+ (\check{q}_1+\check{q}_2)\tau_{\infty,1}+\frac{\tau_{\infty,2}^2}{4}+2t_{\infty^{(1)},0}
\cr
&&\eea}
\normalsize{and }

\beq \td{A}_{\boldsymbol{\alpha}_{\tau_{\infty,1}}}(\lambda)=\begin{pmatrix} -\frac{\lambda+\check{q}_1+\check{q}_2}{2}& \frac{1}{2}\\ \frac{\check{p}_1-\check{p}_2}{\check{q}_1-\check{q}_2}+\check{q}_1^2+\check{q}_1\check{q}_2+\check{q}_2^2+\frac{\tau_{\infty,2}}{2}  & \frac{\lambda+\check{q}_1+\check{q}_2}{2}\end{pmatrix}\eeq 
as well as

\bea \left[\td{A}_{\boldsymbol{\alpha}_{\tau_{\infty,2}}}\right]_{1,1}&=&-\frac{\lambda^2}{4}+ \frac{\check{p}_1-\check{p}_2}{4(\check{q}_1-\check{q}_2)} +\frac{(\check{q}_1+\check{q}_2)^2}{4},\cr
\left[\td{A}_{\boldsymbol{\alpha}_{\tau_{\infty,2}}}\right]_{2,2}&=&-\left[\td{A}_{\boldsymbol{\alpha}_{\tau_{\infty,2}}}\right]_{1,1},\cr
\left[\td{A}_{\boldsymbol{\alpha}_{\tau_{\infty,2}}}\right]_{1,2}&=&\frac{\lambda-(\check{q}_1+\check{q}_2)}{4},\cr
\left[\td{A}_{\boldsymbol{\alpha}_{\tau_{\infty,2}}}\right]_{2,1}&=&\frac{1}{2}\left(\frac{\check{p}_1-\check{p}_2}{\check{q}_1-\check{q}_2}+ \check{q}_1^2+\check{q}_1\check{q}_2+\check{q}_2^2+ \frac{\tau_{\infty,2}}{2}\right)\lambda\cr
&&+\frac{1}{2}\left(\frac{\check{p}_1\check{q}_1-\check{p}_2\check{q}_2}{\check{q}_1-\check{q}_2}+ (\check{q}_1+\check{q}_2)(\check{q}_1^2+\check{q}_2^2+\frac{\tau_{\infty,2}}{2})+\frac{\tau_{\infty,1}}{2}\right).
\eea

Note that these Lax matrices are related by a constant conjugation and time rescaling to those given by H. Chiba in \cite{Chiba} (Section $4.1$(i), with the identification $t_1=\frac{1}{2}\tau_{\infty,1}$, $t_2=\frac{1}{4}\tau_{\infty,2}$, $W_1=-(\check{q}_1+\check{q}_2)$, $\td{W}_2=\check{q}_1\check{q}_2-t_2$, $\alpha_4=2t_{\infty^{(1)},0}$, $V_2=\frac{2(\check{p}_1-\check{p}_2)}{\check{q}_1-\check{q}_2}+2\check{q}_1^2+2\check{q}_1\check{q}_2+2\check{q}_2^2+4t_2$, $V_3=\frac{2(\check{p}_1\check{q}_1-\check{p}_2\check{q}_2)}{\check{q}_1-\check{q}_2}+2\check{q}_1^3+2\check{q}_1^2\check{q}_2+2\check{q}_1\check{q}_2^2+2\check{q}_2^3+4(\check{q}_1+\check{q}_2)t_2+2t_1$). 

Writing $L_C,A_{1,C},A_{2,C}$ for Chiba's matrices and putting
$S=\begin{pmatrix}0&1\\1&0\end{pmatrix}$, the precise identities are
\begin{align}\label{ChibaGaugeDictionary}
\td{L}&=S L_C S^{-1},\nonumber\\
\td{A}_{\boldsymbol\alpha_{\tau_{\infty,1}}}
&=\frac12 S A_{1,C}S^{-1},\qquad
\td{A}_{\boldsymbol\alpha_{\tau_{\infty,2}}}
=\frac14 S A_{2,C}S^{-1}.
\end{align}
The differential equations are compared with Chiba's normalization
at $\hbar=1$.

Following the works of \cite{Chiba}, it is natural to perform a symplectic change of variables 
\bea Q_1&=&-(\check{q}_1+\check{q}_2),\cr
Q_2&=&\check{q}_1\check{q}_2-\frac{1}{4}\tau_{\infty,2},\cr
P_1&=&-\frac{(\check{p}_1\check{q}_1-\check{p}_2\check{q}_2)}{\check{q}_1-\check{q}_2}+\check{q}_1^3+\check{q}_1^2\check{q}_2+\check{q}_1\check{q}_2^2+\check{q}_2^3+\frac{1}{2}(\check{q}_1+\check{q}_2)\tau_{\infty,2}+\frac{1}{2}\tau_{\infty,1},\cr
P_2&=&-\frac{(\check{p}_1-\check{p}_2)}{\check{q}_1-\check{q}_2}+\check{q}_1^2+\check{q}_1\check{q}_2+\check{q}_2^2+\frac{1}{2}\tau_{\infty,2}\cr
&&
\eea
for which Hamiltonian evolutions are polynomial:
\bea \hbar\partial_{\tau_{\infty,1}}Q_1&=&P_2-Q_1^2+Q_2-\frac{1}{4}\tau_{\infty,2},\cr
\hbar\partial_{\tau_{\infty,1}}Q_2&=&P_2Q_1-Q_1Q_2+P_1-\frac{1}{4}\tau_{\infty,2}Q_1-\frac{1}{2}\tau_{\infty,1},\cr
\hbar\partial_{\tau_{\infty,1}}P_1&=&-\frac{1}{2}P_2^2+ Q_2P_2 +2P_1Q_1+\frac{1}{4}\tau_{\infty,2}P_2 -t_{\infty^{(1)},0}+\hbar,\cr
\hbar\partial_{\tau_{\infty,1}}P_2&=&P_2Q_1-P_1
\eea
whose Hamiltonian is 
\bea \text{Ham}^{(\boldsymbol{\alpha}_{\tau_{\infty,1}})}&=&-Q_1Q_2P_2-P_1Q_1^2+\frac{1}{2}P_2^2Q_1+P_1P_2+P_1Q_2-\frac{1}{4}\tau_{\infty,2}P_2Q_1-\frac{1}{4}\tau_{\infty,2}P_1\cr
&&-\frac{1}{2}\tau_{\infty,1}P_2+(t_{\infty^{(1)},0}-\hbar)Q_1
\eea
and
\bea
 \hbar\partial_{\tau_{\infty,2}}Q_1&=&\frac{1}{2}P_2Q_1-\frac{1}{2}Q_1Q_2 -\frac{1}{8}\tau_{\infty,2}Q_1+\frac{1}{2}P_1-\frac{1}{4}\tau_{\infty,1},\cr
\hbar\partial_{\tau_{\infty,2}}Q_2&=&\frac{1}{2}Q_1^2P_2-\frac{1}{2}Q_2P_2+\frac{1}{2}P_1Q_1-\frac{1}{2}Q_2^2 -\frac{1}{8}\tau_{\infty,2}P_2 -\frac{1}{4}\tau_{\infty,1}Q_1+\frac{1}{32}\tau_{\infty,2}^2,\cr
\hbar\partial_{\tau_{\infty,2}}P_1&=&-\frac{1}{2}P_2^2Q_1-\frac{1}{2}P_1P_2+\frac{1}{2}Q_2P_1 +\frac{1}{8}\tau_{\infty,2}P_1+\frac{1}{4}\tau_{\infty,1}P_2,\cr
\hbar\partial_{\tau_{\infty,2}}P_2&=&\frac{1}{4}P_2^2+P_2Q_2+\frac{1}{2}P_1Q_1-\frac{1}{2}(t_{\infty^{(1)},0} -\hbar)\cr
&&
\eea
whose Hamiltonian is 
\bea \text{Ham}^{(\boldsymbol{\alpha}_{\tau_{\infty,2}})}&=&\frac{1}{4}(Q_1^2-Q_2-\frac{1}{4}\tau_{\infty,2})P_2^2+\frac{1}{2}(P_1Q_1-\frac{1}{2}\tau_{\infty,1}Q_1-Q_2^2+\frac{1}{16}\tau_{\infty,2}^2)P_2\cr
&&+\frac{1}{4}P_1^2-\frac{1}{2}(Q_1Q_2+\frac{1}{4}\tau_{\infty,2}Q_1+\frac{1}{2}\tau_{\infty,1})P_1+\frac{1}{2}(t_{\infty^{(1)},0}-\hbar)Q_2.
\eea

Note in particular that $Q_1$ satisfies the following differential
equation with respect to $\tau_{\infty,1}$, with $\tau_{\infty,2}$ fixed:
\begin{equation}\label{GenusTwoScalarEquation}
\begin{aligned}
0={}&\hbar^4\Biggl[
8Q_1^2\frac{\partial^4Q_1}{\partial\tau_{\infty,1}^4}
-16Q_1\frac{\partial Q_1}{\partial\tau_{\infty,1}}
\frac{\partial^3Q_1}{\partial\tau_{\infty,1}^3}
-12Q_1\left(\frac{\partial^2Q_1}{\partial\tau_{\infty,1}^2}\right)^2\\
&\hspace{17mm}
+16\left(\frac{\partial Q_1}{\partial\tau_{\infty,1}}\right)^2
\frac{\partial^2Q_1}{\partial\tau_{\infty,1}^2}\Biggr]\\
&+\hbar^2\Biggl[
-8Q_1(\tau_{\infty,1}+5Q_1^3)
\frac{\partial^2Q_1}{\partial\tau_{\infty,1}^2}
+(8\tau_{\infty,1}-20Q_1^3)
\left(\frac{\partial Q_1}{\partial\tau_{\infty,1}}\right)^2\\
&\hspace{17mm}
-8Q_1\frac{\partial Q_1}{\partial\tau_{\infty,1}}\Biggr]\\
&+20Q_1^7+12\tau_{\infty,2}Q_1^5+8\tau_{\infty,1}Q_1^4\\
&+\bigl(\tau_{\infty,2}^2-24t_{\infty^{(1)},0}+12\hbar\bigr)Q_1^3
-\tau_{\infty,1}^2Q_1.
\end{aligned}
\end{equation}

\section{Outlook}\label{SectionOutlooks}
In this article, we proved that the evolutions of Darboux coordinates with respect to isomonodromic deformations of $\mathfrak{gl}_2(\mathbb{C})$ meromorphic connections are Hamiltonian and provided explicit expressions for these evolutions and Hamiltonians. In addition, we presented a split of the tangent space into trivial and non-trivial deformations with their associated dual trivial and non-trivial times. Using a canonical choice of trivial times, we simplified the general Hamiltonian systems in order to obtain a standard symplectic space of dimension $2g$ and the corresponding Lax pairs. The method developed in the present article opens the way to several generalizations.
\begin{itemize}\item In this article, we assumed that the leading polar coefficient at each pole had distinct eigenvalues. A natural problem is to relax this assumption. This includes both ramified (twisted) formal types and unramified formal types with repeated leading eigenvalues; these cases must be distinguished. In particular, the Painlev\'{e} $1$ case and the Painlev\'{e} $1$ hierarchy would be natural cases to study and results of \cite{MarchalAlameddineP1Hierarchy2023} indicate that the present construction can be adapted to these twisted cases.
\item In this article, we restricted to the case of meromorphic connections in $\mathfrak{gl}_2(\mathbb{C})$ so a natural question is to know if the present setup extends to $\mathfrak{gl}_d(\mathbb{C})$ with $d\geq 2$. In principle, a similar strategy shall be used but it is unclear if all technical issues might be overcome when the rank is arbitrary. We leave this question for future work. In particular, as discussed in the introduction a possible strategy could be to make the connection with the works of Yamakawa \cite{Yamakawa2017TauFA,Yamakawa2019FundamentalTwoForms} more explicit and, since they are valid in $\mathfrak{gl}_d$, to use them to generalise the present article. Let us also mention that one may even generalize the present work to the case of any connected reductive group $G$ over the complex plane admitting a distinguished Borel subgroup which is the natural framework for opers.
\item As mentioned in the introduction, the isospectral approach \cite{Yamakawa2017TauFA,BertolaHarnadHurtubise2022} and the isomonodromic approach using apparent singularities as Darboux coordinates have recently been merged in \cite{MarchalAlameddineIsospectralIsomono2023}. In particular, \cite{MarchalAlameddineIsospectralIsomono2023} complements the present article with the explicit expression of the Lax matrices $(\td{L}(\lambda),\td{A}_{\boldsymbol{\alpha}}(\lambda))$ using those of $(L(\lambda),A_{\boldsymbol{\alpha}}(\lambda))$ and the explicit gauge transformation derived in this paper. Moreover, it provides the change of coordinates between the set of Darboux coordinates $(q_i,p_i)_{1\leq i\leq g}$ used in this paper and the set of ``isospectral Darboux coordinates'' for which isospectral invariants identify with the isomonodromic Hamiltonians. Using this change of coordinates one should thus be able to identify the Hamiltonians and the fundamental two-form defined by Yamakawa  \cite{Yamakawa2017TauFA,Yamakawa2019FundamentalTwoForms} with the ones obtained in this article. This observation may also open the way to the generalization of the present article to $\mathfrak{gl}_d(\mathbb{C})$ and would definitely deserve further investigation.
\end{itemize}

\section*{Acknowledgements}This work is dedicated to John Harnad and his many contributions in integrable systems. This paper is partly a result of the ERC-SyG project, Recursive and Exact New Quantum Theory (ReNewQuantum) which received funding from the European Research Council (ERC) under the European Union's Horizon 2020 research and innovation program under grant agreement No 810573. The work of N.~O.~is supported in part by the NCCR SwissMAP of the Swiss National Science Foundation. The authors would like to thank G. Rembado for fruitful discussions and M. Mazzocco for pointing out the relation of this work with \cite{MartaPaper2022}.

\appendix

\renewcommand{\theequation}{\thesection.\arabic{equation}}

\section{Proof of Proposition \ref{GaugeTransfoProp}}\label{AppendixProofExplicitGaugeTransfo} The proof of Proposition \ref{GaugeTransfoProp} consists in observing that 
\beq \td{G}(\lambda)=\left(G_1(\lambda,\hbar)J(\lambda,\hbar)\right)^{-1}=\begin{pmatrix} 1&0\\ \frac{-Q(\lambda)-(t_{\infty^{(1)},r_\infty-1}\lambda+\eta_0)\underset{j=1}{\overset{g}{\prod}}(\lambda-q_j)}{\underset{s=1}{\overset{n}{\prod}}(\lambda-X_s)^{r_s}}&\frac{\underset{j=1}{\overset{g}{\prod}}(\lambda-q_j)}{ \underset{s=1}{\overset{n}{\prod}}(\lambda-X_s)^{r_s}}\end{pmatrix}\eeq
recovers the matrix \eqref{CompanionMatrix0}. Indeed, we first have that $\td{G}_{2,2}(\lambda)$ is rational in $\lambda$ with poles in $\mathcal{R}$ of order compatible with $F_{\mathcal{R},\mathbf{r}}$. Moreover at $\lambda\to \infty$ it behaves like $\td{G}_{2,2}(\lambda)= \lambda^{r_\infty-3}+O\left(\lambda^{r_\infty-4}\right)$. Finally its zeros are precisely given by $(q_i)_{1\leq i\leq g}$ as required for \eqref{Conditionqi}. Hence $\td{G}_{2,2}(\lambda)=\td{L}_{1,2}(\lambda)$. Similarly, the entry $\td{G}_{2,1}(\lambda)$ is rational in $\lambda$ with poles only in $\mathcal{R}$. The orders at finite poles are compatible with $F_{\mathcal{R},\mathbf{r}}$. Moreover, it satisfies $\td{G}_{2,1}(q_i)=p_i$ for all $i\in \llbracket 1,g\rrbracket$ as required for \eqref{Conditionpi}. At infinity, we get that
\beq \td{G}_{2,1}(\lambda)=-t_{\infty^{(1)},r_\infty-1}\lambda^{r_\infty-2} -\left(\eta_0-t_{\infty^{(1)},r_\infty-1}\left(\sum_{j=1}^gq_j -\sum_{s=1}^nr_sX_s\right)\right)\lambda^{r_\infty-3}+ O\left(\lambda^{r_\infty-4}\right)\eeq
so that the normalization of $\hat{F}_{\mathcal{R},\mathbf{r}}$ and Remark \ref{RemarkCoeff} is verified by taking for $r_\infty\geq 2$
\beq \eta_0=t_{\infty^{(1)},r_\infty-2}+t_{\infty^{(1)},r_\infty-1}\left(\sum_{j=1}^gq_j -\sum_{s=1}^nr_sX_s\right).\eeq
Similarly, for $r_\infty=1$, we have that
\beq \left[\td{L}(\lambda)\right]_{1,1}=-t_{\infty^{(1)},0}\lambda^{-1}+\left(\sum_{s=1}^n\left[X_s\td{L}^{[X_s,0]}+\td{L}^{[X_s,1]}\right]_{1,1} \right)\lambda^{-2}+O\left(\lambda^{-3}\right)\eeq
so that we should take
\beq \eta_0=t_{\infty^{(1)},0}\left(\sum_{j=1}^gq_j -\sum_{s=1}^nr_sX_s\right)-\sum_{s=1}^n\left[X_s\td{L}^{[X_s,0]}+\td{L}^{[X_s,1]}\right]_{1,1}.\eeq
Thus, $\td{G}_{2,1}(\lambda)=\td{L}_{1,1}(\lambda)$ ending the proof.

\section{Proof of Proposition \ref{PropPsiAsymp}}\label{AppendixProofPropPsiAsymp}
Work with a local formal fundamental matrix at infinity. With the
leading eigenvectors labeled and normalized, write
\begin{equation}\label{PsiPsi}
\td{\Psi}^{(\infty)}(\lambda)
=R(\lambda)\operatorname{diag}(e^{\Theta_1(\lambda)},e^{\Theta_2(\lambda)}),
\qquad R(\lambda)=I+R_1\lambda^{-1}+O(\lambda^{-2}),
\end{equation}
where
\[
\Theta_i=-\frac1\hbar\sum_{k=1}^{r_\infty-1}
\frac{t_{\infty^{(i)},k}}{k}\lambda^k
-\frac{t_{\infty^{(i)},0}}\hbar\log\lambda.
\]
The normalization $\td{L}_{1,2}=
\lambda^{r_\infty-3}(1+O(\lambda^{-1}))$ and the $(1,2)$ entry of
$\hbar R'=\td{L}R-\hbar R\Theta'$ give
\[
(R_1)_{12}=\begin{cases}
\bigl(t_{\infty^{(1)},r_\infty-1}
-t_{\infty^{(2)},r_\infty-1}\bigr)^{-1},&r_\infty\ge2,\\
\bigl(t_{\infty^{(1)},0}-t_{\infty^{(2)},0}-\hbar\bigr)^{-1},
&r_\infty=1.
\end{cases}
\]
The denominators are nonzero under the stated genericity and
nonresonance hypotheses. Thus the first-row entries are
$e^{\Theta_1}(1+O(\lambda^{-1}))$ and
$(R_1)_{12}\lambda^{-1}e^{\Theta_2}(1+O(\lambda^{-1}))$.
This proves the extra $-\log\lambda$ in the second scalar exponent.
Their logarithmic derivatives satisfy
\[
\hbar\frac{\psi_i'}{\psi_i}
=-\sum_{k=1}^{r_\infty-1}t_{\infty^{(i)},k}\lambda^{k-1}
-\frac{t_{\infty^{(i)},0}+\hbar\delta_{i,2}}\lambda
+O(\lambda^{-2}).
\]
At a finite pole the nonzero leading first-row components of the two
local eigenvectors give the displayed finite-pole exponents directly.
The assertions are formal, or sectorial asymptotic assertions for
analytic realizations, with the local bases specified in the proposition.

\section{Proof of Proposition \ref{PropLaxMatrix}}\label{AppendixLax}
Let us recall that the wave functions $\psi_1$ and $\psi_2$ are such that for all $s\in \llbracket 1, n\rrbracket$:
 
\bea \label{PsiAsymptotics}\psi_1(\lambda)&\overset{\lambda\to \infty}{=}&\exp\left(-\frac{1}{\hbar}\sum_{k=1}^{r_\infty-1} \frac{t_{\infty^{(1)},k}}{k}\lambda^k -\frac{t_{\infty^{(1)},0}}{\hbar}\ln \lambda+ A_{\infty^{(1)},0} +O\left(\lambda^{-1}\right)\right)\cr
\psi_2(\lambda)&\overset{\lambda\to \infty}{=}&\exp\left(-\frac{1}{\hbar}\sum_{k=1}^{r_\infty-1} \frac{t_{\infty^{(2)},k}}{k}\lambda^k -\frac{t_{\infty^{(2)},0}}{\hbar}\ln \lambda -\ln \lambda+ A_{\infty^{(2)},0}+O\left(\lambda^{-1}\right)\right)\cr
\psi_1(\lambda)&\overset{\lambda\to X_s}{=}&\exp\left(-\frac{1}{\hbar}\sum_{k=1}^{r_s-1} \frac{t_{X_s^{(1)},k}}{k}(\lambda-X_s)^{-k} +\frac{t_{X_s^{(1)},0}}{\hbar}\ln(\lambda-X_s)+ A_{X_s^{(1)},0}+O\left(\lambda-X_s\right)\right)\cr
\psi_2(\lambda)&\overset{\lambda\to X_s}{=}&\exp\left(-\frac{1}{\hbar}\sum_{k=1}^{r_s-1} \frac{t_{X_s^{(2)},k}}{k}(\lambda-X_s)^{-k}+\frac{t_{X_s^{(2)},0}}{\hbar}\ln(\lambda-X_s) + A_{X_s^{(2)},0}+O\left(\lambda-X_s\right)\right).\cr
&&
\eea

Consider $W(\lambda)=\hbar(\psi_1\partial_\lambda\psi_2-
\psi_2\partial_\lambda\psi_1)$. Abel's identity for
$\td{\Psi}$ and the determinant of the rational gauge imply that
$W\exp[-\hbar^{-1}\int_{\lambda_0}^{\lambda}P_1(\xi)d\xi]$
is rational. Pole counts below concern this stripped prefactor,
not $W$ itself. Computing in a local basis at each pole changes
$W$ only by a nonzero $\lambda$-independent constant and gives:
\small{\bea \label{WronskianAsympt} W(\lambda)&\overset{\lambda\to \infty}{=}&\left( \sum_{k=1}^{r_\infty-1} (t_{\infty^{(1)},k} -t_{\infty^{(2)},k}) \lambda^{k-1} +\frac{(t_{\infty^{(1)},0} -t_{\infty^{(2)},0}-\hbar)}{\lambda} +O(\lambda^{-2})\right)\lambda^{-1}\cr
&&\exp\left(-\frac{1}{\hbar}\sum_{k=1}^{r_\infty-1} \frac{(t_{\infty^{(1)},k}+t_{\infty^{(2)},k} ) }{k} \lambda^k -\frac{(t_{\infty^{(1)},0}+t_{\infty^{(2)},0} ) }{\hbar}\ln \lambda +O(1) \right)   \cr
&\overset{\lambda\to \infty}{=}&\left( \sum_{k=1}^{r_\infty-1} (t_{\infty^{(1)},k} -t_{\infty^{(2)},k}) \lambda^{k}+(t_{\infty^{(1)},0} -t_{\infty^{(2)},0}-\hbar) \right)\lambda^{-2}\cr
&&\exp\left(\frac{1}{\hbar}\int_{\lambda_0}^{\lambda}\sum_{j=0}^{r_\infty-2} P_{\infty,j}^{(1)} \tilde{\lambda}^{j}d\tilde{\lambda}-\frac{(t_{\infty^{(1)},0}+t_{\infty^{(2)},0} )}{\hbar}\ln \lambda \right)\kappa_\infty\left( 1+ O(\lambda^{-1})\right)\cr
&&
\eea}
\normalsize{ and} for all $s\in \llbracket 1,n\rrbracket$:
\footnotesize{\bea \label{WronskianAsympt2} W(\lambda)&\overset{\lambda\to X_s}{=}&-\left(\sum_{k=1}^{r_s-1}(t_{X_s^{(1)},k}-t_{X_s^{(2)},k} )(\lambda-X_s)^{-k-1} +\frac{t_{X_s^{(1)},0}-t_{X_s^{(2)},0}}{(\lambda-X_s)}+ O(1)\right)\cr
&&\exp\left(-\frac{1}{\hbar}\sum_{k=1}^{r_s-1} \frac{(t_{X_s^{(1)},k}+t_{X_s^{(2)},k} )}{k}(\lambda-X_s)^{-k} +\frac{t_{X_s^{(1)},0}+t_{X_s^{(2)},0}}{\hbar}\ln(\lambda-X_s)+O(1) \right)   \cr
&\overset{\lambda\to X_s}{=}&-\sum_{k=0}^{r_s-1}(t_{X_s^{(1)},k}-t_{X_s^{(2)},k} )(\lambda-X_s)^{-k-1}\exp\left(\frac{1}{\hbar}\int_{\lambda_0}^{\lambda} \sum_{j=1}^{r_s} \frac{P_{X_s,j}^{(1)}}{(\tilde{\lambda}-X_s)^j}d\tilde{\lambda}
\right)\kappa_s\left(1+O(\lambda-X_s)\right).\cr
&&
\eea}
\normalsize{}
The stripped prefactor has finite poles of order $r_s$ and
behavior proportional to $\lambda^{r_\infty-3}$ at infinity on
the generic chart. Consequently,
\beq \label{WronskianDef} W(\lambda)=\kappa \frac{\underset{i=1}{\overset{g}{\prod}} (\lambda-q_i)}{\underset{s=1}{\overset{n}{\prod}} (\lambda-X_s)^{r_s}} \exp\left(\frac{1}{\hbar}\int_{\lambda_0}^{\lambda} P_1(\td{\lambda})d\td{\lambda}  
\right) \eeq 
where we have denoted $(q_j)_{1\leq j\leq g}$ the zeros of $W$ and an unknown constant $\kappa$. Note that the additional condition \eqref{SumResidues} is necessary for the previous formula to hold around $\lambda\to \infty$.

\medskip

Since $L$ is the companion matrix attached to $\psi_1$ and $\psi_2$, it is a straightforward computation to show that
\bea\label{EntriesInTermsOfYi}
L_{2,2}(\lambda)
&=& \hbar \frac{\partial_\lambda W(\lambda)}{W(\lambda)},\cr
L_{2,1}(\lambda)&=&- Y_1(\lambda) \,  Y_2(\lambda) + \hbar \frac{Y_2(\lambda) \, \partial_\lambda Y_1(\lambda) - Y_1(\lambda) \, \partial_\lambda Y_2(\lambda) }{Y_2(\lambda)-Y_1(\lambda)},
\eea
where we have defined $Y_i(\lambda)= \hbar \frac{1}{\psi_{i}(\lambda)} \frac{\partial \psi_{i}(\lambda)}{\partial \lambda}$ for $i\in\{1,2\}$. Using the relation between $L_{2,2}$ and the Wronskian, we get:
\beq \label{DefL22} L_{2,2}(\lambda)=P_1(\lambda)+\sum_{j=1}^{g} \frac{\hbar}{\lambda-q_j}  -\sum_{s=1}^n \frac{\hbar r_s}{\lambda-X_s}.\eeq
Then, using \eqref{PsiAsymptotics} we have for $r_\infty\geq 3$:

\bea\label{Y1Y2} Y_1(\lambda)&\overset{\lambda \to \infty}{=}&-\sum_{k=1}^{r_\infty-1} t_{\infty^{(1)},k}\lambda^{k-1} -\frac{t_{\infty^{(1)},0}}{\lambda}+ O\left(\lambda^{-2}\right), \cr
Y_2(\lambda)&\overset{\lambda \to \infty}{=}&-\sum_{k=1}^{r_\infty-1} t_{\infty^{(2)},k}\lambda^{k-1} -\frac{t_{\infty^{(2)},0}+\hbar}{\lambda} + O\left(\lambda^{-2}\right), \cr
Y_i(\lambda)&\overset{\lambda \to X_s}{=}&\sum_{k=1}^{r_s-1} t_{X_s^{(i)},k}(\lambda-X_s)^{-k-1} +\frac{t_{X_s^{(i)},0}}{\lambda-X_s}+ O(1),
\eea

so that we get:

\bea \frac{Y_2(\lambda) \, \partial_\lambda Y_1(\lambda) - Y_1(\lambda) \, \partial_\lambda Y_2(\lambda) }{Y_2(\lambda)-Y_1(\lambda)}&\overset{\lambda \to \infty}{=}&O\left(\lambda^{r_\infty-4}\right),\cr
\frac{Y_2(\lambda) \, \partial_\lambda Y_1(\lambda) - Y_1(\lambda) \, \partial_\lambda Y_2(\lambda) }{Y_2(\lambda)-Y_1(\lambda)}&\overset{\lambda \to X_s}{=}&O\left((\lambda-X_s)^{-r_s}\right),
\eea

from which we immediately obtain using \eqref{P2Coeffs} that

\bea
L_{2,1}(\lambda)&\overset{\lambda \to \infty}{=}&-\sum_{j=r_\infty-3}^{2r_\infty-4} P_{\infty,j}^{(2)}\lambda^j -\hbar t_{\infty^{(1)},r_\infty-1}\lambda^{r_\infty-3}+ O\left(\lambda^{r_\infty-4}\right),\cr
L_{2,1}(\lambda)&\overset{\lambda \to X_s}{=}&-\sum_{j=r_s+1}^{2r_s} P_{X_s,j}^{(2)}(\lambda-X_s)^{-j} + O\left((\lambda-X_s)^{-r_s}\right),\cr
L_{2,1}(\lambda)&\overset{\lambda \to q_j}{=}&O\left((\lambda-q_j)^{-1}\right).
\eea

\begin{remark}\label{RemarkLrinftyequal2} For $r_\infty=1$, similar computations lead to
\beq L_{2,1}(\lambda)\overset{\lambda\to \infty}{=}-\frac{t_{\infty^{(1)},0}(t_{\infty^{(2)},0}+\hbar)}{\lambda^2}+ O(\lambda^{-3})\eeq
while for $r_\infty=2$ we obtain:
\beq L_{2,1}(\lambda)\overset{\lambda\to \infty}{=}-t_{\infty^{(1)},1}t_{\infty^{(2)},1}-\frac{t_{\infty^{(1)},1}t_{\infty^{(2)},0}+t_{\infty^{(2)},1}t_{\infty^{(1)},0}+\hbar t_{\infty^{(1)},1}}{\lambda}+O(\lambda^{-2}).\eeq
\end{remark}

Combining the asymptotic behavior at each pole and the fact that $L$ is rational with poles only in $\mathcal{R}$ and simple poles at $\lambda\in \{q_1,\dots,q_g\}$, we get Proposition \ref{PropLaxMatrix}.

\section{Proof of Proposition \ref{Consistencyg}}\label{AppendixGaugeTransfo}
We first observe that from \eqref{CheckLEquations}:
\bea \check{L}_{1,1}(\lambda,\hbar)&\overset{\lambda \to \infty}{=}& \check{L}_{1,1}^{(0)}\lambda^{r_\infty-4}+O\left(\lambda^{r_\infty-5}\right),\cr
\check{L}_{1,2}(\lambda,\hbar)&\overset{\lambda \to \infty}{=}& \lambda^{r_\infty-3}+\check{L}_{1,2}^{(1)}\lambda^{r_\infty-4}+O\left(\lambda^{r_\infty-5}\right),\cr
\check{L}_{2,1}(\lambda,\hbar)&\overset{\lambda \to \infty}{=}& \check{L}_{2,1}^{(0)}\lambda^{r_\infty-1}+\check{L}_{2,1}^{(1)}\lambda^{r_\infty-2}+O\left(\lambda^{r_\infty-3}\right),\cr
\check{L}_{2,2}(\lambda,\hbar)&\overset{\lambda \to \infty}{=}& \check{L}_{2,2}^{(0)}\lambda^{r_\infty-2}+\check{L}_{2,2}^{(1)}\lambda^{r_\infty-3}+O\left(\lambda^{r_\infty-4}\right).
\eea
These asymptotics are valid for any value of $r_\infty\geq 1$. Note also that the terms $\check{L}_{2,1}^{(0)}$ and $\check{L}_{2,1}^{(1)}$ of the entries $\check{L}_{2,1}$ are only determined by the $L_{2,1}(\lambda,\hbar) \frac{\underset{s=1}{\overset{n}{\prod}}(\lambda-X_s)^{r_s}}{\underset{j=1}{\overset{g}{\prod}}(\lambda-q_j)}$ since the other terms are all of lower orders when $\lambda\to \infty$. It gives, using the fact that $\check{L}_{1,2}=\frac{\underset{j=1}{\overset{g}{\prod}}(\lambda-q_j)}{\underset{s=1}{\overset{n}{\prod}}(\lambda-X_s)^{r_s}}$ that:
\bea \label{RelationCheck} \check{L}_{2,1}^{(0)}&=&L_{2,1}^{(0)},\cr
 \check{L}_{2,1}^{(1)}&=&L_{2,1}^{(1)}-L_{2,1}^{(0)}\check{L}_{1,2}^{(1)}
\eea
where we have denoted
\beq L_{2,1}(\lambda,\hbar)\overset{\lambda \to \infty}{=}L_{2,1}^{(0)}\lambda^{2r_\infty-4}+L_{2,1}^{(1)}\lambda^{2r_\infty-5}+O\left(\lambda^{2r_\infty-6}\right).\eeq 
From the fact that $\check{L}_{1,2}=\frac{\underset{j=1}{\overset{g}{\prod}}(\lambda-q_j)}{\underset{s=1}{\overset{n}{\prod}}(\lambda-X_s)^{r_s}}$, it is easy to obtain:
\beq \check{L}_{1,2}^{(1)}=\underset{s=1}{\overset{n}{\sum}} r_s X_s -\underset{j=1}{\overset{g}{\sum}}q_j. \eeq
We perform the following gauge transformation:
\beq \td{\Psi}(\lambda,\hbar)=G(\lambda,\hbar) \check{\Psi}(\lambda,\hbar) \,\,\text{ with } G(\lambda,\hbar)=\begin{pmatrix}1 &0 \\ \eta_1 \lambda+\eta_0 & 1 \end{pmatrix}.\eeq
It satisfies:
\beq \hbar \partial_\lambda \td{\Psi}(\lambda,\hbar)=\td{L}(\lambda,\hbar) \td{\Psi}(\lambda,\hbar), \eeq
with 
\footnotesize{\bea \td{L}(\lambda,\hbar) &=&G(\lambda,\hbar) \check{L}(\lambda,\hbar)G^{-1}(\lambda,\hbar) +\hbar (\partial_\lambda G(\lambda,\hbar)) G^{-1}(\lambda,\hbar)\cr
&=&\begin{pmatrix}
\check L_{1,1}-(\eta_1\lambda+\eta_0)\check L_{1,2}
&\check L_{1,2}\\
\check L_{2,1}+(\eta_1\lambda+\eta_0)(\check L_{1,1}-\check L_{2,2})
-(\eta_1\lambda+\eta_0)^2\check L_{1,2}+\hbar\eta_1
&\check L_{2,2}+(\eta_1\lambda+\eta_0)\check L_{1,2}
\end{pmatrix}.\cr
&&
\eea}
\normalsize{}
Thus, it is clear that the conditions 
\beq \td{L}(\lambda,\hbar)\overset{\lambda\to \infty}{=} \text{diag}(X_1,X_2)\lambda^{r_\infty-2}+ \begin{pmatrix} X&1\\ X&X\end{pmatrix} \lambda^{r_\infty-3}+O\left(\lambda^{r_\infty-4}\right)\eeq  
are equivalent to (note that entry $\check{L}_{1,2}(\lambda,\hbar)$ is already properly normalized at subleading order at infinity)
\bea \label{ConditionsToSatisfy}0&=&\check{L}_{2,1}^{(0)}-\eta_1\left(\eta_1+\check{L}_{2,2}^{(0)} \right) +\hbar \eta_1 \delta_{r_\infty=1},\cr
0&=&\check L_{2,1}^{(1)}
-(2\eta_1+\check L_{2,2}^{(0)})\eta_0
-\eta_1^2\check L_{1,2}^{(1)}
-\eta_1\check L_{2,2}^{(1)}+\hbar\eta_1\delta_{r_\infty=2}.
\eea
 
Using \eqref{RelationCheck}, this is equivalent to saying that the coefficients $(\eta_1,\eta_0)$ are completely determined by $\left(L_{2,1}^{(0)},L_{2,1}^{(1)},\check{L}_{2,2}^{(0)},\check{L}_{2,2}^{(1)}\right)$. Since these coefficients depend on the value of $r_\infty$,  we need to split the computations between the three cases $r_\infty=1$, $r_\infty=2$ and $r_\infty\geq 3$.

\subsection{The case $r_\infty\geq 3$}
Let us assume that $r_\infty\geq 3$. In this case the two leading orders at infinity of $\check{L}_{2,2}(\lambda,\hbar)$ are given by $P_1(\lambda)$. Thus, we get:
\bea \check{L}_{2,2}^{(0)}&=& P_{\infty, r_\infty-2}^{(1)}= -(t_{\infty^{(1)},r_\infty-1}+t_{\infty^{(2)},r_\infty-1}),\cr
\check{L}_{2,2}^{(1)}&=& P_{\infty, r_\infty-3}^{(1)}= -(t_{\infty^{(1)},r_\infty-2}+t_{\infty^{(2)},r_\infty-2}),\cr
L_{2,1}^{(0)}&=&-P_{\infty, 2r_\infty-4}^{(2)}=-t_{\infty^{(1)},r_\infty-1}t_{\infty^{(2)},r_\infty-1},\cr
L_{2,1}^{(1)}&=&-P_{\infty, 2r_\infty-5}^{(2)}=-(t_{\infty^{(1)},r_\infty-1}t_{\infty^{(2)},r_\infty-2}+t_{\infty^{(2)},r_\infty-1}t_{\infty^{(1)},r_\infty-2}).\cr&&
\eea
Thus, conditions \eqref{ConditionsToSatisfy} are equivalent to two sets of solutions $(\eta_1,\eta_0)$ given by 
\bea \label{Trueg0g1}\eta_1&=&t_{\infty^{(1)},r_\infty-1},\cr
\eta_0&=&t_{\infty^{(1)},r_\infty-2}+t_{\infty^{(1)},r_\infty-1}\left(\underset{j=1}{\overset{g}{\sum}}q_j-\underset{s=1}{\overset{n}{\sum}} r_s X_s\right),\cr
\td{L}(\lambda,\hbar)&\overset{\lambda\to\infty}{=}&\text{diag}(-t_{\infty^{(1)},r_\infty-1}, -t_{\infty^{(2)},r_\infty-1} )\lambda^{r_\infty-2} +O\left(\lambda^{r_\infty-3}\right),
\eea
or
\bea \eta_1&=&t_{\infty^{(2)},r_\infty-1},\cr
\eta_0&=&t_{\infty^{(2)},r_\infty-2}+t_{\infty^{(2)},r_\infty-1}\left(\underset{j=1}{\overset{g}{\sum}}q_j-\underset{s=1}{\overset{n}{\sum}} r_s X_s\right),\cr
\td{L}(\lambda,\hbar)&\overset{\lambda\to\infty}{=}&\text{diag}(-t_{\infty^{(2)},r_\infty-1}, -t_{\infty^{(1)},r_\infty-1} )\lambda^{r_\infty-2} +O\left(\lambda^{r_\infty-3}\right),
\eea
but whose asymptotics do not fit those required for $\td{L}$ at infinity so that only \eqref{Trueg0g1} is valid.

\subsection{The case $r_\infty= 2$}
Let us assume that $r_\infty=2$. In this case, the leading orders at infinity of $L_{2,1}(\lambda,\hbar)$ are given by Remark \ref{RemarkLrinftyequal2}. Moreover, the leading orders of $\check{L}_{2,2}(\lambda,\hbar)$ are given by \eqref{P1Coeffs} with the additional residue condition \eqref{SumResidues}.
\bea\check{L}_{2,2}^{(0)}&=& P_{\infty, 0}^{(1)}= -(t_{\infty^{(1)},1}+t_{\infty^{(2)},1}),\cr
\check{L}_{2,2}^{(1)}&=& -(t_{\infty^{(1)},0}+t_{\infty^{(2)},0}),\cr
L_{2,1}^{(0)}&=&-P_{\infty, 0}^{(2)}=-t_{\infty^{(1)},1}t_{\infty^{(2)},1},\cr
L_{2,1}^{(1)}&=&-(t_{\infty^{(1)},1}t_{\infty^{(2)},0}+t_{\infty^{(2)},1}t_{\infty^{(1)},0}+\hbar t_{\infty^{(1)},1}).
\eea
Thus, conditions \eqref{ConditionsToSatisfy} are equivalent to two sets of solutions $(\eta_1,\eta_0)$ given by 
\bea \label{Trueg0g12}\eta_1&=&t_{\infty^{(1)},1},\cr
\eta_0&=&t_{\infty^{(1)},0}+t_{\infty^{(1)},1}\left(\underset{j=1}{\overset{g}{\sum}}q_j-\underset{s=1}{\overset{n}{\sum}} r_s X_s\right),\cr
\td{L}(\lambda,\hbar)&\overset{\lambda\to\infty}{=}&\text{diag}(-t_{\infty^{(1)},1}, -t_{\infty^{(2)},1} )\lambda^{r_\infty-2} +O\left(\lambda^{r_\infty-3}\right),
\eea
or
\bea \eta_1&=&t_{\infty^{(2)},1},\cr
\eta_0&=&t_{\infty^{(2)},0}+\hbar+t_{\infty^{(2)},1}\left(\underset{j=1}{\overset{g}{\sum}}q_j-\underset{s=1}{\overset{n}{\sum}} r_s X_s\right),\cr
\td{L}(\lambda,\hbar)&\overset{\lambda\to\infty}{=}&\text{diag}(-t_{\infty^{(2)},1}, -t_{\infty^{(1)},1} )\lambda^{r_\infty-2} +O\left(\lambda^{r_\infty-3}\right),
\eea
but whose asymptotics do not fit those required for $\td{L}$ at infinity so that only \eqref{Trueg0g12} is valid.

\subsection{The case $r_\infty=1$}

Let us assume that $r_\infty=1$. In this case, the situation is more complicated to obtain the sub-leading orders of $L_{2,1}(\lambda,\hbar)$ and $\check{L}_{2,2}(\lambda,\hbar)$ at infinity. We have:
\beq \check{L}_{2,2}(\lambda,\hbar)=P_1(\lambda)+ O\left(\lambda^{-3}\right)\eeq
with $P_1$ given by \eqref{P1Coeffs}:
\beq P_1(\lambda)=\sum_{s=1}^n\sum_{j=1}^{r_s}\frac{P_{X_s,j}^{(1)}}{(\lambda-X_s)^{j}}\overset{\lambda\to\infty}{=}\frac{\underset{s=1}{\overset{n}{\sum}} P_{X_s,1}^{(1)}}{\lambda}+\frac{\underset{s=1}{\overset{n}{\sum}} \left(P_{X_s,2}^{(1)}\delta_{r_s\geq 2}+ X_sP_{X_s,1}^{(1)}\right)}{\lambda^2}+O\left(\lambda^{-3}\right).
\eeq
The residue condition \eqref{SumResidues} implies that $\underset{s=1}{\overset{n}{\sum}} P_{X_s,1}^{(1)}=-(t_{\infty^{(1)},0}+t_{\infty^{(2)},0})$ but the next term cannot be simplified:
\beq \underset{s=1}{\overset{n}{\sum}} \left(P_{X_s,2}^{(1)}\delta_{r_s\geq 2}+ X_sP_{X_s,1}^{(1)}\right)=\underset{s=1}{\overset{n}{\sum}} \left((t_{X_s^{(1)},1}+t_{X_s^{(2)},1})\delta_{r_s\geq 2}+ X_s(t_{X_s^{(1)},0}+t_{X_s^{(2)},0})\right).\eeq
Thus, we finally get:
\bea \check{L}_{2,2}^{(0)}&=&-(t_{\infty^{(1)},0}+t_{\infty^{(2)},0}),\cr
\check{L}_{2,2}^{(1)}&=&\underset{s=1}{\overset{n}{\sum}} \left((t_{X_s^{(1)},1}+t_{X_s^{(2)},1})\delta_{r_s\geq 2}+ X_s(t_{X_s^{(1)},0}+t_{X_s^{(2)},0})\right).
\eea 
Similarly, the leading order at infinity of the asymptotic expansion of $L_{2,1}(\lambda,\hbar)$ is given by Remark \ref{RemarkLrinftyequal2} while the sub-leading order has to be computed through the definition given in Proposition \ref{PropLaxMatrix}:
\beq L_{2,1}^{(0)}=-t_{\infty^{(1)},0}(t_{\infty^{(2)},0}+\hbar)\eeq
and
\bea L_{2,1}(\lambda,\hbar)&=& -\sum_{s=1}^n\sum_{j=r_s+1}^{2r_s}P_{X_s,j}^{(2)}(\lambda-X_s)^{-j} +\sum_{s=1}^n\sum_{j=1}^{r_s}H_{X_s,j}(\lambda-X_s)^{-j}-\sum_{j=1}^{g} \frac{\hbar p_j}{\lambda-q_j}\cr
&\overset{\lambda\to\infty}{=}&-\frac{t_{\infty^{(1)},0}(t_{\infty^{(2)},0}+\hbar)}{\lambda^2} + \frac{L_{2,1}^{(1)}}{\lambda^3}+O\left(\lambda^{-4}\right)
\eea
with
\bea L_{2,1}^{(1)}&=& -\sum_{s=1}^n(2X_sP_{X_s,2}^{(2)}\delta_{r_s=1}+P_{X_s,3}^{(2)}\delta_{r_s=2}) \cr
&&+ \sum_{s=1}^n(X_s^2H_{X_s,1}+2X_sH_{X_s,2}\delta_{r_s\geq 2}+H_{X_s,3}\delta_{r_s\geq 3})-\hbar \sum_{j=1}^{g}p_jq_j^2.
\eea
Thus, conditions \eqref{ConditionsToSatisfy} are equivalent to two sets of solutions $(\eta_1,\eta_0)$ given by 
\bea\label{Trueg0g13} \eta_1&=&t_{\infty^{(1)},0},\cr
\eta_0&=&\frac{1}{t_{\infty^{(1)},0}-t_{\infty^{(2)},0}}\Big[\cr
&& -\sum_{s=1}^n(2X_sP_{X_s,2}^{(2)}\delta_{r_s=1}+P_{X_s,3}^{(2)}\delta_{r_s=2})\cr
&& + \sum_{s=1}^n(X_s^2H_{X_s,1}+2X_sH_{X_s,2}\delta_{r_s\geq 2}+H_{X_s,3}\delta_{r_s\geq 3})-\hbar \sum_{j=1}^{g}p_jq_j^2\cr
&&-t_{\infty^{(1)},0}\underset{s=1}{\overset{n}{\sum}} \left((t_{X_s^{(1)},1}+t_{X_s^{(2)},1})\delta_{r_s\geq 2}+ X_s(t_{X_s^{(1)},0}+t_{X_s^{(2)},0})\right)\cr
&&+t_{\infty^{(1)},0}(t_{\infty^{(1)},0}-t_{\infty^{(2)},0}-\hbar)\left(\sum_{j=1}^g q_j-\sum_{s=1}^n r_s X_s\right)
\Big],\cr
\td{L}(\lambda,\hbar)&\overset{\lambda\to\infty}{=}&\text{diag}(-t_{\infty^{(1)},0}, -t_{\infty^{(2)},0} )\lambda^{-1} +O\left(\lambda^{-2}\right),\cr
&&
\eea
or
\bea \eta_1&=&t_{\infty^{(2)},0}+\hbar,\cr
\eta_0&=&-\frac{1}{t_{\infty^{(1)},0}-t_{\infty^{(2)},0}-2\hbar}\Big[-\sum_{s=1}^n(2X_sP_{X_s,2}^{(2)}\delta_{r_s=1}+P_{X_s,3}^{(2)}\delta_{r_s=2}) \cr
&&+ \sum_{s=1}^n(X_s^2H_{X_s,1}+2X_sH_{X_s,2}\delta_{r_s\geq 2}+H_{X_s,3}\delta_{r_s\geq 3})-\hbar \sum_{j=1}^{g}p_jq_j^2\cr
&&-(t_{\infty^{(2)},0}+\hbar)\underset{s=1}{\overset{n}{\sum}} \left((t_{X_s^{(1)},1}+t_{X_s^{(2)},1})\delta_{r_s\geq 2}+ X_s(t_{X_s^{(1)},0}+t_{X_s^{(2)},0})\right)\cr
&&-(t_{\infty^{(2)},0}+\hbar)(t_{\infty^{(1)},0}-t_{\infty^{(2)},0} -\hbar) \left(\underset{j=1}{\overset{g}{\sum}}q_j-\underset{s=1}{\overset{n}{\sum}} r_s X_s \right)\Big],\cr
\td{L}(\lambda,\hbar)&\overset{\lambda\to\infty}{=}&\text{diag}(-t_{\infty^{(2)},0}-\hbar, -t_{\infty^{(1)},0}+\hbar)\lambda^{-1} +O\left(\lambda^{-2}\right),\cr
&&
\eea
but whose asymptotics do not fit those required for $\td{L}$ at infinity so that only \eqref{Trueg0g13} is valid.

\section{Proof of Proposition \ref{PropAsymptoticExpansionA12}}\label{AppendixExpansionA}
In order to prove Proposition \ref{PropAsymptoticExpansionA12}, we shall observe that the first line of $A_{\boldsymbol{\alpha}}(\lambda)$ is given by
\bea \label{ExpressionA12A11}
\left[A_{\boldsymbol{\alpha}}(\lambda)\right]_{1,2}&=&\frac{W_{\boldsymbol{\alpha}}(\lambda)}{W(\lambda)}= \frac{Z_{\boldsymbol{\alpha},2}(\lambda)-Z_{\boldsymbol{\alpha},1}(\lambda) }{Y_2(\lambda)-Y_1(\lambda)},\cr
\left[A_{\boldsymbol{\alpha}}(\lambda)\right]_{1,1}&=& \frac{Z_{\boldsymbol{\alpha},1}(\lambda) Y_2(\lambda) -Z_{\boldsymbol{\alpha},2}(\lambda) Y_1(\lambda)}{Y_2(\lambda) -Y_1(\lambda)},
\eea
where we have defined
\bea 
Z_{\boldsymbol{\alpha},i}(\lambda)&=& \frac{\mathcal{L}_{\boldsymbol{\alpha}}[\psi_i(\lambda)]}{\psi_i(\lambda)} \,\,,\,\, \forall \, i\in\llbracket 1,2\rrbracket,\cr
W_{\boldsymbol{\alpha}}(\lambda)&=& \mathcal{L}_{\boldsymbol{\alpha}}[\psi_2(\lambda)] \psi_1(\lambda) - \mathcal{L}_{\boldsymbol{\alpha}}[\psi_1(\lambda)] \psi_2(\lambda).
\eea 
Using \eqref{PsiAsymptotics} we get that $\mathcal{L}_{\boldsymbol{\alpha}}[\Psi_i(\lambda,\hbar)]$ has the following local expansions.
\footnotesize{\bea \label{LPsi}  
\mathcal{L}_{\boldsymbol{\alpha}}[\psi_1(\lambda,\hbar)]&\overset{\lambda\to \infty}{=}&-\left(\sum_{k=1}^{r_\infty-1} \frac{\alpha_{\infty^{(1)},k}}{k}\lambda^k +O(1)\right)\exp\left(-\frac{1}{\hbar}\sum_{k=1}^{r_\infty-1} \frac{t_{\infty^{(1)},k}}{k}\lambda^k -\frac{t_{\infty^{(1)},0}}{\hbar}\ln \lambda+ A_{\infty^{(1)},0} +O\left(\lambda^{-1}\right)\right),\cr
\mathcal{L}_{\boldsymbol{\alpha}}[\psi_2(\lambda,\hbar)]&\overset{\lambda\to \infty}{=}&-\left(\sum_{k=1}^{r_\infty-1} \frac{\alpha_{\infty^{(2)},k}}{k}\lambda^k +O(1)\right)\exp\left(-\frac{1}{\hbar}\sum_{k=1}^{r_\infty-1} \frac{t_{\infty^{(2)},k}}{k}\lambda^k -\frac{t_{\infty^{(2)},0}}{\hbar}\ln \lambda-\ln \lambda+ A_{\infty^{(2)},0} +O\left(\lambda^{-1}\right)\right),\cr
\mathcal{L}_{\boldsymbol{\alpha}}[\psi_i(\lambda,\hbar)]&\overset{\lambda\to X_s}{=}&\left(-\alpha_{X_s}\sum_{k=0}^{r_s-1}t_{X_s^{(i)},k} (\lambda-X_s)^{-k-1} -\sum_{k=1}^{r_s-1} \frac{\alpha_{X_s^{(i)},k}}{k}(\lambda-X_s)^{-k} +O(1)\right)\cr
&&\exp\left(-\frac{1}{\hbar}\sum_{k=1}^{r_s-1} \frac{t_{X_s^{(i)},k}}{k}(\lambda-X_s)^{-k} +\frac{t_{X_s^{(i)},0}}{\hbar}\ln(\lambda-X_s)+ A_{X_s^{(i)},0}+O\left(\lambda-X_s\right)\right).
\eea}
\normalsize{}

Thus, from \eqref{ExpressionA12A11} and the asymptotic expansions \eqref{Y1Y2} and \eqref{LPsi}, we deduce that
\bea \left[A_{\boldsymbol{\alpha}}(\lambda)\right]_{1,2}&\overset{\lambda\to \infty}{=}&\sum_{i=-1}^{r_\infty-3} \frac{\nu^{(\boldsymbol{\alpha})}_{\infty,i}}{\lambda^i} +O\left(\lambda^{-(r_\infty-2)}\right),\cr
\left[A_{\boldsymbol{\alpha}}(\lambda)\right]_{1,2}&\overset{\lambda\to X_s}{=}&\sum_{i=0}^{r_s-1} \nu^{(\boldsymbol{\alpha})}_{{X_s},i}(\lambda-X_s)^i +O\left((\lambda-X_s)^{r_s}\right)
\eea

For $r_\infty\ge2$, the coefficients
$\nu^{(\boldsymbol\alpha)}_{\infty,i}$ with
$-1\le i\le r_\infty-3$ are determined recursively below.
The first displayed formula applies for $r_\infty\ge2$, and the
formula for $\nu^{(\boldsymbol\alpha)}_{\infty,0}$ applies only for
$r_\infty\ge3$. For $r_\infty=1,2$, the missing coefficients are
instead determined by the augmented equations described in
Remark \ref{RemarkSpecialCases}.
\small{\bea \label{Defnuinftyk}
\nu^{(\boldsymbol{\alpha})}_{\infty,-1}&=&\frac{\alpha_{\infty^{(1)},r_\infty-1}-\alpha_{\infty^{(2)},r_\infty-1}}{(r_\infty-1)(t_{\infty^{(1)},r_\infty-1}-t_{\infty^{(2)},r_\infty-1})}\cr
\nu^{(\boldsymbol{\alpha})}_{\infty,0}&=&-\frac{(t_{\infty^{(1)},r_\infty-2}-t_{\infty^{(2)},r_\infty-2})}{(r_\infty-1)(t_{\infty^{(1)},r_\infty-1}-t_{\infty^{(2)},r_\infty-1})^2}(\alpha_{\infty^{(1)},r_\infty-1}-\alpha_{\infty^{(2)},r_\infty-1})\cr
&&+\frac{\alpha_{\infty^{(1)},r_\infty-2}-\alpha_{\infty^{(2)},r_\infty-2}}{(r_\infty-2)(t_{\infty^{(1)},r_\infty-1}-t_{\infty^{(2)},r_\infty-1})}\cr
\nu^{(\boldsymbol{\alpha})}_{\infty, r_\infty-2-k}
&=&\frac{1}{(t_{\infty^{(1)},r_\infty-1}-t_{\infty^{(2)},r_\infty-1})}\left(\frac{\alpha_{\infty^{(1)},k}-\alpha_{\infty^{(2)},k}}{k}-\sum_{i=-1}^{r_\infty-3-k}(t_{\infty^{(1)},k+i+1}-t_{\infty^{(2)},k+i+1})\nu^{(\boldsymbol{\alpha})}_{\infty,i}\right)\cr&&
\eea}
\normalsize{for} $k\in \llbracket 1,r_\infty-1\rrbracket$. Similarly, $\left(\nu^{(\boldsymbol{\alpha})}_{{X_s},i}\right)_{0\leq i\leq r_s-1}$ are defined recursively by
\small{\bea \label{Defnusk}\nu^{(\boldsymbol{\alpha})}_{{X_s},0}&=&-\alpha_{X_s}\cr
(t_{X_s^{(1)},r_s-1}-t_{X_s^{(2)},r_s-1})\,\nu^{(\boldsymbol{\alpha})}_{{X_s}, r_s-k}&=&-\frac{\alpha_{X_s^{(1)},k}-\alpha_{X_s^{(2)},k}}{k}-\alpha_{X_s}(t_{X_s^{(1)},k-1}-t_{X_s^{(2)},k-1})\cr
&&-\sum_{i=0}^{r_s-1-k}(t_{X_s^{(1)},k+i-1}-t_{X_s^{(2)},k+i-1})\nu^{(\boldsymbol{\alpha})}_{{X_s},i}\cr
&=&-\frac{\alpha_{X_s^{(1)},k}-\alpha_{X_s^{(2)},k}}{k}-\sum_{i=1}^{r_s-1-k}(t_{X_s^{(1)},k+i-1}-t_{X_s^{(2)},k+i-1})\nu^{(\boldsymbol{\alpha})}_{{X_s},i}
\eea}
\normalsize{for} $k\in \llbracket 1,r_s-1\rrbracket$.

In particular, equations \eqref{Defnuinftyk} and \eqref{Defnusk} may be rewritten in a matrix form giving Proposition \ref{PropAsymptoticExpansionA12}.

\section{Proof of Proposition \ref{PropA12Form}}\label{ProofEntryA12}
Since $\left[A_{\boldsymbol{\alpha}}(\lambda)\right]_{1,2}$ is a rational function of $\lambda$ with only possible poles in $\mathcal{R}\cup\{q_1,\dots,q_g\}$, the asymptotic expansion at each pole provided by Proposition \ref{PropAsymptoticExpansionA12} implies that

\beq \left[A_{\boldsymbol{\alpha}}(\lambda)\right]_{1,2}=\nu^{(\boldsymbol{\alpha})}_{\infty,-1}\lambda+\nu^{(\boldsymbol{\alpha})}_{\infty,0} + \sum_{j=1}^g \frac{\mu^{(\boldsymbol{\alpha})}_j}{\lambda-q_j}.\eeq
One then observes that the coefficients $\left(\mu^{(\boldsymbol{\alpha})}_j\right)_{1\leq j\leq g}$ are related to $\left(\nu^{(\boldsymbol{\alpha})}_{\infty,i}\right)_{-1\leq i\leq r_\infty-3}$ and $\left(\nu^{(\boldsymbol{\alpha})}_{{X_s},i}\right)_{1\leq i\leq r_s-1}$ through \eqref{Defnuinftyk} and \eqref{Defnusk}. We thus get

\bea \left[A_{\boldsymbol{\alpha}}(\lambda)\right]_{1,2}&=&\nu^{(\boldsymbol{\alpha})}_{\infty,-1}\lambda+\nu^{(\boldsymbol{\alpha})}_{\infty,0} + \sum_{j=1}^g \frac{\mu^{(\boldsymbol{\alpha})}_j}{\lambda-q_j}\cr
&\overset{\lambda\to \infty}{=}&\nu^{(\boldsymbol{\alpha})}_{\infty,-1}\lambda+\nu^{(\boldsymbol{\alpha})}_{\infty,0}+\sum_{k=1}^{\infty} \sum_{j=1}^g\frac{\mu^{(\boldsymbol{\alpha})}_j q_j^{k-1}}{\lambda^{k}}\cr
&\overset{\lambda\to X_s}{=}&\nu^{(\boldsymbol{\alpha})}_{\infty,-1}(\lambda-X_s) +\nu^{(\boldsymbol{\alpha})}_{\infty,-1}X_s+\nu^{(\boldsymbol{\alpha})}_{\infty,0}+ \sum_{k=0}^{\infty}\sum_{j=1}^g \frac{\mu^{(\boldsymbol{\alpha})}_j(-1)^k}{(X_s-q_j)^{k+1}}(\lambda-X_s)^k\cr
&\overset{\lambda\to X_s}{=}&\nu^{(\boldsymbol{\alpha})}_{\infty,-1}X_s+\nu^{(\boldsymbol{\alpha})}_{\infty,0}-\sum_{j=1}^g\frac{\mu^{(\boldsymbol{\alpha})}_j}{q_j-X_s}+\left(\nu^{(\boldsymbol{\alpha})}_{\infty,-1}-\sum_{j=1}^g \frac{\mu^{(\boldsymbol{\alpha})}_j}{(q_j-X_s)^{2}}\right)(\lambda-X_s)\cr
&&-\sum_{k=2}^{\infty}\sum_{j=1}^g\frac{\mu^{(\boldsymbol{\alpha})}_j}{(q_j-X_s)^{k+1}}(\lambda-X_s)^k
\eea 
so that, for all $s\in \llbracket 1,n\rrbracket$,
\bea \label{RelationNuMu} \forall\, k\in \llbracket 1, r_\infty-3\rrbracket\,:\, \sum_{j=1}^g \mu^{(\boldsymbol{\alpha})}_j q_j^{k-1}&=&\nu^{(\boldsymbol{\alpha})}_{\infty,k},\cr
 \sum_{j=1}^g\frac{\mu^{(\boldsymbol{\alpha})}_j}{q_j-X_s}&=&-\nu^{(\boldsymbol{\alpha})}_{{X_s},0}+\nu^{(\boldsymbol{\alpha})}_{\infty,0}+\nu^{(\boldsymbol{\alpha})}_{\infty,-1}X_s,\cr
\sum_{j=1}^g \frac{\mu^{(\boldsymbol{\alpha})}_j}{(q_j-X_s)^{2}}&=&-\nu^{(\boldsymbol{\alpha})}_{{X_s},1}+\nu^{(\boldsymbol{\alpha})}_{\infty,-1},\cr
\forall\, k\in \llbracket 2, r_s-1\rrbracket\,:\, \sum_{j=1}^g\frac{\mu^{(\boldsymbol{\alpha})}_j}{(q_j-X_s)^{k+1}}&=&-\nu^{(\boldsymbol{\alpha})}_{{X_s},k}.
\eea
For $r_\infty\ge3$, these are $g$ relations for $g$ unknowns
$\mu_j^{(\boldsymbol\alpha)}$ and give the square matrix
$\mathbf V$ of Proposition \ref{PropA12Form}. For $r_\infty=1,2$,
use instead the augmented systems including the additional infinity
coefficients described in Remark \ref{RemarkSpecialCases}.

\section{Proof of Proposition \ref{Propcalpha}}\label{AppendixA11}
We may perform the same kind of computations for $\left[A_{\boldsymbol{\alpha}}(\lambda)\right]_{1,1}$ starting from \eqref{ExpressionA12A11}. Note in particular that terms proportional to $\alpha_{X_s}$ cancel in the numerator because of the antisymmetry.  We get that
\beq \left[A_{\boldsymbol{\alpha}}(\lambda)\right]_{1,1}=\sum_{i=0}^{r_\infty-1}c^{(\boldsymbol{\alpha})}_{\infty,i}\lambda^i+\sum_{s=1}^n\sum_{i=1}^{r_s-1}c^{(\boldsymbol{\alpha})}_{{X_s},i}(\lambda-X_s)^{-i}+\sum_{j=1}^g\frac{\rho^{(\boldsymbol{\alpha})}_j}{\lambda-q_j}.\eeq
For any pole $p$ with $N=r_p-1\ge1$, put
$\Delta t_k=t_{p^{(1)},k}-t_{p^{(2)},k}$.
The coefficients are obtained successively for $i=N,N-1,\ldots,1$ from
\[
c_{p,i}^{(\boldsymbol\alpha)}
=\frac{1}{\Delta t_N}\left[
\sum_{k=i}^{N}
\frac{t_{p^{(2)},N+i-k}\alpha_{p^{(1)},k}
-t_{p^{(1)},N+i-k}\alpha_{p^{(2)},k}}{k}
-\sum_{j=i+1}^{N}\Delta t_{N+i-j}c_{p,j}^{(\boldsymbol\alpha)}
\right].
\]
Here $p=\infty$ or $p=X_s$, and all empty sums vanish. Equivalently,
\[
\sum_{j=i}^{N}\Delta t_{N+i-j}c_{p,j}^{(\boldsymbol\alpha)}
=\sum_{k=i}^{N}
\frac{t_{p^{(2)},N+i-k}\alpha_{p^{(1)},k}
-t_{p^{(1)},N+i-k}\alpha_{p^{(2)},k}}{k},
\qquad 1\le i\le N.
\]

Note that $c^{(\boldsymbol{\alpha})}_{\infty,0}$ is not determined but will play no role in the rest of the paper. The previous recursive relations may be rewritten in a matrix form giving Proposition \ref{Propcalpha}.
Finally, the coefficients $\left(\rho^{(\boldsymbol{\alpha})}_j\right)_{1\leq j\leq g}$ are obtained by looking at order $(\lambda-q_j)^{-3}$ of $\mathcal{L}_{\boldsymbol{\alpha}}[L_{2,1}(\lambda)]$.

\section{Proof of Theorem \ref{HamTheorem}}\label{AppendixA}
In this section we prove Theorem \ref{HamTheorem}. 

\subsection{Preliminary results}
We postpone the proof and start with the following lemma:

\begin{lemma}\label{PropsumC} For all $j\in \llbracket 1, g\rrbracket$:
\small{\bea 
&&\sum_{k=0}^{r_\infty-4}\sum_{i=1}^gH_{\infty,k}q_i^k\partial_{q_j} \mu^{\boldsymbol{(\alpha)}}_i+\sum_{s=1}^n\sum_{k=1}^{r_s}\sum_{i=1}^gH_{X_s,k}(q_i-X_s)^{-k}\partial_{q_j}\mu^{\boldsymbol{(\alpha)}}_i\cr
&&=-\mu^{\boldsymbol{(\alpha)}}_j\left(\sum_{k=0}^{r_\infty-4}kH_{\infty,k} q_j^{k-1}-\sum_{s=1}^n\sum_{k=1}^{r_s}kH_{X_s,k}(q_j-X_s)^{-k-1}\right)\cr
&&+  \delta_{r_\infty=2}\left(\partial_{q_j} \nu_{\infty,0}^{(\boldsymbol{\alpha})}\right) \left(\hbar \sum_{i=1}^g p_i-(t_{\infty^{(1)},1}t_{\infty^{(2)},0}+t_{\infty^{(2)},1}t_{\infty^{(1)},0}+\hbar t_{\infty^{(1)},1})\right)\cr
&&+\delta_{r_\infty=1}\left[ \left(\partial_{q_j} \nu_{\infty,-1}^{(\boldsymbol{\alpha})}\right)\left(\hbar \sum_{i=1}^g q_i p_i +\sum_{s=1}^n t_{X_s^{(1)},0}t_{X_s^{(2)},0}\delta_{r_s=1} -t_{\infty^{(1)},0}(t_{\infty^{(2)},0}+\hbar)\right) + \left(\partial_{q_j} \nu_{\infty,0}^{(\boldsymbol{\alpha})}\right)\left(\hbar \sum_{i=1}^g p_i\right) \right].\cr
&&
\eea}
\normalsize{}
\end{lemma}

\begin{proof}The proof follows from the expression relating $\left(\nu^{(\boldsymbol{\alpha})}_{p,k}\right)_{p,k}$ and $\left(\mu_j^{(\boldsymbol{\alpha})}\right)_{1\leq j\leq g}$ given by \eqref{RelationNuMu}. Taking the derivative with respect to $q_j$ and using the fact that the $\left(\nu^{(\boldsymbol{\alpha})}_{p,k}\right)$'s are independent of $q_j$, except for $\nu_{\infty,-1}^{(\boldsymbol{\alpha})}$ and $\nu_{\infty,0}^{(\boldsymbol{\alpha})}$ when $r_\infty\leq 2$, gives:
\bea  \forall\, k\in \llbracket 0, r_\infty-4\rrbracket\,:\, \sum_{i=1}^g (\partial_{q_j}\mu^{(\boldsymbol{\alpha})}_i) q_i^{k}&=&
-k \mu^{(\boldsymbol{\alpha})}_j q_j^{k-1},\cr
\forall\, k\in \llbracket  1, r_s\rrbracket\,:\, \sum_{i=1}^g\frac{\partial_{q_j}\mu^{(\boldsymbol{\alpha})}_i}{(q_i-X_s)^{k}}&=&k\frac{\mu^{(\boldsymbol{\alpha})}_j}{(q_j-X_s)^{k+1}}+\left(\partial_{q_j} \nu_{\infty,-1}^{(\boldsymbol{\alpha})}\right)\delta_{k=2} \cr
&&+\left(\partial_{q_j} \nu_{\infty,0}^{(\boldsymbol{\alpha})}+ X_s \partial_{q_j} \nu_{\infty,-1}^{(\boldsymbol{\alpha})} \right)\delta_{k=1}.
\eea
Thus
\bea&& \sum_{k=0}^{r_\infty-4}\sum_{i=1}^gH_{\infty,k}q_i^k\partial_{q_j} \mu^{\boldsymbol{(\alpha)}}_i+\sum_{s=1}^n\sum_{k=1}^{r_s}\sum_{i=1}^gH_{X_s,k}(q_i-X_s)^{-k}\partial_{q_j}\mu^{\boldsymbol{(\alpha)}}_i\cr
&&=-\mu^{(\boldsymbol{\alpha})}_j \left(\sum_{k=0}^{r_\infty-4} kH_{\infty,k}q_j^{k-1}- \sum_{s=1}^n\sum_{k=1}^{r_s}kH_{X_s,k} \frac{1}{(q_j-X_s)^{k+1}}\right)\cr
&&+\left(\partial_{q_j} \nu_{\infty,-1}^{(\boldsymbol{\alpha})}\right)\left(\sum_{s=1}^n\bigl(H_{X_s,2}\delta_{r_s\geq2}+X_sH_{X_s,1}\bigr)\right)+\left(\partial_{q_j} \nu_{\infty,0}^{(\boldsymbol{\alpha})}\right)\left(\sum_{s=1}^n H_{X_s,1}\right).
\eea
We now recall that $\nu_{\infty,-1}^{(\boldsymbol{\alpha})}$ and $\nu_{\infty,0}^{(\boldsymbol{\alpha})}$ only depend on $(q_j)_{1\leq j\leq g}$ when $r_\infty\leq 2$. In fact for $r_\infty=2$ only $\nu_{\infty,0}^{(\boldsymbol{\alpha})}$ depends on $(q_j)_{1\leq j\leq g}$. In these cases, the sums $\sum_{s=1}^n\bigl(H_{X_s,2}\delta_{r_s\geq2}+X_sH_{X_s,1}\bigr)$ and $\underset{s=1}{\overset{n}{\sum}} H_{X_s,1}$ are determined by \eqref{rinfty2Special} or \eqref{rinfty1Special}. Thus we get:
\bea&& \sum_{k=0}^{r_\infty-4}\sum_{i=1}^gH_{\infty,k}q_i^k\partial_{q_j} \mu^{\boldsymbol{(\alpha)}}_i+\sum_{s=1}^n\sum_{k=1}^{r_s}\sum_{i=1}^gH_{X_s,k}(q_i-X_s)^{-k}\partial_{q_j}\mu^{\boldsymbol{(\alpha)}}_i\cr
&&=-\mu^{(\boldsymbol{\alpha})}_j \left(\sum_{k=0}^{r_\infty-4} kH_{\infty,k}q_j^{k-1}- \sum_{s=1}^n\sum_{k=1}^{r_s}kH_{X_s,k} \frac{1}{(q_j-X_s)^{k+1}}\right)\cr
&&+  \delta_{r_\infty=2}\left(\partial_{q_j} \nu_{\infty,0}^{(\boldsymbol{\alpha})}\right) \left(\hbar \sum_{i=1}^g p_i-(t_{\infty^{(1)},1}t_{\infty^{(2)},0}+t_{\infty^{(2)},1}t_{\infty^{(1)},0}+\hbar t_{\infty^{(1)},1})\right)\cr
&&+\delta_{r_\infty=1}\Big[ \left(\partial_{q_j} \nu_{\infty,-1}^{(\boldsymbol{\alpha})}\right)\left(\hbar \sum_{i=1}^g q_i p_i +\sum_{s=1}^n t_{X_s^{(1)},0}t_{X_s^{(2)},0}\delta_{r_s=1} -t_{\infty^{(1)},0}(t_{\infty^{(2)},0}+\hbar)\right) \cr
&&+ \left(\partial_{q_j} \nu_{\infty,0}^{(\boldsymbol{\alpha})}\right)\left(\hbar \sum_{i=1}^g p_i\right) \Big]\cr
&&
\eea
so that the lemma is proved.
\end{proof}

We may now provide an alternative expression for $\mathcal{L}_{\boldsymbol{\alpha}}[p_j]$:

\begin{proposition}\label{PropLpjbis} Let $j\in \llbracket 1,g\rrbracket$, we have an alternative expression for $\mathcal{L}_{\boldsymbol{\alpha}}[p_j]$:
\small{\bea \label{Lpjbis} \mathcal{L}_{\boldsymbol{\alpha}}[p_j]&=&\hbar \sum_{i\neq j}\frac{(\mu^{(\boldsymbol{\alpha})}_i+\mu^{(\boldsymbol{\alpha})}_j)(p_i-p_j)}{(q_j-q_i)^2} +\frac{\hbar}{2}\displaystyle{\sum_{\substack{(r,s)\in \llbracket 1,g\rrbracket^2 \\ r\neq s }}} \frac{(p_s-p_r)(\partial_{q_j}\mu^{\boldsymbol{(\alpha)}}_r+\partial_{q_j}\mu^{\boldsymbol{(\alpha)}}_s)}{q_s-q_r}\cr
&&-\mu^{(\boldsymbol{\alpha})}_j\left(\td{P}_2'(q_j)-p_j P_1'(q_j)-\hbar p_j\sum_{s=1}^n \frac{r_s}{(q_j-X_s)^2} +\hbar (r_\infty-3)t_{\infty^{(1)},r_\infty-1}q_j^{r_\infty-4}\delta_{r_\infty\geq 3}\right)\cr
&&+\hbar \nu^{(\boldsymbol{\alpha})}_{\infty,-1}p_j+\hbar \sum_{k=1}^{r_\infty-1}kc^{(\boldsymbol{\alpha})}_{\infty,k}q_j^{k-1}-\hbar\sum_{s=1}^n\sum_{k=1}^{r_s-1}kc^{(\boldsymbol{\alpha})}_{{X_s},k}(q_j-X_s)^{-k-1}\cr
&&-\sum_{r=1}^g (\partial_{q_j} \mu^{(\boldsymbol{\alpha})}_r)\left( \td{P}_2(q_r)+ p_r^2-P_1(q_r)p_r +\hbar p_r \sum_{s=1}^n \frac{r_s}{q_r-X_s} +\hbar t_{\infty^{(1)},r_\infty-1}q_r^{r_\infty-3}\delta_{r_\infty\geq 3}\right)\cr
&&+  \delta_{r_\infty=2}\left(\partial_{q_j} \nu_{\infty,0}^{(\boldsymbol{\alpha})}\right) \left(\hbar \sum_{i=1}^g p_i-(t_{\infty^{(1)},1}t_{\infty^{(2)},0}+t_{\infty^{(2)},1}t_{\infty^{(1)},0}+\hbar t_{\infty^{(1)},1})\right)\cr
&&+\delta_{r_\infty=1}\Big[ \left(\partial_{q_j} \nu_{\infty,-1}^{(\boldsymbol{\alpha})}\right)\left(\hbar \sum_{i=1}^g q_i p_i +\sum_{s=1}^n t_{X_s^{(1)},0}t_{X_s^{(2)},0}\delta_{r_s=1} -t_{\infty^{(1)},0}(t_{\infty^{(2)},0}+\hbar)\right) \cr
&&+ \left(\partial_{q_j} \nu_{\infty,0}^{(\boldsymbol{\alpha})}\right)\left(\hbar \sum_{i=1}^g p_i\right) \Big].
\eea}
\normalsize{} 
\end{proposition}

\begin{proof}Using Lemma \ref{PropsumC} we get that the expression \eqref{Lpj} for $\mathcal{L}[p_j]$ becomes:
\footnotesize{\bea \mathcal{L}_{\boldsymbol{\alpha}}[p_j]&=&\hbar \sum_{i\neq j}\frac{(\mu^{(\boldsymbol{\alpha})}_i+\mu^{(\boldsymbol{\alpha})}_j)(p_i-p_j)}{(q_j-q_i)^2} \cr
&&+\mu^{(\boldsymbol{\alpha})}_j\left(p_j P_1'(q_j)+\hbar p_j\sum_{s=1}^n \frac{r_s}{(q_j-X_s)^2}-\td{P}_2'(q_j)  -\hbar (r_\infty-3)t_{\infty^{(1)},r_\infty-1}q_j^{r_\infty-4}\delta_{r_\infty\geq 3}\right)\cr
&&+\hbar \nu^{(\boldsymbol{\alpha})}_{\infty,-1}p_j+\hbar \sum_{k=1}^{r_\infty-1}kc^{(\boldsymbol{\alpha})}_{\infty,k}q_j^{k-1}-\hbar\sum_{s=1}^n\sum_{k=1}^{r_s-1}kc^{(\boldsymbol{\alpha})}_{{X_s},k}(q_j-X_s)^{-k-1}\cr
&&-\sum_{i=1}^g (\partial_{q_j}\mu^{\boldsymbol{(\alpha)}}_i)\left(\sum_{k=0}^{r_\infty-4}H_{\infty,k}q_i^k +\sum_{s=1}^n\sum_{k=1}^{r_s}H_{X_s,k}(q_i-X_s)^{-k}\right)\cr
&&+  \delta_{r_\infty=2}\left(\partial_{q_j} \nu_{\infty,0}^{(\boldsymbol{\alpha})}\right) \left(\hbar \sum_{i=1}^g p_i-(t_{\infty^{(1)},1}t_{\infty^{(2)},0}+t_{\infty^{(2)},1}t_{\infty^{(1)},0}+\hbar t_{\infty^{(1)},1})\right)\cr
&&+\delta_{r_\infty=1}\Big[ \left(\partial_{q_j} \nu_{\infty,-1}^{(\boldsymbol{\alpha})}\right)\left(\hbar \sum_{i=1}^g q_i p_i +\sum_{s=1}^n t_{X_s^{(1)},0}t_{X_s^{(2)},0}\delta_{r_s=1} -t_{\infty^{(1)},0}(t_{\infty^{(2)},0}+\hbar)\right) \cr
&&+ \left(\partial_{q_j} \nu_{\infty,0}^{(\boldsymbol{\alpha})}\right)\left(\hbar \sum_{i=1}^g p_i\right) \Big].
\eea}
\normalsize{We} now use \eqref{DefCi} to get
\footnotesize{\bea \mathcal{L}_{\boldsymbol{\alpha}}[p_j]&=&\hbar \sum_{i\neq j}\frac{(\mu^{(\boldsymbol{\alpha})}_i+\mu^{(\boldsymbol{\alpha})}_j)(p_i-p_j)}{(q_j-q_i)^2} \cr
&&+\mu^{(\boldsymbol{\alpha})}_j\left(p_j P_1'(q_j)+\hbar p_j\sum_{s=1}^n \frac{r_s}{(q_j-X_s)^2}-\td{P}_2'(q_j)  -\hbar (r_\infty-3)t_{\infty^{(1)},r_\infty-1}q_j^{r_\infty-4}\delta_{r_\infty\geq 3}\right)\cr
&&+\hbar \nu^{(\boldsymbol{\alpha})}_{\infty,-1}p_j+\hbar \sum_{k=1}^{r_\infty-1}kc^{(\boldsymbol{\alpha})}_{\infty,k}q_j^{k-1}-\hbar\sum_{s=1}^n\sum_{k=1}^{r_s-1}kc^{(\boldsymbol{\alpha})}_{{X_s},k}(q_j-X_s)^{-k-1}\cr
&&-\sum_{i=1}^g (\partial_{q_j}\mu^{\boldsymbol{(\alpha)}}_i)\Big[p_i^2-P_1(q_i)p_i +\hbar p_i \sum_{s=1}^n \frac{r_s}{q_i-X_s}+\td{P}_2(q_i)\cr
&&+\hbar \sum_{r\neq i}\frac{p_r-p_i}{q_i-q_r}+\hbar t_{\infty^{(1)},r_\infty-1}q_i^{r_\infty-3}\delta_{r_\infty\geq 3}\Big]\cr
\cr
&&+  \delta_{r_\infty=2}\left(\partial_{q_j} \nu_{\infty,0}^{(\boldsymbol{\alpha})}\right) \left(\hbar \sum_{i=1}^g p_i-(t_{\infty^{(1)},1}t_{\infty^{(2)},0}+t_{\infty^{(2)},1}t_{\infty^{(1)},0}+\hbar t_{\infty^{(1)},1})\right)\cr
&&+\delta_{r_\infty=1}\Big[ \left(\partial_{q_j} \nu_{\infty,-1}^{(\boldsymbol{\alpha})}\right)\left(\hbar \sum_{i=1}^g q_i p_i +\sum_{s=1}^n t_{X_s^{(1)},0}t_{X_s^{(2)},0}\delta_{r_s=1} -t_{\infty^{(1)},0}(t_{\infty^{(2)},0}+\hbar)\right) \cr
&&+ \left(\partial_{q_j} \nu_{\infty,0}^{(\boldsymbol{\alpha})}\right)\left(\hbar \sum_{i=1}^g p_i\right) \Big]\cr
&=&\hbar \sum_{i\neq j}\frac{(\mu^{(\boldsymbol{\alpha})}_i+\mu^{(\boldsymbol{\alpha})}_j)(p_i-p_j)}{(q_j-q_i)^2} \cr
&&+\mu^{(\boldsymbol{\alpha})}_j\left(p_j P_1'(q_j)+\hbar p_j\sum_{s=1}^n \frac{r_s}{(q_j-X_s)^2}-\td{P}_2'(q_j)  -\hbar (r_\infty-3)t_{\infty^{(1)},r_\infty-1}q_j^{r_\infty-4}\delta_{r_\infty\geq 3}\right)\cr
&&+\hbar \nu^{(\boldsymbol{\alpha})}_{\infty,-1}p_j+\hbar \sum_{k=1}^{r_\infty-1}kc^{(\boldsymbol{\alpha})}_{\infty,k}q_j^{k-1}-\hbar\sum_{s=1}^n\sum_{k=1}^{r_s-1}kc^{(\boldsymbol{\alpha})}_{{X_s},k}(q_j-X_s)^{-k-1}\cr
&&-\sum_{i=1}^g (\partial_{q_j}\mu^{\boldsymbol{(\alpha)}}_i)\left(p_i^2-P_1(q_i)p_i+\hbar p_i \sum_{s=1}^n \frac{r_s}{q_i-X_s}+\td{P}_2(q_i)+\hbar t_{\infty^{(1)},r_\infty-1}q_i^{r_\infty-3}\delta_{r_\infty\geq 3}\right)\cr
&&+  \delta_{r_\infty=2}\left(\partial_{q_j} \nu_{\infty,0}^{(\boldsymbol{\alpha})}\right) \left(\hbar \sum_{i=1}^g p_i-(t_{\infty^{(1)},1}t_{\infty^{(2)},0}+t_{\infty^{(2)},1}t_{\infty^{(1)},0}+\hbar t_{\infty^{(1)},1})\right)\cr
&&+\delta_{r_\infty=1}\Big[ \left(\partial_{q_j} \nu_{\infty,-1}^{(\boldsymbol{\alpha})}\right)\left(\hbar \sum_{i=1}^g q_i p_i +\sum_{s=1}^n t_{X_s^{(1)},0}t_{X_s^{(2)},0}\delta_{r_s=1} -t_{\infty^{(1)},0}(t_{\infty^{(2)},0}+\hbar)\right) \cr
&&+ \left(\partial_{q_j} \nu_{\infty,0}^{(\boldsymbol{\alpha})}\right)\left(\hbar \sum_{i=1}^g p_i\right) \Big]\cr
&&+\hbar\sum_{i=1}^g (\partial_{q_j}\mu^{\boldsymbol{(\alpha)}}_i)\sum_{r\neq i}\frac{p_r-p_i}{q_r-q_i}.
\eea}
\normalsize{In} the final double sum, the quotient
$(p_r-p_i)/(q_r-q_i)$ is symmetric in $i,r$. Hence the antisymmetric
part of the coefficient cancels, and
\[
\sum_i\sum_{r\ne i}\partial_{q_j}\mu_i^{(\boldsymbol\alpha)}
\frac{p_r-p_i}{q_r-q_i}
=\frac12\sum_i\sum_{r\ne i}
\bigl(\partial_{q_j}\mu_i^{(\boldsymbol\alpha)}
+\partial_{q_j}\mu_r^{(\boldsymbol\alpha)}\bigr)
\frac{p_r-p_i}{q_r-q_i}.
\]
Therefore,
\footnotesize{\bea \mathcal{L}_{\boldsymbol{\alpha}}[p_j]&=&\hbar \sum_{i\neq j}\frac{(\mu^{(\boldsymbol{\alpha})}_i+\mu^{(\boldsymbol{\alpha})}_j)(p_i-p_j)}{(q_j-q_i)^2} \cr
&&-\mu^{(\boldsymbol{\alpha})}_j\left(\td{P}_2'(q_j)-p_j P_1'(q_j)-\hbar p_j\sum_{s=1}^n \frac{r_s}{(q_j-X_s)^2}+\hbar (r_\infty-3)t_{\infty^{(1)},r_\infty-1}q_j^{r_\infty-4}\delta_{r_\infty\geq 3}\right)\cr
&&+\hbar \nu^{(\boldsymbol{\alpha})}_{\infty,-1}p_j+\hbar \sum_{k=1}^{r_\infty-1}kc^{(\boldsymbol{\alpha})}_{\infty,k}q_j^{k-1}-\hbar\sum_{s=1}^n\sum_{k=1}^{r_s-1}kc^{(\boldsymbol{\alpha})}_{{X_s},k}(q_j-X_s)^{-k-1}\cr
&&-\sum_{i=1}^g (\partial_{q_j}\mu^{\boldsymbol{(\alpha)}}_i)\left(p_i^2-P_1(q_i)p_i+\hbar p_i \sum_{s=1}^n \frac{r_s}{q_i-X_s}+\td{P}_2(q_i)+\hbar t_{\infty^{(1)},r_\infty-1}q_i^{r_\infty-3}\delta_{r_\infty\geq 3}\right)\cr
&&+\frac{\hbar}{2}\sum_{i=1}^g\sum_{r\neq i} \frac{(p_r-p_i)(\partial_{q_j}\mu^{\boldsymbol{(\alpha)}}_i+\partial_{q_j}\mu^{\boldsymbol{(\alpha)}}_r)}{q_r-q_i}\cr
&&+  \delta_{r_\infty=2}\left(\partial_{q_j} \nu_{\infty,0}^{(\boldsymbol{\alpha})}\right) \left(\hbar \sum_{i=1}^g p_i-(t_{\infty^{(1)},1}t_{\infty^{(2)},0}+t_{\infty^{(2)},1}t_{\infty^{(1)},0}+\hbar t_{\infty^{(1)},1})\right)\cr
&&+\delta_{r_\infty=1}\Big[ \left(\partial_{q_j} \nu_{\infty,-1}^{(\boldsymbol{\alpha})}\right)\left(\hbar \sum_{i=1}^g q_i p_i +\sum_{s=1}^n t_{X_s^{(1)},0}t_{X_s^{(2)},0}\delta_{r_s=1} -t_{\infty^{(1)},0}(t_{\infty^{(2)},0}+\hbar)\right)\cr
&& + \left(\partial_{q_j} \nu_{\infty,0}^{(\boldsymbol{\alpha})}\right)\left(\hbar \sum_{i=1}^g p_i\right) \Big].
\eea}
\normalsize{}
\end{proof}

\subsection{Proof of the Theorem \ref{HamTheorem}}

We may now proceed to the proof of Theorem \ref{HamTheorem}. We recall that the Hamiltonian is given by:

\bea\label{HamComputation}&&\text{Ham}^{(\boldsymbol{\alpha})}(\mathbf{q},\mathbf{p})=-\frac{\hbar}{2}\displaystyle{\sum_{\substack{(i,j)\in \llbracket 1,g\rrbracket^2 \\ i\neq j }}} \frac{(\mu^{(\boldsymbol{\alpha})}_i+\mu^{(\boldsymbol{\alpha})}_j)(p_i-p_j)}{q_i-q_j} -\hbar \sum_{j=1}^{g} (\nu^{(\boldsymbol{\alpha})}_{\infty,0} p_j+\nu^{(\boldsymbol{\alpha})}_{\infty,-1}q_jp_j) \cr
&&+\sum_{j=1}^{g}\mu^{(\boldsymbol{\alpha})}_j\left[p_j^2-P_1(q_j)p_j +\hbar p_j \sum_{s=1}^n \frac{r_s}{q_j-X_s}+\td{P}_2(q_j)+\hbar t_{\infty^{(1)},r_\infty-1}q_j^{r_\infty-3}\delta_{r_\infty\geq 3}\right]\cr
&&-\hbar \sum_{j=1}^g\left[\sum_{k=1}^{r_\infty-1}c^{(\boldsymbol{\alpha})}_{\infty,k}q_j^{k}+\sum_{s=1}^n\sum_{k=1}^{r_s-1}c^{(\boldsymbol{\alpha})}_{{X_s},k}(q_j-X_s)^{-k}\right]\cr
&&+  \delta_{r_\infty=2}(t_{\infty^{(1)},1}t_{\infty^{(2)},0}+t_{\infty^{(2)},1}t_{\infty^{(1)},0}+\hbar t_{\infty^{(1)},1}) \nu_{\infty,0}^{(\boldsymbol{\alpha})}\cr
&&-\delta_{r_\infty=1}\left(\sum_{s=1}^n t_{X_s^{(1)},0}t_{X_s^{(2)},0}\delta_{r_s=1} -t_{\infty^{(1)},0}(t_{\infty^{(2)},0}+\hbar)\right)\nu_{\infty,-1}^{(\boldsymbol{\alpha})}.\cr
&&
\eea

It is a straightforward computation from \eqref{DefHam} and from the fact that the $\left(\nu^{(\boldsymbol{\alpha})}_{p,k}, c^{(\boldsymbol{\alpha})}_{p,k}\right)_{p,k}$ are independent of $q_j$ (except for $r_\infty=2$ where $\nu_{\infty,0}^{(\boldsymbol{\alpha})}$ depends on $(q_j)_{1\leq j\leq g}$ and for $r_\infty=1$ where both $\nu_{\infty,0}^{(\boldsymbol{\alpha})}$ and $\nu_{\infty,-1}^{(\boldsymbol{\alpha})}$ depend on $(q_j)_{1\leq j\leq g}$) to get that
\small{\bea -\frac{\partial \text{Ham}^{(\boldsymbol{\alpha})}(\mathbf{q},\mathbf{p})}{\partial q_j}&=&\hbar\displaystyle{\sum_{\substack{i\in \llbracket 1,g\rrbracket \\ i\neq j }}} \frac{(\mu^{(\boldsymbol{\alpha})}_i+\mu^{(\boldsymbol{\alpha})}_j)(p_i-p_j)}{(q_i-q_j)^2}\cr
&&+\frac{\hbar}{2}\displaystyle{\sum_{\substack{(r,s)\in \llbracket 1,g\rrbracket^2 \\ r\neq s }}} \frac{(\partial_{q_j}\mu^{(\boldsymbol{\alpha})}_r+\partial_{q_j}\mu^{(\boldsymbol{\alpha})}_s)(p_r-p_s)}{q_r-q_s}+\hbar \nu^{(\boldsymbol{\alpha})}_{\infty,-1}p_j \cr
&&-\sum_{i=1}^{g}\partial_{q_j}(\mu^{(\boldsymbol{\alpha})}_i)\left[p_i^2-P_1(q_i)p_i+\hbar p_i \sum_{s=1}^n \frac{r_s}{q_i-X_s}+\td{P}_2(q_i) +\hbar t_{\infty^{(1)},r_\infty-1}q_i^{r_\infty-3}\delta_{r_\infty\geq 3}\right]\cr
&&-\mu^{(\boldsymbol{\alpha})}_j\left[\td{P}_2'(q_j)-p_j P_1'(q_j)-\hbar p_j \sum_{s=1}^n \frac{r_s}{(q_j-X_s)^2} +\hbar t_{\infty^{(1)},r_\infty-1}(r_\infty-3)q_j^{r_\infty-4}\delta_{r_\infty\geq 3}\right]\cr
&&+\hbar \left[\sum_{k=1}^{r_\infty-1}kc^{(\boldsymbol{\alpha})}_{\infty,k}q_j^{k-1}-\sum_{s=1}^n\sum_{k=1}^{r_s-1}kc^{(\boldsymbol{\alpha})}_{{X_s},k}(q_j-X_s)^{-k-1}\right]\cr
&&+  \delta_{r_\infty=2}\left(\hbar \sum_{i=1}^g p_i-(t_{\infty^{(1)},1}t_{\infty^{(2)},0}+t_{\infty^{(2)},1}t_{\infty^{(1)},0}+\hbar t_{\infty^{(1)},1})\right) \partial_{q_j}\nu_{\infty,0}^{(\boldsymbol{\alpha})}\cr
&&+\delta_{r_\infty=1}\Big[\left(\hbar \sum_{i=1}^g q_ip_i+ \sum_{s=1}^n t_{X_s^{(1)},0}t_{X_s^{(2)},0}\delta_{r_s=1} -t_{\infty^{(1)},0}(t_{\infty^{(2)},0}+\hbar)\right)\partial_{q_j}\nu_{\infty,-1}^{(\boldsymbol{\alpha})}\cr
&&+\hbar \left(\sum_{i=1}^g p_i\right)\partial_{q_j}\nu_{\infty,0}^{(\boldsymbol{\alpha})} \Big] \cr
&\overset{\text{\eqref{Lpjbis}}}{=}&\mathcal{L}_{\boldsymbol{\alpha}}[p_j].
\eea}
\normalsize{Similarly} a straightforward computation using the fact that $\left(\nu^{(\boldsymbol{\alpha})}_{p,k}\right)_{p,k}$ (including $\nu_{\infty,-1}^{(\boldsymbol{\alpha})}$ and $\nu_{\infty,0}^{(\boldsymbol{\alpha})}$ when $r_\infty\leq 2$) and $\left(c^{(\boldsymbol{\alpha})}_{p,k}\right)_{p,k}$ and $\left(\mu^{(\boldsymbol{\alpha})}_i\right)_{1\leq i\leq g}$ are independent of $p_j$ gives:
\small{\beq \frac{\partial \text{Ham}^{(\boldsymbol{\alpha})}(\mathbf{q},\mathbf{p})}{\partial p_j}=-\hbar\displaystyle{\sum_{\substack{i\in \llbracket 1,g\rrbracket \\ i\neq j }}} \frac{\mu^{(\boldsymbol{\alpha})}_i+\mu^{(\boldsymbol{\alpha})}_j}{q_j-q_i} -\hbar \nu^{(\boldsymbol{\alpha})}_{\infty,0} -\hbar\nu^{(\boldsymbol{\alpha})}_{\infty,-1}q_j +\mu^{(\boldsymbol{\alpha})}_j\left(2p_j -P_1(q_j) +\hbar \sum_{s=1}^n \frac{r_s}{q_j-X_s}\right)
\eeq}
\normalsize{which} is exactly $\mathcal{L}_{\boldsymbol{\alpha}}[q_j]$ given by \eqref{Lqj}.

\medskip

The last step is to verify that from Propositions \ref{PropA12Form} and \ref{PropDefCi2}:
\small{\bea &&\sum_{j=1}^g \mu_j^{\boldsymbol{(\alpha)}}\left(p_j^2- P_1(q_j)p_j + p_j\underset{s=1}{\overset{n}{\sum}} \frac{\hbar r_s}{q_j-X_s}+\td{P}_2(q_j)+\hbar \underset{i\neq j}{\sum}\frac{p_i-p_j}{q_j-q_i}+\hbar t_{\infty^{(1)},r_\infty-1}q_j^{r_\infty-3}\delta_{r_\infty\geq 3}\right) \cr
&&=\begin{pmatrix}\mu_1^{\boldsymbol{(\alpha)}},\dots,\mu_{g} ^{\boldsymbol{(\alpha)}}\end{pmatrix} \begin{pmatrix} p_1^2- P_1(q_1)p_1 + p_1\underset{s=1}{\overset{n}{\sum}} \frac{\hbar r_s}{q_1-X_s}+\td{P}_2(q_1)+\hbar \underset{i\neq 1}{\sum}\frac{p_i-p_1}{q_1-q_i}+\hbar t_{\infty^{(1)},r_\infty-1}q_1^{r_\infty-3}\delta_{r_\infty\geq 3}\\
\vdots\\
\vdots\\
p_g^2- P_1(q_g)p_g + p_g\underset{s=1}{\overset{n}{\sum}} \frac{\hbar r_s}{q_g-X_s}+\td{P}_2(q_g)+\hbar \underset{i\neq g}{\sum}\frac{p_i-p_g}{q_g-q_i}+\hbar t_{\infty^{(1)},r_\infty-1}q_g^{r_\infty-3}\delta_{r_\infty\geq 3}\end{pmatrix}\cr
&&=\begin{pmatrix}\mu_1^{\boldsymbol{(\alpha)}},\dots,\mu_{g} ^{\boldsymbol{(\alpha)}}\end{pmatrix}\mathbf{V}^t\begin{pmatrix} \mathbf{H}_\infty\\\mathbf{H}_{X_1}\\ \vdots\\ \mathbf{H}_{X_n}\end{pmatrix}=\begin{pmatrix}(\boldsymbol{\nu}_{\infty}^{\boldsymbol{(\alpha)}})^t, (\boldsymbol{\nu}_{X_1}^{\boldsymbol{(\alpha)}})^t,\dots,(\boldsymbol{\nu}_{{X_n}} ^{\boldsymbol{(\alpha)}})^t\end{pmatrix}\begin{pmatrix} \mathbf{H}_\infty\\\mathbf{H}_{X_1}\\ \vdots\\ \mathbf{H}_{X_n}\end{pmatrix}\cr
&&=\sum_{k=0}^{r_\infty-4} \nu_{\infty,k+1}^{\boldsymbol{(\alpha)}}H_{\infty,k}-\sum_{s=1}^n\sum_{k=2}^{r_s}\nu_{{X_s},k-1}^{\boldsymbol{(\alpha)}}H_{X_s,k}-\sum_{s=1}^n \nu_{{X_s},0}^{\boldsymbol{(\alpha)}}H_{X_s,1}\cr
&&+\nu_{\infty,-1}^{\boldsymbol{(\alpha)}}\sum_{s=1}^n\left(X_s H_{X_s,1}+H_{X_s,2}\delta_{r_s\geq 2}\right)+\nu_{\infty,0}^{\boldsymbol{(\alpha)}}\sum_{s=1}^n H_{X_s,1}\cr
&&
\eea}
\normalsize{so} that \eqref{HamComputation} becomes (using \eqref{ConditionsAddrinftyequal2} and \eqref{ConditionsAddrinftyequal1})
\small{\bea\text{Ham}^{(\boldsymbol{\alpha})}(\mathbf{q},\mathbf{p})&=&\sum_{k=0}^{r_\infty-4} \nu_{\infty,k+1}^{\boldsymbol{(\alpha)}}H_{\infty,k}-\sum_{s=1}^n\sum_{k=2}^{r_s}\nu_{{X_s},k-1}^{\boldsymbol{(\alpha)}}H_{X_s,k}+\sum_{s=1}^n \alpha_{X_s}H_{X_s,1}\cr
&&-\hbar \sum_{j=1}^g\left[\sum_{k=1}^{r_\infty-1}c^{(\boldsymbol{\alpha})}_{\infty,k}q_j^{k}+\sum_{s=1}^n\sum_{k=1}^{r_s-1}c^{(\boldsymbol{\alpha})}_{{X_s},k}(q_j-X_s)^{-k}\right]\cr
&&+\nu_{\infty,-1}^{\boldsymbol{(\alpha)}}\sum_{s=1}^n\left(X_s H_{X_s,1}+H_{X_s,2}\delta_{r_s\geq 2}\right)+\nu_{\infty,0}^{\boldsymbol{(\alpha)}}\sum_{s=1}^n H_{X_s,1}\cr
&&-  \delta_{r_\infty\in \{1,2\}}\left(\sum_{s=1}^n H_{X_s,1}-\hbar \sum_{j=1}^g p_j\right) \nu_{\infty,0}^{(\boldsymbol{\alpha})}\cr
&&-\delta_{r_\infty=1}\left(\sum_{s=1}^n X_s H_{X_s,1}+\sum_{s=1}^n H_{X_s,2}\delta_{r_s\geq 2} -\hbar \sum_{j=1}^g q_j p_j\right)\nu_{\infty,-1}^{(\boldsymbol{\alpha})}\cr
 &&-\hbar \nu^{(\boldsymbol{\alpha})}_{\infty,0}\sum_{j=1}^{g} p_j-\hbar \nu^{(\boldsymbol{\alpha})}_{\infty,-1}\sum_{j=1}^g q_jp_j.
\eea}
\normalsize{}

\section{Proof of Proposition \ref{TrivialSubspace}}\label{AppendixB}
Let us take $j\in \llbracket 1, r_\infty-1\rrbracket$ and consider $\mathcal{L}_{\mathbf{v}_{\infty,j}}$. The r.h.s. of \eqref{RelationNuAlphaInfty} is null and the r.h.s of \eqref{RelationNuAlphas} is trivially null too. Thus, all $(\nu^{(\mathbf{v}_{\infty,j})}_{\infty,k})_{-1\leq k\leq r_\infty-3}$ and all $(\nu^{(\mathbf{v}_{\infty,j})}_{{X_s},k})_{1\leq s\leq n, 1\leq k\leq r_s-1}$ are vanishing. In the same way, the r.h.s. of \eqref{Relationckalphas} is trivially vanishing so that for all $s\in \llbracket 1,n\rrbracket$, all $(c^{(\mathbf{v}_{\infty,j})}_{{X_s},k})_{1\leq k\leq r_s-1}$ are vanishing. On the contrary, the r.h.s. of \eqref{Relationckalphainfty} is non-zero. More precisely, the r.h.s. of \eqref{Relationckalphainfty} is $\left(0,\dots,0, \frac{t_{\infty^{(2)},r_\infty-1}-t_{\infty^{(1)},r_\infty-1}}{j},\dots, \frac{t_{\infty^{(2)},r_\infty-j}-t_{\infty^{(1)},r_\infty-j}}{j}\right)^t$ which is precisely $-\frac{1}{j}$ times the $(r_\infty-j)^{\text{th}}$ column of $M_\infty$ given by \eqref{MatrixMInfty}. Hence, $(c^{(\mathbf{v}_{\infty,j})}_{\infty,r_\infty-1},\dots,c^{(\mathbf{v}_{\infty,j})}_{\infty,1})^t=-\frac{1}{j}\mathbf{e}_{r_\infty-j}$, i.e. $c^{(\mathbf{v}_{\infty,j})}_{\infty,i}= -\frac{1}{j} \delta_{i,j}$ for all $i\in \llbracket 1,r_\infty-1\rrbracket$.

\medskip

Let us now take $s_0\in \llbracket 1,n\rrbracket$ and $j\in \llbracket 1, r_{s_0}-1\rrbracket$ and consider $\mathcal{L}_{\mathbf{v}_{X_{s_0},j}}$. The r.h.s. of\eqref{RelationNuAlphaInfty} and \eqref{RelationNuAlphas} are trivially null too. Thus, all $(\nu^{(\mathbf{v}_{X_{s_0},j})}_{\infty,k})_{-1\leq k\leq r_\infty-3}$ and all $(\nu^{(\mathbf{v}_{X_{s_0},j})}_{{X_s},k})_{1\leq s\leq n, 1\leq k\leq r_s-1}$ are vanishing. In the same way, the r.h.s. of \eqref{Relationckalphainfty} is vanishing so that all $(c^{(\mathbf{v}_{X_{s_0},j})}_{\infty,k})_{1\leq k\leq r_\infty-1}$ are vanishing. Similarly, the r.h.s. of \eqref{Relationckalphas} is vanishing for $s\neq s_0$ so that all $(c^{(\mathbf{v}_{X_{s_0},j})}_{{X_s},k})_{1\leq k\leq r_s-1}$ are vanishing for $s\neq s_0$. On the contrary, the r.h.s. of \eqref{Relationckalphas} is nonzero for $s=s_0$ and is given by  $\left(0,\dots,0, \frac{t_{X_{s_0}^{(2)},r_{s_0}-1}-t_{X_{s_0}^{(1)},r_{s_0}-1}}{j},\dots, \frac{t_{X_{s_0}^{(2)},r_{s_0}-j}-t_{X_{s_0}^{(1)},r_{s_0}-j}}{j}\right)^t$ which is precisely $-\frac{1}{j}$ times the $(r_{s_0}-j)^{\text{th}}$ column of $M_{s_0}$ given by \eqref{MatrixMs}. Hence, $\left(c^{(\mathbf{v}_{X_{s_0},j})}_{{X_{s_0}},r_{s_0}-1},\dots,c^{(\mathbf{v}_{X_{s_0},j})}_{{X_{s_0}},1}\right)^t=-\frac{1}{j}\mathbf{e}_{r_{s_0}-j}$, i.e. $c^{(\mathbf{v}_{X_{s_0},j})}_{{X_{s_0}},i}= -\frac{1}{j} \delta_{i,j}$ for all $i\in \llbracket 1,r_{s_0}-1\rrbracket$.

\section{Proof of Proposition \ref{TrivialSubspace2}}\label{AppendixC}
\subsection{The case of $\mathbf{u}_{\infty,j}$}
Let $j\in \llbracket 1,r_\infty-1\rrbracket$ and we consider $\mathcal{L}_{\mathbf{u}_{\infty,j}}$. For $i\in \llbracket 1,2\rrbracket$, it corresponds to
\beq \left(\alpha_{\infty^{(i)},1},\dots,\alpha_{\infty^{(i)},r_\infty-1}\right)^t=\left(t_{\infty^{(i)},r_\infty-j},\dots,j t_{\infty^{(i)},r_\infty-1},0\dots,0\right)^t\eeq
\sloppy{so that the r.h.s. of \eqref{RelationNuAlphaInfty} is $\left(0\dots,0,(t_{\infty^{(1)},r_\infty-1}-t_{\infty^{(2)},r_\infty-1}) ,\dots,  (t_{\infty^{(1)},r_\infty-j}-t_{\infty^{(2)},r_\infty-j}) \right)^t$ which is precisely the $(r_\infty-j)^{\text{th}}$ column of $M_\infty$. Since $M_\infty$ is invertible, we necessarily have $\left( \nu^{(\mathbf{u}_{\infty,j})}_{\infty,-1},\dots,\nu^{(\mathbf{u}_{\infty,j})}_{\infty,r_\infty-3} \right)^t=\mathbf{e}_{r_\infty-j}$, i.e.}
\beq \nu^{(\mathbf{u}_{\infty,j})}_{\infty,k}=\delta_{k,r_\infty-j-2}\,\, ,\,\, \forall \, k\in \llbracket -1, r_\infty-3\rrbracket.\eeq
It is then obvious that the r.h.s. of \eqref{RelationNuAlphas} is null so that we get
\beq \nu^{(\mathbf{u}_{\infty,j})}_{{X_s},k}=0 \,\,,\,\forall \,(s,k) \in \llbracket 1,n\rrbracket\times\llbracket 1,r_{s}-1\rrbracket.\eeq
In the same way, it is obvious that the r.h.s. of \eqref{Relationckalphas} is null so that we get
\beq c^{(\mathbf{u}_{\infty,j})}_{{X_s},k}=0\,\,,\,\forall \,(s,k) \in \llbracket 1,n\rrbracket\times\llbracket 1,r_{s}-1\rrbracket.\eeq

Finally, let us compute the r.h.s. of \eqref{Relationckalphainfty}. The $i^{\text{th}}$ line, denoted $\text{RHS}_i$ is given by
\beq \text{RHS}_i=\underset{r=r_\infty-i}{\overset{r_\infty-1}{\sum}} \frac{t_{\infty^{(2)},2r_\infty-1-i-r}\alpha_{\infty^{(1)},r}-t_{\infty^{(1)},2r_\infty-1-i-r}\alpha_{\infty^{(2)},r}}{r}.\eeq
However, by definition $\alpha_{\infty^{(1)},r}=r t_{\infty^{(1)},r_\infty-1-j+r}\delta_{1\leq r\leq j}$ and $\alpha_{\infty^{(2)},r}=r t_{\infty^{(2)},r_\infty-1-j+r}\delta_{1\leq r\leq j}$ so that 
\bea \text{RHS}_i&=&\underset{r=r_\infty-i}{\overset{r_\infty-1}{\sum}} (t_{\infty^{(2)},2r_\infty-1-i-r} t_{\infty^{(1)},r_\infty-1-j+r}-t_{\infty^{(1)},2r_\infty-1-i-r}t_{\infty^{(2)},r_\infty-1-j+r})\delta_{1\leq r\leq j}\cr
&=&\underset{r=r_\infty-i}{\overset{r_\infty-1}{\sum}} t_{\infty^{(2)},2r_\infty-1-i-r} t_{\infty^{(1)},r_\infty-1-j+r}\delta_{1\leq r\leq j}\cr
&&-\underset{r=r_\infty-i}{\overset{r_\infty-1}{\sum}}t_{\infty^{(1)},2r_\infty-1-i-r}t_{\infty^{(2)},r_\infty-1-j+r}\delta_{1\leq r\leq j}\cr
&=&\underset{r=r_\infty-i}{\overset{j}{\sum}} t_{\infty^{(2)},2r_\infty-1-i-r} t_{\infty^{(1)},r_\infty-1-j+r}
-\underset{r=r_\infty-i}{\overset{j}{\sum}}t_{\infty^{(1)},2r_\infty-1-i-r}t_{\infty^{(2)},r_\infty-1-j+r}\cr
&\overset{s=r_\infty+j-i-r}{=}&\underset{r=r_\infty-i}{\overset{j}{\sum}} t_{\infty^{(2)},2r_\infty-1-i-r} t_{\infty^{(1)},r_\infty-1-j+r}-\underset{s=r_\infty-i}{\overset{j}{\sum}}t_{\infty^{(1)},r_\infty-1-j+s}t_{\infty^{(2)},2r_\infty-1-i-s}\cr
&=&0.\cr
&&
\eea
Thus, the r.h.s. of \eqref{Relationckalphainfty} is null and so are all $\left(c^{(\mathbf{u}_{\infty,j})}_{\infty,k}\right)_{1\leq k\leq r_\infty-1}$.

\subsection{The case of $\mathbf{u}_{X_s,j}$}
Let $s\in\llbracket 1,n\rrbracket$ and $j\in \llbracket 1,r_s-1\rrbracket$ and we consider $\mathcal{L}_{\mathbf{u}_{X_s,j}}$. The only non-zero deformation parameters are 
\beq \alpha_{X_{s'}^{(i)},r}=r t_{X_{s'}^{(i)},r_{s'}-1-j+r}\delta_{s',s}\delta_{1\leq r\leq j}.\eeq
Thus, it is obvious that the r.h.s. of \eqref{RelationNuAlphaInfty} and \eqref{RelationNuAlphas} are null for $s'\neq s$ so that 
\bea 0&=&\nu^{(\mathbf{u}_{X_s,j})}_{\infty,k} \,\,,\,\forall \,k \in \llbracket -1 ,r_{\infty}-3\rrbracket, \cr
0&=&\nu^{(\mathbf{u}_{X_s,j})}_{{X_{s'}},k} \,\,,\,\forall \,(s',k) \in \left(\llbracket 1,n\rrbracket\setminus\{s\}\right) \times\llbracket 1,r_{s'}-1\rrbracket.
\eea
Similarly, the r.h.s. of \eqref{Relationckalphainfty} and \eqref{Relationckalphas} are null for $s'\neq s$ so that
\bea
0&=&c^{(\mathbf{u}_{X_s,j})}_{\infty,k}\,\,,\,\forall \,k \in \llbracket 1 ,r_{\infty}-1\rrbracket,\cr
0&=&c^{(\mathbf{u}_{X_s,j})}_{{X_{s'}},k}\,\,,\,\forall \,(s',k)\in \left(\llbracket 1,n\rrbracket\setminus\{s\}\right) \times\llbracket 1,r_{s'}-1\rrbracket.
\eea

Let us now turn to the r.h.s. of \eqref{RelationNuAlphas} for $s'=s$. By definition we have
\beq \left(\alpha_{X_s^{(i)},1},\dots,\alpha_{X_s^{(i)},r_s-1}\right)^t=\left(t_{X_s^{(i)},r_s-j},\dots,j t_{X_s^{(i)},r_s-1},0\dots,0\right)^t.\eeq
Since $\alpha_{X_s}=0$, we get that the r.h.s. of \eqref{RelationNuAlphas} is 
\beqq\left(0\dots,0,-(t_{X_s^{(1)},r_s-1}-t_{X_s^{(2)},r_s-1}) ,\dots,  -(t_{X_s^{(1)},r_s-j}-t_{X_s^{(2)},r_s-j}) \right)^t\eeqq
which is precisely the opposite of $(r_s-j)^{\text{th}}$ column of $M_s$. Since $M_s$ is invertible, we necessarily have $\left( \nu^{(\mathbf{u}_{X_s,j})}_{X_s,1},\dots,\nu^{(\mathbf{u}_{X_s,j})}_{X_s,r_s-1} \right)^t=-\mathbf{e}_{r_s-j}$, i.e.
\beq \nu^{(\mathbf{u}_{X_s,j})}_{X_s,k}=-\delta_{k,r_s-j}\,\, ,\,\, \forall \, k\in \llbracket 1, r_s-1\rrbracket.\eeq

Finally, let us look at the r.h.s. of \eqref{Relationckalphas} for $s'=s$. The $i^{\text{th}}$ line, denoted $\text{RHS}_i$ is given by
\beq \text{RHS}_i=\underset{k=r_s-i}{\overset{r_s-1}{\sum}} \frac{t_{X_s^{(2)},2r_s-1-i-k}\alpha_{X_s^{(1)},k}-t_{X_s^{(1)},2r_s-1-i-k}\alpha_{X_s^{(2)},k}}{k}.\eeq
However, by definition $\alpha_{X_s^{(1)},r}=r t_{X_s^{(1)},r_s-1-j+r}\delta_{1\leq r\leq j}$ and $\alpha_{X_s^{(2)},r}=r t_{X_s^{(2)},r_s-1-j+r}\delta_{1\leq r\leq j}$ so that 
\bea \text{RHS}_i&=&\underset{r=r_s-i}{\overset{r_s-1}{\sum}} (t_{X_s^{(2)},2r_s-1-i-r} t_{X_s^{(1)},r_s-1-j+r}-t_{X_s^{(1)},2r_s-1-i-r}t_{X_s^{(2)},r_s-1-j+r})\delta_{1\leq r\leq j}\cr
&=&\underset{r=r_s-i}{\overset{r_s-1}{\sum}} t_{X_s^{(2)},2r_s-1-i-r} t_{X_s^{(1)},r_s-1-j+r}\delta_{1\leq r\leq j}\cr
&&-\underset{r=r_s-i}{\overset{r_s-1}{\sum}}t_{X_s^{(1)},2r_s-1-i-r}t_{X_s^{(2)},r_s-1-j+r}\delta_{1\leq r\leq j}\cr
&=&\underset{r=r_s-i}{\overset{j}{\sum}} t_{X_s^{(2)},2r_s-1-i-r} t_{X_s^{(1)},r_s-1-j+r}
-\underset{r=r_s-i}{\overset{j}{\sum}}t_{X_s^{(1)},2r_s-1-i-r}t_{X_s^{(2)},r_s-1-j+r}\cr
&\overset{u=r_s+j-i-r}{=}&\underset{r=r_s-i}{\overset{j}{\sum}} t_{X_s^{(2)},2r_s-1-i-r} t_{X_s^{(1)},r_s-1-j+r}-\underset{u=r_s-i}{\overset{j}{\sum}}t_{X_s^{(1)},r_s-1-j+u}t_{X_s^{(2)},2r_s-1-i-u}\cr
&=&0.\cr
&&
\eea
Thus, the r.h.s. of \eqref{Relationckalphas} is null and so are all $\left(c^{(\mathbf{u}_{X_s,j})}_{{X_s},k}\right)_{1\leq k\leq r_s-1}$.

\section{Proof of Proposition \ref{TrivialSubspace3}}\label{AppendixD}
\subsection{The case of $\mathcal{L}_{\mathbf{a}}$}
The main idea of the proof is that the connections between $\left(\nu^{(\mathbf{a})}_{\infty,k}\right)_{-1\leq k\leq r_\infty-3}$, $\left(\nu^{(\mathbf{a})}_{{X_s},k}\right)_{1\leq s\leq n, 1\leq k\leq r_s-1}$ and $\left(\alpha_{\infty^{(i)},k}\right)_{1\leq k\leq r_\infty-1}$, $\left(\alpha_{X_s^{(i)},k}\right)_{1\leq s \leq n , 1\leq k\leq r_s-1}$ given by \eqref{Defnuinftyk} and \eqref{Defnusk} are independent of $\alpha_{X_s}$. The same is true for the connection between $\left(c^{(\mathbf{a})}_{\infty,k}\right)_{1\leq k\leq r_\infty-1}$, $\left(c^{(\mathbf{a})}_{{X_s},k}\right)_{1\leq s\leq n, 1\leq k\leq r_s-1}$ and $\left(\alpha_{\infty^{(i)},k}\right)_{1\leq k\leq r_\infty-1}$, $\left(\alpha_{X_s^{(i)},k}\right)_{1\leq s \leq n , 1\leq k\leq r_s-1}$ given by \eqref{Relationckalphainfty} and \eqref{Relationckalphas} are independent of $\alpha_{X_s}$.

More precisely, we observe from \eqref{Defnuinftyk} that  $\left(\nu^{(\mathbf{a})}_{\infty,k}\right)_{-1\leq k\leq r_\infty-3}$ are only determined by the term involving $\mathcal{L}_{\mathbf{u}_{\infty,r_\infty-1}}$. Hence, from Proposition \ref{TrivialSubspace2}, we get that $\nu^{(\mathbf{a})}_{\infty,k}=\delta_{k,-1}$ for all $k\in \llbracket -1,r_\infty-3\rrbracket$. Similarly from \eqref{Relationckalphainfty} we have $c^{(\mathbf{a})}_{\infty,k}=0$ for all $k\in \llbracket 1,r_\infty-1\rrbracket$.

Similarly, from \eqref{Defnusk} we observe that for any $s\in \llbracket 1,n\rrbracket$, $\left(\nu^{(\mathbf{a})}_{{X_s},k}\right)_{1\leq k\leq r_s-1}$ are only determined by the term involving $\mathcal{L}_{\mathbf{u}_{X_s,r_s-1}}$. Hence, from Proposition \ref{TrivialSubspace2}, we get that $\nu^{(\mathbf{a})}_{{X_s},k}=\delta_{k,1}$ for all $k\in \llbracket 1,r_s-1\rrbracket$. Similarly from \eqref{Relationckalphas} we have $c^{(\mathbf{a})}_{{X_s},k}=0$ for all $k\in \llbracket 1,r_s-1\rrbracket$.

Finally from \eqref{Defnusk}, we get that for all $s\in \llbracket 1,n\rrbracket$, $\nu^{(\mathbf{a})}_{{X_s},0}=X_s$.

\subsection{The case of $\mathcal{L}_{\mathbf{b}}$}

Similarly, we observe from \eqref{Defnuinftyk} that  $\left(\nu^{(\mathbf{b})}_{\infty,k}\right)_{-1\leq k\leq r_\infty-3}$ are only determined by the term involving $\mathcal{L}_{\mathbf{u}_{\infty,r_\infty-2}}$. Hence, from Proposition \ref{TrivialSubspace2}, we get that $\nu^{(\mathbf{b})}_{\infty,k}=\delta_{k,0}$ for all $k\in \llbracket -1,r_\infty-3\rrbracket$. In the same way, from \eqref{Relationckalphainfty} we have $c^{(\mathbf{b})}_{\infty,k}=0$ for all $k\in \llbracket 1,r_\infty-1\rrbracket$.

Similarly, from \eqref{Defnusk} and \eqref{Relationckalphas} it is trivial that $\nu^{(\mathbf{b})}_{{X_s},k}=0$ and $c^{(\mathbf{b})}_{{X_s},k}=0$ for all $k\in \llbracket 1,r_s-1\rrbracket$.  

Finally from \eqref{Defnusk}, we get that for all $s\in \llbracket 1,n\rrbracket$, $\nu^{(\mathbf{b})}_{{X_s},0}=1$.

\section{Proof of Theorem \ref{TheoremTrivialSubspace1}}\label{AppendixE}
\subsection{The case of $\mathcal{L}_{\mathbf{v}_{\infty,k}}$} 
Let $k\in \llbracket 1,r_\infty-1\rrbracket$. $\mathcal{L}_{\mathbf{v}_{\infty,k}}$ gives vanishing $\left(\nu^{(\mathbf{v}_{\infty,k})}_{\infty,m}\right)_{-1\leq m\leq r_\infty-3}$ and vanishing $\left(\nu^{(\mathbf{v}_{\infty,k})}_{{X_s},m}\right)_{1\leq s\leq n, 1\leq m\leq r_s-1}$ from Proposition \ref{TrivialSubspace}. Moreover, it also gives vanishing $\left(\nu^{(\mathbf{v}_{\infty,k})}_{{X_s},0}\right)_{1\leq s\leq n}$ because $\alpha_{X_{s}}=0$ for all $s\in\llbracket 1,n\rrbracket$. Thus, the r.h.s. of \eqref{RelationNuMuMatrixForm} given by \eqref{DefRHSmunu} is vanishing so that all $\left(\mu^{(\mathbf{v}_{\infty,k})}_{j}\right)_{1\leq j\leq g}$ are null. Thus we trivially get $\mathcal{L}_{\mathbf{v}_{\infty,k}}[q_j]=0$ for all $j\in \llbracket 1, g\rrbracket$ from \eqref{Lqj}. Moreover, $\mathcal{L}_{\mathbf{v}_{\infty,k}}[p_j]$ given by \eqref{Lpj} reduces to $\hbar \underset{i=1}{\overset{r_\infty-1}{\sum}}ic^{(\mathbf{v}_{\infty,k})}_{\infty,i}q_j^{i-1}-\hbar\underset{s=1}{\overset{n}{\sum}} \underset{i=1}{\overset{r_s-1}{\sum}}ic^{(\mathbf{v}_{\infty,k})}_{{X_s},i}(q_j-X_s)^{-i-1}$. But Proposition \ref{TrivialSubspace} implies that this reduces only to  $-\hbar q_j^{k-1}$. Thus, we end up with
\beq \mathcal{L}_{\mathbf{v}_{\infty,k}}[q_j]=0\,\,,\,\,\mathcal{L}_{\mathbf{v}_{\infty,k}}[p_j]=-\hbar q_j^{k-1}.\eeq

\subsection{The case of $\mathcal{L}_{\mathbf{v}_{X_s,k}}$}
Let $s_0\in \llbracket 1,n\rrbracket$ and $k\in \llbracket 1,r_{s_0}-1\rrbracket$. $\mathcal{L}_{\mathbf{v}_{X_{s_0},k}}$ gives vanishing $\left(\nu^{(\mathbf{v}_{X_{s_0},k})}_{\infty,m}\right)_{-1\leq m\leq r_\infty-3}$ and vanishing $\left(\nu^{(\mathbf{v}_{X_{s_0},k})}_{{X_s},m}\right)_{1\leq s\leq n, 1\leq m\leq r_s-1}$ from Proposition \ref{TrivialSubspace}. Moreover, it also gives vanishing $\left(\nu^{(\mathbf{v}_{X_{s_0},k})}_{{X_s},0}\right)_{1\leq s\leq n}$ because $\alpha_{X_{s}}=0$ for all $s\in\llbracket 1,n\rrbracket$. Thus, the r.h.s. of \eqref{RelationNuMuMatrixForm} given by \eqref{DefRHSmunu} is vanishing so that all $\left(\mu^{(\mathbf{v}_{X_{s_0},k})}_{j}\right)_{1\leq j\leq g}$ are null. Thus we trivially get $\mathcal{L}_{\mathbf{v}_{X_{s_0},k}}[q_j]=0$ for all $j\in \llbracket 1, g\rrbracket$ from \eqref{Lqj}. Moreover, $\mathcal{L}_{\mathbf{v}_{X_{s_0},k}}[p_j]$ given by \eqref{Lpj} reduces to $\hbar \underset{i=1}{\overset{r_\infty-1}{\sum}}ic^{(\mathbf{v}_{X_{s_0},k})}_{\infty,i}q_j^{i-1}-\hbar\underset{s=1}{\overset{n}{\sum}} \underset{i=1}{\overset{r_s-1}{\sum}}ic^{(\mathbf{v}_{X_{s_0},k})}_{{X_s},i}(q_j-X_s)^{-i-1}$. But Proposition \ref{TrivialSubspace} implies that this reduces only to  $\hbar (q_j-X_{s_0})^{-k-1}$. Thus, we end up with
\beq \mathcal{L}_{\mathbf{v}_{X_{s_0},k}}[q_j]=0\,\,,\,\,\mathcal{L}_{\mathbf{v}_{X_{s_0},k}}[p_j]=\hbar (q_j-X_{s_0})^{-k-1}.\eeq

\subsection{The case of $\mathcal{L}_{\mathbf{a}}$}

We use results of Proposition \ref{TrivialSubspace3} into \eqref{DefRHSmunu}. We get:
\beq \boldsymbol{\nu}^{(\mathbf{a})}_\infty=\mathbf{0} \,\,,\,\, \boldsymbol{\nu}^{(\mathbf{a})}_{X_s}=\begin{pmatrix} -\nu^{(\mathbf{a})}_{{X_s},0}+\nu^{(\mathbf{a})}_{\infty,0}+\nu^{(\mathbf{a})}_{\infty,-1}X_s\\ -\nu^{(\mathbf{a})}_{{X_s},1}+\nu^{(\mathbf{a})}_{\infty,-1}\\-\nu^{(\mathbf{a})}_{{X_s},2}\\   \vdots \\ -\nu^{(\mathbf{a})}_{{X_s},r_s-1}\end{pmatrix}=\mathbf{0}.\eeq
Hence from \eqref{RelationNuMuMatrixForm} we trivially get, 
\beq \forall \, j\in \llbracket 1,g\rrbracket \,:\, \mu^{(\mathbf{a})}_j=0.\eeq
Finally, we get from \eqref{Lqj}, \eqref{Lpj} that for all $j\in \llbracket 1,g\rrbracket$: 
\beq \mathcal{L}_{\mathbf{a}}[q_j]=-\hbar q_j\,\,,\,\,\mathcal{L}_{\mathbf{a}}[p_j]=\hbar p_j.\eeq

\subsection{The case of $\mathcal{L}_{\mathbf{b}}$}
We use results of Proposition \ref{TrivialSubspace3} into \eqref{DefRHSmunu}. We get:
\beq \boldsymbol{\nu}^{(\mathbf{b})}_\infty=\mathbf{0} \,\,,\,\, \boldsymbol{\nu}^{(\mathbf{b})}_{X_s}=\begin{pmatrix} -\nu^{(\mathbf{b})}_{{X_s},0}+\nu^{(\mathbf{b})}_{\infty,0}+\nu^{(\mathbf{b})}_{\infty,-1}X_s\\ -\nu^{(\mathbf{b})}_{{X_s},1}+\nu^{(\mathbf{b})}_{\infty,-1}\\-\nu^{(\mathbf{b})}_{{X_s},2}\\   \vdots \\ -\nu^{(\mathbf{b})}_{{X_s},r_s-1}\end{pmatrix}=\mathbf{0}.\eeq
Hence from \eqref{RelationNuMuMatrixForm} we trivially get, 
\beq \forall \, j\in \llbracket 1,g\rrbracket \,:\, \mu^{(\mathbf{b})}_j=0.\eeq
Finally, we get from \eqref{Lqj}, \eqref{Lpj} that for all $j\in \llbracket 1,g\rrbracket$: 
\beq \mathcal{L}_{\mathbf{b}}[q_j]=-\hbar \,\,,\,\,\mathcal{L}_{\mathbf{b}}[p_j]=0.\eeq

\section{Proof of Proposition \ref{TheoDualIsorgeq3} }\label{AppendixH}
From the definition of the trivial and isomonodromic times given in each of the four cases (Definitions \ref{DefTrivialrgeq3}, \ref{DefTrivialrequal2}, \ref{DefTrivialrequal1} and \ref{DefTrivialrequal1n1}), one may obtain the irregular times by inverting the change of coordinates. In fact, it is obvious that the only non-trivial cases are $(t_{\infty^{(1)},k}-t_{\infty^{(2)},k})_{1\leq k\leq r_\infty-1}$, that we may combine with $T_{\infty,k}$ in order to obtain $t_{\infty^{(1)},k}$ and $t_{\infty^{(2)},k}$.

Let us prove the following Lemma:

\begin{lemma}\label{TechnicalLemma} For complex numbers $A\ne0$ and $B$ and any $r\in \llbracket 2, r_\infty-2\rrbracket$:
\small{\bea &&\sum_{j=2}^{r}\frac{(r_\infty-2-j)!}{(r_\infty-2-r)! (r-j)!}A^{r_\infty-1-j}B^{r-j}\cr
&&\left(\sum_{i=0}^{j-2}\frac{(-1)^i (r_\infty-2-j+i)!}{i!(r_\infty-2-j)!}\frac{B^i (t_{\infty^{(1)},r_\infty-1-j+i}-t_{\infty^{(2)},r_\infty-1-j+i})}{A^{r_\infty-1-j}}+ \frac{(-1)^{j-1}(r_\infty-2)!}{j (j-2)! (r_\infty-2-j)!} B^jA^j\right)\cr
&&=-\frac{(r_\infty-2)!}{(r_\infty-2-r)! r!}B^rA^{r_\infty-1}+ (t_{\infty^{(1)},r_\infty-1-r}-t_{\infty^{(2)},r_\infty-1-r}).\cr
&&
\eea}\normalsize{}
\end{lemma}

\begin{proof}We have \footnotesize{\bea 
&&\sum_{j=2}^{r} \frac{(r_\infty-2-j)!}{(r_\infty-2-r)! (r-j)!}A^{r_\infty-1-j}B^{r-j}\sum_{i=0}^{j-2}\frac{(-1)^i (r_\infty-2-j+i)!}{i!(r_\infty-2-j)!}\frac{B^i (t_{\infty^{(1)},r_\infty-1-j+i}-t_{\infty^{(2)},r_\infty-1-j+i})}{A^{r_\infty-1-j}}\cr
&&+ \sum_{j=2}^{r} \frac{(r_\infty-2-j)!}{(r_\infty-2-r)! (r-j)!}A^{r_\infty-1-j}B^{r-j}\frac{(-1)^{j-1}(r_\infty-2)!}{j (j-2)! (r_\infty-2-j)!}B^jA^j\cr
&&=\sum_{j=2}^r\frac{(-1)^{j-1}(r_\infty-2)!}{(r_\infty-2-r)!(r-j)! j (j-2)!}B^r A^{r_\infty-1}\cr
&&+\sum_{j=2}^{r}\sum_{i=0}^{j-2} \frac{(-1)^i(r_\infty-2-j+i)!}{(r_\infty-2-r)!(r-j)!i!} B^{r+i-j} (t_{\infty^{(1)},r_\infty-1-j+i}-t_{\infty^{(2)},r_\infty-1-j+i})\cr
&\overset{i=j-s}{=}&\frac{(r_\infty-2)!}{(r_\infty-2-r)!}B^rA^{r_\infty-1}\sum_{j=2}^r \frac{(-1)^{j-1}}{(r-j)!j (j-2)!}\cr
&&+ \sum_{j=2}^{r}\sum_{s=2}^{j} \frac{(-1)^{j-s}(r_\infty-2-s)!}{(r_\infty-2-r)!(r-j)!(j-s)!} B^{r-s} (t_{\infty^{(1)},r_\infty-1-s}-t_{\infty^{(2)},r_\infty-1-s})\cr
&=&\frac{(r_\infty-2)!}{(r_\infty-2-r)!}B^rA^{r_\infty-1}\left(-\frac{1}{r!}\right)+\sum_{s=2}^{r}\sum_{j=s}^{r} \frac{(-1)^{j-s}(r_\infty-2-s)!}{(r_\infty-2-r)!(r-j)!(j-s)!} B^{r-s} (t_{\infty^{(1)},r_\infty-1-s}-t_{\infty^{(2)},r_\infty-1-s})\cr
&\overset{p=j-s}{=}&-\frac{(r_\infty-2)!}{(r_\infty-2-r)! r!}B^rA^{r_\infty-1}+\sum_{s=2}^{r}\left(\sum_{p=0}^{r-s} \frac{(-1)^{p}}{(r-s-p)!p!}\right) \frac{(r_\infty-2-s)!}{(r_\infty-2-r)!} B^{r-s} (t_{\infty^{(1)},r_\infty-1-s}-t_{\infty^{(2)},r_\infty-1-s})\cr
&=&-\frac{(r_\infty-2)!}{(r_\infty-2-r)! r!}B^rA^{r_\infty-1}+\sum_{s=2}^{r}\left(\delta_{r-s=0}\right) \frac{(r_\infty-2-s)!}{(r_\infty-2-r)!} B^{r-s} (t_{\infty^{(1)},r_\infty-1-s}-t_{\infty^{(2)},r_\infty-1-s})\cr
&=&-\frac{(r_\infty-2)!}{(r_\infty-2-r)! r!}B^rA^{r_\infty-1}+ (t_{\infty^{(1)},r_\infty-1-r}-t_{\infty^{(2)},r_\infty-1-r}).\cr
&&
\eea}\normalsize{}
\end{proof}
Lemma \ref{TechnicalLemma} implies that for all $k\in \llbracket 1, r_\infty-3\rrbracket$:
\bea \label{LemmaTechnical2}&&t_{\infty^{(1)},k}-t_{\infty^{(2)},k}= \frac{(r_\infty-2)!}{(k-1)! (r_\infty-1-k)!}B^{r_\infty-1-k} A^{r_\infty-1}\cr
&&+\sum_{j=2}^{r_\infty-1-k}\frac{(r_\infty-2-j)!}{(k-1)! (r_\infty-1-k-j)!}A^{r_\infty-1-j}B^{r_\infty-1-k-j} \cr
&&\Big(\sum_{i=0}^{j-2}\frac{(-1)^i (r_\infty-2-j+i)!}{i!(r_\infty-2-j)!}\frac{B^i (t_{\infty^{(1)},r_\infty-1-j+i}-t_{\infty^{(2)},r_\infty-1-j+i})}{A^{r_\infty-1-j}}\cr
&&+ \frac{(-1)^{j-1}(r_\infty-2)!}{j (j-2)! (r_\infty-2-j)!} B^jA^j
\Big).\cr
&&
\eea
Let us now observe that in the case $r_\infty\geq 3$ the expression for $(\tau_{\infty,k})_{1\leq k\leq r_\infty-3}$ is given by
\small{\bea 2^{-\frac{k}{r_\infty-1}}\tau_{\infty,k}&=& \sum_{i=0}^{r_\infty-k-3} \frac{(-1)^i (k+i-1)!}{i! (k-1)! (r_\infty-2)^i}\frac{(t_{\infty^{(1)},r_\infty-2}-t_{\infty^{(2)},r_\infty-2})^i(t_{\infty^{(1)},k+i}-t_{\infty^{(2)},k+i})}{(t_{\infty^{(1)},r_\infty-1}-t_{\infty^{(2)},r_\infty-1})^{\frac{i(r_\infty-1)+k}{r_\infty-1}}}\cr
&&+\frac{(-1)^{r_\infty-k-2} (r_\infty-3)!}{(r_\infty-1-k)(r_\infty-k-3)!(k-1)!(r_\infty-2)^{r_\infty-k-2}}\frac{(t_{\infty^{(1)},r_\infty-2}-t_{\infty^{(2)},r_\infty-2})^{r_\infty-1-k}}{(t_{\infty^{(1)},r_\infty-1}-t_{\infty^{(2)},r_\infty-1})^{\frac{(r_\infty-2)(r_\infty-1-k)}{r_\infty-1}}}\cr&&
\eea}
\normalsize{so} that performing the change of variables $k=r_\infty-1-j$ leads to the formula for all $j\in \llbracket 2,r_\infty-2\rrbracket$:
\footnotesize{\bea &&2^{-\frac{r_\infty-1-j}{r_\infty-1}}\tau_{\infty,r_\infty-1-j}= \sum_{i=0}^{j-2} \frac{(-1)^i (r_\infty-2-j+i)!}{i! (r_\infty-2-j)!} \left(\frac{(t_{\infty^{(1)},r_\infty-2}-t_{\infty^{(2)},r_\infty-2})}{(r_\infty-2)(t_{\infty^{(1)},r_\infty-1}-t_{\infty^{(2)},r_\infty-1})}\right)^i\cr
&& \frac{1}{\left((t_{\infty^{(1)},r_\infty-1}-t_{\infty^{(2)},r_\infty-1})^{\frac{1}{r_\infty-1}}\right)^{r_\infty-1-j}}(t_{\infty^{(1)},r_\infty-1-j+i}-t_{\infty^{(2)},r_\infty-1-j+i})\cr
&&+\frac{(-1)^{j-1} (r_\infty-2)!}{j(j-2)!(r_\infty-2-j)!}\left(\frac{t_{\infty^{(1)},r_\infty-2}-t_{\infty^{(2)},r_\infty-2}}{(r_\infty-2)(t_{\infty^{(1)},r_\infty-1}-t_{\infty^{(2)},r_\infty-1})}\right)^j\left((t_{\infty^{(1)},r_\infty-1}-t_{\infty^{(2)},r_\infty-1})^{\frac{1}{r_\infty-1}}\right)^j.\cr
&&
\eea}
\normalsize{} Thus, applying Lemma \ref{TechnicalLemma} with
\bea B&=&\frac{t_{\infty^{(1)},r_\infty-2}-t_{\infty^{(2)},r_\infty-2}}{(r_\infty-2)(t_{\infty^{(1)},r_\infty-1}-t_{\infty^{(2)},r_\infty-1})}\cr
A&=&(t_{\infty^{(1)},r_\infty-1}-t_{\infty^{(2)},r_\infty-1})^{\frac{1}{r_\infty-1}}
\eea
we get from \eqref{LemmaTechnical2} for all $k\in \llbracket 1, r_\infty-3\rrbracket$:
\small{\bea \label{LemmaTechnical3}&&t_{\infty^{(1)},k}-t_{\infty^{(2)},k}= \frac{(r_\infty-2)!}{(k-1)! (r_\infty-1-k)!}\left((t_{\infty^{(1)},r_\infty-1}-t_{\infty^{(2)},r_\infty-1})^{\frac{1}{r_\infty-1}}\right)^{r_\infty-1}\cr
&& \left(\frac{t_{\infty^{(1)},r_\infty-2}-t_{\infty^{(2)},r_\infty-2}}{(r_\infty-2)(t_{\infty^{(1)},r_\infty-1}-t_{\infty^{(2)},r_\infty-1})}\right)^{r_\infty-1-k}\cr
&&+\sum_{j=2}^{r_\infty-1-k}\frac{(r_\infty-2-j)!}{(k-1)! (r_\infty-1-k-j)!}\left(\frac{t_{\infty^{(1)},r_\infty-2}-t_{\infty^{(2)},r_\infty-2}}{(r_\infty-2)(t_{\infty^{(1)},r_\infty-1}-t_{\infty^{(2)},r_\infty-1})}\right)^{r_\infty-1-k-j}\cr
&&\left((t_{\infty^{(1)},r_\infty-1}-t_{\infty^{(2)},r_\infty-1})^{\frac{1}{r_\infty-1}}\right)^{r_\infty-1-j}2^{-\frac{r_\infty-1-j}{r_\infty-1}}\tau_{\infty,r_\infty-1-j}.
\eea}
\normalsize{}

The definition of $T_1$ and $T_2$ implies that
\bea t_{\infty^{(1)},r_\infty-1}-t_{\infty^{(2)},r_\infty-1}&=&2T_2^{r_\infty-1},\cr
t_{\infty^{(1)},r_\infty-2}-t_{\infty^{(2)},r_\infty-2}&=&2(r_\infty-2) T_{1} T_{2}^{r_\infty-2},
\eea
so that
\bea B&=&\frac{t_{\infty^{(1)},r_\infty-2}-t_{\infty^{(2)},r_\infty-2}}{(r_\infty-2)(t_{\infty^{(1)},r_\infty-1}-t_{\infty^{(2)},r_\infty-1})}=\frac{(r_\infty-2) T_{1} T_{2}^{r_\infty-2}}{(r_\infty-2)T_2^{r_\infty-1}}=\frac{T_1}{T_2}\cr
A&=&(t_{\infty^{(1)},r_\infty-1}-t_{\infty^{(2)},r_\infty-1})^{\frac{1}{r_\infty-1}}=2^{\frac{1}{r_\infty-1}}T_2.
\eea

Hence from \eqref{LemmaTechnical3} we obtain for all $k\in \llbracket 1, r_\infty-3\rrbracket$:
\small{\bea &&t_{\infty^{(1)},k}-t_{\infty^{(2)},k}= \frac{2(r_\infty-2)!}{(k-1)! (r_\infty-1-k)!}T_2^{r_\infty-1} \left(\frac{T_1}{T_2}\right)^{r_\infty-1-k}\cr
&&+\sum_{j=2}^{r_\infty-1-k}\frac{(r_\infty-2-j)!}{(k-1)! (r_\infty-1-k-j)!}\left(\frac{T_1}{T_2}\right)^{r_\infty-1-j-k}T_2^{r_\infty-1-j}\tau_{\infty,r_\infty-1-j}\cr
&&=\frac{2(r_\infty-2)!}{(k-1)! (r_\infty-1-k)!}T_2^{k} T_1^{r_\infty-1-k}\cr
&&+T_2^{k}\sum_{j=2}^{r_\infty-1-k}2^{\frac{r_\infty-1-j}{r_\infty-1}}\frac{(r_\infty-2-j)!}{(k-1)! (r_\infty-1-k-j)!}T_1^{r_\infty-1-j-k}2^{-\frac{r_\infty-1-j}{r_\infty-1}}\tau_{\infty,r_\infty-1-j}\cr
&=& \frac{2(r_\infty-2)!}{(k-1)! (r_\infty-1-k)!}T_2^{k} T_1^{r_\infty-1-k}+T_2^{k}\sum_{j=2}^{r_\infty-1-k}\frac{(r_\infty-2-j)!}{(k-1)! (r_\infty-1-k-j)!}T_1^{r_\infty-1-j-k}\tau_{\infty,r_\infty-1-j}.\cr
&&
\eea}
\normalsize{}

\section{Proof of Proposition \ref{PropositionT1T2}}\label{AppendixT1T2}
Let us start with $r_\infty\geq 3$. In this case $T_{1}=\frac{t_{\infty^{(1)},r_\infty-2}-t_{\infty^{(2)},r_\infty-2}}{2^{\frac{1}{r_\infty-1}}(r_\infty-2)(t_{\infty^{(1)},r_\infty-1}-t_{\infty^{(2)},r_\infty-1})^{\frac{r_\infty-2}{r_\infty-1}}}$ and $T_{2}=\left(\frac{t_{\infty^{(1)},r_\infty-1}-t_{\infty^{(2)},r_\infty-1}}{2}\right)^{\frac{1}{r_\infty-1}}$. By symmetry, it is obvious that $(\mathcal{L}_{\mathbf{v}_{\infty,k}})_{1\leq k\leq r_\infty-1}$ and $(\mathcal{L}_{\mathbf{v}_{X_s,k}})_{1\leq s\leq n, 1\leq k\leq r_s-1}$ acts trivially on $T_1$ and $T_2$. 
We also have
\bea \mathcal{L}_{\mathbf{a}}[T_1]&=&\hbar\Big((r_\infty-2)\frac{t_{\infty^{(1)},r_\infty-2}-t_{\infty^{(2)},r_\infty-2}}{2^{\frac{1}{r_\infty-1}}(r_\infty-2)(t_{\infty^{(1)},r_\infty-1}-t_{\infty^{(2)},r_\infty-1})^{\frac{r_\infty-2}{r_\infty-1}}}\cr
&&- (r_\infty-1)\frac{r_\infty-2}{r_\infty-1}\frac{t_{\infty^{(1)},r_\infty-2}-t_{\infty^{(2)},r_\infty-2}}{2^{\frac{1}{r_\infty-1}}(r_\infty-2)(t_{\infty^{(1)},r_\infty-1}-t_{\infty^{(2)},r_\infty-1})^{\frac{r_\infty-2}{r_\infty-1}}} \Big)\cr
&=&0,\cr 
\mathcal{L}_{\mathbf{b}}[T_1]&=&\hbar\frac{(r_\infty-2)(t_{\infty^{(1)},r_\infty-1}-t_{\infty^{(2)},r_\infty-1})}{2^{\frac{1}{r_\infty-1}}(r_\infty-2)(t_{\infty^{(1)},r_\infty-1}-t_{\infty^{(2)},r_\infty-1})^{\frac{r_\infty-2}{r_\infty-1}}}
\cr
&=&\hbar T_2,\cr
\mathcal{L}_{\mathbf{a}}[T_2]&=&\hbar \left(\frac{t_{\infty^{(1)},r_\infty-1}-t_{\infty^{(2)},r_\infty-1}}{2}\right)^{\frac{1}{r_\infty-1}}=\hbar T_2,\cr
\mathcal{L}_{\mathbf{b}}[T_2]&=&0.\cr
&&
\eea

Let us continue with the case $r_\infty= 2$. In this case, we have $T_1=-X_1\left(\frac{t_{\infty^{(1)},1}-t_{\infty^{(2)},1}}{2}\right)$. and $T_2=\left(\frac{t_{\infty^{(1)},1}-t_{\infty^{(2)},1}}{2}\right)$. Since
\bea \mathcal{L}_{\mathbf{a}}&=&\hbar(t_{\infty^{(1)},1}\partial_{t_{\infty^{(1)},1}}+t_{\infty^{(2)},1}\partial_{t_{\infty^{(2)},1}})-\hbar \sum_{s=1}^n X_s\partial_{X_s}\cr
&&-\hbar\sum_{s=1}^n\sum_{r=1}^{r_s-1} r(t_{X_s^{(1)},r}\partial_{t_{X_s^{(1)},r}}+t_{X_s^{(2)},r}\partial_{t_{X_s^{(2)},r}}),\cr
\mathcal{L}_{\mathbf{b}}&=&-\hbar \sum_{s=1}^n \partial_{X_s},\cr
&&
\eea
we immediately get that $\mathcal{L}_{\mathbf{a}}[T_1]=0$, $\mathcal{L}_{\mathbf{b}}[T_1]=\hbar T_2$, $\mathcal{L}_{\mathbf{a}}[T_2]=\hbar T_2$ and $\mathcal{L}_{\mathbf{b}}[T_2]=0$.

\medskip

Let us now consider $r_\infty=1$ and $n\geq 2$. In this case we have $T_1=-\frac{X_1}{X_2-X_1}$ and $T_2=\frac{1}{X_2-X_1}$. The action of $\mathcal{L}_{\mathbf{a}}$ and $\mathcal{L}_{\mathbf{b}}$ on these two quantities respectively reduces only to $-\hbar \underset{s=1}{\overset{n}{\sum}} X_s\partial_{X_s}$ and $-\hbar \underset{s=1}{\overset{n}{\sum}} \partial_{X_s}$ so that the result is trivial.

\medskip

Finally, let us consider $r_\infty=1$ and $n=1$. In this case we have $T_1=-X_1\left(\frac{t_{X_1^{(1)},r_1-1}-t_{X_1^{(2)},r_1-1}}{2}\right)^{-\frac{1}{r_1-1}}$ and $T_2=\left(\frac{t_{X_1^{(1)},r_1-1}-t_{X_1^{(2)},r_1-1}}{2}\right)^{-\frac{1}{r_1-1}}$. The action of $\mathcal{L}_{\mathbf{b}}$ is immediate and the action of $\mathcal{L}_{\mathbf{a}}$ provides $\mathcal{L}_{\mathbf{a}}[T_1]=0$ and $\mathcal{L}_{\mathbf{a}}[T_2]=\hbar T_2$.

\section{Proof of Theorem \ref{MainTheotau}}\label{AppendixF}

Let us look at solutions $f(X_s,t_{\infty^{(i)},k}, t_{X_s^{(i)},k})$ of the set of partial differential equations:
\bea \label{DiffSystem}
0&=&\mathcal{L}_{v_{\infty,k}}[f] \,\,,\,\, \forall \, k\in \llbracket 1, r_\infty-1\rrbracket, \cr
0&=&\mathcal{L}_{v_{X_s,k}}[f]\,\,,\,\, \forall \, (s,k)\in \llbracket 1,n\rrbracket\times \llbracket 1, r_s-1\rrbracket,\cr
0&=&\mathcal{L}_{\mathbf{a}}[f],\cr
0&=&\mathcal{L}_{\mathbf{b}}[f].
\eea
 
The first two sets of equations remove dependence on the trace
coordinates $T_{p,k}=t_{p^{(1)},k}+t_{p^{(2)},k}$, $k\ge1$.
Thus $f$ is initially a function of the differences
$\td{t}_{p,k}=t_{p^{(1)},k}-t_{p^{(2)},k}$ and the pole
positions. This conclusion uses differences at finite poles as well
as at infinity.

Work locally with the branches and nonvanishing conditions specified
for the time coordinates. The inverse formulas in the preceding
propositions show that the trace coordinates, $T_1,T_2$, and the
independent isomonodromic times form a coordinate system. Direct
substitution of the time definitions in $\mathcal L_{\mathbf a}$ and
$\mathcal L_{\mathbf b}$ gives
\[
\mathcal L_{\mathbf a}\tau=
\mathcal L_{\mathbf b}\tau=0
\quad(\tau\in\mathcal T_{\mathrm{iso}}).
\]
For the infinity times this also follows by translating the difference
polynomial to remove its subleading coefficient and then rescaling
its leading coefficient; these are precisely the combinations in
Definition \ref{DefTrivialrgeq3}. The finite times and scaled positions
are invariant by their displayed powers of $T_2$ and their affine
transformation law. The same substitution applies to the other three
choices of time coordinates.

After quotienting by the trace directions, Proposition
\ref{PropositionT1T2} gives
\[
\mathcal L_{\mathbf a}=\hbar T_2\partial_{T_2},
\qquad
\mathcal L_{\mathbf b}=\hbar T_2\partial_{T_1}.
\]
Terms in the trace-coordinate directions have already been eliminated
and therefore do not act on $f$. Since $\hbar T_2\ne0$, the two
remaining equations imply $\partial_{T_1}f=\partial_{T_2}f=0$.
Conversely every function of the independent isomonodromic times is
annihilated by all the displayed trivial directions. This proves both
invariance and completeness, including the small-$r_\infty$ cases.

\section{Proof of Theorem \ref{MainTheotau0}}\label{AppendixSymplecticForm}
In this section, we prove that the initial fundamental two-form
\bea \Omega&=&\hbar\sum_{j=1}^g dq_j \wedge dp_j -\sum_{s=1}^n \sum_{i=1}^2\sum_{k=1}^{r_s-1} dt_{X_s^{(i)},k}\wedge d\text{Ham}^{(\mathbf{e}_{X_s^{(i)},k})}\cr
&&-\sum_{i=1}^2\sum_{k=1}^{r_\infty-1} dt_{\infty^{(i)},k}\wedge d\text{Ham}^{(\mathbf{e}_{\infty^{(i)},k})} -\sum_{s=1}^n dX_s\wedge d\,\text{Ham}^{(\mathbf{e}_{X_s})}
\eea
is equal to
\beq \label{tdOmega}\td{\Omega}=\hbar\sum_{j=1}^g d\check{q}_j \wedge d\check{p}_j- \sum_{\tau \in \mathcal{T}_{\text{iso}}} d\tau \wedge d \text{Ham}^{(\boldsymbol{\alpha}_{\tau})}\eeq
where the isomonodromic times $\mathcal{T}_{\text{iso}}=\left\{\tau_k\right\}_{1\leq k\leq g}$ are given by Definitions \ref{DefTrivialrgeq3}, \ref{DefTrivialrequal2}, \ref{DefTrivialrequal1} and \ref{DefTrivialrequal1n1} while the corresponding vectors $\left(\boldsymbol{\alpha}_{\tau_k}\right)_{1\leq k\leq g}$ in the tangent space are given by Propositions \ref{TheoDualIsorgeq3}, \ref{TheoDualIsorequal2}, \ref{TheoDualIsorequal1} and \ref{TheoDualIsorequal1n1} depending on the value of $r_\infty$ and $n$.

Throughout the proof, $\hbar\ne0$ and all formal-monodromy parameters
are fixed, and we work on a local generic chart with $T_2\ne0$.
We first carry out the cancellations using convenient Hamiltonian
representatives. In this computational part, Hamiltonian identities
are understood up to functions of the times alone, and identities
involving $\Omega$ are understood modulo two-forms depending only on
the times. No term involving a Darboux differential is omitted.
Remark \ref{RemarkPureTimeAppendixP} below explains how the compatible
normalization makes the final identity exact on the total space.
All sums over empty index sets are understood to vanish.

\subsection{First step: reduction to $\mathfrak{sl}_2(\mathbb{C})$}
Let us first observe that we have from Theorem \ref{HamTheorem} that:
\bea \text{Ham}^{(\mathbf{v}_{\infty,j})}(\mathbf{q},\mathbf{p})&=&\frac{\hbar}{j}\sum_{i=1}^g q_i^{\,j} \,\,,\,\forall\, j\in \llbracket 1,r_\infty-1\rrbracket,\cr
\text{Ham}^{(\mathbf{v}_{X_s,j})}(\mathbf{q},\mathbf{p})&=&\frac{\hbar}{j}\sum_{i=1}^g (q_i-X_s)^{-j} \,\,,\,\forall\, (s,j)\in \llbracket 1,n\rrbracket\times\llbracket 1,r_s-1\rrbracket.
\eea

Let us now introduce an intermediate step and define some intermediate Darboux coordinates $(\hat{\mathbf{q}},\hat{\mathbf{p}})$ by
\beq \hat{q}_j:=q_j \,\,,\,\, \hat{p}_j:=p_j-\frac{1}{2}P_1(q_j) \,\,,\,\, \forall\, j\in \llbracket 1,g\rrbracket.\eeq
Note that they are related to $(\check{\mathbf{q}},\check{\mathbf{p}})$ by the relation
\beq \check{q}_j=T_2 \hat{q}_j+T_1 \,\,,\,\, \check{p}_j=T_2^{-1} \hat{p}_j \,\,,\,\, \forall\, j\in \llbracket 1,g\rrbracket.\eeq
The associated vectors in the tangent space are defined by:
\bea \mathbf{w}_{\infty,k}&:=&\mathbf{e}_{\infty^{(1)},k}- \mathbf{e}_{\infty^{(2)},k}\,\,,\,\, \forall\, k\in \llbracket 1, r_\infty-1\rrbracket,\cr
\mathbf{w}_{X_s,k}&:=&\mathbf{e}_{X_s^{(1)},k}- \mathbf{e}_{X_s^{(2)},k}\,\,,\,\, \forall\, (s,k)\in \llbracket 1,n\rrbracket\times\llbracket 1, r_s-1\rrbracket.
\eea
From the fact that $\mathcal{L}_{\mathbf{v}_{X_s,j}}[\hat{p}_i]=-\partial_{q_i}\text{Ham}^{(\mathbf{v}_{X_s,j})}-\frac{1}{2}\mathcal{L}_{\mathbf{v}_{X_s,j}}[P_1](q_i)$, it is straightforward to observe that
\bea \label{SuperIDHam}\text{Ham}^{(\mathbf{v}_{\infty,j})}(\hat{\mathbf{q}},\hat{\mathbf{p}})&=&0\,\,,\,\forall\, j\in \llbracket 1,r_\infty-1\rrbracket,\cr
\text{Ham}^{(\mathbf{v}_{X_s,j})}(\hat{\mathbf{q}},\hat{\mathbf{p}})&=&0\,\,,\,\forall\, (s,j)\in \llbracket 1,n\rrbracket\times\llbracket 1,r_s-1\rrbracket
\eea
where we recall that 
\bea \mathbf{v}_{\infty,k}&:=&\mathbf{e}_{\infty^{(1)},k}+ \mathbf{e}_{\infty^{(2)},k}\,\,,\,\, \forall\, k\in \llbracket 1, r_\infty-1\rrbracket,\cr
\mathbf{v}_{X_s,k}&:=&\mathbf{e}_{X_s^{(1)},k}+ \mathbf{e}_{X_s^{(2)},k}\,\,,\,\, \forall\, (s,k)\in \llbracket 1,n\rrbracket\times\llbracket 1, r_s-1\rrbracket.
\eea
Since we have
\small{\bea \text{Ham}^{(\mathbf{v}_{\infty,k})}(\mathbf{q},\mathbf{p})&=&\text{Ham}^{(\mathbf{e}_{\infty^{(1)},k})}(\mathbf{q},\mathbf{p})+\text{Ham}^{(\mathbf{e}_{\infty^{(2)},k})}(\mathbf{q},\mathbf{p})\,\,,\,\,, \forall\, k\in \llbracket 1, r_\infty-1\rrbracket,\cr
 \text{Ham}^{(\mathbf{w}_{\infty,k})}(\mathbf{q},\mathbf{p})&=&\text{Ham}^{(\mathbf{e}_{\infty^{(1)},k})}(\mathbf{q},\mathbf{p})-\text{Ham}^{(\mathbf{e}_{\infty^{(2)},k})}(\mathbf{q},\mathbf{p})\,\,,\,\, \forall\, k\in \llbracket 1, r_\infty-1\rrbracket,\cr
\text{Ham}^{(\mathbf{v}_{X_s,k})}(\mathbf{q},\mathbf{p})&=&\text{Ham}^{(\mathbf{e}_{X_s^{(1)},k})}(\mathbf{q},\mathbf{p})+\text{Ham}^{(\mathbf{e}_{X_s^{(2)},k})}(\mathbf{q},\mathbf{p})\,\,,\,\, \forall\, (s,k)\in \llbracket 1,n\rrbracket\times\llbracket 1, r_s-1\rrbracket,\cr
\text{Ham}^{(\mathbf{w}_{X_s,k})}(\mathbf{q},\mathbf{p})&=&\text{Ham}^{(\mathbf{e}_{X_s^{(1)},k})}(\mathbf{q},\mathbf{p})-\text{Ham}^{(\mathbf{e}_{X_s^{(2)},k})}(\mathbf{q},\mathbf{p})\,\,,\,\, \forall\, (s,k)\in \llbracket 1,n\rrbracket\times\llbracket 1, r_s-1\rrbracket,\cr
&&
\eea}
\normalsize{it} is a straightforward computation to get:
\bea\label{SympPart1}
dq_j\wedge dp_j&=&d\hat q_j\wedge d\hat p_j
+\frac12\sum_{k=1}^{r_\infty-1}q_j^{k-1}
dT_{\infty,k}\wedge d\hat q_j\cr
&&-\frac12\sum_{s=1}^n\sum_{k=1}^{r_s-1}
(q_j-X_s)^{-k-1}dT_{X_s,k}\wedge d\hat q_j\cr
&&+\frac12\sum_{s=1}^n\sum_{k=0}^{r_s-1}
(k+1)T_{X_s,k}(q_j-X_s)^{-k-2}dq_j\wedge dX_s.
\eea
The last line comes from the motion of the poles; its $k=0$ terms
remain present even though the formal-monodromy parameters are fixed.
At the level of coordinates, we shall introduce the corresponding times:
\bea T_{\infty,k}&=&t_{\infty^{(1)},k}+t_{\infty^{(2)},k} \,\,,\,\, \hat{T}_{\infty,k}:=t_{\infty^{(1)},k}-t_{\infty^{(2)},k}\,\,,\,\, \forall\, k\in \llbracket 0, r_\infty-1\rrbracket,\cr
T_{X_s,k}&=&t_{X_s^{(1)},k}+t_{X_s^{(2)},k} \,\,,\,\, \hat{T}_{X_s,k}:=t_{X_s^{(1)},k}-t_{X_s^{(2)},k}\,\,,\,\, \forall\, (s,k)\in \llbracket 1,n\rrbracket\times\llbracket 0, r_s-1\rrbracket.\cr
&&
\eea
We get that
\small{\bea &&\sum_{s=1}^n \sum_{i=1}^2\sum_{k=1}^{r_s-1} dt_{X_s^{(i)},k}\wedge d\,\text{Ham}^{(\mathbf{e}_{X_s^{(i)},k})}(\mathbf{q},\mathbf{p})+\sum_{i=1}^2\sum_{k=1}^{r_\infty-1} dt_{\infty^{(i)},k}\wedge d\,\text{Ham}^{(\mathbf{e}_{\infty^{(i)},k})}(\mathbf{q},\mathbf{p})\cr
&&=\frac{1}{2}\sum_{s=1}^n\sum_{k=1}^{r_s-1} ( dT_{X_s,k}\wedge d\,\text{Ham}^{(\mathbf{v}_{X_s,k})}(\mathbf{q},\mathbf{p}) +d\hat{T}_{X_s,k}\wedge d\,\text{Ham}^{(\mathbf{w}_{X_s,k})})(\mathbf{q},\mathbf{p})\cr
&&+
\frac{1}{2}\sum_{k=1}^{r_\infty-1} ( dT_{\infty,k}\wedge d\,\text{Ham}^{(\mathbf{v}_{\infty,k})}(\mathbf{q},\mathbf{p}) +d\hat{T}_{\infty,k}\wedge d\,\text{Ham}^{(\mathbf{w}_{\infty,k})})(\mathbf{q},\mathbf{p})\cr
&&=\frac{1}{2}\sum_{p\in \mathcal{R}}\sum_{k=1}^{r_p-1} d\hat{T}_{p,k}\wedge d\,\text{Ham}^{(\mathbf{w}_{p,k})}(\mathbf{q},\mathbf{p})+\frac{\hbar}{2}\sum_{j=1}^g\sum_{k=1}^{r_\infty-1}q_j^{k-1}dT_{\infty,k}\wedge dq_j\cr
&&-\frac{\hbar}{2}\sum_{j=1}^g\sum_{s=1}^n\sum_{k=1}^{r_s-1}
(q_j-X_s)^{-k-1}dT_{X_s,k}\wedge dq_j\cr
&&+\frac{\hbar}{2}\sum_{j=1}^g\sum_{s=1}^n\sum_{k=1}^{r_s-1}
(q_j-X_s)^{-k-1}dT_{X_s,k}\wedge dX_s.
\eea}
\normalsize{Combining} this identity with \eqref{SympPart1}, the
$dT_{\infty,k}\wedge dq_j$ and $dT_{X_s,k}\wedge dq_j$ terms cancel.
The remaining moving-pole terms are
\[
\frac{\hbar}{2}\sum_{s=1}^n dX_s\wedge
d\left(\sum_{j=1}^g\sum_{k=0}^{r_s-1}
T_{X_s,k}(q_j-X_s)^{-k-1}\right).
\]
They are incorporated in the transformed position Hamiltonians:
\begin{equation}\label{MovingPoleHamiltonianP}
\text{Ham}^{(\mathbf w_s)}(\hat{\mathbf q},\hat{\mathbf p})
=\left.\text{Ham}^{(\mathbf e_{X_s})}(\mathbf q,\mathbf p)
\right|_{q_j=\hat q_j,\ p_j=\hat p_j+P_1(\hat q_j)/2}
-\frac\hbar2\sum_{j=1}^g\sum_{k=0}^{r_s-1}
T_{X_s,k}(\hat q_j-X_s)^{-k-1}.
\end{equation}
Thus we obtain the intermediate identity
\beq \label{IntermediateIdentity}\Omega=\hbar\sum_{j=1}^g d\hat{q}_j\wedge d\hat{p}_j - \frac{1}{2}\sum_{p\in \mathcal{R}}\sum_{k=1}^{r_p-1} d\hat{T}_{p,k}\wedge d\,\text{Ham}^{(\mathbf{w}_{p,k})}(\hat{\mathbf{q}},\hat{\mathbf{p}}) -\sum_{s=1}^n dX_s\wedge d\,\text{Ham}^{(\mathbf{w}_s)}(\hat{\mathbf{q}},\hat{\mathbf{p}}).\eeq

Let us also mention that Theorem \ref{HamTheorem} and Proposition \ref{TrivialSubspace3} imply that
\bea\label{ExplicitComputation0bis} 
\text{Ham}^{(\mathbf{a})}(\mathbf{q},\mathbf{p})&=&-\hbar \sum_{j=1}^g q_jp_j-\delta_{r_\infty=1}\left(\sum_{s=1}^n t_{X_s^{(1)},0}t_{X_s^{(2)},0}\delta_{r_s=1} -t_{\infty^{(1)},0}(t_{\infty^{(2)},0}+\hbar)\right),\cr
\text{Ham}^{(\mathbf{b})}(\mathbf{q},\mathbf{p})&=&-\hbar \sum_{j=1}^g p_j+\delta_{r_\infty=2}\left(t_{\infty^{(1)},1}t_{\infty^{(2)},0}+t_{\infty^{(2)},1}t_{\infty^{(1)},0}+\hbar t_{\infty^{(1)},1}\right).\cr
&&
\eea
The quantities proportional to $\delta_{r_\infty=1}$ and $\delta_{r_\infty=2}$ do not depend on the Darboux coordinates and thus do not play any role in the Hamiltonian system. For the computational part of the proof we choose representatives without these terms, while retaining their effect on the total two-form through Remark \ref{RemarkPureTimeAppendixP}. Thus we use
\beq\label{ExplicitComputation0} 
\text{Ham}^{(\mathbf{a})}(\mathbf{q},\mathbf{p})=-\hbar \sum_{j=1}^g q_jp_j\,,\, 
\text{Ham}^{(\mathbf{b})}(\mathbf{q},\mathbf{p})=-\hbar \sum_{j=1}^g p_j.
\eeq
However, one has to be extremely careful when changing $(\mathbf{q},\mathbf{p})\leftrightarrow (\hat{\mathbf{q}},\hat{\mathbf{p}})$ because the change of coordinates is time dependent (because of the position of poles and times arising in $P_1$). 

Using Proposition \ref{PropositionT1T2} and Theorem \ref{TheoremTrivialSubspace1} we have for all $j\in \llbracket 1,g\rrbracket$:
\beq \mathcal{L}_{\mathbf{a}}[\hat{q}_j]=-\hbar \hat{q}_j \,\,,\,\, \mathcal{L}_{\mathbf{a}}[\hat{p}_j]=\hbar \hat{p}_j\,,\,\mathcal{L}_{\mathbf{b}}[\hat{q}_j]=-\hbar \,\,,\,\, \mathcal{L}_{\mathbf{b}}[\hat{p}_j]=0.
\eeq
Thus, we end up with
\beq\label{ExplicitComputation}\text{Ham}^{(\mathbf{a})}(\hat{\mathbf{q}},\hat{\mathbf{p}})=- \hbar\sum_{j=1}^g\hat{q}_j\hat{p}_j\,,\, 
\text{Ham}^{(\mathbf{b})}(\hat{\mathbf{q}},\hat{\mathbf{p}})=- \hbar\sum_{j=1}^g\hat{p}_j.
\eeq
Note in particular that there are no $\left(P_1(q_j)\right)_{1\leq j\leq g}$ terms in the previous formulas as one would have obtained by simply replacing $(\mathbf{q},\mathbf{p})$ in terms of $(\hat{\mathbf{q}},\hat{\mathbf{p}})$ in \eqref{ExplicitComputation0}.

The second step of the proof is to reduce \eqref{IntermediateIdentity} into \eqref{tdOmega}. Since the definition of the trivial and isomonodromic times depends on the values of $r_\infty$ and $n$, we shall study separately each case although the strategy is identical for all cases.

\subsection{The case $r_\infty\geq 3$}

For $r_\infty=3$, all contributions involving infinity
isomonodromic times are absent; in particular, terms proportional to
$r_\infty-3$ involving $\tau_{\infty,r_\infty-3}$ are omitted.

For $r_\infty\geq 3$, Definition \ref{DefTrivialrgeq3} and Proposition \ref{InverseRelationsrgeq3} are equivalent to
\footnotesize{\bea \label{InitialFormrgeq3}T_2&=&\left(\frac{\hat{T}_{\infty,r_\infty-1}}{2}\right)^{\frac{1}{r_\infty-1}}\cr
T_1&=&\frac{\hat{T}_{\infty,r_\infty-2}}{(r_\infty-2)2^{\frac{1}{r_\infty-1}}(\hat{T}_{\infty,r_\infty-1})^{\frac{r_\infty-2}{r_\infty-1}}}\cr
\tau_{\infty,j}&=&2^{\frac{j}{r_\infty-1}}\Big[ \sum_{i=0}^{r_\infty-j-3} \frac{(-1)^i (j+i-1)!}{i! (j-1)! (r_\infty-2)^i}\frac{(\hat{T}_{\infty,r_\infty-2})^i\hat{T}_{\infty,j+i}}{(\hat{T}_{\infty,r_\infty-1})^{\frac{i(r_\infty-1)+j}{r_\infty-1}}}\cr
&&+\frac{(-1)^{r_\infty-j-2} (r_\infty-3)!}{(r_\infty-1-j)(r_\infty-j-3)!(j-1)!(r_\infty-2)^{r_\infty-j-2}}\frac{(\hat{T}_{\infty,r_\infty-2} )^{r_\infty-1-j}}{(\hat{T}_{\infty,r_\infty-1} )^{\frac{(r_\infty-2)(r_\infty-1-j)}{r_\infty-1}}}\Big]\,,\, \forall \, j\in \llbracket 1,r_\infty-3\rrbracket\cr
\tau_{X_s,k}&=&\hat{T}_{X_s,k}\left(\frac{\hat{T}_{\infty,r_\infty-1}}{2}\right)^{\frac{k}{r_\infty-1}}\,,\,\, \forall\, (s,k)\in \llbracket 1,n\rrbracket\times\llbracket 1,r_s-1\rrbracket\cr
\td{X}_s&=&X_s\left(\frac{\hat{T}_{\infty,r_\infty-1}}{2}\right)^{\frac{1}{r_\infty-1}}+\frac{\hat{T}_{\infty,r_\infty-2}}{(r_\infty-2)2^{\frac{1}{r_\infty-1}}(\hat{T}_{\infty,r_\infty-1})^{\frac{r_\infty-2}{r_\infty-1}}}\,,\, \, \forall\, s\in \llbracket 1,n\rrbracket\cr
&&
\eea}
\normalsize{and} the inverse relations
\footnotesize{\bea\label{Homogeneity} \hat{T}_{\infty,r_\infty-1}&=&2T_2^{r_\infty-1}\cr
\hat{T}_{\infty,r_\infty-2}&=&2(r_\infty-2)T_1 T_2^{r_\infty-2}\cr
\hat{T}_{\infty,k}&=&T_2^k\left(2\binom{r_\infty-2}{k-1}T_1^{r_\infty-1-k}+\sum_{j=2}^{r_\infty-1-k} \binom{r_\infty-2-j}{k-1}T_1^{r_\infty-1-j-k}\tau_{\infty,r_\infty-1-j}\right) \,,\,\, \forall\, k\in \llbracket 1,r_\infty-3\rrbracket\cr
\hat{T}_{X_s,k}&=&T_2^{-k}\tau_{X_s,k}\,,\,\, \forall\, (s,k)\in \llbracket 1,n\rrbracket\times\llbracket 1,r_s-1\rrbracket\cr
X_s&=&T_2^{-1}(\td{X}_s-T_1)\,,\,\, \forall\, s\in \llbracket 1,n\rrbracket.\cr
&&
\eea}
\normalsize{We} thus obtain at the level of differentials
\footnotesize{\bea\label{differentialrgeq3}
d\hat{T}_{\infty,r_\infty-1}&=&2(r_\infty-1)T_2^{r_\infty-2}dT_2\cr
d\hat{T}_{\infty,r_\infty-2}&=&2(r_\infty-2)T_2^{r_\infty-2}dT_1+2(r_\infty-2)^2T_1T_2^{r_\infty-3}dT_2\cr
d\hat{T}_{\infty,k}&=&kT_2^{k-1}\left(2\binom{r_\infty-2}{k-1}T_1^{r_\infty-1-k}+\sum_{m=2}^{r_\infty-1-k} \binom{r_\infty-2-m}{k-1}T_1^{r_\infty-1-m-k}\tau_{\infty,r_\infty-1-m}\right)dT_2\cr
&&+2T_2^k(r_\infty-1-k)\binom{r_\infty-2}{k-1}T_1^{r_\infty-2-k}dT_1\cr
&&+T_2^k\sum_{m=2}^{r_\infty-2-k} (r_\infty-1-m-k)\binom{r_\infty-2-m}{k-1}T_1^{r_\infty-2-m-k}\tau_{\infty,r_\infty-1-m}dT_1\cr
&&+T_2^k\sum_{m=2}^{r_\infty-1-k} \binom{r_\infty-2-m}{k-1}T_1^{r_\infty-1-m-k}d\tau_{\infty,r_\infty-1-m} \,,\,\, \forall\, k\in \llbracket 1,r_\infty-3\rrbracket\cr
d\hat{T}_{X_s,k}&=&-kT_2^{-k-1}\tau_{X_s,k}dT_2+T_2^{-k}d\tau_{X_s,k}\,,\,\, \forall\, (s,k)\in \llbracket 1,n\rrbracket\times\llbracket 1,r_s-1\rrbracket\cr
dX_s&=&-T_2^{-2}(\td{X}_s-T_1)dT_2+T_2^{-1}d\td{X}_s-T_2^{-1}dT_1 \,,\,\, \forall\, s\in \llbracket 1,n\rrbracket.
\eea}
\normalsize{At} the level of Hamiltonians, since $\text{Ham}^{(\mathbf{w}_{\infty,k})}=2\text{Ham}^{(\boldsymbol{\alpha}_{\hat{T}_{\infty,k}})}$ for all $k\in \llbracket 1,r_\infty-1\rrbracket$ and $\text{Ham}^{(\mathbf{w}_{X_s,k})}=2\text{Ham}^{(\hat{T}_{X_s,k})}$ for all $(s,k)\in \llbracket 1,n\rrbracket\times \llbracket 1,r_s-1\rrbracket$, this is equivalent to
\footnotesize{\bea\label{OldHamiltonians} \text{Ham}^{(\boldsymbol{\alpha}_{T_1})}(\hat{\mathbf{q}},\hat{\mathbf{p}})&=&(r_\infty-2)T_2^{r_\infty-2}\text{Ham}^{(\mathbf{w}_{\infty,r_\infty-2})}(\hat{\mathbf{q}},\hat{\mathbf{p}})-T_2^{-1}\sum_{s=1}^n\text{Ham}^{(\mathbf{w}_{s})}(\hat{\mathbf{q}},\hat{\mathbf{p}})\cr
&&+\frac{1}{2}\sum_{k=1}^{r_\infty-3}2T_2^k(r_\infty-1-k)\binom{r_\infty-2}{k-1}T_1^{r_\infty-2-k}\text{Ham}^{(\mathbf{w}_{\infty,k})}(\hat{\mathbf{q}},\hat{\mathbf{p}})\cr
&&+\frac{1}{2}\sum_{k=1}^{r_\infty-3}T_2^k\sum_{m=2}^{r_\infty-2-k} (r_\infty-1-m-k)\binom{r_\infty-2-m}{k-1}T_1^{r_\infty-2-m-k}\tau_{\infty,r_\infty-1-m}\text{Ham}^{(\mathbf{w}_{\infty,k})}(\hat{\mathbf{q}},\hat{\mathbf{p}})\cr
\text{Ham}^{(\boldsymbol{\alpha}_{T_2})}(\hat{\mathbf{q}},\hat{\mathbf{p}})&=&(r_\infty-1)T_2^{r_\infty-2}\text{Ham}^{(\mathbf{w}_{\infty,r_\infty-1})}(\hat{\mathbf{q}},\hat{\mathbf{p}})+(r_\infty-2)^2T_1T_2^{r_\infty-3}\text{Ham}^{(\mathbf{w}_{\infty,r_\infty-2})}(\hat{\mathbf{q}},\hat{\mathbf{p}})\cr
&&+\frac{1}{2}\sum_{k=1}^{r_\infty-3}2kT_2^{k-1}\binom{r_\infty-2}{k-1}T_1^{r_\infty-1-k}\text{Ham}^{(\mathbf{w}_{\infty,k})}(\hat{\mathbf{q}},\hat{\mathbf{p}})\cr
&&+\frac{1}{2}\sum_{k=1}^{r_\infty-3}kT_2^{k-1}\sum_{m=2}^{r_\infty-1-k} \binom{r_\infty-2-m}{k-1}T_1^{r_\infty-1-m-k}\tau_{\infty,r_\infty-1-m}\text{Ham}^{(\mathbf{w}_{\infty,k})}(\hat{\mathbf{q}},\hat{\mathbf{p}})\cr
&&-\frac12\sum_{s=1}^n\sum_{k=1}^{r_s-1}kT_2^{-k-1}\tau_{X_s,k}\text{Ham}^{(\mathbf{w}_{X_s,k})}(\hat{\mathbf{q}},\hat{\mathbf{p}})-T_2^{-2}\sum_{s=1}^n(\td{X}_s-T_1)\text{Ham}^{(\mathbf{w}_{s})}(\hat{\mathbf{q}},\hat{\mathbf{p}})\cr
\text{Ham}^{(\boldsymbol{\alpha}_{\tau_{\infty,i}})}(\hat{\mathbf{q}},\hat{\mathbf{p}})&=&\frac{1}{2}\sum_{k=1}^{i}T_2^k\binom{i-1}{k-1}T_1^{i-k}\text{Ham}^{(\mathbf{w}_{\infty,k})}(\hat{\mathbf{q}},\hat{\mathbf{p}}) \,,\, \forall\, i\in \llbracket 1, r_\infty-3\rrbracket\cr
\text{Ham}^{(\boldsymbol{\alpha}_{\tau_{X_s,k}})}(\hat{\mathbf{q}},\hat{\mathbf{p}})&=&\frac{1}{2}T_2^{-k}\text{Ham}^{(\mathbf{w}_{X_s,k})}(\hat{\mathbf{q}},\hat{\mathbf{p}})\,,\,\,\forall\, (s,k)\in \llbracket 1,n\rrbracket\times\llbracket 1,r_s-1\rrbracket \cr
\text{Ham}^{(\boldsymbol{\alpha}_{\td{X}_s})}(\hat{\mathbf{q}},\hat{\mathbf{p}})&=&T_2^{-1}\text{Ham}^{(\mathbf{w}_s)}(\hat{\mathbf{q}},\hat{\mathbf{p}})\,,\,\, \forall\, s\in \llbracket 1,n\rrbracket.\cr
&&
\eea}
\normalsize{Homogeneity} in $T_2$ in \eqref{Homogeneity} and Definition \ref{TrivialVectors}  implies that
\beq \label{Larinftygeq3}\mathcal{L}_{\mathbf{a}}=\hbar T_2\partial_{T_2}+\frac{1}{2}\sum_{k=1}^{r_\infty-1}kT_{\infty,k}\mathcal{L}_{\mathbf{v}_{\infty,k}}-\frac{1}{2}\sum_{s=1}^n\sum_{k=1}^{r_s-1}k T_{X_s,k}\mathcal{L}_{\mathbf{v}_{X_s,k}}.
\eeq
From \eqref{Homogeneity}, we also get:
\footnotesize{\bea \partial_{T_1}&=&-T_2^{-1}\sum_{s=1}^n\partial_{X_s}+T_2^{-1}(r_\infty-2)\hat{T}_{\infty,r_\infty-1}\partial_{\hat{T}_{\infty,r_\infty-2}}+2\sum_{k=1}^{r_\infty-3}T_2^k(r_\infty-1-k)\binom{r_\infty-2}{k-1}T_1^{r_\infty-2-k}\partial_{\hat{T}_{\infty,k}}\cr
&&+\sum_{k=1}^{r_\infty-3}T_2^k\sum_{j=2}^{r_\infty-1-k} (r_\infty-1-j-k)\binom{r_\infty-2-j}{k-1}T_1^{r_\infty-2-j-k}\tau_{\infty,r_\infty-1-j}\partial_{\hat{T}_{\infty,k}}\cr
&=&-T_2^{-1}\sum_{s=1}^n\partial_{X_s}+T_2^{-1}(r_\infty-2)\hat{T}_{\infty,r_\infty-1}\partial_{\hat{T}_{\infty,r_\infty-2}}+T_2^{-1}\sum_{k=1}^{r_\infty-3}k\hat{T}_{\infty,k+1}\partial_{\hat{T}_{\infty,k}}.\cr
&&
\eea}
\normalsize{Thus}, we get from Definition \ref{TrivialVectors} that
\beq \label{Lbrinftygeq3}\mathcal{L}_{\mathbf{b}}=\hbar T_2\partial_{T_1}+\frac{1}{2}\sum_{k=1}^{r_\infty-2}kT_{\infty,k+1}\mathcal{L}_{\mathbf{v}_{\infty,k}}.\eeq
Using \eqref{SuperIDHam} into \eqref{Larinftygeq3} and \eqref{Lbrinftygeq3} we finally obtain 
\beq \label{SecondHamrinftygeq3}\text{Ham}^{(\mathbf{b})}(\hat{\mathbf{q}},\hat{\mathbf{p}})=T_2\text{Ham}^{(\boldsymbol{\alpha}_{T_1})}(\hat{\mathbf{q}},\hat{\mathbf{p}})\,,\,\text{Ham}^{(\mathbf{a})}(\hat{\mathbf{q}},\hat{\mathbf{p}})=T_2\text{Ham}^{(\boldsymbol{\alpha}_{T_2})}(\hat{\mathbf{q}},\hat{\mathbf{p}}).
\eeq
Let us now observe that \eqref{ExplicitComputation} implies that
\footnotesize{\bea \label{IdentityFinalbis}
&&T_2^{-1}dT_2\wedge d\text{Ham}^{(\mathbf{a})}(\hat{\mathbf{q}},\hat{\mathbf{p}})+T_2^{-1}dT_1\wedge d \text{Ham}^{(\mathbf{b})}(\hat{\mathbf{q}},\hat{\mathbf{p}})-T_2^{-2}\text{Ham}^{(\mathbf{b})}(\hat{\mathbf{q}},\hat{\mathbf{p}})dT_1\wedge dT_2\cr
&&= -\hbar\left(T_2^{-1}\sum_{j=1}^g \hat{q}_j dT_2 \wedge  d\hat{p}_j +T_2^{-1}\sum_{j=1}^g \hat{p}_j dT_2 \wedge  d\hat{q}_j+T_2^{-1}\sum_{j=1}^g dT_1\wedge d\hat{p}_j-T_2^{-2}\sum_{j=1}^g\hat{p}_j dT_1\wedge d T_2\right)
\eea}
\normalsize{so} that the shifted Darboux coordinates satisfy
\footnotesize{\beq\label{DarbouxShiftedpj} \sum_{j=1}^g d\check{q}_j \wedge d\check{p}_j =\sum_{j=1}^gd\hat{q}_j \wedge d\hat{p}_j + T_2^{-1}\sum_{j=1}^g\hat{q}_j  dT_2 \wedge d\hat{p}_j + T_2^{-1} \sum_{j=1}^g\hat{p}_j dT_2\wedge d\hat{q}_j  +T_2^{-1}\sum_{j=1}^gdT_1 \wedge d\hat{p}_j- T_2^{-2}\sum_{j=1}^g\hat{p}_j dT_1 \wedge dT_2\eeq}
\normalsize{i.e.}
\footnotesize{\beq\label{IdentityFinal2bis} \hbar\sum_{j=1}^gd\hat{q}_j \wedge d\hat{p}_j=\hbar\sum_{j=1}^g d\check{q}_j \wedge d\check{p}_j+T_2^{-1}dT_2\wedge d\text{Ham}^{(\mathbf{a})}(\hat{\mathbf{q}},\hat{\mathbf{p}})+T_2^{-1}dT_1\wedge d \text{Ham}^{(\mathbf{b})}(\hat{\mathbf{q}},\hat{\mathbf{p}})-T_2^{-2}\text{Ham}^{(\mathbf{b})}(\hat{\mathbf{q}},\hat{\mathbf{p}})dT_1\wedge dT_2.
\eeq}

\normalsize{We} now need to invert the Hamiltonian relations \eqref{OldHamiltonians}. In order to do that, we observe that the following identities hold for any numbers $(h_j)_{j\geq 1}$ and $(\tau_k)_{k\geq 1}$:
\footnotesize{\bea \label{IdentitiesBinomial}&&\sum_{k=1}^{r_\infty-3}\sum_{j=1}^k(r_\infty-1-k)(-1)^{j+k}\binom{r_\infty-2}{k-1}\binom{k-1}{j-1}T_1^{r_\infty-2-j}h_j=-(r_\infty-2)\sum_{j=1}^{r_\infty-3}(-1)^{j+r_\infty}\binom{r_\infty-3}{j-1}T_1^{r_\infty-2-j}h_j\cr
&&\sum_{k=1}^{r_\infty-3}\sum_{m=2}^{r_\infty-2-k}\sum_{j=1}^k(-1)^{j+k} (r_\infty-1-m-k)\binom{r_\infty-2-m}{k-1}\binom{k-1}{j-1}T_1^{r_\infty-2-m-j}\tau_{r_\infty-1-m}h_j=\sum_{k=1}^{r_\infty-4}k \tau_{k+1}h_k\cr
&&\sum_{k=1}^{r_\infty-3}\sum_{j=1}^kk(-1)^{j+k}\binom{r_\infty-2}{k-1}\binom{k-1}{j-1}T_1^{r_\infty-1-j}h_j\cr
&&=-\sum_{m=1}^{r_\infty-3}(r_\infty^2-mr_\infty+2m-4r_\infty+3)(-1)^{m+r_\infty}\binom{r_\infty-2}{m-1}T_1^{r_\infty-1-m}h_m\cr
&&\sum_{k=1}^{r_\infty-3}\sum_{m=2}^{r_\infty-1-k}\sum_{j=1}^k(-1)^{j+k} k\binom{r_\infty-2-m}{k-1}\binom{k-1}{j-1}T_1^{r_\infty-1-m-j}\tau_{r_\infty-1-m}h_j=\sum_{k=1}^{r_\infty-3}k\tau_kh_k-T_1\sum_{k=1}^{r_\infty-4}k \tau_{k+1}h_k.\cr
&&
\eea
}
\normalsize{Using} these identities, we may invert the previous Hamiltonian relations \eqref{OldHamiltonians}:
\footnotesize{\bea\label{Haminverted}&& \text{Ham}^{(\mathbf{w}_{\infty,k})}(\hat{\mathbf{q}},\hat{\mathbf{p}})=2T_2^{-k}\sum_{j=1}^k(-1)^{j+k}\binom{k-1}{j-1}T_1^{k-j}\text{Ham}^{(\boldsymbol{\alpha}_{\tau_{\infty,j}})}(\hat{\mathbf{q}},\hat{\mathbf{p}}) \,,\,\,\forall \, k\in \llbracket 1,r_\infty-3\rrbracket\cr
&&\text{Ham}^{(\mathbf{w}_{X_s,k})}(\hat{\mathbf{q}},\hat{\mathbf{p}})=2T_2^{k}\text{Ham}^{(\boldsymbol{\alpha}_{\tau_{X_s,k}})}(\hat{\mathbf{q}},\hat{\mathbf{p}})\,,\,\,\forall\, (s,k)\in \llbracket 1,n\rrbracket\times\llbracket 1,r_s-1\rrbracket \cr
&&\text{Ham}^{(\mathbf{w}_s)}(\hat{\mathbf{q}},\hat{\mathbf{p}})=T_2\text{Ham}^{(\boldsymbol{\alpha}_{\td{X}_s})}(\hat{\mathbf{q}},\hat{\mathbf{p}})\,,\,\, \forall\, s\in \llbracket 1,n\rrbracket\cr
&&\text{Ham}^{(\mathbf{w}_{\infty,r_\infty-2})}(\hat{\mathbf{q}},\hat{\mathbf{p}})=\frac{1}{(r_\infty-2)T_2^{r_\infty-2}}\text{Ham}^{(\boldsymbol{\alpha}_{T_1})}(\hat{\mathbf{q}},\hat{\mathbf{p}})+\frac{1}{(r_\infty-2)T_2^{r_\infty-2}}\sum_{s=1}^n\text{Ham}^{(\boldsymbol{\alpha}_{\td{X}_s})}(\hat{\mathbf{q}},\hat{\mathbf{p}})\cr
&&+\frac{2}{T_2^{r_\infty-2}}\sum_{j=1}^{r_\infty-3}(-1)^{j+r_\infty}\binom{r_\infty-3}{j-1}T_1^{r_\infty-2-j}\text{Ham}^{(\boldsymbol{\alpha}_{\tau_{\infty,j}})}(\hat{\mathbf{q}},\hat{\mathbf{p}})-\frac{1}{(r_\infty-2)T_2^{r_\infty-2}}\sum_{k=1}^{r_\infty-4}k \tau_{\infty,k+1}\text{Ham}^{(\boldsymbol{\alpha}_{\tau_{\infty,k}})}(\hat{\mathbf{q}},\hat{\mathbf{p}})\cr
&&\text{Ham}^{(\mathbf{w}_{\infty,r_\infty-1})}(\hat{\mathbf{q}},\hat{\mathbf{p}})=\frac{1}{(r_\infty-1)T_2^{r_\infty-2}}\text{Ham}^{(\boldsymbol{\alpha}_{T_2})}(\hat{\mathbf{q}},\hat{\mathbf{p}})+\frac{1}{(r_\infty-1)T_2^{r_\infty-1}}\sum_{s=1}^n\sum_{k=1}^{r_s-1}k\tau_{X_s,k}\text{Ham}^{(\boldsymbol{\alpha}_{\tau_{X_s,k}})}(\hat{\mathbf{q}},\hat{\mathbf{p}})\cr
&&+\frac{1}{(r_\infty-1)T_2^{r_\infty-1}}\sum_{s=1}^n(\td{X}_s-T_1)\text{Ham}^{(\boldsymbol{\alpha}_{\td{X}_s})}(\hat{\mathbf{q}},\hat{\mathbf{p}})\cr
&&+\frac{2}{(r_\infty-1)T_2^{r_\infty-1}}\sum_{m=1}^{r_\infty-3}(r_\infty^2-mr_\infty+2m-4r_\infty+3)(-1)^{m+r_\infty}\binom{r_\infty-2}{m-1}T_1^{r_\infty-1-m}\text{Ham}^{(\boldsymbol{\alpha}_{\tau_{\infty,m}})}(\hat{\mathbf{q}},\hat{\mathbf{p}})\cr
&&+\frac{1}{(r_\infty-1)T_2^{r_\infty-1}}\left(T_1\sum_{k=1}^{r_\infty-4}k \tau_{\infty,k+1}\text{Ham}^{(\boldsymbol{\alpha}_{\tau_{\infty,k}})}(\hat{\mathbf{q}},\hat{\mathbf{p}})-\sum_{k=1}^{r_\infty-3}k\tau_{\infty,k}\text{Ham}^{(\boldsymbol{\alpha}_{\tau_{\infty,k}})}(\hat{\mathbf{q}},\hat{\mathbf{p}})\right)\cr
&&-\frac{(r_\infty-2)}{(r_\infty-1)T_2^{r_\infty-1}}T_1 \text{Ham}^{(\boldsymbol{\alpha}_{T_1})}(\hat{\mathbf{q}},\hat{\mathbf{p}})-\frac{(r_\infty-2)}{(r_\infty-1)T_2^{r_\infty-1}}T_1\sum_{s=1}^n \text{Ham}^{(\boldsymbol{\alpha}_{\td{X}_s})}(\hat{\mathbf{q}},\hat{\mathbf{p}})\cr
&&-2\frac{(r_\infty-2)^2}{(r_\infty-1)T_2^{r_\infty-1}}T_1\sum_{j=1}^{r_\infty-3}(-1)^{j+r_\infty}\binom{r_\infty-3}{j-1}T_1^{r_\infty-2-j}\text{Ham}^{(\boldsymbol{\alpha}_{\tau_{\infty,j}})}(\hat{\mathbf{q}},\hat{\mathbf{p}})\cr
&&+\frac{(r_\infty-2)}{(r_\infty-1)T_2^{r_\infty-1}}T_1\sum_{k=1}^{r_\infty-4}k\tau_{\infty,k+1}\text{Ham}^{(\boldsymbol{\alpha}_{\tau_{\infty,k}})}(\hat{\mathbf{q}},\hat{\mathbf{p}}).\cr
&&
\eea}

\normalsize{Note in particular that we have}
\footnotesize{\bea \label{Hamspecial}
&&\text{Ham}^{(\mathbf{w}_{\infty,r_\infty-1})}(\hat{\mathbf{q}},\hat{\mathbf{p}})=\frac{1}{(r_\infty-1)T_2^{r_\infty-2}}\text{Ham}^{(\boldsymbol{\alpha}_{T_2})}(\hat{\mathbf{q}},\hat{\mathbf{p}})-\frac{(r_\infty-2)}{(r_\infty-1)T_2^{r_\infty-1}}T_1 \text{Ham}^{(\boldsymbol{\alpha}_{T_1})}(\hat{\mathbf{q}},\hat{\mathbf{p}})\cr
&& +\frac{1}{(r_\infty-1)T_2^{r_\infty-1}}\sum_{s=1}^n\sum_{k=1}^{r_s-1}k\tau_{X_s,k}\text{Ham}^{(\boldsymbol{\alpha}_{\tau_{X_s,k}})}(\hat{\mathbf{q}},\hat{\mathbf{p}})+\frac{1}{(r_\infty-1)T_2^{r_\infty-1}}\sum_{s=1}^n(\td{X}_s-(r_\infty-1)T_1)\text{Ham}^{(\boldsymbol{\alpha}_{\td{X}_s})}(\hat{\mathbf{q}},\hat{\mathbf{p}})\cr
&&+\frac{1}{T_2^{r_\infty-1}}T_1\sum_{k=1}^{r_\infty-4}k\tau_{\infty,k+1}\text{Ham}^{(\boldsymbol{\alpha}_{\tau_{\infty,k}})}(\hat{\mathbf{q}},\hat{\mathbf{p}})-\frac{1}{(r_\infty-1)T_2^{r_\infty-1}}\sum_{k=1}^{r_\infty-3}k\tau_{\infty,k}\text{Ham}^{(\boldsymbol{\alpha}_{\tau_{\infty,k}})}(\hat{\mathbf{q}},\hat{\mathbf{p}})\cr
&&-\frac{2}{T_2^{r_\infty-1}}\sum_{m=1}^{r_\infty-3}(-1)^{m+r_\infty}T_1^{r_\infty-1-m}\binom{r_\infty-2}{m-1} \text{Ham}^{(\boldsymbol{\alpha}_{\tau_{\infty,m}})}(\hat{\mathbf{q}},\hat{\mathbf{p}}).\cr
&&
\eea}
\normalsize{Combining} \eqref{SecondHamrinftygeq3} with \eqref{Haminverted} and \eqref{Hamspecial}, we obtain: 
\small{\bea\label{Hamconverted} \text{Ham}^{(\mathbf{w}_{\infty,k})}(\hat{\mathbf{q}},\hat{\mathbf{p}})&=&2T_2^{-k}\sum_{j=1}^k(-1)^{j+k}\binom{k-1}{j-1}T_1^{k-j}\text{Ham}^{(\boldsymbol{\alpha}_{\tau_{\infty,j}})}(\hat{\mathbf{q}},\hat{\mathbf{p}}) \,,\,\,\forall \, k\in \llbracket 1,r_\infty-3\rrbracket\cr
\text{Ham}^{(\mathbf{w}_{X_s,k})}(\hat{\mathbf{q}},\hat{\mathbf{p}})&=&2T_2^{k}\text{Ham}^{(\boldsymbol{\alpha}_{\tau_{X_s,k}})}(\hat{\mathbf{q}},\hat{\mathbf{p}})\,,\,\,\forall\, (s,k)\in \llbracket 1,n\rrbracket\times\llbracket 1,r_s-1\rrbracket \cr
\text{Ham}^{(\mathbf{w}_s)}(\hat{\mathbf{q}},\hat{\mathbf{p}})&=&T_2\text{Ham}^{(\boldsymbol{\alpha}_{\td{X}_s})}(\hat{\mathbf{q}},\hat{\mathbf{p}})\,,\,\, \forall\, s\in \llbracket 1,n\rrbracket\cr
\text{Ham}^{(\mathbf{w}_{\infty,r_\infty-2})}(\hat{\mathbf{q}},\hat{\mathbf{p}})&=&\frac{1}{(r_\infty-2)T_2^{r_\infty-1}}\text{Ham}^{(\mathbf{b})}(\hat{\mathbf{q}},\hat{\mathbf{p}})+\frac{1}{(r_\infty-2)T_2^{r_\infty-2}}\sum_{s=1}^n\text{Ham}^{(\boldsymbol{\alpha}_{\td{X}_s})}(\hat{\mathbf{q}},\hat{\mathbf{p}})\cr
&&+\frac{2}{T_2^{r_\infty-2}}\sum_{j=1}^{r_\infty-3}(-1)^{j+r_\infty}\binom{r_\infty-3}{j-1}T_1^{r_\infty-2-j}\text{Ham}^{(\boldsymbol{\alpha}_{\tau_{\infty,j}})}(\hat{\mathbf{q}},\hat{\mathbf{p}})\cr
&&-\frac{1}{(r_\infty-2)T_2^{r_\infty-2}}\sum_{k=1}^{r_\infty-4}k \tau_{\infty,k+1}\text{Ham}^{(\boldsymbol{\alpha}_{\tau_{\infty,k}})}(\hat{\mathbf{q}},\hat{\mathbf{p}})\cr
\text{Ham}^{(\mathbf{w}_{\infty,r_\infty-1})}(\hat{\mathbf{q}},\hat{\mathbf{p}})&=&\frac{1}{(r_\infty-1)T_2^{r_\infty-1}}\text{Ham}^{(\mathbf{a})}(\hat{\mathbf{q}},\hat{\mathbf{p}})-\frac{(r_\infty-2)}{(r_\infty-1)T_2^{r_\infty}}T_1 \text{Ham}^{(\mathbf{b})}(\hat{\mathbf{q}},\hat{\mathbf{p}})\cr
&& +\frac{1}{(r_\infty-1)T_2^{r_\infty-1}}\sum_{s=1}^n\sum_{k=1}^{r_s-1}k\tau_{X_s,k}\text{Ham}^{(\boldsymbol{\alpha}_{\tau_{X_s,k}})}(\hat{\mathbf{q}},\hat{\mathbf{p}})\cr
&&+\frac{1}{(r_\infty-1)T_2^{r_\infty-1}}\sum_{s=1}^n(\td{X}_s-(r_\infty-1)T_1)\text{Ham}^{(\boldsymbol{\alpha}_{\td{X}_s})}(\hat{\mathbf{q}},\hat{\mathbf{p}})\cr
&&+\frac{1}{T_2^{r_\infty-1}}T_1\sum_{k=1}^{r_\infty-4}k\tau_{\infty,k+1}\text{Ham}^{(\boldsymbol{\alpha}_{\tau_{\infty,k}})}(\hat{\mathbf{q}},\hat{\mathbf{p}})\cr
&&-\frac{1}{(r_\infty-1)T_2^{r_\infty-1}}\sum_{k=1}^{r_\infty-3}k\tau_{\infty,k}\text{Ham}^{(\boldsymbol{\alpha}_{\tau_{\infty,k}})}(\hat{\mathbf{q}},\hat{\mathbf{p}})\cr
&&-\frac{2}{T_2^{r_\infty-1}}\sum_{m=1}^{r_\infty-3}(-1)^{m+r_\infty}T_1^{r_\infty-1-m}\binom{r_\infty-2}{m-1} \text{Ham}^{(\boldsymbol{\alpha}_{\tau_{\infty,m}})}(\hat{\mathbf{q}},\hat{\mathbf{p}}).
\eea}
\normalsize{We} may now use \eqref{Larinftygeq3}, \eqref{Lbrinftygeq3} and \eqref{Hamconverted} to compute \eqref{IntermediateIdentity}. For compactness, we shall drop the notation of the variables $(\hat{\mathbf{q}},\hat{\mathbf{p}})$ in the various Hamiltonians unless necessary. The different pieces arising in \eqref{IntermediateIdentity} are:
\footnotesize{\bea \label{Piece1} &&-\sum_{s=1}^n dX_s\wedge d\text{Ham}^{(X_s)}=-\sum_{s=1}^n\left(-T_2^{-2}(\td{X}_s-T_1)dT_2+T_2^{-1}d\td{X}_s-T_2^{-1}dT_1 \right)\wedge \left(T_2d\text{Ham}^{(\boldsymbol{\alpha}_{\td{X}_s})}+\text{Ham}^{(\boldsymbol{\alpha}_{\td{X}_s})}dT_2\right)\cr
&&=-\sum_{s=1}^n d\td{X}_s\wedge d\text{Ham}^{(\boldsymbol{\alpha}_{\td{X}_s})} +\textcolor{green}{\sum_{s=1}^nT_2^{-1}(\td{X}_s-T_1)dT_2\wedge d\text{Ham}^{(\boldsymbol{\alpha}_{\td{X}_s)}}+\sum_{s=1}^nT_2^{-1}\text{Ham}^{(\boldsymbol{\alpha}_{\td{X}_s})} dT_2 \wedge d\td{X}_s}\cr
&&\textcolor{blue}{+\sum_{s=1}^ndT_1\wedge d\text{Ham}^{(\boldsymbol{\alpha}_{\td{X}_s})}}\textcolor{orange}{+T_2^{-1}\sum_{s=1}^n\text{Ham}^{(\boldsymbol{\alpha}_{\td{X}_s})}dT_1\wedge dT_2}.\cr
&&
\eea
\bea\label{Piece2} &&-\frac{1}{2}\sum_{s=1}^n\sum_{k=1}^{r_s-1} d\hat{T}_{X_s,k}\wedge d\text{Ham}^{(\mathbf{w}_{X_s,k})}\cr
&&=-\frac{1}{2}\sum_{s=1}^n\sum_{k=1}^{r_s-1}\left(-kT_2^{-k-1}\tau_{X_s,k}dT_2+T_2^{-k}d\tau_{X_s,k}\right)\wedge\left(2kT_2^{k-1}\text{Ham}^{(\boldsymbol{\alpha}_{\tau_{X_s,k}})}dT_2+2T_2^{k}d\text{Ham}^{(\boldsymbol{\alpha}_{\tau_{X_s,k}})}\right)\cr
&&=-\sum_{s=1}^n\sum_{k=1}^{r_s-1}d\tau_{X_s,k}\wedge d\text{Ham}^{(\boldsymbol{\alpha}_{\tau_{X_s,k}})}\cr
&&\textcolor{pink}{+T_2^{-1}\sum_{s=1}^n\sum_{k=1}^{r_s-1}k\tau_{X_s,k}dT_2\wedge d\text{Ham}^{(\boldsymbol{\alpha}_{\tau_{X_s,k}})}+T_2^{-1} \sum_{s=1}^n\sum_{k=1}^{r_s-1} k\text{Ham}^{(\boldsymbol{\alpha}_{\tau_{X_s,k}})}dT_2\wedge d\tau_{X_s,k}}.\cr
&&
\eea
\bea \label{Piece3}&&-\frac{1}{2}d\hat{T}_{\infty,r_\infty-1}\wedge d\text{Ham}^{(\mathbf{w}_{\infty,r_\infty-1})}=-T_2^{-1}dT_2\wedge d \text{Ham}^{(\mathbf{a})}+(r_\infty-2)T_2^{-2}dT_2\wedge d(T_1\text{Ham}^{(\mathbf{b})})\cr
&&\textcolor{pink}{-T_2^{-1}\sum_{s=1}^n\sum_{k=1}^{r_s-1}kdT_2\wedge d(\tau_{X_s,k}\text{Ham}^{(\boldsymbol{\alpha}_{\tau_{X_s,k}})})}\textcolor{green}{-T_2^{-1}\sum_{s=1}^ndT_2\wedge d(\td{X}_s\text{Ham}^{(\boldsymbol{\alpha}_{\td{X}_s})} )}\cr
&&\textcolor{orange}{+(r_\infty-1) T_2^{-1}\sum_{s=1}^n\text{Ham}^{(\boldsymbol{\alpha}_{\td{X}_s})} dT_2\wedge d T_1}\textcolor{green}{+(r_\infty-1) T_2^{-1}T_1 \sum_{s=1}^ndT_2\wedge d\text{Ham}^{(\boldsymbol{\alpha}_{\td{X}_s})}}\cr
&&+(r_\infty-1)T_2^{-1}\sum_{k=1}^{r_\infty-4}k\tau_{\infty,k+1}\text{Ham}^{(\boldsymbol{\alpha}_{\tau_{\infty,k}})} dT_1\wedge dT_2- (r_\infty-1)T_2^{-1}T_1\sum_{k=1}^{r_\infty-4}kdT_2 \wedge d(\tau_{\infty,k+1}\text{Ham}^{(\boldsymbol{\alpha}_{\tau_{\infty,k}})})\cr
&&+T_2^{-1}dT_2\wedge\sum_{k=1}^{r_\infty-3}kd(\tau_{\infty,k}\text{Ham}^{(\boldsymbol{\alpha}_{\tau_{\infty,k}})})\cr
&&+2(r_\infty-1)T_2^{-1} dT_2 \wedge \sum_{m=1}^{r_\infty-3}(-1)^{m+r_\infty}\binom{r_\infty-2}{m-1}d\left(T_1^{r_\infty-1-m} \text{Ham}^{(\boldsymbol{\alpha}_{\tau_{\infty,m}})}\right)\cr
&&
\eea
\bea\label{Piece4}  &&-\frac{1}{2}d\hat{T}_{\infty,r_\infty-2}\wedge d\text{Ham}^{(\mathbf{w}_{\infty,r_\infty-2})}=-\left((r_\infty-2)T_2^{r_\infty-2}dT_1+(r_\infty-2)^2T_1T_2^{r_\infty-3}dT_2\right)\wedge \left( d\text{Ham}^{(\mathbf{w}_{\infty,r_\infty-2})}\right)\cr
&&=-T_2^{-1}dT_1 \wedge d\text{Ham}^{(\mathbf{b})}\textcolor{blue}{-\sum_{s=1}^n dT_1 \wedge  d\text{Ham}^{(\boldsymbol{\alpha}_{\td{X}_s})}}\cr
&&+2T_2^{-1}(r_\infty-2)^2\sum_{j=1}^{r_\infty-3}(-1)^{j+r_\infty}\binom{r_\infty-3}{j-1}T_1^{r_\infty-2-j}\text{Ham}^{(\boldsymbol{\alpha}_{\tau_{\infty,j}})} dT_1\wedge dT_2\cr
&&-2(r_\infty-2) \sum_{j=1}^{r_\infty-3}(-1)^{j+r_\infty}\binom{r_\infty-3}{j-1}T_1^{r_\infty-2-j} dT_1 \wedge d\text{Ham}^{(\boldsymbol{\alpha}_{\tau_{\infty,j}})}+\sum_{k=1}^{r_\infty-4}k  dT_1 \wedge d(\tau_{\infty,k+1}\text{Ham}^{(\boldsymbol{\alpha}_{\tau_{\infty,k}})})\cr
&&-(r_\infty-2)T_1 T_2^{-2}dT_2 \wedge d\text{Ham}^{(\mathbf{b})}\textcolor{green}{-(r_\infty-2)T_2^{-1}T_1\sum_{s=1}^n dT_2 \wedge d\text{Ham}^{(\boldsymbol{\alpha}_{\td{X}_s})}}\cr
&&-2(r_\infty-2)^2T_1T_2^{-1} \sum_{j=1}^{r_\infty-3}(-1)^{j+r_\infty}\binom{r_\infty-3}{j-1} dT_2 \wedge d(T_1^{r_\infty-2-j}\text{Ham}^{(\boldsymbol{\alpha}_{\tau_{\infty,j}})})\cr
&&+(r_\infty-2)T_1T_2^{-1}\sum_{k=1}^{r_\infty-4}k dT_2 \wedge d(\tau_{\infty,k+1}\text{Ham}^{(\boldsymbol{\alpha}_{\tau_{\infty,k}})})\cr
&&+(r_\infty-1)T_2^{-2}\text{Ham}^{(\mathbf{b})} dT_1\wedge dT_2 \textcolor{orange}{+(r_\infty-2)T_2^{-1}\sum_{s=1}^n \text{Ham}^{(\boldsymbol{\alpha}_{\td{X}_s})}dT_1\wedge dT_2}\cr
&&-(r_\infty-2)T_2^{-1}\sum_{k=1}^{r_\infty-4}k\tau_{\infty,k+1}\text{Ham}^{(\boldsymbol{\alpha}_{\tau_{\infty,k}})} dT_1\wedge dT_2.\cr
&&
\eea}
\normalsize{Let} us observe that we have
\footnotesize{\bea &&-T_2^{-1}dT_2\wedge d \text{Ham}^{(\mathbf{a})}+(r_\infty-2)T_2^{-2}dT_2\wedge d(T_1\text{Ham}^{(\mathbf{b})})-T_2^{-1}dT_1 \wedge d\text{Ham}^{(\mathbf{b})}-(r_\infty-2)T_1 T_2^{-2}dT_2 \wedge d\text{Ham}^{(\mathbf{b})}\cr
&&=-T_2^{-1}dT_2\wedge d \text{Ham}^{(\mathbf{a})}-T_2^{-1}dT_1 \wedge d\text{Ham}^{(\mathbf{b})}+(r_\infty-2)T_2^{-2}\text{Ham}^{(\mathbf{b})} dT_2\wedge dT_1.\cr
&&
\eea}
\normalsize{Moreover}, it is a trivial check to see that the terms in \eqref{Piece1}, \eqref{Piece2}, \eqref{Piece3} and \eqref{Piece4} simplify so that we obtain the intermediate identity:
\footnotesize{\bea \label{IntermediateIdentityrgeq3}&&-\sum_{s=1}^n dX_s\wedge d\text{Ham}^{(X_s)}-\frac{1}{2}\sum_{s=1}^n\sum_{k=1}^{r_s-1} d\hat{T}_{X_s,k}\wedge d\text{Ham}^{(\mathbf{w}_{X_s,k})}-\frac{1}{2}d\hat{T}_{\infty,r_\infty-1}\wedge d\text{Ham}^{(\mathbf{w}_{\infty,r_\infty-1})}\cr
&&-\frac{1}{2}d\hat{T}_{\infty,r_\infty-2}\wedge d\text{Ham}^{(\mathbf{w}_{\infty,r_\infty-2})}=-\sum_{s=1}^n\sum_{k=1}^{r_s-1}d\tau_{X_s,k}\wedge d\text{Ham}^{(\boldsymbol{\alpha}_{\tau_{X_s,k}})}-\sum_{s=1}^n d\td{X}_s\wedge d\text{Ham}^{(\boldsymbol{\alpha}_{\td{X}_s})}\cr
&&-T_2^{-1}dT_2\wedge d \text{Ham}^{(\mathbf{a})}-T_2^{-1}dT_1 \wedge d\text{Ham}^{(\mathbf{b})}+(r_\infty-2)T_2^{-2}\text{Ham}^{(\mathbf{b})} dT_2\wedge dT_1\cr
&&+(r_\infty-1)T_2^{-1}\sum_{k=1}^{r_\infty-4}k\tau_{\infty,k+1}\text{Ham}^{(\boldsymbol{\alpha}_{\tau_{\infty,k}})} dT_1\wedge dT_2\cr
&&\textcolor{pink}{- (r_\infty-1)T_2^{-1}T_1\sum_{k=1}^{r_\infty-4}kdT_2 \wedge d(\tau_{\infty,k+1}\text{Ham}^{(\boldsymbol{\alpha}_{\tau_{\infty,k}})})}+T_2^{-1}dT_2\wedge\sum_{k=1}^{r_\infty-3}kd(\tau_{\infty,k}\text{Ham}^{(\boldsymbol{\alpha}_{\tau_{\infty,k}})})\cr
&&+2(r_\infty-1)T_2^{-1} dT_2 \wedge \sum_{m=1}^{r_\infty-3}(-1)^{m+r_\infty}d\left(T_1^{r_\infty-1-m}\binom{r_\infty-2}{m-1} \text{Ham}^{(\boldsymbol{\alpha}_{\tau_{\infty,m}})}\right)\cr
&&+2T_2^{-1}(r_\infty-2)^2\sum_{j=1}^{r_\infty-3}(-1)^{j+r_\infty}\binom{r_\infty-3}{j-1}T_1^{r_\infty-2-j}\text{Ham}^{(\boldsymbol{\alpha}_{\tau_{\infty,j}})} dT_1\wedge dT_2\cr
&&\textcolor{brown}{-2(r_\infty-2) \sum_{j=1}^{r_\infty-3}(-1)^{j+r_\infty}\binom{r_\infty-3}{j-1}T_1^{r_\infty-2-j} dT_1 \wedge d\text{Ham}^{(\boldsymbol{\alpha}_{\tau_{\infty,j}})}}\textcolor{yellow}{+\sum_{k=1}^{r_\infty-4}k  dT_1 \wedge d(\tau_{\infty,k+1}\text{Ham}^{(\boldsymbol{\alpha}_{\tau_{\infty,k}})})}\cr
&&-2(r_\infty-2)^2T_1T_2^{-1} \sum_{j=1}^{r_\infty-3}(-1)^{j+r_\infty}\binom{r_\infty-3}{j-1} dT_2 \wedge d( T_1^{r_\infty-2-j}\text{Ham}^{(\boldsymbol{\alpha}_{\tau_{\infty,j}})})\cr
&&\textcolor{pink}{+(r_\infty-2)T_1T_2^{-1}\sum_{k=1}^{r_\infty-4}k dT_2 \wedge d(\tau_{\infty,k+1}\text{Ham}^{(\boldsymbol{\alpha}_{\tau_{\infty,k}})})}\cr
&&+(r_\infty-1)T_2^{-2}\text{Ham}^{(\mathbf{b})} dT_1\wedge dT_2-(r_\infty-2)T_2^{-1}\sum_{k=1}^{r_\infty-4}k\tau_{\infty,k+1}\text{Ham}^{(\boldsymbol{\alpha}_{\tau_{\infty,k}})} dT_1\wedge dT_2.\cr
&&
\eea}
\normalsize{In} order to simplify the previous quantity, we need to compute the last missing terms of \eqref{IntermediateIdentity} using \eqref{Hamconverted}. They correspond to
\small{\bea\label{MissingTerm} -\frac{1}{2}\sum_{k=1}^{r_\infty-3}d\hat{T}_{\infty,k}\wedge d\text{Ham}^{(\mathbf{w}_{\infty,k})}
&=&\sum_{k=1}^{r_\infty-3}d\hat{T}_{\infty,k}\wedge\Big(kT_2^{-k-1}\sum_{j=1}^k(-1)^{j+k}\binom{k-1}{j-1}T_1^{k-j}\text{Ham}^{(\boldsymbol{\alpha}_{\tau_{\infty,j}})} dT_2\cr
&&-  T_2^{-k}\sum_{j=1}^k(-1)^{j+k}(k-j)\binom{k-1}{j-1}\text{Ham}^{(\boldsymbol{\alpha}_{\tau_{\infty,j}})}T_1^{k-j-1} d T_1\cr
&&-  T_2^{-k}\sum_{j=1}^k(-1)^{j+k}\binom{k-1}{j-1}T_1^{k-j}d \text{Ham}^{(\boldsymbol{\alpha}_{\tau_{\infty,j}})}\Big).
\eea}
\normalsize{To} obtain compact formulas, we note that we have the following identities for any $(h_j)_{j\geq 1}$ and $(\tau_k)_{k\geq 1}$:
\footnotesize{\bea &&\sum_{k=1}^{r_\infty-3}\sum_{j=1}^kk(r_\infty-1-k)(-1)^{j+k}\binom{r_\infty-2}{k-1}\binom{k-1}{j-1}T_1^{r_\infty-2-j}h_j\cr
&&=-(r_\infty-2)^2\sum_{m=1}^{r_\infty-3}(-1)^{m+r_\infty}\binom{r_\infty-3}{m-1}T_1^{r_\infty-2-m}h_m -(r_\infty-2)(r_\infty-3)T_1h_{r_\infty-3}\cr
&&\sum_{k=1}^{r_\infty-3}\sum_{m=2}^{r_\infty-2-k}\sum_{j=1}^k(-1)^{j+k}k(r_\infty-1-m-k)\binom{r_\infty-2-m}{k-1}\binom{k-1}{j-1}T_1^{r_\infty-2-m-j}\tau_{r_\infty-1-m}h_j\cr
&&=\sum_{k=1}^{r_\infty-4}k^2 \tau_{k+1}h_k -T_1\sum_{k=1}^{r_\infty-5}k(k+1) \tau_{k+2}h_k\cr
&&\sum_{k=1}^{r_\infty-3}\sum_{m=2}^{r_\infty-1-k}\sum_{j=1}^k(-1)^{j+k}k\binom{r_\infty-2-m}{k-1}\binom{k-1}{j-1}T_1^{r_\infty-1-m-j}\tau_{r_\infty-1-m}h_j\cr
&&=\sum_{k=1}^{r_\infty-3}k \tau_{k}h_k -T_1\sum_{k=1}^{r_\infty-4}k \tau_{k+1}h_k\cr
&&\sum_{k=1}^{r_\infty-3}\sum_{j=1}^kk(k-j)(-1)^{j+k}\binom{r_\infty-2}{k-1}\binom{k-1}{j-1}T_1^{r_\infty-2-j}h_j\cr
&&=-\sum_{m=1}^{r_\infty-4}(-1)^{m+r_\infty}(r_\infty-2)(r_\infty^2-(m+5)r_\infty+2m+5 )\binom{r_\infty-3}{m-1}T_1^{r_\infty-2-m}h_m\cr
&&\sum_{k=1}^{r_\infty-3}\sum_{m=2}^{r_\infty-1-k}\sum_{j=1}^k(-1)^{j+k}k(k-j)\binom{r_\infty-2-m}{k-1}\binom{k-1}{j-1}T_1^{r_\infty-2-m-j}\tau_{r_\infty-1-m}h_j\cr
&&=-\sum_{k=1}^{r_\infty-4}k(k+1)\tau_{k+1}h_k+T_1\sum_{k=1}^{r_\infty-5}k(k+1)\tau_{k+2}h_k\cr
&&\sum_{k=1}^{r_\infty-3}\sum_{m=2}^{r_\infty-1-k}\sum_{j=1}^k(-1)^{j+k}(k-j)\binom{r_\infty-2-m}{k-1}\binom{k-1}{j-1}T_1^{r_\infty-2-m-j}\tau_{r_\infty-1-m}h_j
\cr
&&=-\sum_{k=1}^{r_\infty-4}k\tau_{k+1}h_k.
\eea}

\normalsize{Let} us now compute each of the three contributions of \eqref{MissingTerm} using \eqref{differentialrgeq3}. The first contribution is:
\footnotesize{\bea\label{C1} &&C_1=2T_2^{-1}\sum_{k=1}^{r_\infty-3}\sum_{j=1}^k(-1)^{j+k}k(r_\infty-1-k)\binom{k-1}{j-1}\binom{r_\infty-2}{k-1}T_1^{r_\infty-2-j}\text{Ham}^{(\boldsymbol{\alpha}_{\tau_{\infty,j}})}dT_1\wedge dT_2\cr
&&+T_2^{-1}\sum_{k=1}^{r_\infty-3}\sum_{j=1}^k\sum_{m=2}^{r_\infty-2-k}(-1)^{j+k}k(r_\infty-1-m-k)\binom{k-1}{j-1} \binom{r_\infty-2-m}{k-1}T_1^{r_\infty-2-m-j}\text{Ham}^{(\boldsymbol{\alpha}_{\tau_{\infty,j}})}\tau_{\infty,r_\infty-1-m} dT_1\wedge dT_2\cr
&&+T_2^{-1}\sum_{k=1}^{r_\infty-3}\sum_{j=1}^k\sum_{m=2}^{r_\infty-1-k}(-1)^{j+k}k\binom{k-1}{j-1} \binom{r_\infty-2-m}{k-1}T_1^{r_\infty-1-m-j}\text{Ham}^{(\boldsymbol{\alpha}_{\tau_{\infty,j}})}d\tau_{\infty,r_\infty-1-m}\wedge dT_2\cr
&&=-2T_2^{-1}(r_\infty-2)^2\sum_{m=1}^{r_\infty-3}(-1)^{m+r_\infty}\binom{r_\infty-3}{m-1}T_1^{r_\infty-2-m}\text{Ham}^{(\boldsymbol{\alpha}_{\tau_{\infty,m}})}dT_1\wedge dT_2
\cr
&&-2(r_\infty-2)(r_\infty-3)T_2^{-1}T_1\text{Ham}^{(\boldsymbol{\alpha}_{\tau_{\infty,r_\infty-3}})}dT_1\wedge dT_2\cr
&&\textcolor{orange}{+T_2^{-1}\sum_{k=1}^{r_\infty-4}k^2 \tau_{\infty,k+1}\text{Ham}^{(\boldsymbol{\alpha}_{\tau_{\infty,k}})}dT_1\wedge dT_2} \textcolor{green}{-T_2^{-1}T_1\sum_{k=1}^{r_\infty-5}k(k+1) \tau_{\infty,k+2}\text{Ham}^{(\boldsymbol{\alpha}_{\tau_{\infty,k}})}dT_1\wedge dT_2}\cr
&&+T_2^{-1}\sum_{k=1}^{r_\infty-3}k\text{Ham}^{(\boldsymbol{\alpha}_{\tau_{\infty,k}})} d\tau_{\infty,k}\wedge dT_2-\textcolor{pink}{T_2^{-1}T_1\sum_{k=1}^{r_\infty-4}k \text{Ham}^{(\boldsymbol{\alpha}_{\tau_{\infty,k}})} d\tau_{\infty,k+1}\wedge dT_2}.\cr
&&
\eea}
\normalsize{The second contribution is}:
\footnotesize{\bea\label{C2}
&&C_2=-2T_2^{-1}\sum_{k=1}^{r_\infty-3}\sum_{j=1}^k(-1)^{j+k}k(k-j)\binom{k-1}{j-1}\binom{r_\infty-2}{k-1}T_1^{r_\infty-2-j}\text{Ham}^{(\boldsymbol{\alpha}_{\tau_{\infty,j}})} dT_2\wedge d T_1\cr
&&-T_2^{-1}\sum_{k=1}^{r_\infty-3}\sum_{j=1}^k\sum_{m=2}^{r_\infty-1-k}(-1)^{j+k}k(k-j)\binom{k-1}{j-1} \binom{r_\infty-2-m}{k-1}T_1^{r_\infty-2-m-j}\tau_{\infty,r_\infty-1-m}\text{Ham}^{(\boldsymbol{\alpha}_{\tau_{\infty,j}})}dT_2\wedge d T_1\cr
&&-\sum_{k=1}^{r_\infty-3}\sum_{j=1}^k\sum_{m=2}^{r_\infty-1-k} (-1)^{j+k}(k-j)\binom{k-1}{j-1} \binom{r_\infty-2-m}{k-1}T_1^{r_\infty-2-m-j}\text{Ham}^{(\boldsymbol{\alpha}_{\tau_{\infty,j}})}d\tau_{\infty,r_\infty-1-m}\wedge d T_1\cr
&&=2T_2^{-1}\sum_{m=1}^{r_\infty-4}(-1)^{m+r_\infty}(r_\infty-2)(r_\infty^2-(m+5)r_\infty+2m+5)\binom{r_\infty-3}{m-1}T_1^{r_\infty-2-m}\text{Ham}^{(\boldsymbol{\alpha}_{\tau_{\infty,m}})}dT_2\wedge d T_1\cr
&&+T_2^{-1}\sum_{k=1}^{r_\infty-4}k(\textcolor{orange}{k}+1) \tau_{\infty,k+1}\text{Ham}^{(\boldsymbol{\alpha}_{\tau_{\infty,k}})}dT_2\wedge dT_1 \textcolor{green}{-T_2^{-1}T_1\sum_{k=1}^{r_\infty-5}k(k+1)\text{Ham}^{(\boldsymbol{\alpha}_{\tau_{\infty,k}})}\tau_{\infty,k+2} dT_2\wedge dT_1}\cr
&&\textcolor{yellow}{+\sum_{k=1}^{r_\infty-4}k\text{Ham}^{(\boldsymbol{\alpha}_{\tau_{\infty,k}})}d\tau_{\infty,k+1}\wedge dT_1}.\cr
&&
\eea}
\normalsize{The third contribution is}
\footnotesize{\bea\label{C3}
&&C_3=-2T_2^{-1}\sum_{k=1}^{r_\infty-3}\sum_{j=1}^k(-1)^{j+k}k\binom{k-1}{j-1}\binom{r_\infty-2}{k-1}T_1^{r_\infty-1-j}dT_2\wedge d \text{Ham}^{(\boldsymbol{\alpha}_{\tau_{\infty,j}})}\cr
&&-T_2^{-1}\sum_{k=1}^{r_\infty-3}\sum_{j=1}^k\sum_{m=2}^{r_\infty-1-k} (-1)^{j+k}k\binom{k-1}{j-1}\binom{r_\infty-2-m}{k-1}T_1^{r_\infty-1-m-j}\tau_{\infty,r_\infty-1-m}dT_2\wedge d \text{Ham}^{(\boldsymbol{\alpha}_{\tau_{\infty,j}})}\cr
&&-2\sum_{k=1}^{r_\infty-3}\sum_{j=1}^k(-1)^{j+k}\binom{k-1}{j-1}(r_\infty-1-k)\binom{r_\infty-2}{k-1}T_1^{r_\infty-2-j}dT_1\wedge d \text{Ham}^{(\boldsymbol{\alpha}_{\tau_{\infty,j}})}\cr
&&-\sum_{k=1}^{r_\infty-3}\sum_{j=1}^k\sum_{m=2}^{r_\infty-2-k} (-1)^{j+k}(r_\infty-1-m-k)\binom{k-1}{j-1}\binom{r_\infty-2-m}{k-1}T_1^{r_\infty-2-m-j}\tau_{\infty,r_\infty-1-m}dT_1\wedge d \text{Ham}^{(\boldsymbol{\alpha}_{\tau_{\infty,j}})}\cr
&&-\sum_{k=1}^{r_\infty-3}\sum_{j=1}^k\sum_{m=2}^{r_\infty-1-k} \binom{k-1}{j-1}\binom{r_\infty-2-m}{k-1}T_1^{r_\infty-1-m-j}(-1)^{j+k} d\tau_{\infty,r_\infty-1-m}\wedge d \text{Ham}^{(\boldsymbol{\alpha}_{\tau_{\infty,j}})}\cr
&&= 2T_2^{-1}\sum_{m=1}^{r_\infty-3} (-1)^{m+r_\infty}(r_\infty^2-(m+4)r_\infty+2m+3) \binom{r_\infty-2}{m-1}T_1^{r_\infty-1-m}dT_2\wedge d \text{Ham}^{(\boldsymbol{\alpha}_{\tau_{\infty,m}})}\cr
&&-T_2^{-1}\sum_{k=1}^{r_\infty-3} k \tau_{\infty,k} dT_2\wedge d \text{Ham}^{(\boldsymbol{\alpha}_{\tau_{\infty,k}})} \textcolor{pink}{+T_2^{-1}T_1\sum_{k=1}^{r_\infty-4} k \tau_{\infty,k+1} dT_2\wedge d \text{Ham}^{(\boldsymbol{\alpha}_{\tau_{\infty,k}})}}\cr
&&\textcolor{brown}{+2\sum_{m=1}^{r_\infty-3} (-1)^{m+r_\infty}(r_\infty-2) \binom{r_\infty-3}{m-1}T_1^{r_\infty-2-m}dT_1\wedge d \text{Ham}^{(\boldsymbol{\alpha}_{\tau_{\infty,m}})}}\cr
&&\textcolor{yellow}{-\sum_{k=1}^{r_\infty-4} k\tau_{\infty,k+1} dT_1\wedge d \text{Ham}^{(\boldsymbol{\alpha}_{\tau_{\infty,k}})}}-\sum_{k=1}^{r_\infty-3} d\tau_{\infty,k}\wedge d\text{Ham}^{(\boldsymbol{\alpha}_{\tau_{\infty,k}})}.\cr
&&
\eea}
\normalsize{Observe} that some of the terms of \eqref{IntermediateIdentityrgeq3}, \eqref{C1}, \eqref{C2} and \eqref{C3} simplify so that we are left with
\footnotesize{\bea
&&-\sum_{s=1}^n dX_s\wedge d\text{Ham}^{(X_s)}-\frac{1}{2}\sum_{s=1}^n\sum_{k=1}^{r_s-1} d\hat{T}_{X_s,k}\wedge d\text{Ham}^{(\mathbf{w}_{X_s,k})}-\frac{1}{2}\sum_{k=1}^{r_\infty-1}d\hat{T}_{\infty,k}\wedge d\text{Ham}^{(\mathbf{w}_{\infty,k})}\cr
&&=-\sum_{s=1}^n\sum_{k=1}^{r_s-1}d\tau_{X_s,k}\wedge d\text{Ham}^{(\boldsymbol{\alpha}_{\tau_{X_s,k}})}-\sum_{s=1}^n d\td{X}_s\wedge d\text{Ham}^{(\boldsymbol{\alpha}_{\td{X}_s})}-\sum_{k=1}^{r_\infty-3} d\tau_{\infty,k}\wedge d\text{Ham}^{(\boldsymbol{\alpha}_{\tau_{\infty,k}})}\cr
&&-T_2^{-1}dT_2\wedge d \text{Ham}^{(\mathbf{a})}-T_2^{-1}dT_1 \wedge d\text{Ham}^{(\mathbf{b})}\textcolor{orange}{-(r_\infty-1)T_2^{-2}\text{Ham}^{(\mathbf{b})} dT_1\wedge dT_2}+T_2^{-2}\text{Ham}^{(\mathbf{b})} dT_1\wedge dT_2\cr
&&\textcolor{red}{+(r_\infty-1)T_2^{-1}\sum_{k=1}^{r_\infty-4}k\tau_{\infty,k+1}\text{Ham}^{(\boldsymbol{\alpha}_{\tau_{\infty,k}})} dT_1\wedge dT_2}\textcolor{green}{+T_2^{-1}dT_2\wedge\sum_{k=1}^{r_\infty-3}kd(\tau_{\infty,k}\text{Ham}^{(\boldsymbol{\alpha}_{\tau_{\infty,k}})})}\cr
&&+2(r_\infty-1)T_2^{-1} dT_2 \wedge \sum_{m=1}^{r_\infty-3}(-1)^{m+r_\infty}d\left(T_1^{r_\infty-1-m}\binom{r_\infty-2}{m-1} \text{Ham}^{(\boldsymbol{\alpha}_{\tau_{\infty,m}})}\right)\cr
&&+2T_2^{-1}(r_\infty-2)^2\sum_{j=1}^{r_\infty-3}(-1)^{j+r_\infty}\binom{r_\infty-3}{j-1}T_1^{r_\infty-2-j}\text{Ham}^{(\boldsymbol{\alpha}_{\tau_{\infty,j}})} dT_1\wedge dT_2\cr
&&-2(r_\infty-2)^2T_1T_2^{-1} \sum_{j=1}^{r_\infty-3}(-1)^{j+r_\infty}\binom{r_\infty-3}{j-1} dT_2 \wedge d(T_1^{r_\infty-2-j}\text{Ham}^{(\boldsymbol{\alpha}_{\tau_{\infty,j}})})\cr
&&-2T_2^{-1}(r_\infty-2)^2\sum_{m=1}^{r_\infty-3}(-1)^{m+r_\infty}\binom{r_\infty-3}{m-1}T_1^{r_\infty-2-m}\text{Ham}^{(\boldsymbol{\alpha}_{\tau_{\infty,m}})}dT_1\wedge dT_2\cr
&&-2T_2^{-1}T_1(r_\infty-2)(r_\infty-3)\text{Ham}^{(\boldsymbol{\alpha}_{\tau_{\infty,r_\infty-3}})}dT_1\wedge dT_2\cr
&&+2T_2^{-1}\sum_{m=1}^{r_\infty-4}(-1)^{m+r_\infty}(r_\infty-2)(r_\infty^2-(m+5)r_\infty+2m+5 )\binom{r_\infty-3}{m-1}T_1^{r_\infty-2-m}\text{Ham}^{(\boldsymbol{\alpha}_{\tau_{\infty,m}})}dT_2\wedge d T_1\cr
&&\textcolor{green}{+T_2^{-1}\sum_{k=1}^{r_\infty-3}k\text{Ham}^{(\boldsymbol{\alpha}_{\tau_{\infty,k}})} d\tau_{\infty,k}\wedge dT_2}\textcolor{red}{+T_2^{-1}\sum_{k=1}^{r_\infty-4}k \tau_{\infty,k+1}\text{Ham}^{(\boldsymbol{\alpha}_{\tau_{\infty,k}})}dT_2\wedge dT_1}\cr
&&+2T_2^{-1}\sum_{m=1}^{r_\infty-3} (-1)^{m+r_\infty}(r_\infty^2-(m+4)r_\infty+2m+3) \binom{r_\infty-2}{m-1}T_1^{r_\infty-1-m}dT_2\wedge d \text{Ham}^{(\boldsymbol{\alpha}_{\tau_{\infty,m}})}\cr
&&\textcolor{green}{-T_2^{-1}\sum_{k=1}^{r_\infty-3} k \tau_{\infty,k} dT_2\wedge d \text{Ham}^{(\boldsymbol{\alpha}_{\tau_{\infty,k}})}}\textcolor{orange}{+(r_\infty-1)T_2^{-2}\text{Ham}^{(\mathbf{b})} dT_1\wedge dT_2}\cr
&&\textcolor{red}{-(r_\infty-2)T_2^{-1}\sum_{k=1}^{r_\infty-4}k\tau_{\infty,k+1}\text{Ham}^{(\boldsymbol{\alpha}_{\tau_{\infty,k}})} dT_1\wedge dT_2}.\cr
&&
\eea}
\normalsize{Note that many terms simplify} so that the last identity is equivalent to:
\footnotesize{\bea 
&&-\sum_{s=1}^n dX_s\wedge d\text{Ham}^{(X_s)}-\frac{1}{2}\sum_{s=1}^n\sum_{k=1}^{r_s-1} d\hat{T}_{X_s,k}\wedge d\text{Ham}^{(\mathbf{w}_{X_s,k})}-\frac{1}{2}\sum_{k=1}^{r_\infty-1}d\hat{T}_{\infty,k}\wedge d\text{Ham}^{(\mathbf{w}_{\infty,k})}\cr
&&=-\sum_{s=1}^n\sum_{k=1}^{r_s-1}d\tau_{X_s,k}\wedge d\text{Ham}^{(\boldsymbol{\alpha}_{\tau_{X_s,k}})}-\sum_{s=1}^n d\td{X}_s\wedge d\text{Ham}^{(\boldsymbol{\alpha}_{\td{X}_s})}-\sum_{k=1}^{r_\infty-3} d\tau_{\infty,k}\wedge d\text{Ham}^{(\boldsymbol{\alpha}_{\tau_{\infty,k}})}\cr
&&-T_2^{-1}dT_2\wedge d \text{Ham}^{(\mathbf{a})}-T_2^{-1}dT_1 \wedge d\text{Ham}^{(\mathbf{b})}+T_2^{-2}\text{Ham}^{(\mathbf{b})} dT_1\wedge dT_2\cr
&&-2(r_\infty-1)T_2^{-1}  \sum_{m=1}^{r_\infty-3}(-1)^{m+r_\infty}(r_\infty-1-m)\binom{r_\infty-2}{m-1}T_1^{r_\infty-2-m}\text{Ham}^{(\boldsymbol{\alpha}_{\tau_{\infty,m}})} dT_1 \wedge dT_2\cr
&&+2(r_\infty-1)T_2^{-1}  \sum_{m=1}^{r_\infty-3}(-1)^{m+r_\infty}\binom{r_\infty-2}{m-1} T_1^{r_\infty-1-m}dT_2 \wedge d\text{Ham}^{(\boldsymbol{\alpha}_{\tau_{\infty,m}})}\cr
&&+2T_2^{-1}(r_\infty-2)^2\sum_{j=1}^{r_\infty-3}(-1)^{j+r_\infty}\binom{r_\infty-3}{j-1}T_1^{r_\infty-2-j}\text{Ham}^{(\boldsymbol{\alpha}_{\tau_{\infty,j}})} dT_1\wedge dT_2\cr
&&+2(r_\infty-2)^2T_2^{-1} \sum_{j=1}^{r_\infty-3}(-1)^{j+r_\infty}(r_\infty-2-j)\binom{r_\infty-3}{j-1}T_1^{r_\infty-2-j}\text{Ham}^{(\boldsymbol{\alpha}_{\tau_{\infty,j}})} dT_1 \wedge dT_2\cr
&&-2(r_\infty-2)^2T_2^{-1} \sum_{j=1}^{r_\infty-3}(-1)^{j+r_\infty}\binom{r_\infty-3}{j-1}T_1^{r_\infty-1-j} dT_2 \wedge d\text{Ham}^{(\boldsymbol{\alpha}_{\tau_{\infty,j}})}\cr
&&-2T_2^{-1}(r_\infty-2)^2\sum_{m=1}^{r_\infty-3}(-1)^{m+r_\infty}\binom{r_\infty-3}{m-1}T_1^{r_\infty-2-m}\text{Ham}^{(\boldsymbol{\alpha}_{\tau_{\infty,m}})}dT_1\wedge dT_2 \cr
&&-2(r_\infty-2)(r_\infty-3) T_2^{-1}T_1 \text{Ham}^{(\boldsymbol{\alpha}_{\tau_{\infty,r_\infty-3}})}dT_1\wedge dT_2 \cr
&&-2T_2^{-1}\sum_{m=1}^{r_\infty-4}(-1)^{m+r_\infty}(r_\infty-2)(r_\infty^2-(m+5)r_\infty+2m+5 )\binom{r_\infty-3}{m-1}T_1^{r_\infty-2-m}\text{Ham}^{(\boldsymbol{\alpha}_{\tau_{\infty,m}})}dT_1\wedge d T_2\cr
&&+2T_2^{-1}\sum_{m=1}^{r_\infty-3} (-1)^{m+r_\infty}(r_\infty^2-(m+4)r_\infty+2m+3) \binom{r_\infty-2}{m-1}T_1^{r_\infty-1-m}dT_2\wedge d \text{Ham}^{(\boldsymbol{\alpha}_{\tau_{\infty,m}})}\cr
&&= -\sum_{s=1}^n\sum_{k=1}^{r_s-1}d\tau_{X_s,k}\wedge d\text{Ham}^{(\boldsymbol{\alpha}_{\tau_{X_s,k}})}-\sum_{s=1}^n d\td{X}_s\wedge d\text{Ham}^{(\boldsymbol{\alpha}_{\td{X}_s})}-\sum_{k=1}^{r_\infty-3} d\tau_{\infty,k}\wedge d\text{Ham}^{(\boldsymbol{\alpha}_{\tau_{\infty,k}})}\cr
&&-T_2^{-1}dT_2\wedge d \text{Ham}^{(\mathbf{a})}-T_2^{-1}dT_1 \wedge d\text{Ham}^{(\mathbf{b})}+T_2^{-2}\text{Ham}^{(\mathbf{b})} dT_1\wedge dT_2\cr
&&
\eea}
\normalsize{where} terms proportional to $dT_2\wedge d \text{Ham}^{(\tau_{\infty,m})}$ easily cancel while most terms proportional to $dT_1\wedge dT_2$ also cancel (note that the last line does not contribute for $m=r_\infty-3$) because we have for all  $m\in \llbracket 1,r_\infty-4\rrbracket$ the identity:
\footnotesize{\beq  (r_\infty-1) (r_\infty-1-m)\binom{r_\infty-2}{m-1} -(r_\infty-2)^2(r_\infty-2-m)\binom{r_\infty-3}{m-1}+(r_\infty-2)(r_\infty^2-(m+5)r_\infty+2m+5) \binom{r_\infty-3}{m-1}=0.  \eeq}
\normalsize{In} the end, using \eqref{IdentityFinal2bis}, we are left with
\bea \Omega&=&\hbar\sum_{j=1}^g d\check{q}_j \wedge d\check{p}_j - \sum_{k=1}^{r_\infty-3} d\tau_{\infty,k}\wedge d\text{Ham}^{(\boldsymbol{\alpha}_{\tau_{\infty,k}})}(\check{\mathbf{q}},\check{\mathbf{p}})\cr
&&- \sum_{s=1}^n\sum_{k=1}^{r_s-1} d\tau_{X_s,k}\wedge d\text{Ham}^{(\boldsymbol{\alpha}_{\tau_{X_s,k}})}(\check{\mathbf{q}},\check{\mathbf{p}})-\sum_{s=1}^n d\td{X}_s\wedge d\text{Ham}^{(\boldsymbol{\alpha}_{\td{X}_s})}(\check{\mathbf{q}},\check{\mathbf{p}}).\eea

\subsection{The case $r_\infty=2$}
Let us study the case $r_\infty=2$. In this case we have from Proposition \ref{InverseRelationsrequal2}:
\bea T_2&=&\frac{1}{2}\hat{T}_{\infty,1},\cr
T_1&=&-\frac{1}{2}\hat{T}_{\infty,1}X_1,\cr
\tau_{X_s,k}&=&\left(\frac{\hat{T}_{\infty,1}}{2}\right)^{k}\hat{T}_{X_s,k}\,\,,\,\, \forall\, (s,k)\in \llbracket 1,n\rrbracket\times \llbracket 1,r_s-1\rrbracket,\cr
\td{X}_s&=&\frac{1}{2}(X_s-X_1)\hat{T}_{\infty,1}\,\,,\,\, \forall\, s\in \llbracket 2,n\rrbracket,
\eea
or equivalently
\bea \hat{T}_{\infty,1}&=&2T_2,\cr
X_1&=&-T_1 T_2^{-1},\cr
X_s&=&(\td{X}_s-T_1)T_2^{-1}\,\,,\,\, \forall\, s\in \llbracket 2,n\rrbracket,\cr
\hat{T}_{X_s,k}&=&T_2^{-k}\tau_{X_s,k} \,\,,\,\, \forall\, (s,k)\in \llbracket 1,n\rrbracket\times \llbracket 1,r_s-1\rrbracket.
\eea
At the level of differentials we obtain:
\bea \label{diffrinftyequal2}d\hat{T}_{\infty,1}&=&2dT_2,\cr
dX_1&=&-T_2^{-1}dT_1+T_1 T_2^{-2}dT_2,\cr
d\hat{T}_{X_s,k}&=&-kT_2^{-k-1}\tau_{X_s,k}dT_2+T_2^{-k}d\tau_{X_s,k} \,\,,\,\, \forall\, (s,k)\in \llbracket 1,n\rrbracket\times \llbracket 1,r_s-1\rrbracket,\cr
dX_s&=&-T_2^{-2}(\td{X}_s-T_1)dT_2-T_2^{-1}dT_1+T_2^{-1}d\td{X}_s\,\,,\,\, \forall\, s\in \llbracket 2,n\rrbracket.\cr
&&
\eea
At the level of Hamiltonians, since $\text{Ham}^{(\mathbf{w}_{\infty,1})}=2\text{Ham}^{(\boldsymbol{\alpha}_{\hat{T}_{\infty,1}})}$ and $\text{Ham}^{(\mathbf{w}_{X_s,k})}=2\text{Ham}^{(\boldsymbol{\alpha}_{\hat{T}_{X_s,k}})}$ for all $(s,k)\in \llbracket 1,n\rrbracket\times \llbracket 1,r_s-1\rrbracket$ this is equivalent to
\bea \text{Ham}^{(\boldsymbol{\alpha}_{T_1})}(\hat{\mathbf{q}},\hat{\mathbf{p}})&=&-T_2^{-1}\sum_{s=1}^n\text{Ham}^{(\mathbf{w}_s)}(\hat{\mathbf{q}},\hat{\mathbf{p}}),\cr
\text{Ham}^{(\boldsymbol{\alpha}_{T_2})}(\hat{\mathbf{q}},\hat{\mathbf{p}})&=&\text{Ham}^{(\mathbf{w}_{\infty,1})}(\hat{\mathbf{q}},\hat{\mathbf{p}})-T_2^{-1}\sum_{s=1}^nX_s \text{Ham}^{(\mathbf{w}_s)}(\hat{\mathbf{q}},\hat{\mathbf{p}})\cr
&&-\frac{1}{2}\sum_{s=1}^n\sum_{k=1}^{r_s-1}kT_2^{-k-1}\tau_{X_s,k}\text{Ham}^{(\mathbf{w}_{X_s,k})}(\hat{\mathbf{q}},\hat{\mathbf{p}}),\cr
\text{Ham}^{(\boldsymbol{\alpha}_{\tau_{X_s,k}})}(\hat{\mathbf{q}},\hat{\mathbf{p}})&=&\frac{1}{2}T_2^{-k}\text{Ham}^{(\mathbf{w}_{X_s,k})}(\hat{\mathbf{q}},\hat{\mathbf{p}})\,\,,\,\, \forall\, (s,k)\in \llbracket 1,n\rrbracket\times \llbracket 1,r_s-1\rrbracket,\cr
\text{Ham}^{(\boldsymbol{\alpha}_{\td{X}_s})}(\hat{\mathbf{q}},\hat{\mathbf{p}})&=&T_2^{-1}\text{Ham}^{(\mathbf{w}_s)}(\hat{\mathbf{q}},\hat{\mathbf{p}})\,\,,\,\, \forall\, s\in \llbracket 2,n\rrbracket.\cr
&&
\eea
A straightforward computation also provides
\bea \partial_{X_1}&=&-T_2\partial_{T_1}-T_2\sum_{s=2}^n\partial_{\td{X}_s},\cr
\partial_{\hat{T}_{\infty,1}}&=&\frac{1}{2}\partial_{T_2}+\frac{1}{2}T_1T_2^{-1}\partial_{T_1}+\frac{1}{2}T_2^{-1}\sum_{s=1}^n\sum_{k=1}^{r_s-1}k\tau_{X_s,k}\partial_{\tau_{X_s,k}}+\frac{1}{2}T_2^{-1}\sum_{s=2}^n\td{X}_s\partial_{\td{X}_s},\cr
\partial_{\hat{T}_{X_s,k}}&=&T_2^k \partial_{\tau_{X_s,k}} \,\,,\,\, \forall\, (s,k)\in \llbracket 1,n\rrbracket\times \llbracket 1,r_s-1\rrbracket,\cr
\partial_{X_s}&=&T_2\partial_{\td{X}_s} \,\,,\,\, \forall\, s\in \llbracket 2,n\rrbracket,\cr
&&
\eea
i.e.
\bea\label{Partialrinftyequal2} \text{Ham}^{(\mathbf{w}_1)}(\hat{\mathbf{q}},\hat{\mathbf{p}})&=&-T_2\text{Ham}^{(\boldsymbol{\alpha}_{T_1})}(\hat{\mathbf{q}},\hat{\mathbf{p}})-T_2\sum_{s=2}^n\text{Ham}^{(\boldsymbol{\alpha}_{\td{X}_s})}(\hat{\mathbf{q}},\hat{\mathbf{p}}),\cr
\text{Ham}^{(\mathbf{w}_{\infty,1})}(\hat{\mathbf{q}},\hat{\mathbf{p}})&=&\text{Ham}^{(T_2)}(\hat{\mathbf{q}},\hat{\mathbf{p}})+T_1T_2^{-1}\text{Ham}^{(\boldsymbol{\alpha}_{T_1})}(\hat{\mathbf{q}},\hat{\mathbf{p}})\cr
&&+T_2^{-1}\sum_{s=1}^n\sum_{k=1}^{r_s-1}k\tau_{X_s,k}\text{Ham}^{(\boldsymbol{\alpha}_{\tau_{X_s,k}})}(\hat{\mathbf{q}},\hat{\mathbf{p}})+T_2^{-1}\sum_{s=2}^n\td{X}_s\text{Ham}^{(\boldsymbol{\alpha}_{\td{X}_s})}(\hat{\mathbf{q}},\hat{\mathbf{p}}),\cr
\text{Ham}^{(\mathbf{w}_{X_s,k})}(\hat{\mathbf{q}},\hat{\mathbf{p}})&=&2T_2^k\text{Ham}^{(\boldsymbol{\alpha}_{\tau_{X_s,k}})}(\hat{\mathbf{q}},\hat{\mathbf{p}}) \,\,,\,\, \forall\, (s,k)\in \llbracket 1,n\rrbracket\times \llbracket 1,r_s-1\rrbracket,\cr
\text{Ham}^{(\mathbf{w}_s)}(\hat{\mathbf{q}},\hat{\mathbf{p}})&=&T_2\text{Ham}^{(\boldsymbol{\alpha}_{\td{X}_s})}(\hat{\mathbf{q}},\hat{\mathbf{p}}) \,\,,\,\, \forall\, s\in \llbracket 2,n\rrbracket.\cr
&&
\eea

Let us now observe that Definition \ref{TrivialVectors} and the fact that $\mathcal{L}_{\mathbf{w}_{\infty,1}}=2\hbar \partial_{\hat{T}_{\infty,1}}$ and $\mathcal{L}_{\mathbf{w}_{X_s,k}}=2\hbar \partial_{\hat{T}_{X_s,k}}$ for all $(s,k)\in \llbracket 1,n\rrbracket\times \llbracket 1,r_s-1\rrbracket$ provides
\beq \mathcal{L}_{\mathbf{a}}=\hbar T_2\partial_{T_2}+\frac{1}{2}T_{\infty,1}\mathcal{L}_{\mathbf{v}_{\infty,1}}-\frac{1}{2}\sum_{s=1}^n\sum_{k=1}^{r_s-1}k T_{X_s,k}\mathcal{L}_{\mathbf{v}_{X_s,k}}\,\,,\,\,\mathcal{L}_{\mathbf{b}}=\hbar T_2\partial_{T_1}
\eeq
i.e. using \eqref{SuperIDHam}
\beq \label{SecondHamrinftyequal2}\text{Ham}^{(\mathbf{b})}(\hat{\mathbf{q}},\hat{\mathbf{p}})=T_2\text{Ham}^{(\boldsymbol{\alpha}_{T_1})}(\hat{\mathbf{q}},\hat{\mathbf{p}})\,\,,\,\,
\text{Ham}^{(\mathbf{a})}(\hat{\mathbf{q}},\hat{\mathbf{p}})=T_2\text{Ham}^{(\boldsymbol{\alpha}_{T_2})}(\hat{\mathbf{q}},\hat{\mathbf{p}}).
\eeq
Combining \eqref{Partialrinftyequal2} and \eqref{SecondHamrinftyequal2}, we get
\bea\label{Partialrinftyequal2new} \text{Ham}^{(\mathbf{w}_1)}(\hat{\mathbf{q}},\hat{\mathbf{p}})&=&-\text{Ham}^{(\mathbf{b})}(\hat{\mathbf{q}},\hat{\mathbf{p}})-T_2\sum_{s=2}^n\text{Ham}^{(\boldsymbol{\alpha}_{\td{X}_s})}(\hat{\mathbf{q}},\hat{\mathbf{p}}),\cr
\text{Ham}^{(\mathbf{w}_{\infty,1})}(\hat{\mathbf{q}},\hat{\mathbf{p}})&=&T_2^{-1}\text{Ham}^{(\mathbf{a})}(\hat{\mathbf{q}},\hat{\mathbf{p}})+T_1T_2^{-2}\text{Ham}^{(\mathbf{b})}(\hat{\mathbf{q}},\hat{\mathbf{p}}),\cr
&&+T_2^{-1}\sum_{s=1}^n\sum_{k=1}^{r_s-1}k\tau_{X_s,k}\text{Ham}^{(\boldsymbol{\alpha}_{\tau_{X_s,k}})}(\hat{\mathbf{q}},\hat{\mathbf{p}})+T_2^{-1}\sum_{s=2}^n\td{X}_s\text{Ham}^{(\boldsymbol{\alpha}_{\td{X}_s})}(\hat{\mathbf{q}},\hat{\mathbf{p}}),\cr
\text{Ham}^{(\mathbf{w}_{X_s,k})}(\hat{\mathbf{q}},\hat{\mathbf{p}})&=&2T_2^k\text{Ham}^{(\boldsymbol{\alpha}_{\tau_{X_s,k}})}(\hat{\mathbf{q}},\hat{\mathbf{p}}) \,\,,\,\, \forall\, (s,k)\in \llbracket 1,n\rrbracket\times \llbracket 1,r_s-1\rrbracket,\cr
\text{Ham}^{(\mathbf{w}_s)}(\hat{\mathbf{q}},\hat{\mathbf{p}})&=&T_2\text{Ham}^{(\boldsymbol{\alpha}_{\td{X}_s})}(\hat{\mathbf{q}},\hat{\mathbf{p}}) \,\,,\,\, \forall\, s\in \llbracket 2,n\rrbracket.\cr
&&
\eea
Let us now insert this result and \eqref{diffrinftyequal2} into \eqref{IntermediateIdentity}:
\small{\bea  &&\frac{1}{2}\sum_{p\in \mathcal{R}}\sum_{k=1}^{r_p-1} d\hat{T}_{p,k}\wedge d\,\text{Ham}^{(\mathbf{w}_{p,k})}(\hat{\mathbf{q}},\hat{\mathbf{p}}) +\sum_{s=1}^n dX_s\wedge d\,\text{Ham}^{(\mathbf{w}_s)}(\hat{\mathbf{q}},\hat{\mathbf{p}})\cr
&=&\frac{1}{2}\left(2dT_2\right)\wedge\Big(T_2^{-1}d\text{Ham}^{(\mathbf{a})}+T_2^{-2}d(T_1\text{Ham}^{(\mathbf{b})})+T_2^{-1}\sum_{s=1}^n\sum_{k=1}^{r_s-1}k\text{Ham}^{(\boldsymbol{\alpha}_{\tau_{X_s,k}})} d\tau_{X_s,k}\cr
&&+T_2^{-1}\sum_{s=1}^n\sum_{k=1}^{r_s-1}k\tau_{X_s,k}d\text{Ham}^{(\boldsymbol{\alpha}_{\tau_{X_s,k}})}+T_2^{-1}\sum_{s=2}^n\text{Ham}^{(\boldsymbol{\alpha}_{\td{X}_s})}d\td{X}_s+T_2^{-1}\sum_{s=2}^n\td{X}_sd\text{Ham}^{(\boldsymbol{\alpha}_{\td{X}_s})}\Big)\cr
&&+\frac{1}{2}\sum_{s=1}^n\sum_{k=1}^{r_s-1}\left(-kT_2^{-k-1}\tau_{X_s,k}dT_2+T_2^{-k}d\tau_{X_s,k}\right)\wedge\left(2kT_2^{k-1}\text{Ham}^{(\boldsymbol{\alpha}_{\tau_{X_s,k}})} dT_2+   2T_2^k d\text{Ham}^{(\boldsymbol{\alpha}_{\tau_{X_s,k}})} \right)\cr
&&+\left(-T_2^{-1}dT_1+T_1 T_2^{-2}dT_2\right)\wedge\left(-d\text{Ham}^{(\mathbf{b})}-\sum_{s=2}^n\text{Ham}^{(\boldsymbol{\alpha}_{\td{X}_s})}dT_2 -T_2\sum_{s=2}^nd\text{Ham}^{(\boldsymbol{\alpha}_{\td{X}_s})}\right)\cr
&&+\sum_{s=2}^n \left(-T_2^{-2}(\td{X}_s-T_1)dT_2-T_2^{-1}dT_1+T_2^{-1}d\td{X}_s\right)\wedge \left(\text{Ham}^{(\boldsymbol{\alpha}_{\td{X}_s})}dT_2+T_2d\text{Ham}^{(\boldsymbol{\alpha}_{\td{X}_s})}\right)\cr
&=& T_2^{-1}dT_2\wedge d\text{Ham}^{(\mathbf{a})} +T_2^{-2}dT_2\wedge d(T_1\text{Ham}^{(\mathbf{b})})+ \left(T_2^{-1}dT_1-T_1 T_2^{-2}dT_2\right)\wedge d\text{Ham}^{(\mathbf{b})}\cr 
&&\textcolor{blue}{+T_2^{-1}\sum_{s=1}^n\sum_{k=1}^{r_s-1}k\text{Ham}^{(\boldsymbol{\alpha}_{\tau_{X_s,k}})} dT_2\wedge d\tau_{X_s,k} +T_2^{-1}\sum_{s=1}^n\sum_{k=1}^{r_s-1}k\tau_{X_s,k}dT_2\wedge  d\text{Ham}^{(\boldsymbol{\alpha}_{\tau_{X_s,k}})}}\cr
&&\textcolor{orange}{+T_2^{-1}\sum_{s=2}^n\text{Ham}^{(\boldsymbol{\alpha}_{\td{X}_s})}dT_2\wedge d\td{X}_s} \textcolor{green}{+T_2^{-1}\sum_{s=2}^n\td{X}_s \,dT_2\wedge d\text{Ham}^{(\boldsymbol{\alpha}_{\td{X}_s})}}+\sum_{s=1}^n\sum_{k=1}^{r_s-1}d\tau_{X_s,k}\wedge d\text{Ham}^{(\boldsymbol{\alpha}_{\tau_{X_s,k}})}\cr
&&\textcolor{blue}{-T_2^{-1}\sum_{s=1}^n\sum_{k=1}^{r_s-1}k\tau_{X_s,k}dT_2\wedge d\text{Ham}^{(\boldsymbol{\alpha}_{\tau_{X_s,k}})}+T_2^{-1}\sum_{s=1}^n\sum_{k=1}^{r_s-1} k\text{Ham}^{(\boldsymbol{\alpha}_{\tau_{X_s,k}})} d\tau_{X_s,k}\wedge dT_2}\cr
&&\textcolor{pink}{+T_2^{-1}\sum_{s=2}^n\text{Ham}^{(\boldsymbol{\alpha}_{\td{X}_s})}dT_1\wedge dT_2+\sum_{s=2}^ndT_1\wedge d\text{Ham}^{(\boldsymbol{\alpha}_{\td{X}_s})}}\textcolor{red}{-T_1 T_2^{-1}\sum_{s=2}^ndT_2\wedge d\text{Ham}^{(\boldsymbol{\alpha}_{\td{X}_s})}}\cr
&&-T_2^{-1}\sum_{s=2}^n(\textcolor{green}{\td{X}_s}\textcolor{red}{-T_1}) dT_2 \wedge d\text{Ham}^{(\boldsymbol{\alpha}_{\td{X}_s})}\textcolor{pink}{-T_2^{-1}\sum_{s=2}^n \text{Ham}^{(\boldsymbol{\alpha}_{\td{X}_s})}dT_1\wedge dT_2}\cr
&&\textcolor{pink}{-  \sum_{s=2}^ndT_1\wedge d\text{Ham}^{(\boldsymbol{\alpha}_{\td{X}_s})}}\textcolor{orange}{+T_2^{-1}\sum_{s=2}^n\text{Ham}^{(\boldsymbol{\alpha}_{\td{X}_s})}d\td{X}_s\wedge dT_2} +\sum_{s=2}^nd\td{X}_s\wedge d\text{Ham}^{(\boldsymbol{\alpha}_{\td{X}_s})}\cr
&=&T_2^{-1}dT_2\wedge d\text{Ham}^{(\mathbf{a})} +T_2^{-2}dT_2\wedge d(T_1\text{Ham}^{(\mathbf{b})})+ \left(T_2^{-1}dT_1-T_1 T_2^{-2}dT_2\right)\wedge d\text{Ham}^{(\mathbf{b})}\cr 
&&+\sum_{s=1}^n\sum_{k=1}^{r_s-1}d\tau_{X_s,k}\wedge d\text{Ham}^{(\boldsymbol{\alpha}_{\tau_{X_s,k}})}+\sum_{s=2}^nd\td{X}_s\wedge d\text{Ham}^{(\boldsymbol{\alpha}_{\td{X}_s})}\cr
&=&-\left(T_2^{-2}T_1 dT_2-T_2^{-1}dT_1\right)\wedge d\text{Ham}^{(\mathbf{b})}(\hat{\mathbf{q}},\hat{\mathbf{p}})+T_2^{-2}dT_2\wedge \left(d(T_1\text{Ham}^{(\mathbf{b})}(\hat{\mathbf{q}},\hat{\mathbf{p}}))+d(T_2\text{Ham}^{(\mathbf{a})}(\hat{\mathbf{q}},\hat{\mathbf{p}})) \right)\cr
&&+\sum_{s=1}^n\sum_{k=1}^{r_s-1}d\tau_{X_s,k}\wedge d\text{Ham}^{(\boldsymbol{\alpha}_{\tau_{X_s,k}})}(\hat{\mathbf{q}},\hat{\mathbf{p}})+\sum_{s=2}^nd\td{X}_s\wedge d\text{Ham}^{(\boldsymbol{\alpha}_{\td{X}_s})}(\hat{\mathbf{q}},\hat{\mathbf{p}}).\cr
&&
\eea
}
\normalsize{Finally,} using \eqref{ExplicitComputation} and \eqref{IdentityFinal2bis}, equation \eqref{IntermediateIdentity} can be rewritten as follows: 
\beq \Omega=\hbar\sum_{j=1}^g d\check{q}_j \wedge d\check{p}_j - \sum_{s=1}^n\sum_{k=1}^{r_s-1} d\tau_{X_s,k}\wedge d\text{Ham}^{(\boldsymbol{\alpha}_{\tau_{X_s,k}})}(\check{\mathbf{q}},\check{\mathbf{p}})-\sum_{s=2}^n d\td{X}_s\wedge d\text{Ham}^{(\boldsymbol{\alpha}_{\td{X}_s})}(\check{\mathbf{q}},\check{\mathbf{p}}).\eeq

\subsection{The case $r_\infty=1$ and $n\geq 2$}
Let us study the case $r_\infty=1$ and $n\geq 2$. In this case we have from Proposition \ref{InverseRelationsrequal1}:
\bea \hat{T}_{X_s,k}&=&T_2^{-k}\tau_{X_s,k} \,\,,\,\, \forall \,(s,k)\in \llbracket 1,n\rrbracket\times\llbracket 1,r_s-1\rrbracket,\cr
X_1&=&-T_{2}^{-1}T_1,\cr
X_2&=&T_2^{-1}(1-T_1),\cr
X_s&=&T_2^{-1}(\td{X}_s-T_1) \,\,,\,\, \forall \,s \in \llbracket 3,n\rrbracket.
\eea
Thus we immediately get
\bea d\hat{T}_{X_s,k}&=&-kT_{2}^{-k-1}\tau_{X_s,k}dT_2+T_2^{-k}d\tau_{X_s,k}\,\,,\,\, \forall \,(s,k)\in \llbracket 1,n\rrbracket\times\llbracket 1,r_s-1\rrbracket,\cr
dX_1&=&T_2^{-2}T_1dT_2-T_2^{-1}dT_1,\cr
dX_2&=&-T_2^{-2}(1-T_1)dT_2-T_2^{-1}dT_1,\cr
dX_s&=&-T_2^{-2}(\td{X}_s-T_1)dT_2+T_2^{-1}d\td{X}_s-T_2^{-1}dT_1 \,\,,\,\, \forall \,s \in \llbracket 3,n\rrbracket.\cr
&&
\eea
This provides
\bea \text{Ham}^{(\boldsymbol{\alpha}_{T_1})}(\hat{\mathbf{q}},\hat{\mathbf{p}})&=&-T_2^{-1}\sum_{s=1}^n \text{Ham}^{(\mathbf{w}_s)}(\hat{\mathbf{q}},\hat{\mathbf{p}})\cr
&&-\frac12\sum_{s=1}^n\sum_{k=1}^{r_s-1}kT_2^{-k-1}\tau_{X_s,k}\text{Ham}^{(\mathbf w_{X_s,k})}(\hat{\mathbf q},\hat{\mathbf p}),\cr
\text{Ham}^{(\boldsymbol{\alpha}_{T_2})}(\hat{\mathbf{q}},\hat{\mathbf{p}}) &=&-T_2^{-1}\sum_{s=1}^n X_s \text{Ham}^{(\mathbf{w}_s)}(\hat{\mathbf{q}},\hat{\mathbf{p}}),\cr
\text{Ham}^{(\boldsymbol{\alpha}_{\td{X}_s})}(\hat{\mathbf{q}},\hat{\mathbf{p}})&=&T_2^{-1}\text{Ham}^{(\mathbf{w}_s)}(\hat{\mathbf{q}},\hat{\mathbf{p}})\,\,,\,\, \forall\, s\in \llbracket 3,n\rrbracket,\cr
\text{Ham}^{(\boldsymbol{\alpha}_{\tau_{X_s,k}})}(\hat{\mathbf{q}},\hat{\mathbf{p}})&=&\frac{1}{2}T_2^{-k}\text{Ham}^{(\mathbf{w}_{X_s,k})}(\hat{\mathbf{q}},\hat{\mathbf{p}})\,\,,\,\, \forall \,(s,k)\in \llbracket 1,n\rrbracket\times\llbracket 1,r_s-1\rrbracket.\cr
&&
\eea
Let us now observe that
\bea \partial_{X_1}&=&-T_2(1-T_1)\partial_{T_1}+T_2^2\partial_{T_2}+T_2\sum_{s=3}^n(\td{X}_s-1)\partial_{\td{X}_s}+T_2\sum_{s=1}^n\sum_{k=1}^{r_s-1}k \tau_{X_s,k}\partial_{\tau_{X_s,k}},\cr
\partial_{X_2}&=&-T_1T_2\partial_{T_1}-T_2^2\partial_{T_2}-T_2\sum_{s=3}^n\td{X}_s\partial_{\td{X}_s} -T_2\sum_{s=1}^n\sum_{k=1}^{r_s-1} k\tau_{X_s,k}\partial_{\tau_{X_s,k}},\cr
\partial_{X_s}&=&T_2\partial_{\td{X}_s} \,\,,\,\, \forall \, s\in \llbracket 3,n\rrbracket,\cr
\partial_{\hat{T}_{X_s,k}}&=&T_2^{k}\partial_{\tau_{X_s,k}} \,\,,\,\, \forall \,(s,k)\in \llbracket 1,n\rrbracket\times\llbracket 1,r_s-1\rrbracket,\cr
&&
\eea
i.e. 
\bea\label{FirstHam} \text{Ham}^{(\mathbf{w}_1)}&=&-T_2(1-T_1)\text{Ham}^{(\boldsymbol{\alpha}_{T_1})}+T_2^2\text{Ham}^{(\boldsymbol{\alpha}_{T_2})}+T_2\sum_{s=3}^n(\td{X}_s-1)\text{Ham}^{(\boldsymbol{\alpha}_{\td{X}_s})}\cr
&&+T_2\sum_{s=1}^n\sum_{k=1}^{r_s-1}k \tau_{X_s,k}\text{Ham}^{(\boldsymbol{\alpha}_{\tau_{X_s,k}})},\cr
\text{Ham}^{(\mathbf{w}_2)}&=&-T_1T_2\text{Ham}^{(\boldsymbol{\alpha}_{T_1})}-T_2^2\text{Ham}^{(\boldsymbol{\alpha}_{T_2})}-T_2\sum_{s=3}^n\td{X}_s\text{Ham}^{(\boldsymbol{\alpha}_{\td{X}_s})} \cr
&&-T_2\sum_{s=1}^n\sum_{k=1}^{r_s-1} k\tau_{X_s,k}\text{Ham}^{(\boldsymbol{\alpha}_{\tau_{X_s,k}})},\cr
\text{Ham}^{(\mathbf{w}_s)}&=&T_2\text{Ham}^{(\boldsymbol{\alpha}_{\td{X}_s})}\,\,,\,\, \forall \, s\in \llbracket 3,n\rrbracket,\cr
\text{Ham}^{(\mathbf{w}_{X_s,k})}&=&2T_2^{k}\text{Ham}^{(\boldsymbol{\alpha}_{\tau_{X_s,k}})}\,\,,\,\, \forall \,(s,k)\in \llbracket 1,n\rrbracket\times\llbracket 1,r_s-1\rrbracket.\cr
&&
\eea

We have also from Definitions \eqref{Defa} and \eqref{Defb}:
\small{\bea \mathcal{L}_{\mathbf{b}}&=&-\hbar\partial_{X_1}-\hbar \partial_{X_2}-\hbar \sum_{s=3}^n \partial_{X_s}=\hbar T_2\partial_{T_1},\cr
\mathcal{L}_{\mathbf{a}}&=&-\frac{1}{2}\sum_{s=1}^n\sum_{r=1}^{r_s-1}rT_{X_s,r}\mathcal{L}_{\mathbf{v}_{X_s,r}}-\frac{1}{2}\sum_{s=1}^n\sum_{r=1}^{r_s-1}r\hat{T}_{X_s,r}\mathcal{L}_{\mathbf{w}_{X_s,r}}-\hbar X_1\partial_{X_1}-\hbar X_2\partial_{X_2}-\hbar \sum_{s=3}^n X_s\partial_{X_s}\cr
&=&\hbar T_2 \partial_{T_2}-\frac{1}{2}\sum_{s=1}^n\sum_{r=1}^{r_s-1}rT_{X_s,r}\mathcal{L}_{\mathbf{v}_{X_s,r}},\cr
&&
\eea}
\normalsize{i.e.} using \eqref{SuperIDHam}
\beq \label{SecondHam}\text{Ham}^{(\mathbf{b})}(\hat{\mathbf{q}},\hat{\mathbf{p}})= T_2\text{Ham}^{(\boldsymbol{\alpha}_{T_1})}(\hat{\mathbf{q}},\hat{\mathbf{p}})\,\,,\,\,
\text{Ham}^{(\mathbf{a})}(\hat{\mathbf{q}},\hat{\mathbf{p}})=T_2\text{Ham}^{(\boldsymbol{\alpha}_{T_2})}(\hat{\mathbf{q}},\hat{\mathbf{p}}).
\eeq
Combining \eqref{FirstHam} and \eqref{SecondHam} implies
\bea\label{ThirdHam} \text{Ham}^{(\mathbf{w}_1)}&=&-(1-T_1)\text{Ham}^{(\mathbf{b})}+T_2\text{Ham}^{(\mathbf{a})}+T_2\sum_{s=3}^n(\td{X}_s-1)\text{Ham}^{(\boldsymbol{\alpha}_{\td{X}_s})}\cr
&&+T_2\sum_{s=1}^n\sum_{k=1}^{r_s-1}k \tau_{X_s,k}\text{Ham}^{(\boldsymbol{\alpha}_{\tau_{X_s,k}})},\cr
\text{Ham}^{(\mathbf{w}_2)}&=&-T_1\text{Ham}^{(\mathbf{b})}-T_2\text{Ham}^{(\mathbf{a})}-T_2\sum_{s=3}^n\td{X}_s\text{Ham}^{(\boldsymbol{\alpha}_{\td{X}_s})} \cr&&-T_2\sum_{s=1}^n\sum_{k=1}^{r_s-1} k\tau_{X_s,k}\text{Ham}^{(\boldsymbol{\alpha}_{\tau_{X_s,k}})},\cr
\text{Ham}^{(\mathbf{w}_s)}&=&T_2\text{Ham}^{(\boldsymbol{\alpha}_{\td{X}_s})}\,\,,\,\, \forall \, s\in \llbracket 3,n\rrbracket,\cr
\text{Ham}^{(\mathbf{w}_{X_s,k})}&=&2T_2^{k}\text{Ham}^{(\boldsymbol{\alpha}_{\tau_{X_s,k}})}\,\,,\,\, \forall \,(s,k)\in \llbracket 1,n\rrbracket\times\llbracket 1,r_s-1\rrbracket.\cr
&&
\eea

Let us now insert these results into \eqref{IntermediateIdentity}
\small{\bea  &&\frac{1}{2}\sum_{p\in \mathcal{R}}\sum_{k=1}^{r_p-1} d\hat{T}_{p,k}\wedge d\,\text{Ham}^{(\mathbf{w}_{p,k})}(\hat{\mathbf{q}},\hat{\mathbf{p}}) +\sum_{s=1}^n dX_s\wedge d\,\text{Ham}^{(\mathbf{w}_s)}(\hat{\mathbf{q}},\hat{\mathbf{p}})\cr
&&=\sum_{s=1}^n\sum_{k=1}^{r_s-1}\left(-kT_{2}^{-1}\tau_{X_s,k}dT_2+d\tau_{X_s,k}\right)\wedge\left(kT_2^{-1}\text{Ham}^{(\boldsymbol{\alpha}_{\tau_{X_s,k}})}dT_2+ d\text{Ham}^{(\boldsymbol{\alpha}_{\tau_{X_s,k}})} \right)\cr
&&+\left(T_2^{-2}T_1dT_2-T_2^{-1}dT_1\right)\wedge\Big(d(-(1-T_1)\text{Ham}^{(\mathbf{b})})+d(T_2\text{Ham}^{(\mathbf{a})})+\sum_{s=3}^n(\td{X}_s-1)\text{Ham}^{(\boldsymbol{\alpha}_{\td{X}_s})}dT_2\cr
&&+ T_2\sum_{s=3}^nd((\td{X}_s-1)\text{Ham}^{(\boldsymbol{\alpha}_{\td{X}_s})}) +\sum_{s=1}^n\sum_{k=1}^{r_s-1}k \tau_{X_s,k}\text{Ham}^{(\boldsymbol{\alpha}_{\tau_{X_s,k}})} dT_2+ T_2\sum_{s=1}^n\sum_{k=1}^{r_s-1}k d(\tau_{X_s,k}\text{Ham}^{(\boldsymbol{\alpha}_{\tau_{X_s,k}})}) \Big)\cr
&&-(T_2^{-2}(1-T_1)dT_2+T_2^{-1}dT_1)\wedge \Big(-d(T_1\text{Ham}^{(\mathbf{b})})-d(T_2\text{Ham}^{(\mathbf{a})})-\sum_{s=3}^n\td{X}_s\text{Ham}^{(\boldsymbol{\alpha}_{\td{X}_s})} dT_2-T_2\sum_{s=3}^nd(\td{X}_s\text{Ham}^{(\boldsymbol{\alpha}_{\td{X}_s})})  \cr
&&-\sum_{s=1}^n\sum_{k=1}^{r_s-1} k\tau_{X_s,k}\text{Ham}^{(\boldsymbol{\alpha}_{\tau_{X_s,k}})}dT_2- T_2 \sum_{s=1}^n\sum_{k=1}^{r_s-1} kd(\tau_{X_s,k}\text{Ham}^{(\boldsymbol{\alpha}_{\tau_{X_s,k}})})\Big)\cr
&&+\sum_{s=3}^n \left(-T_2^{-1}(\td{X}_s-T_1)dT_2+d\td{X}_s-dT_1\right)\wedge\left( d\text{Ham}^{(\boldsymbol{\alpha}_{\td{X}_s})}+T_2^{-1}\text{Ham}^{(\boldsymbol{\alpha}_{\td{X}_s})} dT_2\right)\cr
&&= \sum_{s=1}^n\sum_{k=1}^{r_s-1}d\tau_{X_s,k}\wedge d\text{Ham}^{(\boldsymbol{\alpha}_{\tau_{X_s,k}})}+\sum_{s=3}^n d\td{X}_s\wedge d\text{Ham}^{(\boldsymbol{\alpha}_{\td{X}_s})}\textcolor{blue}{-T_2^{-1}\sum_{s=1}^n\sum_{k=1}^{r_s-1}k dT_2\wedge d(\tau_{X_s,k}\text{Ham}^{(\boldsymbol{\alpha}_{\tau_{X_s,k}})})}\cr
&&\textcolor{blue}{+\left(T_2^{-2}T_1dT_2-T_2^{-1}dT_1\right)\wedge\left(\sum_{s=1}^n\sum_{k=1}^{r_s-1}k \tau_{X_s,k}\text{Ham}^{(\boldsymbol{\alpha}_{\tau_{X_s,k}})} dT_2+ T_2\sum_{s=1}^n\sum_{k=1}^{r_s-1}k d(\tau_{X_s,k}\text{Ham}^{(\boldsymbol{\alpha}_{\tau_{X_s,k}})})\right)}\cr
&&\textcolor{blue}{+(T_2^{-2}(1-T_1)dT_2+T_2^{-1}dT_1)\wedge\left(\sum_{s=1}^n\sum_{k=1}^{r_s-1} k\tau_{X_s,k}\text{Ham}^{(\boldsymbol{\alpha}_{\tau_{X_s,k}})}dT_2+ T_2 \sum_{s=1}^n\sum_{k=1}^{r_s-1} kd(\tau_{X_s,k}\text{Ham}^{(\boldsymbol{\alpha}_{\tau_{X_s,k}})})\right)}\cr
&&\textcolor{green}{+\left(T_2^{-2}T_1dT_2-T_2^{-1}dT_1\right)\wedge \left(\sum_{s=3}^n(\td{X}_s-1)\text{Ham}^{(\boldsymbol{\alpha}_{\td{X}_s})}dT_2+ T_2\sum_{s=3}^nd((\td{X}_s-1)\text{Ham}^{(\boldsymbol{\alpha}_{\td{X}_s})})\right)}\cr
&&\textcolor{green}{+(T_2^{-2}(1-T_1)dT_2+T_2^{-1}dT_1)\wedge\left(\sum_{s=3}^n\td{X}_s\text{Ham}^{(\boldsymbol{\alpha}_{\td{X}_s})} dT_2+T_2\sum_{s=3}^nd(\td{X}_s\text{Ham}^{(\boldsymbol{\alpha}_{\td{X}_s})})\right)}\cr
&&\textcolor{green}{+\sum_{s=3}^n \left(-T_2^{-1}(\td{X}_s-T_1)dT_2-dT_1\right)\wedge \left( d\text{Ham}^{(\boldsymbol{\alpha}_{\td{X}_s})}+T_2^{-1}\text{Ham}^{(\boldsymbol{\alpha}_{\td{X}_s})} dT_2\right)+T_2^{-1}\sum_{s=3}^n\text{Ham}^{(\boldsymbol{\alpha}_{\td{X}_s})} d\td{X}_s \wedge dT_2}\cr
&&+\left(T_2^{-2}T_1dT_2-T_2^{-1}dT_1\right)\wedge\left(-d((1-T_1)\text{Ham}^{(\mathbf{b})})+d(T_2\text{Ham}^{(\mathbf{a})})\right)\cr
&&-(T_2^{-2}(1-T_1)dT_2+T_2^{-1}dT_1)\wedge\left(-d(T_1\text{Ham}^{(\mathbf{b})})-d(T_2\text{Ham}^{(\mathbf{a})})\right)\cr
&&= \sum_{s=1}^n\sum_{k=1}^{r_s-1}d\tau_{X_s,k}\wedge d\text{Ham}^{(\boldsymbol{\alpha}_{\tau_{X_s,k}})}+\sum_{s=3}^n d\td{X}_s\wedge d\text{Ham}^{(\boldsymbol{\alpha}_{\td{X}_s})}\cr
&&+\left(T_2^{-2}T_1dT_2-T_2^{-1}dT_1\right)\wedge\left(-d((1-T_1)\text{Ham}^{(\mathbf{b})})+d(T_2\text{Ham}^{(\mathbf{a})})\right)\cr
&&-(T_2^{-2}(1-T_1)dT_2+T_2^{-1}dT_1)\wedge\left(-d(T_1\text{Ham}^{(\mathbf{b})})-d(T_2\text{Ham}^{(\mathbf{a})})\right)\cr 
&&= \sum_{s=1}^n\sum_{k=1}^{r_s-1}d\tau_{X_s,k}\wedge d\text{Ham}^{(\boldsymbol{\alpha}_{\tau_{X_s,k}})}(\hat{\mathbf{q}},\hat{\mathbf{p}})+\sum_{s=3}^n d\td{X}_s\wedge d\text{Ham}^{(\boldsymbol{\alpha}_{\td{X}_s})}(\hat{\mathbf{q}},\hat{\mathbf{p}})\cr
&&-\left(T_2^{-2}T_1dT_2-T_2^{-1}dT_1\right)\wedge d\text{Ham}^{(\mathbf{b})}(\hat{\mathbf{q}},\hat{\mathbf{p}})+T_2^{-2}dT_2\wedge\left(d(T_1\text{Ham}^{(\mathbf{b})}(\hat{\mathbf{q}},\hat{\mathbf{p}}))+d(T_2\text{Ham}^{(\mathbf{a})}(\hat{\mathbf{q}},\hat{\mathbf{p}}))\right).\cr
&& 
\eea}
\normalsize{Using} \eqref{IdentityFinal2bis} we finally get
\beq \Omega=\hbar\sum_{j=1}^g d\check{q}_j \wedge d\check{p}_j- \sum_{s=1}^n\sum_{k=1}^{r_s-1} d\tau_{X_s,k}\wedge d\text{Ham}^{(\boldsymbol{\alpha}_{\tau_{X_s,k}})}(\check{\mathbf{q}},\check{\mathbf{p}})-\sum_{s=3}^n d\td{X}_s\wedge d\text{Ham}^{(\boldsymbol{\alpha}_{\td{X}_s})}(\check{\mathbf{q}},\check{\mathbf{p}}).\eeq

\subsection{The case $r_\infty=1$ and $n=1$}
Let us study the case $r_\infty=1$ and $n=1$. In this case we have from Proposition \ref{InverseRelationsrequal1n1}:
\bea \hat{T}_{X_1,r_1-1}&=&2T_2^{-(r_1-1)},\cr
\hat{T}_{X_1,k}&=&T_2^{-k}\tau_{X_1,k} \,\,,\,\, \forall \, k\in \llbracket 1,r_1-2\rrbracket,\cr
X_1&=&-T_2^{-1} T_1.
\eea
Thus we immediately get
\bea d\hat{T}_{X_1,r_1-1}&=&-2(r_1-1)T_2^{-r_1}dT_2,\cr
d\hat{T}_{X_1,k}&=&-kT_2^{-k-1}\tau_{X_1,k}dT_2+T_2^{-k}d\tau_{X_1,k} \,\,,\,\, \forall \, k\in \llbracket 1,r_1-2\rrbracket,\cr
dX_1&=&-T_2^{-1}dT_1+T_1 T_2^{-2}dT_2,
\eea
i.e. using the fact that $\partial_{\hat{T}_{X_1,k}}=\frac{1}{2}\partial_{\mathbf{w}_{X_1,k}}$ (or using Proposition \ref{TheoDualIsorequal1n1})
\bea \label{HamNewcoord1}\text{Ham}^{(\boldsymbol{\alpha}_{T_1})}(\hat{\mathbf{q}},\hat{\mathbf{p}})&=&-T_2^{-1}\text{Ham}^{(\mathbf{w}_1)}(\hat{\mathbf{q}},\hat{\mathbf{p}}),\cr
\text{Ham}^{(\boldsymbol{\alpha}_{\tau_{X_1,k}})}(\hat{\mathbf{q}},\hat{\mathbf{p}})&=&\frac{1}{2}T_2^{-k}\text{Ham}^{(\mathbf{w}_{X_1,k})}(\hat{\mathbf{q}},\hat{\mathbf{p}})\,\,,\,\, \forall \, k\in \llbracket 1,r_1-2\rrbracket,\cr
\text{Ham}^{(\boldsymbol{\alpha}_{T_2})}(\hat{\mathbf{q}},\hat{\mathbf{p}})&=&T_1 T_2^{-2}\text{Ham}^{(\mathbf{w}_1)}(\hat{\mathbf{q}},\hat{\mathbf{p}})-(r_1-1)T_2^{-r_1}\text{Ham}^{(\mathbf{w}_{X_1,r_1-1})}(\hat{\mathbf{q}},\hat{\mathbf{p}})\cr
&&-\frac{1}{2}\sum_{k=1}^{r_1-2}k T_2^{-k-1}\tau_{X_1,k}\text{Ham}^{(\mathbf{w}_{X_1,k})}(\hat{\mathbf{q}},\hat{\mathbf{p}})
\eea
where we have denoted $\text{Ham}^{(\boldsymbol{\alpha}_{\tau})}(\hat{\mathbf{q}},\hat{\mathbf{p}})$ the Hamiltonian associated to the evolution $\hbar \partial_{\tau}$. Definitions \eqref{Defa} and \eqref{Defb} of $\mathcal{L}_{\mathbf{a}}$ and $\mathcal{L}_{\mathbf{b}}$ provide
\small{\bea \mathcal{L}_{\mathbf{a}}&=&-\hbar\sum_{r=1}^{r_1-1}r \frac{t_{X_1^{(1)},r}+t_{X_1^{(2)},r}}{2}\left(\partial_{t_{X_1^{(1)},r}}+\partial_{t_{X_1^{(2)},r}}\right) -\hbar\sum_{r=1}^{r_1-1}r \frac{t_{X_1^{(1)},r}-t_{X_1^{(2)},r}}{2}\left(\partial_{t_{X_1^{(1)},r}}-\partial_{t_{X_1^{(2)},r}}\right)-\hbar X_1\partial_{X_1}\cr
&=&-\frac{1}{2}\sum_{r=1}^{r_1-1}rT_{X_1,r}\mathcal{L}_{\mathbf{v}_{X_1,r}}-\frac{1}{2}\sum_{r=1}^{r_1-1}r\hat{T}_{X_1,r}\mathcal{L}_{\mathbf{w}_{X_1,r}}-\hbar X_1\partial_{X_1}\cr
\mathcal{L}_{\mathbf{b}}&=&-\hbar\partial_{X_1}\cr
&&
\eea}
\normalsize{so} that using \eqref{HamNewcoord1}
\bea\label{HamNewcoord2} \text{Ham}^{(\mathbf{b})}(\hat{\mathbf{q}},\hat{\mathbf{p}})&=&-\text{Ham}^{(\mathbf{w}_1)}(\hat{\mathbf{q}},\hat{\mathbf{p}})= T_2 \,\text{Ham}^{(\boldsymbol{\alpha}_{T_1})}(\hat{\mathbf{q}},\hat{\mathbf{p}}),\cr
\text{Ham}^{(\mathbf{a})}(\hat{\mathbf{q}},\hat{\mathbf{p}})
&=&-\frac{1}{2}\sum_{r=1}^{r_1-1}rT_{X_1,r}\text{Ham}^{(\mathbf{v}_{X_1,r})}(\hat{\mathbf{q}},\hat{\mathbf{p}})+T_2\text{Ham}^{(\boldsymbol{\alpha}_{T_2})}(\hat{\mathbf{q}},\hat{\mathbf{p}}).\cr
&&
\eea
Combining \eqref{HamNewcoord1} and \eqref{HamNewcoord2} and using \eqref{SuperIDHam} we obtain:
\bea\label{Hamshort1} \text{Ham}^{(\mathbf{w}_1)}(\hat{\mathbf{q}},\hat{\mathbf{p}})&=& -\text{Ham}^{(\mathbf{b})}(\hat{\mathbf{q}},\hat{\mathbf{p}})=\hbar\sum_{j=1}^g\hat{p}_j,\cr
\text{Ham}^{(\mathbf{w}_{X_1,k})}(\hat{\mathbf{q}},\hat{\mathbf{p}})&=&2 T_2^{k}\text{Ham}^{(\boldsymbol{\alpha}_{\tau_{X_1,k}})}(\hat{\mathbf{q}},\hat{\mathbf{p}})\,\,,\,\, \forall \, k\in \llbracket 1,r_1-2\rrbracket,\cr
\text{Ham}^{(\mathbf{w}_{X_1,r_1-1})}(\hat{\mathbf{q}},\hat{\mathbf{p}})
&=&-\frac{1}{r_1-1}T_2^{r_1-1}\text{Ham}^{(\mathbf{a})}(\hat{\mathbf{q}},\hat{\mathbf{p}})-\frac{1}{r_1-1}T_1 T_2^{r_1-2}\text{Ham}^{(\mathbf{b})}(\hat{\mathbf{q}},\hat{\mathbf{p}})\cr
&&-\frac{1}{r_1-1}T_2^{r_1-1}\sum_{k=1}^{r_1-2}k\tau_{X_1,k} \text{Ham}^{(\boldsymbol{\alpha}_{\tau_{X_1,k}})}(\hat{\mathbf{q}},\hat{\mathbf{p}}). 
\eea
A trivial computation from \eqref{ExplicitComputation} gives:
\beq -\frac{1}{r_1-1}T_2^{r_1-1}\left(\text{Ham}^{(\mathbf{a})}(\hat{\mathbf{q}},\hat{\mathbf{p}})+T_1 T_2^{-1}\text{Ham}^{(\mathbf{b})}(\hat{\mathbf{q}},\hat{\mathbf{p}})\right)=\frac{\hbar}{r_1-1}T_2^{r_1-1}\left(\sum_{j=1}^g \hat{q}_j\hat{p}_j+T_1T_2^{-1}\sum_{j=1}^g \hat{p}_j\right)
\eeq
so that
\small{\beq \label{Hamshort2}\text{Ham}^{(\mathbf{w}_{X_1,r_1-1})}(\hat{\mathbf{q}},\hat{\mathbf{p}})=\frac{\hbar}{r_1-1}T_2^{r_1-1}\left(\sum_{j=1}^g \hat{q}_j\hat{p}_j+T_1T_2^{-1}\sum_{j=1}^g \hat{p}_j\right)-\frac{1}{r_1-1}T_2^{r_1-1}\sum_{k=1}^{r_1-2}k\tau_{X_1,k} \text{Ham}^{(\boldsymbol{\alpha}_{\tau_{X_1,k}})}(\hat{\mathbf{q}},\hat{\mathbf{p}}).
\eeq}
\normalsize{Finally} we get from \eqref{Hamshort1} and \eqref{Hamshort2}:
\small{\bea &&\frac{1}{2}\sum_{k=1}^{r_1-1} d\hat{T}_{X_1,k}\wedge d\text{Ham}^{(\mathbf{w}_{X_1,k})}(\hat{\mathbf{q}},\hat{\mathbf{p}})+ dX_1\wedge d\text{Ham}^{(\mathbf{w}_1)}(\hat{\mathbf{q}},\hat{\mathbf{p}})=\cr
&&-T_2^{-r_1}dT_2\wedge \Big(T_2^{r_1-1}d\left(\hbar\sum_{j=1}^g \hat{q}_j\hat{p}_j+\hbar T_1T_2^{-1}\sum_{j=1}^g \hat{p}_j\right)-T_2^{r_1-1}\sum_{k=1}^{r_1-2}kd(\tau_{X_1,k} \text{Ham}^{(\boldsymbol{\alpha}_{\tau_{X_1,k}})}(\hat{\mathbf{q}},\hat{\mathbf{p}}))\Big)\cr
&&+\frac{1}{2}\sum_{k=1}^{r_1-2} \left(-kT_2^{-k-1}\tau_{X_1,k}dT_2+T_2^{-k}d\tau_{X_1,k}\right)\wedge \left(2k T_2^{k-1} \text{Ham}^{(\boldsymbol{\alpha}_{\tau_{X_1,k}})}(\hat{\mathbf{q}},\hat{\mathbf{p}}) dT_2+2 T_2^{k} d\text{Ham}^{(\boldsymbol{\alpha}_{\tau_{X_1,k}})}(\hat{\mathbf{q}},\hat{\mathbf{p}}) \right)\cr
&&-\left(-T_2^{-1}dT_1+T_1 T_2^{-2}dT_2\right)\wedge d\text{Ham}^{(\mathbf{b})}(\hat{\mathbf{q}},\hat{\mathbf{p}})\cr
&&=\sum_{k=1}^g d\tau_{X_1,k} \wedge d\text{Ham}^{(\boldsymbol{\alpha}_{\tau_{X_1,k}})}(\hat{\mathbf q},\hat{\mathbf p})\cr
&& -T_2^{-1} dT_2\wedge d\left(\hbar\sum_{j=1}^g \hat{q}_j\hat{p}_j+\hbar T_1T_2^{-1}\sum_{j=1}^g \hat{p}_j\right)+T_2^{-1}dT_2 \wedge \sum_{k=1}^{r_1-2}kd(\tau_{X_1,k} \text{Ham}^{(\boldsymbol{\alpha}_{\tau_{X_1,k}})}(\hat{\mathbf{q}},\hat{\mathbf{p}}))\cr
&&-T_2^{-1}\sum_{k=1}^{r_1-2}k\tau_{X_1,k}dT_2\wedge d\text{Ham}^{(\boldsymbol{\alpha}_{\tau_{X_1,k}})}(\hat{\mathbf{q}},\hat{\mathbf{p}})  +T_2^{-1}\sum_{k=1}^{r_1-2}k \text{Ham}^{(\boldsymbol{\alpha}_{\tau_{X_1,k}})}(\hat{\mathbf{q}},\hat{\mathbf{p}}) d\tau_{X_1,k} \wedge dT_2\cr
&&+\hbar\sum_{j=1}^g\left(-T_2^{-1}dT_1+T_1 T_2^{-2}dT_2\right)\wedge d\hat{p}_j\cr
&&=\sum_{k=1}^g d\tau_{X_1,k} \wedge d\text{Ham}^{(\boldsymbol{\alpha}_{\tau_{X_1,k}})}(\hat{\mathbf q},\hat{\mathbf p})\cr
&&-\hbar T_2^{-1} dT_2\wedge d\left(\sum_{j=1}^g \hat{q}_j\hat{p}_j+T_1T_2^{-1}\sum_{j=1}^g \hat{p}_j\right)+\hbar\sum_{j=1}^g\left(-T_2^{-1}dT_1+T_1 T_2^{-2}dT_2\right)\wedge d\hat{p}_j.\cr
&&
\eea}
\normalsize{From} \eqref{DarbouxShiftedpj}, equation \eqref{IntermediateIdentity} eventually reduces to 
\beq
\Omega=\hbar\sum_{j=1}^g d\check{q}_j\wedge d\check{p}_j - \sum_{k=1}^g d\tau_{X_1,k} \wedge d\text{Ham}^{(\boldsymbol{\alpha}_{\tau_{X_1,k}})}(\check{\mathbf{q}},\check{\mathbf{p}}).
\eeq

\begin{remark}[Terms depending only on the times]
\label{RemarkPureTimeAppendixP}
The preceding computations display the cancellations associated
with the trivial deformations. There remains a normalization
issue for an identity of total two-forms. Indeed, replacing
$H_a$ by $H_a+c_a(\mathbf t)$ does not change its Hamiltonian
vector field, but changes the fundamental two-form by
\[
-\sum_a dt_a\wedge dc_a
=d\left(\sum_a c_a\,dt_a\right).
\]
Thus functions of the times alone cannot simply be discarded when
proving the exact equality in Theorem \ref{MainTheotau0}.

Here is a compatible normalization that completes the computation.
Let $H_a$ be the flat representatives in Definition
\ref{DefinitionSymplecticForm}, and choose local logarithms in
\begin{align*}
\Phi(\lambda;\mathbf t)
={}&-\sum_{k=1}^{r_\infty-1}\frac{T_{\infty,k}}{k}\lambda^k
-\sum_{s=1}^n\sum_{k=1}^{r_s-1}\frac{T_{X_s,k}}{k}
(\lambda-X_s)^{-k}
+\sum_{s=1}^n T_{X_s,0}\log(\lambda-X_s),\\
\partial_\lambda\Phi={}&P_1.
\end{align*}
Write $y=(u,\boldsymbol\tau)$ for the trivial and non-trivial time
coordinates, and use
$q_j=(\check q_j-T_1)/T_2$ and
$p_j=T_2\check p_j+P_1(q_j)/2$ below. The exact transformed
Hamiltonian coefficients are
\begin{align}\label{TransformedHamiltonianCoefficientsP}
\mathcal H_A={}&
\sum_a H_a(\mathbf q,\mathbf p;\mathbf t(y))
\frac{\partial t_a}{\partial y_A}
+\hbar\sum_{j=1}^g\check p_j
\left(\frac{\partial T_1}{\partial y_A}
+q_j\frac{\partial T_2}{\partial y_A}\right)
\nonumber\\
&+\frac\hbar2\sum_{j=1}^g
\left.\partial_{y_A}\Phi(\lambda;\mathbf t(y))\right|_{\lambda=q_j},
\end{align}
where the last partial derivative is taken at fixed $\lambda$.
With $G=\frac\hbar2\sum_j\Phi(q_j;\mathbf t)$, direct
differentiation gives
\[
\hbar\sum_jp_j\,dq_j-\sum_aH_a\,dt_a
=\hbar\sum_j\check p_j\,d\check q_j
+dG-\sum_A\mathcal H_A\,dy_A.
\]
In particular, this identity retains the moving-pole correction
\eqref{MovingPoleHamiltonianP} as well as all terms depending only
on the times. It also preserves flatness:
\[
\hbar(\partial_{y_A}\mathcal H_B-\partial_{y_B}\mathcal H_A)
-\{\mathcal H_A,\mathcal H_B\}_{\check{\mathbf q},\check{\mathbf p}}=0.
\]
The trivial deformations annihilate $\check q_j,\check p_j$ by
Theorem \ref{TheoremTrivialSubspace1} and Proposition
\ref{PropositionT1T2}. Hence each $\mathcal H_{u_b}$ depends only
on the times. Flatness then gives
$\partial_{u_b}\mathcal H_{u_c}=\partial_{u_c}\mathcal H_{u_b}$
and
$\partial_{u_b}\mathcal H_{\tau_i}
=\partial_{\tau_i}\mathcal H_{u_b}$.
On a product chart whose trivial-time factor is star-shaped about
$u^\ast$, put
\begin{equation}\label{TrivialTimePrimitiveP}
f(u,\boldsymbol\tau)=\int_0^1\sum_b(u_b-u_b^\ast)
\mathcal H_{u_b}(u^\ast+\rho(u-u^\ast),\boldsymbol\tau)\,d\rho.
\end{equation}
Then $\partial_{u_b}f=\mathcal H_{u_b}$, and the choice
\begin{equation}\label{ReducedHamiltonianNormalizationP}
\text{Ham}^{(\boldsymbol\alpha_{\tau_i})}
(\check{\mathbf q},\check{\mathbf p};\boldsymbol\tau)
:=\mathcal H_{\tau_i}-\partial_{\tau_i}f
\end{equation}
is independent of $u$ and generates the same reduced evolution.
Moreover,
\[
\sum_A\mathcal H_A\,dy_A
=df+\sum_i\text{Ham}^{(\boldsymbol\alpha_{\tau_i})}\,d\tau_i.
\]
Taking $-d$ of the canonical one-form identity therefore proves
\[
\Omega=\hbar\sum_jd\check q_j\wedge d\check p_j
-\sum_i d\tau_i\wedge
d\text{Ham}^{(\boldsymbol\alpha_{\tau_i})}
\]
as an exact equality, completing the proof in all four cases.
This normalization only fixes terms independent of the Darboux
coordinates. It changes none of the Hamiltonian vector fields or
the cancellations displayed above. A further addition
$c_i(\boldsymbol\tau)$ preserves this total two-form precisely when
$\sum_i c_i(\boldsymbol\tau)\,d\tau_i$ is closed.
\end{remark}

\section{Proof of Theorem \ref{TheoMonodromies}}\label{AppendixMonodromies}
Consider the scalar-twist family specified in the theorem. A scalar
gauge $\td{\Psi}\mapsto
\exp[\hbar^{-1}\int_{\lambda_0}^{\lambda}f(\xi)d\xi]\td{\Psi}$
changes the geometric connection by
\[
\td{L}\mapsto\td{L}+f(\lambda)I_2,
\qquad P_1\mapsto P_1+2f.
\]
Here $f$ is rational with the prescribed poles, and its residues,
including the one at infinity, sum to zero. For variations of trace
formal monodromies one takes
$f(\lambda)=\underset{s}{\sum} \frac{a_s}{\lambda-X_s}$.
The off-diagonal entry $\td{L}_{1,2}$ is unchanged, so its simple
zeros satisfy $q_j\mapsto q_j$. Since
$p_j=\td{L}_{1,1}(q_j)$, one has
\[
p_j\mapsto p_j+f(q_j),\qquad
p_j-\tfrac12P_1(q_j)\mapsto p_j-\tfrac12P_1(q_j).
\]
The affine parameters $T_1,T_2$ depend on differences of irregular
times and on pole positions, which this twist leaves unchanged.
Thus both $\check q_j=T_2q_j+T_1$ and
$\check p_j=T_2^{-1}\left(p_j-\frac{P_1(q_j)}{2}\right)$ are invariant.

This also fixes the signs in a wave-function calculation. Let
$\widehat\Psi_i=\exp[-(2\hbar)^{-1}\int_{\lambda_0}^{\lambda}
P_1(\xi)d\xi]\Psi_i$ and
$\widehat Y_i=\hbar\frac{\widehat\Psi_i'}{\widehat\Psi_i}$.
Then $Y_i=\widehat Y_i+\frac{P_1}{2}$. On the generic apparent chart,
\[
p_j=-\frac1\hbar\operatorname*{Res}_{\lambda=q_j}L_{2,1}
=-\operatorname*{Res}_{\lambda=q_j}
\frac{Y_2Y_1'-Y_1Y_2'}{Y_2-Y_1}
=\widehat p_j+\frac12P_1(q_j).
\]
The invariance conclusion concerns the stated scalar-twist
identification. It does not claim independence for families whose
trace-free initial conditions themselves vary with the trace parameters.

\section{Proof of Theorem \ref{HamTheoremReduced}}\label{ProofTheoremHamTheoremReduced}
Let us start with a trivial lemma
\begin{lemma}\label{LemmaToeplitz} Let us denote $S^{(d)}$ the $d\times d$ matrix with $1$ on the antidiagonal and $0$ everywhere else. In other words $S^{(d)}_{i,j}=\delta_{j=d+1-i}$. Then, for any $d\times d$ Toeplitz matrix $T$:
\beq S^{(d)} T^tS^{(d)}=T.\eeq
\end{lemma}
The proof of the lemma is immediate from the fact that $T_{i,j}=\gamma_{i-j}$ for all $(i,j)\in\llbracket 1,d\rrbracket^2$ for some $(\gamma_{-(d-1)},\dots,\gamma_{d-1})$, so that for any $(i,j)\in \llbracket 1,d\rrbracket^2$:
\beq [S^{(d)}T^tS^{(d)}]_{i,j}=\sum_{k=1}^{d}\sum_{m=1}^d [S^{(d)}]_{i,k}T_{m,k}[S^{(d)}]_{m,j}\overset{k=d+1-i}{\underset{m=d+1-j}{=}}T_{d+1-j,d+1-i}=\gamma_{i-j}=T_{i,j}\eeq
Note also that $\left(S^{(d)}\right)^2=I_d$ and that $S^{(d)}$ acts on vectors by reversing the order of the entries. 

\medskip

Let us then mention that \eqref{cobservation} eliminates part of the Hamiltonian in Theorem \ref{HamTheorem}. Then, for $r_\infty=1$, all additional terms identically vanish. For $r_\infty=2$, additional terms in Theorem \ref{HamTheorem} proportional to $\nu_{\infty,0}^{(\boldsymbol{\alpha})}$ identically vanish and terms for $r_\infty=2$, since $t_{\infty^{(1)},r_\infty-1}$ is fixed because of the definition of $T_2$, Proposition \ref{PropAsymptoticExpansionA12} implies that $\nu_{\infty,-1}^{(\boldsymbol{\alpha}_\tau)}=0$ so that all additional terms vanish. Finally for $r_\infty\geq 3$, the choice of $T_1$ and $T_2$ implies that the leading and subleading infinity times $t_{\infty^{(i)},r_\infty-1}$ and $t_{\infty^{(i)},r_\infty-2}$ are fixed on both sheets so that $\nu_{\infty,-1}^{(\boldsymbol{\alpha}_\tau)}=\nu_{\infty,0}^{(\boldsymbol{\alpha}_\tau)}=0$ from Proposition \ref{PropAsymptoticExpansionA12}. Thus, \eqref{DefHamReduced} is proved. 

\medskip

Let us now consider a deformation with respect to $\tau_{\infty,j}$ with $j\in \llbracket 1,r_\infty-3\rrbracket$. In this case, the only non-vanishing
components are $\alpha_{\infty^{(1)},j}=\frac12$ and $\alpha_{\infty^{(2)},j}=-\frac12$. Consequently, $\nu^{(\boldsymbol{\alpha}_{\infty,j})}_{{X_s},k}=0$ for all $(s,k)\in \llbracket 1,n\rrbracket\times\llbracket 1,r_s-1\rrbracket$. Hence we get
\beq \label{nuinftyReduced}\td{M}_\infty\boldsymbol{\nu}^{(\infty)}_j=\frac{1}{j}\mathbf{e}_{r_\infty-2-j} \,\text{ with }\, \boldsymbol{\nu}^{(\infty)}_j=\left( \nu^{(\boldsymbol{\alpha}_{\infty,j})}_{\infty,1}, \nu^{(\boldsymbol{\alpha}_{\infty,j})}_{\infty,2}, \dots, \nu^{(\boldsymbol{\alpha}_{\infty,j})}_{\infty,r_\infty-3}\right)^t. \eeq 
It gives
\beq \label{HamreducedSpecial}\text{Ham}^{(\boldsymbol{\alpha}_{\infty,j})}(\mathbf{\check{q}},\mathbf{\check{p}})=\left(\boldsymbol{\nu}^{(\infty)}_j\right)^t \mathbf{H_\infty}\, \text{ with }\,\mathbf{H_\infty}=\left(H_{\infty,0},\dots, H_{\infty,r_\infty-4}\right)^t.\eeq
Thus we have,
\beq \text{Ham}^{(\boldsymbol{\alpha}_{\infty,j})}(\mathbf{q},\mathbf{p})=\left(\boldsymbol{\nu}^{(\infty)}_j\right)^t \td{M}_\infty^t (\td{M}_\infty^t)^{-1}\mathbf{H_\infty}=\frac{1}{j}(\mathbf{e}_{r_\infty-2-j})^t(\td{M}_\infty^t)^{-1}\mathbf{H_\infty}.\eeq
Since $\text{Ham}^{(\boldsymbol{\alpha}_{\infty,j})}(\mathbf{q},\mathbf{p})$ is scalar, it is equal to its transpose and we get for all $j\in \llbracket 1,r_\infty-3\rrbracket$:
\beq j\text{Ham}^{(\boldsymbol{\alpha}_{\infty,j})}(\mathbf{q},\mathbf{p})=(\mathbf{H_\infty})^t \td{M}_\infty^{-1}\mathbf{e}_{r_\infty-2-j}.\eeq
Taking $i=r_\infty-2-j$, we get that for all $i\in \llbracket 1,r_\infty-3\rrbracket$:
\beq (r_\infty-2-i)\text{Ham}^{(\boldsymbol{\alpha}_{\infty,r_\infty-2-i})}(\mathbf{q},\mathbf{p})=(\mathbf{H_\infty})^t \td{M}_\infty^{-1}\mathbf{e}_{i}\eeq
$(\mathbf{H_\infty})^t \td{M}_\infty^{-1}$ is a row vector and $(\mathbf{H_\infty})^t \td{M}_\infty^{-1}\mathbf{e}_i$ extracts its $i^{\text{th}}$ entry. Thus, the last identities are equivalent to
\beq \left((r_\infty-3)\text{Ham}^{(\boldsymbol{\alpha}_{\infty,r_\infty-3})},(r_\infty-4)\text{Ham}^{(\boldsymbol{\alpha}_{\infty,r_\infty-4})},\dots ,\text{Ham}^{(\boldsymbol{\alpha}_{\infty,1})}\right)=(\mathbf{H_\infty})^t \td{M}_\infty^{-1}\eeq
which is equivalent after transposition to
\beq \begin{pmatrix}(r_\infty-3)\text{Ham}^{(\alpha_{\tau_{\infty,r_\infty -3}})}\\ \vdots\\\text{Ham}^{(\alpha_{\tau_{\infty,1}})} \end{pmatrix}=\left(\td{M}_{\infty}^t\right)^{-1}\begin{pmatrix}H_{\infty,0}\\ \vdots\\ H_{\infty,r_\infty-4}\end{pmatrix}.\eeq
Thus, we get
\beq (S^{(r_\infty-3)} \td{M}_\infty^t S^{(r_\infty-3)}) S^{(r_\infty-3)}\begin{pmatrix}(r_\infty-3)\text{Ham}^{(\boldsymbol{\alpha}_{\infty,r_\infty-3})}\\ \vdots \\\text{Ham}^{(\boldsymbol{\alpha}_{\infty,1})}\end{pmatrix}=S^{(r_\infty-3)} \begin{pmatrix}H_{\infty,0}\\ \vdots \\ H_{\infty,r_\infty-4}\end{pmatrix} \eeq
i.e. using Lemma \ref{LemmaToeplitz}:
\beq \td{M}_\infty S^{(r_\infty-3)}\begin{pmatrix}(r_\infty-3)\text{Ham}^{(\boldsymbol{\alpha}_{\infty,r_\infty-3})}\\\vdots \\\text{Ham}^{(\boldsymbol{\alpha}_{\infty,1})}\end{pmatrix}=S^{(r_\infty-3)} \begin{pmatrix}H_{\infty,0}\\ \vdots \\ H_{\infty,r_\infty-4}\end{pmatrix} \eeq
which is equivalent to \eqref{NewHamReduced}.

\medskip

Similar computations can be carried out for deformations with respect to $\tau_{X_s,j}$ for $(s,j)\in \llbracket 1,n\rrbracket\times\llbracket 1,r_s-1\rrbracket$, excluding the fixed leading finite time when $r_\infty=1,n=1$. In this case, the only non-vanishing components are $\alpha_{X_s^{(1)},j}=\frac12$ and $\alpha_{X_s^{(2)},j}=-\frac12$. Consequently, $\nu^{(\boldsymbol{\alpha}_{X_s,j})}_{{X_u},k}=0$ for all $u\neq s$ and $\nu^{(\boldsymbol{\alpha}_{X_s,j})}_{\infty,k}=0$ for all $k\in \llbracket 1, r_\infty-3\rrbracket$ and 
\beq \label{nuXsReduced}M_s\left(  \nu^{(\boldsymbol{\alpha}_{X_s,j})}_{{X_s},1} \dots, \nu^{(\boldsymbol{\alpha}_{X_s,j})}_{{X_s},r_s-1}\right)^t=-\frac{1}{j}\mathbf{e}_{r_s-j}.\eeq
The corresponding Hamiltonian is given by:
\beq \text{Ham}^{(\boldsymbol{\alpha}_{X_s,j})}(\mathbf{\check{q}},\mathbf{\check{p}})=-\sum_{k=2}^{r_s}\nu_{X_s,k-1}^{(\boldsymbol{\alpha}_{X_s,j})}H_{X_s,k}:=-\left(\boldsymbol{\nu}_j^{(X_s)}\right)^t \begin{pmatrix}H_{X_s,2}\\\vdots\\H_{X_s,r_s}\end{pmatrix}.
\eeq
The rest of the computation is identical to the case at infinity presented above and provides \eqref{NewHamReduced}.

\medskip

Finally consider an independent position $\td{X}_s=X_s$.
The allowed indices are $s\ge1$ for $r_\infty\ge3$, $s\ge2$ for
$r_\infty=2$, and $s\ge3$ for $r_\infty=1,n\ge2$.
There is no independent position when $r_\infty=1,n=1$.
For such an index the only nonzero deformation component is
$\alpha_{X_s}=1$, and therefore
\beq \text{Ham}^{(\boldsymbol{\alpha}_{\td{X}_s})}(\mathbf{\check{q}},\mathbf{\check{p}})=H_{X_s,1}.\eeq
Together with the irregular-time formulas, this proves
\eqref{NewHamReduced} for every independent reduced time.

\newpage
\bibliographystyle{plain}
\bibliography{Biblio}

\end{document}